\documentclass[letter,11pt]{article}
\usepackage[utf8]{inputenc}
\usepackage{verbatim}
\usepackage{titling}
\usepackage{cancel}
\usepackage{enumitem} 
\setlist{noitemsep,topsep=0pt,parsep=0pt} 
\usepackage[margin=1cm]{caption} 
\usepackage[perpage]{footmisc}

\usepackage{subfig}
\usepackage{mathtools} 
\usepackage{multirow}
\usepackage{diagbox}
\usepackage[all]{xy}
\usepackage{tikz}
\usetikzlibrary{arrows,backgrounds,calc,fit,decorations.pathreplacing,decorations.markings,shapes.geometric}

\tikzset{every fit/.append style=text badly centered}

\usepackage{tyson}
\numberwithin{equation}{section}
\usepackage[bookmarks=true,hypertexnames=false]{hyperref}
\hypersetup{colorlinks=true, citecolor=blue, linkcolor=red, urlcolor=blue}
\makeatletter

\newcommand{\Rmnum}[1]{\expandafter\@slowromancap\romannumeral #1@}
\makeatother
\newcommand{\ii}{\mathfrak{i}}

\newcommand{\Holant}{\operatorname{Holant}}

\newcommand{\Holantb}{\operatorname{Holant}^b}

\newcommand{\holant}[2]{\ensuremath{\Holant(#1\mid #2)}}

\newcommand{\CSP}{\operatorname{\#CSP}}

\newenvironment{remark}{\medskip{\bfseries \noindent Remark:}}{\par\medskip}{\par\medskip}

\tikzstyle{internal} = [draw, fill, shape=circle]
\tikzstyle{external} = [shape=circle]
\tikzstyle{square}   = [draw, fill, rectangle]
\tikzstyle{triangle} = [draw, fill, regular polygon, regular polygon sides=3, inner sep=3pt]
\tikzstyle{pentagon} = [draw, fill, regular polygon, regular polygon sides=5, inner sep=2pt, minimum size=14pt]

\begin{document}

\title{{\bf A Dichotomy for Real Boolean Holant Problems}}

\vspace{0.3in}

\author{Shuai Shao\thanks{Department of Computer Sciences, University of Wisconsin-Madison. Supported by NSF CCF-1714275.
 } \\ {\tt sh@cs.wisc.edu}
\and
Jin-Yi Cai$^\ast$ \\ \tt jyc@cs.wisc.edu}

\date{}
\maketitle

\thispagestyle{empty}
\bibliographystyle{plain}

\begin{abstract}
We prove a complexity dichotomy for Holant problems on the boolean domain with arbitrary sets of real-valued constraint functions.
These constraint functions need not be symmetric nor do we assume any auxiliary functions as in previous results.
It is proved that 
for every set $\mathcal{F}$ of real-valued constraint functions, $\Holant(\mathcal{F})$ is either P-time computable or \#P-hard.
The classification has an explicit criterion.
This is the culmination of much research on this problem, and it uses previous results and techniques from many researchers.
Some particularly intriguing concrete functions $f_6$, $f_8$ and their associated families with  extraordinary closure properties related to  Bell states in quantum information theory play an important role in this proof.

 \end{abstract}

\newpage

\pagenumbering{roman}

\section*{{Prologue}}
\addcontentsline{toc}{section}{\protect\numberline{}{Prologue}}
\vspace{.1in}

{\small
\sl
The young knight Siegfried and his tutor Jeyoda set out for their life journey
together. Their aim is to pacify the real land of Holandica,
to bring order and unity.

\vspace{.05in}

In the past decade, brave knights have battled innumerable demons and 
creatures, and have conquered the more familiar portion of Holandica, 
called Holandica Symmetrica. Along the way they have also been victorious
by channeling the power of various deities known as Unary Oracles.
In the past few years this brotherhood of the intrepid have also gained great 
power from the beneficent god Orieneuler and enhanced their skills
in a more protected world ruled by Count Seaspie.

\vspace{.05in}

``But prepared we must be,'' Jeyoda reminds Siegfried, ``arduous,
our journey will become.'' As the real land of Holandica is teeming with unknowns,
who knows what wild beasts and creatures they may encounter. Siegfried nods,
but in his heart  he is confident that his valor and power will be equal
to the challenge.

\vspace{.05in}

They have recently discovered a treasure sword. This is their second gift
from the Cathedral Orthogonia, more splendid and more powerful than the 
first. In their initial encounters with the minion creatures in their
journey, the second sword from Cathedral Orthogonia proved to be invincible. 

\vspace{.05in}

These initial victories laid the foundation for their journey, but also
a cautious optimism sets in. Perhaps with their new powerful sword in hand, 
final victory will not be that far away.

\vspace{.05in}

Just as they savor these initial victories, things start to change.
As they enter the Kingdom of Degree-Six everything becomes strange.
Subliminally they feel the presence of a cunning enemy hiding in the darkness.
Gradually they are convinced that this enemy possesses a special power
that eludes the ongoing campaign, and in particular their magic sword.
After a series of difficult and protracted battles with many twists and turns, their nemesis,
the Lord of Intransigence slowly reveals his face. The Lord of Intransigence
has a suit of magic armor, called the Bell Spell, that hides and protects 
him so well that the sword of Cathedral Orthogonia cannot touch him.
 
\vspace{.05in}

Siegfried and Jeyoda know that in order to conquer the Lord of Intransigence,
they need all the skills and wisdom they have. Although the Lord of Intransigence
has a strong armor, he has a weakness. The armor is maintained by
four little elfs called the Bell Binaries. The next battle is tough and long.
Siegfried and Jeyoda hit upon the idea of convincing the Bell Binaries
to stage a mutiny. With his four little elfs turning against him, his armor
loses its magic, and the Kingdom of Six-degree is conquered. In the aftermath
of this victory, Siegfried and Jeyoda also collect some valuable treasures
that will come in handy in their next campaign.
 
\vspace{.05in}

After defeating the Lord of Intransigence, Siegfried and Jeyoda enter
the Land of Degree-Eight. Now they are very careful. After meticulous
reconnaissance, they finally identify the most fearsome enemy,
the Queen of the Night. Taking a page from their battle with the
Lord of Intransigence they look for opportunities to gain help from
within the enemy camp. However, the Queen of the Night has the strongest
protective coat called the Strong Bell Spell. This time there is no way
to summon help from within the Queen's own camp. In fact, her protective 
armor is so strong that any encounter with Siegfried and Jeyoda's sword
makes her magically disappear in a puff of white smoke.

\vspace{.05in}

But, everyone has a weakness. For the Queen, her vanishing act also
brings the downfall. After plotting the strategy for a long time, 
Siegfried and Jeyoda use a magical potion to create from nothing
the helpers needed to defeat the Queen.

\vspace{.05in}

Buoyed by their victory, they summon their last strength to secure
the Land of Degree-Eight and beyond. Finally they bring complete order
to the entire real land of Holandica. At their celebratory banquet,
they want to share the laurels with   Knight  Ming and Knight Fu who provided
invaluable assistance in their journey; but the two brave and generous knights have retreated to their
Philosopher's Temple and are nowhere to be found.}

{\hypersetup{linkcolor=black}
  \tableofcontents}

\newpage
\pagenumbering{arabic}
\setcounter{page}{1}
\section{Introduction}
Counting problems arise in many different fields, e.g., statistical physics, economics and machine learning. 
In order to study the complexity of counting problems, several natural frameworks have been proposed. 
Two well studied frameworks are counting constraint satisfaction problems (\#CSP) \cite{bulatov-ccsp, dyer-richerby, bulatov2012csp, cai-chen-lu-nonnegative-csp, cai-chen-csp} and counting graph homomorphisms (\#GH)  \cite{dyer-green-gh, bulatov2005gh, goldberg2010gh, cai-chen-lu-gh} which is a special case of \#CSP.
These frameworks are expressive enough so that they can express many natural counting problems but also specific enough so that complete complexity classifications can be established. 

Holant problems are a more expressive framework which generalizes \#CSP and \#GH.  It is a broad class of sum-of-products computation. 
Unlike \#CSP and \#GH for which full complexity dichotomies have been established,
the understanding of Holant problems, even  restricted to
the Boolean domain, is still limited.
In this paper, we establish the first Holant dichotomy
on the Boolean domain with arbitrary real-valued constraint functions.
These constraint functions need not be symmetric nor do we assume any
auxiliary functions (as in previous results).

A Holant problem on the Boolean domain is parameterized by a set of constraint functions, also called signatures; such a signature
maps  $\{0, 1\}^n \rightarrow \mathbb{C}$
for some $n > 0$. 
Let $\mathcal{F}$ be any fixed set of signatures. 
A signature grid
$\Omega=(G, \pi)$ over $\mathcal{F}$
 is a tuple, where $G = (V,E)$
is a graph without isolated vertices,
 $\pi$ labels each $v\in V$ with a signature
$f_v\in\mathcal{F}$ of arity ${\operatorname{deg}(v)}$,
and labels the incident edges
$E(v)$ at $v$ with input variables of $f_v$.
We consider all 0-1 edge assignments $\sigma$, and
each gives an evaluation
$\prod_{v\in V}f_v(\sigma|_{E(v)})$, where $\sigma|_{E(v)}$
denotes the restriction of $\sigma$ to $E(v)$.

\begin{definition}[Holant problems]
The input to the problem {\rm Holant}$(\mathcal{F})$
is a signature grid $\Omega=(G, \pi)$ over $\mathcal{F}$.
The output is  the partition function
\begin{equation*}
\Holant({\Omega}) =\sum\limits_{\sigma:E(G)\rightarrow\{0, 1\}}\prod_{v\in V(G)}f_v(\sigma|_{E_{(v)}}).    
\end{equation*}
Bipartite Holant problems $\holant{\mathcal{F}}{\mathcal{G}}$
are {Holant} problems
 over bipartite graphs $H = (U,V,E)$,
where each vertex in $U$ or $V$ is labeled by a signature in $\mathcal{F}$ or $\mathcal{G}$ respectively.
When $\{f\}$ is a singleton set, we write $\Holant(\{f\})$ as $\Holant(f)$ and  $\Holant(\{f\}\cup \mathcal{F})$ as $\Holant(f, \mathcal{F}).$ 
\end{definition}

Weighted \#CSP  is a special class of  Holant problems.
So are all weighted \#GH.
Other problems  expressible as Holant problems include counting matchings and perfect matchings~\cite{valiant1979complexity}, counting weighted Eulerian orientations (\#EO problems)~\cite{mihail1996number, cai-fu-shao-eo}, computing the partition functions of six-vertex models~\cite{pauling1935structure, cai-fu-xia-six} and eight-vertex models~\cite{baxter1971eight, cai-fu-eight}, and a host of other, if not almost all, vertex models from statistical physics~\cite{baxter2004six}. 
It is proved that counting perfect matchings cannot be expressed by \#GH \cite{Freedman-et-al,cai-artem}. 
Thus, Holant problems are provably more expressive. 

Progress has been made in  the complexity classification of Holant problems.
When all signatures are restricted to be \emph{symmetric}, a full dichotomy is proved \cite{cai-guo-tyson-vanishing}. 
When asymmetric signatures are allowed, some dichotomies are proved for special families of Holant problems 
by assuming that certain auxiliary  signatures are available, e.g.,  $\Holant^\ast$, $\Holant^+$ and $\Holant^c$~\cite{cai-lu-xia-holant*, Backens-holant-plus, CLX-HOLANTC, Backens-Holant-c}.
Without assuming auxiliary signatures a Holant dichotomy is established  for non-negative real-valued signatures~\cite{wang-lin}, and for all real-valued signatures where a signature of odd arity is present \cite{realodd}.
In this paper, we prove a full complexity dichotomy for Holant problems with real values. 
\begin{theorem}\label{main-theorem}
Let $\mathcal{F}$ be a set of real-valued signatures.
If $\mathcal{F}$ satisfies the tractability condition {\rm(\ref{main-thr})} in Theorem \ref{thm-main-thr},
then $\Holant(\mathcal{F})$ is polynomial-time computable;
otherwise, $\Holant(\mathcal{F})$ is \#P-hard.
\end{theorem}

This theorem is the culmination of a large part of previous research on dichotomy
theorems on Holant problems, and it uses much of the previously established results and techniques.
However, as it turned out, the journey to this theorem has been arduous.
The overall plan of the proof is by induction on arities of signatures 
in $\mathcal{F}$.
Since a dichotomy is proved when $\mathcal{F}$ contains a signature of odd arity, we only need to consider signatures of even arity.
For signatures of small arity $2$ or $4$ (base cases) and large arity at least 10,
we given an induction proof based on results of \#CSP, \#EO problems and eight-vertex models.
However, 
two  signatures $f_6$ and $f_8$ (and their associated
families)
of  arity $6$ and $8$, are discovered 
which have  extraordinary closure properties; we call
them Bell properties \cite{realodd}.
These amazing signatures are wholly unexpected, and their existence
presented a formidable obstacle to the induction proof.
All four binary Bell  signatures (related to Bell states \cite{bell1964einstein} in quantum information theory) are realizable from $f_6$ by gadget construction.  We introduce
$\Holantb$ problems where the four binary Bell  signatures are available.
This is specifically to  handle the signature $f_6$. 
We prove a \#P-hardness result for $\Holantb(f_6, \mathcal{F})$. 
In this proof, we find other miraculous  signatures with special structures such that all signatures realized from them by merging gadgets are affine signatures, while themselves are not affine signatures.
In order to handle the signature $f_8$,
we introduce Holant problems with limited appearance, 
where some signatures are only allowed to appear a limited number of times in all instances. 
We turn the obstacle of the closure property of $f_8$
in our favor to 
prove  non-constructively a P-time
reduction from $\Holantb(f_8, \mathcal{F})$ to $\Holant(f_8, \mathcal{F})$.
In fact, it is provable that except $=_2$, the other three binary Bell  signatures are \emph{not} realizable from $f_8$ by gadget construction. 
However, 
we 
show that we can realize, in the sense of a non-constructive complexity reduction,
the desired binary Bell signatures which appear an unlimited number of
times.
This utilizes the framework where these signatures occur  only a limited number of
times.
Then, we give a \#P-hardness result for $\Holantb(f_8, \mathcal{F})$  similar to $\Holantb(f_6, \mathcal{F})$. 


\section{Preliminaries}
\subsection{Definitions and notations}
Let $f$ be a complex-valued signature.
If $\overline{f({\alpha})}=f({\overline{\alpha}})$ for all $\alpha$ where $\overline{f({\alpha})}$ denotes the complex conjugation of $f(\alpha)$ and $\overline{\alpha}$ denotes the bit-wise complement of $\alpha$, we say $f$ satisfies  \emph{Arrow Reversal Symmetry} ({\sc ars}).
We may also use $f^\alpha$ to denote $f(\alpha)$.
We use ${\rm wt}({\alpha})$ to denote the Hamming weight of $\alpha$.
The support $\mathscr{S}(f)$ of a signature $f$ is 
$\{\alpha \in \mathbb{Z}_2^{n} \mid f(\alpha) \neq 0\}$.
We say $f$ has support of size $k$ if $|\mathscr{S}(f)|=k$.
 If $\mathscr{S}(f)=\emptyset$, i.e., $f$ is identically $0$, we say $f$ is a zero signature and denote it by $f\equiv 0$.
Otherwise, $f$ is a nonzero signature.
Let $\mathscr{E}_n=\{\alpha \in \mathbb{Z}^n_2 \mid {\rm wt}(\alpha)$ is even$\}$, and $\mathscr{O}_n=\{\alpha \in \mathbb{Z}^n_2 \mid {\rm wt}(\alpha)$ is odd$\}$.
A signature $f$ of arity $n$ has even or odd parity if $\mathscr{S}(f)\subseteq\mathscr{E}_n$ or $\mathscr{S}(f)\subseteq\mathscr{O}_n$ respectively. In both cases, we say that $f$ has parity. 
Let $\mathscr{H}_{2n}=\{\alpha \in \mathbb{Z}_2^{2n} \mid
{\rm wt}(\alpha)= n\}$.
A signature $f$ of arity $2n$ has half-weighted support if $\mathscr{S}(f) \subseteq \mathscr{H}_{2n}$.
We call such a signature an \emph{Eularian orientaion} (EO) signature.
For $\alpha\in \mathbb{Z}^n_2$ and $1\leqslant i \leqslant n$, we use $\alpha_i$ to denote the value of $\alpha$ on bit $i$.

Counting  constraint satisfaction problems (\#CSP) can be expressed as Holant problems.
We use $=_n$ to denote the {\sc Equality} signature of arity $n$,
which takes value $1$ on the all-0 and all-1 inputs  
and $0$ elsewhere. (We denote the $n$-bits all-0 and all-1 strings by $\Vec{0}^n$ and $\Vec{1}^n$ respectively.  We may omit the superscript $n$ when it is clear from the context.)
Let $\mathcal{EQ}=\{=_1, =_2, \ldots, =_n, \ldots\}$ denote the set of all {\sc Equality} signatures.

\begin{lemma}[\cite{jcbook}]
$\CSP(\mathcal{F})\equiv_T \Holant(\mathcal{EQ} \mid \mathcal{F})$.
\end{lemma}

We use $\neq_2$ to denote the binary {\sc Disequality} signature with truth table $(0, 1, 1, 0)$.
We generalize this notion to signatures of 
higher arities. 
A signature $f$ of arity $2n$ is called a  {\sc Disequality} signature of arity $2n$, denoted by $\neq_{2n}$, if $f =1$ when $(x_1 \neq x_2) \wedge \ldots \wedge (x_{2n-1} \neq x_{2n})$,
and 0 otherwise. By permuting its variables
the  {\sc Disequality} signature  of arity $2n$ also defines
$(2n-1)(2n-3)\cdots1$ functions which we also call {\sc Disequality} signatures.
These signatures are equivalent for the complexity of Holant problems; once we have one
we have them all.
Let $\mathcal{DEQ}=\{\neq_2, \neq_4, \ldots, \neq_{2n}, \ldots \}$ denote the set of all {\sc Disequality} signatures.

We use $=_2^-$ to denote the binary signature $(1, 0, 0, -1)$ and $\neq_2^-$ to denote the binary signature $(0, 1, -1, 0)$.
We may also write $=_2$ as $=_2^+$ and $\neq_2$ as $\neq_2^+$. 
Let $\mathcal{B}=\{=^+_2, =_2^-, \neq_2^+, \neq_2^-\}$.
We call them Bell signatures which correspond to Bell states $|\Phi^+\rangle=|00\rangle+|11\rangle$, $|\Phi^-\rangle=|00\rangle-|11\rangle$, 
$|\Psi^+\rangle=|01\rangle+|10\rangle$ and 
$|\Psi^-\rangle=|01\rangle-|10\rangle$ in quantum information science \cite{bell1964einstein}.

A signature $f$ of arity $n \geqslant 2$ can be expressed as 
a $2^k \times 2^{n-k}$ matrix $M_{S_k}(f)$ where $S_k$ is a set of $k$ many variables among all $n$ variables of $f$.
 The matrix $M_{S_k}(f)$ lists all $2^n$ many entries of $f$ with the assignments of variables in $S_k$\footnote{ Given a set of variables, without other specification, we always list them in the cardinal order i.e., from variables with the smallest index to the largest index.} listed in lexicographic order (from $\Vec{0}^{k}$ to $\Vec{1}^{k}$) as row index and the assignments of the other $n-k$ many variables  in lexicographic order as column index.
In particular, $f$ can be expressed as a $2\times 2^{n-1}$ matrix $M_{i}(f)$
which lists the $2^n$ entries of $f$ with the assignments of variable $x_i$ as row index (from $x_i=0$ to $x_i=1$) and the assignments of the other
 $n-1$ variables in lexicographic order as column index.
 Then, $$M_{i}(f)=\begin{bmatrix}
f^{0,00\ldots0} & f^{0, 00\ldots1} & \ldots & f^{0, 11\ldots1}\\
f^{1,00\ldots0} & f^{1, 00\ldots1} & \ldots & f^{1, 11\ldots1}\\
\end{bmatrix}=\begin{bmatrix}
{\bf {f}}^{0}_{i}\\
{\bf {f}}^{1}_{i}\\
\end{bmatrix},$$
where $ {\bf {f}}^{a}_{i}$ 
 denotes the row vector indexed by $x_i=a$ in $M_{i}(f)$.
 Similarly, $f$ can also be expressed as a $4\times 2^{n-2}$ matrix with the assignments of two variables $x_i$ and $x_j$ as row index. Then,
$$M_{i j}(f)=\begin{bmatrix}
f^{00,00\ldots0} & f^{00, 00\ldots1} & \ldots & f^{00, 11\ldots1}\\
f^{01,00\ldots0} & f^{01, 00\ldots1} & \ldots & f^{01, 11\ldots1}\\
f^{10,00\ldots0} & f^{10, 00\ldots1} & \ldots & f^{10, 11\ldots1}\\
f^{11,00\ldots0} & f^{11, 00\ldots1} & \ldots & f^{11, 11\ldots1}\\
\end{bmatrix}=\begin{bmatrix}
{\bf {f}}^{00}_{ij}\\
{\bf {f}}^{01}_{ij}\\
{\bf {f}}^{10}_{ij}\\
{\bf {f}}^{11}_{ij}\\
\end{bmatrix},$$
where $ {\bf {f}}^{ab}_{ij}$ 
denotes the row vector indexed by $(x_i, x_j)=(a, b)$ in $M_{i j}(f)$.
For $=_2$, it has the 2-by-2 signature matrix $M(=_2)=I_2=\left[\begin{smallmatrix}
1 & 0\\
0 & 1\\
\end{smallmatrix}\right]$.
For $\neq_2$, $M(\neq_2)=N_2=\left[\begin{smallmatrix}
0 & 1\\
1 & 0\\
\end{smallmatrix}\right].$

\subsection{Holographic transformation}
To introduce the idea of holographic transformation,
it is convenient to consider bipartite graphs.
For a general graph,
we can always transform it into a bipartite graph while preserving the Holant value,
as follows.
For each edge in the graph,
we replace it by a path of length two.
(This operation is called the \emph{2-stretch} of the graph and yields the edge-vertex incidence graph.)
Each new vertex is assigned the binary \textsc{Equality} signature $=_2$. Thus, we have $\holant{=_2}{\mathcal{F}}\equiv_T \Holant(\mathcal{F})$.

For an invertible $2$-by-$2$ matrix $T \in {\bf GL}_2({\mathbb{C}})$
 and a signature $f$ of arity $n$, written as
a column vector (covariant tensor) $f \in \mathbb{C}^{2^n}$, we denote by
$Tf = T^{\otimes n} f$ the transformed signature.
  For a signature set $\mathcal{F}$,
define $T\mathcal{F} = \{T f \mid  f \in \mathcal{F}\}$ the set of
transformed signatures.
For signatures written as
 row vectors (contravariant tensors) we define
$f T^{-1}$ and  $\mathcal{F} T^{-1}$ similarly.
Whenever we write $T f$ or $T \mathcal{F}$,
we view the signatures as column vectors;
similarly for $f T^{-1}$ or $\mathcal{F} T^{-1}$ as row vectors.
We can also represent $Tf$ as the matrix $M_{S_k}(Tf)$ with the assignments of variables in $S_k$ as row index and the assignments of the other $n-k$ variables as column index. 
Then,  we  have $M_{S_k}(Tf)=T^{\otimes k}M_{S_k}(f)(T^{\tt T})^{\otimes n-k}$. Similarly, $M_{S_k}(fT^{-1})=({T^{-1}}^{\tt T})^{\otimes k}M_{S_k}(f)(T^{-1})^{\otimes n-k}$.

Let $T \in {\bf GL}_2({\mathbb{C}})$.
The holographic transformation defined by $T$ is the following operation:
given a signature grid $\Omega = (H, \pi)$ of $\holant{\mathcal{F}}{\mathcal{G}}$,
for the same bipartite graph $H$,
we get a new signature grid $\Omega' = (H, \pi')$ of $\holant{\mathcal{F} T^{-1}}{T \mathcal{G}}$ by replacing each signature in
$\mathcal{F}$ or $\mathcal{G}$ with the corresponding signature in $\mathcal{F} T^{-1}$ or $T \mathcal{G}$.

\begin{theorem}[\cite{valiant-holo-08}]
 For every $T \in {\bf GL}_2({\mathbb{C}})$,
  $\Holant(\mathcal{F} \mid \mathcal{G}) \equiv_T \Holant(\mathcal{F} T^{-1} \mid T \mathcal{G}).$
\end{theorem}

Therefore,
a holographic transformation does not change the complexity of the Holant problem in the bipartite setting. 
Let  ${{\bf O}}_2(\mathbb{R})\subseteq \mathbb{R}^{2 \times 2}$ be the set of all 2-by-2 real orthogonal matrices.
 We denote ${{\bf O}}_2(\mathbb{R})$ by ${\bf O}_2$.
For all $Q\in {\bf O}_2$,
since $(=_2)Q^{-1}=(=_2)$,
  $\holant{=_2}{\mathcal{F}}\equiv_T \holant{=_2}{Q\mathcal{F}}.$

A particular holographic transformation that will be commonly  used in this paper is the transformation defined by 
$Z^{-1}=\frac{1}{\sqrt{2}}\left[\begin{smallmatrix} 
1 & - \ii \\
1 & \ii \\
\end{smallmatrix}
\right].$ 
Note that $(=_2)Z=(\neq_2)$. 
Thus, $\holant{=_2}{\mathcal{F}}\equiv_T \holant{\neq_2}{Z^{-1}{\mathcal{F}}}.$
We denote $Z^{-1}{\mathcal{F}}$ by $\widehat{{\mathcal{F}}}$ and $Z^{-1}f$ by $\widehat{f}$.
It is known that $f$ and $\widehat{f}$ have the following relation.
\begin{lemma}[\cite{cai-fu-shao-eo}]\label{real-ars}
A (complex-valued) signature $f$ is a real-valued signature iff $\widehat{f}$ satisfies {\sc ars}.
\end{lemma} 

We say a real-valued binary signature $f(x_1, x_2)$ is orthogonal if  $M_1(f)M_1^{\tt T}(f)=\lambda I_2$ for some real $\lambda> 0$.
Since $M_2(f)=M_1^{\tt T}(f)$, $M_1(f)M^{\tt T}_1(f)=\lambda I_2$ iff $M_2(f)M^{\tt T}_2(f)=\lambda I_2$.
The following fact is easy to check.
\begin{lemma}\label{q-parity}
A binary signature $f$ is orthogonal or a zero signature iff $\widehat{f}$ has parity and {\sc ars}. 
\end{lemma}
\begin{proof}
Consider $M_1(f)$ and $M_1(\widehat{f})=M_1(Z^{-1}{f})=Z^{-1}M_1(f)(Z^{-1})^{\tt T}$.
Then, $M_1(f)=\left[\begin{smallmatrix} 
a & b \\
-b & a \\
\end{smallmatrix}
\right]$ iff $M_1(\widehat{f})=\left[\begin{smallmatrix} 
0 & a+b\ii \\
a-b\ii & 0 \\
\end{smallmatrix}
\right],$
 and $M_1(f)=\left[\begin{smallmatrix} 
a & b \\
b & -a \\
\end{smallmatrix}
\right]$ iff $M_1(\widehat{f})=\left[\begin{smallmatrix} 
a-b\ii & 0 \\
0 & a+b\ii \\
\end{smallmatrix}
\right].$
Also, $f\equiv 0$ iff $\widehat{f}\equiv 0$ which also has parity.
\end{proof}

Let $\mathcal{O}$ denote the set of all binary orthogonal signatures and the binary zero signature. 
Then, $\widehat{\mathcal{O}}=Z^{-1}\mathcal{O}$ is the set of all binary signatures with {\sc ars} and parity (including the binary zero signature).
Note that $\mathcal{B}\subseteq \mathcal{O}$ and 
$\widehat{\mathcal{B}}\subseteq \widehat{\mathcal{O}}$.
Here the transformed set $$\widehat{\mathcal{B}}=\left\{\widehat{=^+_2}, \widehat{=_2^-}, \widehat{\neq_2^+},  \widehat{\neq_2^-}\right\}=\{\neq_2, =_2, (-\ii) \cdot =_2^-, \ii \cdot \neq_2^-\}.$$

For every $Q\in {\bf O}_2$, let $\widehat{Q}=Z^{-1}QZ$. 
Then, $\widehat{Q}\widehat{\mathcal{F}}=(Z^{-1}QZ)(Z^{-1}\mathcal{F})=Z^{-1}(Q\mathcal{F})=\widehat{Q\mathcal{F}}$.
Thus, $$\holant{\neq_2}{\widehat{\mathcal{F}}}\equiv_T\holant{=_2}{\mathcal{F}}\equiv_T\holant{=_2}{Q\mathcal{F}}\equiv_T\holant{\neq_2}{\widehat{Q}\widehat{\mathcal{F}}}.$$
Let $\widehat{{\bf O}_2}=\{\widehat{Q}=Z^{-1}QZ\mid Q\in {\bf O}_2\}.$
Then, $\widehat{{\bf O}_2}=\{\left[\begin{smallmatrix}
\alpha & 0\\
0 & \bar \alpha\\
\end{smallmatrix}\right], \left[\begin{smallmatrix}
0 & \alpha\\
\bar \alpha & 0\\
\end{smallmatrix}\right]\mid \alpha\in \mathbb{C}, |\alpha|=1\}.$
 Note that the  notation $\widehat{\cdot}$ on a matrix $Q\in{\bf O}_2$ is \emph{not} the same  as the notation $\widehat{\cdot}$ on a signature $f\in \mathcal{O}$.
 Suppose that $M_1(f)=Q\in{\bf O}_2$. Since $Q=Z\widehat{Q}Z^{-1}$ and $Z^{-1}(Z^{-1})^{\tt T}=N_2$, $$M_1(\widehat{f})=Z^{-1}Q(Z^{-1})^{\tt T}=Z^{-1}(Z\widehat{Q}Z^{-1})(Z^{-1})^{\tt T}=\widehat{Q}N_2\neq \widehat{Q}.$$

\subsection{Signature factorization}

Recall that by our definition, every (complex valued) signature  has arity at least one.
A nonzero signature $g$ \emph{divides} $f$ denoted by $g \mid f$, if there is a signature $h$ such that  $f=g \otimes h$
(with possibly a permutation of variables) or there is a constant $\lambda$ such that $f= \lambda \cdot g$.
In the latter case, if $\lambda \neq 0$, then we also have $f \mid g$ since $g= \frac{1}{\lambda} \cdot f$.
For nonzero signatures, if both $g\mid f$ and $f \mid g$, then they are nonzero constant multiples of
each other, and 
we say $g$ is an \emph{associate} of $f$, 
denoted by $g \sim f$.
In terms of  this division relation,
the notions of \emph{irreducible}  signatures  and \emph{prime} signatures have been defined.
They are proved equivalent and thus, the \emph{unique prime factorization} (UPF) of signatures is established \cite{cai-fu-shao-eo}.

A nonzero signature $f$ is irreducible  if there are no signatures $g$ and $h$ such that $f=g\otimes h$. 
A nonzero signature $f$ is a prime signature
if $f \mid g\otimes h$
implies that $f \mid g$ or $f \mid h$.  These notions
are equivalent.
We say a signature $f$ is reducible  if $f = g \otimes h$,
for some signatures $g$ and $h$. All zero signatures 
of arity greater than 1 are reducible.
A prime factorization of a signature $f$ is $f=g_1\otimes \ldots \otimes g_k$ up to a permutation of
variables, where each $g_i$ is irreducible. 

\begin{lemma}[Unique prime factorization~\cite{cai-fu-shao-eo}]\label{unique}
Every nonzero signature $f$ has a prime factorization.
If  $f$ has  prime factorizations
 $f=g_1\otimes \ldots \otimes g_k$ and $f=h_1\otimes \ldots \otimes h_\ell$,
both up to a permutation of variables,
then $k=\ell$ and after reordering the factors we have $g_i \sim h_i$  for all $i$. 
\end{lemma}

\begin{lemma}[\cite{cai-fu-shao-eo}]
 let $f$ be a real-valued reducible signature, then there exists a factorization $f=g\otimes h$ such that $g$ and $h$ are both real-valued signatures.
 
Equivalently, let $\widehat f$ be a reducible signature satisfying {\sc ars}, then there exists a factorization $\widehat f=\widehat g\otimes \widehat h$ such that $\widehat g$ and $\widehat h$ both satisfy {\sc ars}.
\end{lemma}

In the following, when we say that a real-valued reducible signature $f$ has a factorization $g\otimes h$, we always assume that $g$ and $h$ are real-valued. 
Equivalently, when we say a signature $\widehat{f}$ satisfying {\sc ars} has a factorization $\widehat{g}\otimes \widehat{h}$, we always assume that $\widehat{g}$ and $\widehat{h}$ satisfy {\sc ars}.

For a signature set $\mathcal{F}$,  we use $\mathcal{F}^{\otimes k}$ $(k\geqslant 1)$ to denote the set $\{\lambda\bigotimes^k_{i=1}f_i\mid \lambda \in \mathbb{R}\backslash\{0\}, f_i\in \mathcal{F}\}$. 
Here,  $\lambda$ denotes a normalization scalar. 
In this paper, we only consider the normalization by nonzero real constants.
Note that $\mathcal{F}^{\otimes 1}$ contains all signatures obtained from $\mathcal{F}$ by normalization.
We use $\mathcal{F}^{\otimes}$ to denote $\bigcup^\infty_{k=1}\mathcal{F}^{\otimes k}.$

If a vertex $v$ in a signature grid 
is labeled by a reducible signature $f=g\otimes h$, we can replace the vertex $v$ 
by two vertices $v_1$ and $v_2$ and label $v_1$ with $g$ and $v_2$ with $h$, respectively.
The incident edges of $v$ become incident edges of $v_1$ and $v_2$ respectively
according to the partition of variables of $f$ in the tensor product of $g$ and $h$.  This does not
change the Holant value.
On the other hand,  Lin and Wang proved  that, from a real-valued reducible signature $f=g\otimes h\not\equiv 0$ we can freely replace $f$ by $g$ and $h$ while preserving the complexity of a  Holant problem.

\begin{lemma}[\cite{wang-lin}]\label{lin-wang}
If a nonzero real-valued signature $f$ has a real factorization $g\otimes h$, then $$\Holant(g, h, \mathcal{F})\equiv_T\Holant(f, \mathcal{F}) \text{ and } \holant{\neq_2}{\widehat{g}, \widehat{h}, \widehat{F}}\equiv_T \holant{\neq_2}{\widehat{f}, \widehat{F}}$$ for any signature set $\mathcal{F}$ $(\widehat{\mathcal{F}})$. 
We say $g$ $(\widehat{g})$ and $h$ $(\widehat{h})$ are realizable from $f$ $(\widehat{f})$ by factorization.
\end{lemma}

\subsection{Gadget construction}
One basic tool used throughout the paper is gadget construction.
An $\mathcal{F}$-gate is similar to a signature grid $(G, \pi)$ for $\Holant(\mathcal{F})$ except that $G = (V,E,D)$ is a graph with internal edges $E$ and dangling edges $D$.
The dangling edges $D$ define input variables for the $\mathcal{F}$-gate.
We denote the regular edges in $E$ by $1, 2, \dotsc, m$ and the dangling edges in $D$ by $m+1, \dotsc, m+n$.
Then the  $\mathcal{F}$-gate  defines a function $f$
\[
f(y_1, \dotsc, y_n) = \sum_{\sigma: E \rightarrow\{0, 1\}} \prod_{v\in V}f_v(\hat{\sigma}\mid_{E(v)})
\]
where $(y_1, \dotsc, y_n) \in \{0, 1\}^n$ is an assignment on the dangling edges, $\hat{\sigma}$ is the extension of  $\sigma$ on $E$ by the assignment $(y_1, \ldots, y_m)$, and $f_v$ is the signature assigned at each vertex $v \in V$.
This function $f$ is called the signature of the $\mathcal{F}$-gate.
There may be no internal edges in an $\mathcal{F}$-gate at all. In this case, $f$ is simply a tensor product of these signatures $f_v$, i.e., $f={\bigotimes}_{v\in V}f_v$ (with possibly a permutation of variables).
We say a signature $f$ is \emph{realizable} from a signature set $\mathcal{F}$ by gadget construction
if $f$ is the signature of an 
 $\mathcal{F}$-gate. 
If $f$ is realizable from a set $\mathcal{F}$,
then we can freely add $f$ into $\mathcal{F}$ while preserving the complexity (Lemma 1.3 in \cite{jcbook}). 
\begin{lemma}[\cite{jcbook}]
If $f$ is realizable from a set $\mathcal{F}$, then $\Holant(f, \mathcal{F})\equiv_T\Holant(\mathcal{F})$.
\end{lemma}
Note that,  if we view $\holant{=_2}{\mathcal{F}}$ as
the edge-vertex incidence graph form of $\Holant(\mathcal{F})$,
then it is equivalent
to label every edge by $=_2$;
similarly in the setting of  $\holant{\neq_2}{\widehat{\mathcal{F}}}$, every edge 
is labeled by $\neq_2$.

\begin{lemma}
If $f$ is realizable from a real-valued signature set $\mathcal{F}$ (in the setting of  $\holant{=_2}{\mathcal{F}}$), then $f$ is also real-valued.
Equivalently, if $\widehat {f}$ is realizable from a signature set $\widehat{\mathcal{F}}$ satisfying {\sc ars} (in the setting of  $\holant{\neq_2}{\widehat{\mathcal{F}}}$), then $\widehat f$ also satisfies {\sc ars}.
\end{lemma}

A basic gadget construction is \emph{merging}. In the setting of  $\holant{=_2}{\mathcal{F}}$,  given a signature $f\in \mathcal{F}$ of arity $n$, we can connect two variables $x_i$ and $x_j$ of $f$ using $=_2$, and this operation gives a signature of arity $n-2$. 
We use $\partial_{ij}f$ or $\partial^+_{ij}f$ to denote this signature and $\partial_{ij}f=f^{00}_{ij}+f^{11}_{ij}$, where $f^{ab}_{ij}$\footnote{We use
$f^{ab}_{ij}$ to denote a function, and ${\bf f}^{ab}_{ij}$ to denote a vector that lists the truth table of $f^{ab}_{ij}$ in a given order.} denotes the signature obtained by setting $(x_i, x_j)=(a, b)\in \{0, 1\}^2$. 
While in the setting of $\holant{\neq_2}{\widehat{\mathcal{F}}}$, the above merging gadget is equivalent to connecting two variables  $x_i$ and $x_j$ of $\widehat{f}$ using $\neq_2$. We denote the resulting signature by $\widehat{\partial}_{ij}\widehat{f}$ or $\widehat{\partial}^+_{ij}\widehat{f}$, and we have $\widehat{\partial_{ij}f}=\widehat{\partial}_{ij}\widehat{f}={\widehat{f}}^{01}_{ij}+{\widehat f}^{10}_{ij}.$
If $\neq_2$ is available (i.e., it either belongs to or can be realized from $\mathcal{F}$) in $\holant{=_2}{{\mathcal{F}}}$,
we can also connect two variables $x_i$ and $x_j$ of $f$ using $\neq_2$.
We denote the resulting signature by $\partial^{\widehat{+}}_{ij}f$.
The merging gadget ${\widehat\partial^{+}}_{ij}$ is the same as $\partial^{\widehat+}_{ij}$, we use different notations to distinguish whether this gadget is used in the setting of $\holant{=_2}{ \mathcal{F}}$ or $\holant{\neq}{\widehat{\mathcal{F}}}.$
Also, if $=_2^-$ and $\neq_2^-$ are available in $\holant{=_2}{{\mathcal{F}}}$, then we can  construct
$\partial^-_{ij}f$ and $\partial^{\widehat-}_{ij}f$ by connecting
$x_i$ and $x_j$  using  $=_2^-$ and $\neq_2^-$ respectively.
We also call $\partial^-_{ij}$ and $\partial^{\widehat-}_{ij}$ merging gadgets. 
Without other specification, by default a merging gadget refers to $\partial_{ij}$ in the setting of $\holant{=_2}{\mathcal F}$. Similarly by  default a merging gadget  refers to $\widehat{\partial}_{ij}$ in the setting of $\holant{\neq_2}{\widehat{\mathcal{F}}}.$



The following  lemma gives a relation between a signature $\widehat{f}$ and signatures $\widehat{\partial}_{ij}\widehat f$. 

\begin{lemma}[\cite{realodd}]\label{lem-zero_2}
Let $\widehat{f}$ be a signature of arity $n\geqslant 3$.
If $\widehat{f}(\alpha) \neq 0$ for some ${\rm wt}(\alpha)=k\neq 0$ and $k\neq n$, then there is a pair of indices $\{i, j\}$ such that $\widehat{\partial}_{ij}\widehat{f}({\beta})\neq 0$ for some ${\rm wt}(\beta)=k-1$. 
In particular, 
if for all pairs of indices $\{i, j\}$, 
$\widehat \partial_{ij} \widehat f\equiv0$, then $\widehat{ f}(\alpha)=0$ for all $\alpha$ with ${\rm wt}(\alpha)\neq 0$ and $n$. 
\end{lemma}

When $\widehat{f}$ is an EO signature satisfying {\sc ars}, the following relation between $\widehat{f}$ and  $\widehat{\partial}_{ij}\widehat f$ can be easily obtained following the proofs of Lemmas 4.3 and 4.5 in \cite{cai-fu-shao-eo}. 
Let $\mathcal{D}=\{\neq_2\}$.
Then $\mathcal{D}^{\otimes}=\{\lambda \cdot (\neq_2)^{\otimes k}\mid \lambda\in \mathbb{R}\backslash\{0\}, k \geqslant 1\}$ is the set of tensor products of $\neq_2$ up to nonzero real scalars. 
\begin{lemma}\label{lem-eo}
Let $\widehat{f}$ be a $2n$-ary \rm{EO} signature satisfying {\sc ars}. 
\begin{itemize}
    \item When $2n=8$, if for all pairs of indices $\{i, j\}$, $\widehat{\partial}_{ij}\widehat{f}\in {\mathcal{D}}^{\otimes}$, and there exists some $\neq_2(x_i, x_j)$ and two pairs of indices $\{u, v\}$ and $\{s, t\}$ where $\{u, v\}\cap\{s, t\}\neq \emptyset$ such that $\neq_2(x_i, x_j)\mid \widehat\partial_{uv}\widehat f, \widehat\partial_{st}\widehat f$, then $\widehat{f}\in {\mathcal{D}}^{\otimes}$ and $\neq_2(x_i, x_j)\mid \widehat f$.
   \item When $2n \geqslant 10$, if for all pairs of indices $\{i, j\}$, $\widehat{\partial}_{ij}\widehat{f}\in {\mathcal{D}}^{\otimes}$, then $\widehat{f}\in {\mathcal{D}}^{\otimes}$.
\end{itemize}

\end{lemma}

Another gadget construction that connects a nonzero binary signature $b$ with a signature $f$ is called \emph{extending}.
 An extending gadget connects one variable of $f$ with one variable of $b$ using $=_2$ in the setting of $\holant{=_2}{\mathcal{F}}$, 
 and connects one variable of $\widehat{f}$ with one variable of $\widehat{b}$ using $\neq_2$ in the setting of $\holant{\neq_2}{\widehat{\mathcal{F}}}$. 
 By extending an irreducible signature using $=_2$ or $\neq_2$, we still get an irreducible signature. 
  A particular extending gadget is to extend 
 $f$  with binary signatures in $\mathcal{B}^{\otimes 1}$  using $=_2$ in the setting of $\Holant(\mathcal{F}).$
 We use $\{f\}^{\mathcal{B}}_{=_2}$ to denote the set of signatures realizable by extending some variables of ${f}$ with binary signatures in ${\mathcal{B}}^{\otimes 1}$ using $=_2$  (recall that ${\mathcal{B}}^{\otimes 1}$ 
 allows all nonzero real normalization scalars).
 Equivalently, this gadget is to extend 
 $\widehat f$  with binary signatures in $\widehat{\mathcal{B}}$  using $\neq_2$ in the setting of $\holant{\neq_2}{\widehat{\mathcal{F}}}.$
 We use $\{\widehat{f}\}^{\widehat{\mathcal{B}}}_{\neq_2}$ to denote the set of signatures realizable by extending some variables of $\widehat{f}$ with binary signatures in $\widehat{\mathcal{B}}^{\otimes 1}$ using $\neq_2$.
 If $\widehat{g}\in \{\widehat{f}\}_{\neq_2}^{\widehat{\mathcal{B}}}$, then we can say that the extending gadget by $\widehat{\mathcal{B}}$ defines a relation between $\widehat{g}$ and $\widehat{f}$.
 Clearly, by extending variables of $\widehat{f}$ with $\neq_2 \in \widehat{\mathcal{B}}$ (using $\neq_2$), we still get $\widehat{f}$. 
 Thus, $\widehat{f} \in \{\widehat{f}\}_{\neq_2}^{\widehat{\mathcal{B}}}$.
 So this relation is reflexive.
 The following lemma shows that this relation is symmetric and transitive, thus it is an equivalence relation.
 \begin{lemma}\label{lem-extending}
1. $\widehat{g}\in \{\widehat{f}\}^{\widehat{\mathcal{B}}}_{\neq_2}$ iff $\widehat{f}\in \{\widehat{g}\}^{\widehat{\mathcal{B}}}_{\neq_2}.$
~~~2. If  $\widehat{h}\in \{\widehat{g}\}^{\widehat{\mathcal{B}}}_{\neq_2}$ and  $\widehat{g}\in \{\widehat{f}\}^{\widehat{\mathcal{B}}}_{\neq_2}$, then  $\widehat{h}\in \{\widehat{f}\}^{\widehat{\mathcal{B}}}_{\neq_2}.$
 \end{lemma}
 \begin{proof}
 Note that for any $\widehat{b}\in \widehat{\mathcal{B}}^{\otimes 1}$, 
if we connect any variable of $\widehat{b}$ with another arbitrary variable of a copy of the same $\widehat{b}$ using $\neq_2$, then we  get $\neq_2$ after normalization.
Also, by extending a variable of $\widehat{f}$ with $\neq_2$ (using $\neq_2$), we  still get $\widehat{f}$.
Suppose that $\widehat{g}\in \{\widehat{f}\}^{\widehat{\mathcal{B}}}_{\neq_2}$, and  it is realized  by extending certain variables $x_i$ of $\widehat{f}$ with certain $b_i\in \mathcal{\widehat{\mathcal{B}}}$.
Then, by extending each of these variables $x_i$ of $\widehat{g}$ with exactly the same $b_i\in \mathcal{\widehat{\mathcal{B}}}$, we will get $\widehat{f}$ after normalization. 
Thus, $\widehat{f}\in \{\widehat{g}\}^{\widehat{\mathcal{B}}}_{\neq_2}.$
The other direction is proved by  exchanging $\widehat{f}$ and $\widehat{g}$.
Thus, $\widehat{g}\in \{\widehat{f}\}^{\widehat{\mathcal{B}}}_{\neq_2}$ iff $\widehat{f}\in \{\widehat{g}\}^{\widehat{\mathcal{B}}}_{\neq_2}.$ 

Also, note that for any $\widehat{b^1}, \widehat{b^2}\in \widehat{\mathcal{B}}^{\otimes 1}$, by connecting an arbitrary variable of $\widehat{b^1}$ with  an arbitrary variable of $\widehat{b^2}$ using $\neq_2$, we  still get a signature in $\widehat{\mathcal{B}}^{\otimes 1}$.
Suppose that $\widehat{h}$ is realized by extending some variables $x_i$ of $\widehat{g}$ with some $b_i^1\in \widehat{\mathcal{B}}^{\otimes 1}$.
We may assume every variable $x_i$ of $\widehat{g}$
has been so connected as $\ne_2 \in \widehat{\mathcal{B}}^{\otimes 1}$. Similarly
we can assume
$\widehat{g}$ is realized by extending every variable $x_i$ of $\widehat{f}$ with some $b_i^2\in \widehat{\mathcal{B}}^{\otimes 1}.$
Let $b_i$ be the signature realized by connecting $b_i^1$ and $b^2_i$ (using $\neq_2$). 
Then, $\widehat{h}$ can be realized  by extending each variable $x_i$ of $\widehat{f}$ with $b_i\in \widehat{\mathcal{B}}^{\otimes 1}$.
Thus, $\widehat{h}\in \{\widehat{f}\}^{\widehat{\mathcal{B}}}_{\neq_2}$.
 \end{proof}
\begin{remark}
As a corollary, if  $\widehat{g}\in \{\widehat{f}\}^{\widehat{\mathcal{B}}}_{\neq_2}$, then $\{\widehat{g}\}^{\widehat{\mathcal{B}}}_{\neq_2}=\{\widehat{f}\}^{\widehat{\mathcal{B}}}_{\neq_2}$.
\end{remark}

\begin{lemma}\label{lem-binary-sim}
Let $\widehat{b_1}(x_1, x_2), \widehat{b_2}(y_1, y_2)\in \widehat{\mathcal{O}}$. 
If by connecting the variable $x_1$ of $\widehat{b_1}$ and the variable $y_1$ of $\widehat{b_2}$ using $\neq_2$, we get $\lambda\cdot \neq_2(x_2, y_2)$ for some $\lambda\in \mathbb{R}\backslash\{0\}$, then $\widehat{b_1}\sim \widehat{b_2}$. 
Moreover,  by connecting the variable $x_2$ of $\widehat{b_1}$ and the variable $y_2$ of $\widehat{b_2}$, we will get $\lambda\cdot \neq_2(x_1, y_1).$
\end{lemma}
\begin{proof}
We prove this lemma  
in the setting of 
$\Holant(\mathcal{F})$
after the transformation $Z$ back.
Now, $b_1=Z\widehat{b_1}\in \mathcal{O}$ and $b_2=Z\widehat{b_2}\in \mathcal{O}$.

Consider matrices $M_1(b_1)=M^{\tt T}_2(b_1)$ and $M_1(b_2)=M^{\tt T}_2(b_2)$.
Since $b_1, b_2\in \mathcal{O}$,  both $M_1(b_1)$ and $M_1(b_2)$
are  real multiples of real orthogonal matrices, of which
there are two types, either rotations or reflections.
For such matrices $X, Y$, 
to get $X^{\tt T} Y = \lambda I_2$  for some $\lambda \in
\mathbb{R}\backslash\{0\}$,
$X$ and $Y$ must be
either   both
reflections, or both rotations of the same angle, up to
nonzero real multiples.
First suppose
$M_1(b_1)= \left[\begin{smallmatrix} 
a & b \\
b & -a \\
\end{smallmatrix}
\right]$, reflection. Then
by connecting  $x_1$ of $b_1$ and  $y_1$ of $b_2$ using $=_2$ 
we get $\lambda \cdot =_2(x_2, y_2)$, i.e., 
$M^{\tt T}_1(b_1)M_1(b_2)=\lambda I_2.$
This implies that $b_2$ is the same reflection
up to a nonzero scalar, i.e., $b_2\sim b_1$. 
Similarly, for a rotation
$M_1(b_1) = \left[\begin{smallmatrix} 
a & b \\
-b & a \\
\end{smallmatrix}
\right]$, $M^{\tt T}_1(b_1)M_1(b_2)=\lambda I_2$
implies that $b_2$ is  also a rotation of the same angle as $b_1$
up to a nonzero scalar, thus $b_2\sim b_1$.
In either case,
 by connecting the variable $x_2$ of ${b_1}$ and the variable $y_2$ of ${b_2}$, we will get
 $$M^{\tt T}_2(b_1)M_2(b_2)=M_1(b_1)M^{\tt T}_1(b_2)
 =\lambda I_2.$$
This means that we get the signature $\lambda\cdot =_2(x_1, y_1).$
The statement of the lemma follows from this after a $Z^{-1}$ transformation.
\end{proof}




A gadget construction often used in this paper is \emph{mating}. 
Given a real-valued signature $f$ of arity $n\geqslant 2$, we connect two copies of $f$ in the following manner:
Fix a set $S$ of $n-m$ variables among all $n$ variables of $f$. For each $x_k\in S$, connect $x_k$ of one copy of $f$ with $x_k$ of the other copy using $=_2$. 
The variables 
that are not in $S$ are called dangling variables.
In this paper, we only consider the case that $m=1$ or $2$.
For $m=1$, there is one dangling variable $x_i$. Then, 
the mating construction  realizes a signature of arity $2$, denoted by $\mathfrak m_{i}f$.
It can be represented by matrix multiplication. 
We have 
\begin{equation}\label{m-form}
M(\mathfrak m_{i}f)=M_{i}(f)I_2^{\otimes (n-1)}M^{\tt T}_{i}(f)
=\begin{bmatrix}
{\bf {f}}^{0}_i\\
{\bf {f}}^{1}_i\\
\end{bmatrix}
\left[\begin{matrix}
{{\bf {f}}^{0}_i}^{\tt T} &{{\bf {f}}^{1}_i}^{\tt T}\\
\end{matrix}\right]
=\left[\begin{matrix}
|{\bf f}_i^0|^2 &  \langle {\bf f}_i^0, {\bf f}_i^1 \rangle\\
\langle {\bf f}_i^0, {\bf f}_i^1 \rangle & |{\bf f}_i^1|^2,\\
\end{matrix}\right]
\end{equation}
where $\langle \cdot, \cdot\rangle$ denotes the 
inner product and $|\cdot|$ denotes the  norm defined by this inner product.
(We will use the same notation  $\langle \cdot, \cdot\rangle$ to denote the complex inner product (with conjugation) below. The notation is consistent.)
Note that $|\langle{\bf f}^{0}_i, {\bf f}^{1}_i\rangle|^2\leqslant|{\bf f}^{0}_i|^2|{\bf f}^{1}_i|^2$ by the Cauchy-Schwarz inequality. 
Similarly, in the setting of $\holant{\neq_2}{\widehat{\mathcal{F}}}$, the above mating operation is equivalent to connecting variables in $S$ using $\neq_2$. We denote the resulting signature by  $\widehat{\mathfrak{m}}_i\widehat{f}$,
which is the same as $\widehat{{\mathfrak{m}}_i{f}}$, and we have
\begin{equation*}
M(\widehat{\mathfrak{m}}_i\widehat{f})=
M_{i}(\widehat f)N_2^{\otimes n-1}M^{\tt T}_{i}(\widehat f)=
\left[\begin{matrix}
\widehat{{\bf f}}_i^0\\
\widehat{{\bf f}}_i^1
\end{matrix}\right]
\left[\begin{matrix}
0 & 1 \\
 1 & 0
\end{matrix}\right]^{\otimes (n-1)}
\left[\begin{matrix}
{\widehat{{\bf f}_i^{0}}}^{\tt T} &{\widehat{{\bf f}_i^{1}}}^{\tt T}
\end{matrix}\right].\\
\end{equation*}
Note that (in general complex-valued) $\widehat{f}$ satisfies the {\sc ars} since $f$ is real, we have \begin{equation*}
    N_2^{\otimes (n-1)}{\widehat{{\bf f}_i^{0}}}^{\tt T}=(\widehat f^{0,11\ldots1}, \widehat f^{0, 11\ldots0}, \ldots, \widehat f^{0, 00\ldots0})^{\tt T}\\=(\overline{\widehat f^{1,00\ldots0}}, \overline{\widehat f^{1,00\ldots1}}, \ldots, \overline{\widehat f^{1, 11\ldots1}})={\overline{\widehat{{\bf f}}_i^{1}}}^{\tt T}.
\end{equation*}
Thus, we have
\begin{equation}\label{eqn:ars}
M(\widehat{\mathfrak{m}}_i\widehat{f})=
\left[\begin{matrix}
\widehat{{\bf f}}_i^0\\
\widehat{{\bf f}}_i^1
\end{matrix}\right]
\left[\begin{matrix}
0 & 1 \\
 1 & 0
\end{matrix}\right]^{\otimes (n-1)}
\left[\begin{matrix}
{\widehat{{\bf f}_i^{0}}}^{\tt T} &{\widehat{{\bf f}_i^{1}}}^{\tt T}
\end{matrix}\right]\\
=\left[\begin{matrix}
\widehat{{\bf f}}_i^0\\
\widehat{{\bf f}}_i^1
\end{matrix}\right]
\left[\begin{matrix}
{\overline{\widehat{{\bf f}}_i^{1}}}^{\tt T} &\overline{{\widehat{{\bf f}}_i^{0}}}^{\tt T}
\end{matrix}\right]
=\left[\begin{matrix}
\langle \widehat{{\bf f}}_i^0, \widehat{{\bf f}}_i^1 \rangle & |\widehat{{\bf f}}_i^0|^2\\
|\widehat{{\bf f}}_i^1|^2 & \langle \widehat{{\bf f}}_i^1, \widehat{{\bf f}}_i^0 \rangle
\end{matrix}\right].
\end{equation}
If there are two dangling variables $x_i$ and $x_j$, we use $\mathfrak{m}_{ij}f$ and $\widehat{\mathfrak{m}}_{ij}\widehat f$ to denote the  signatures realized by mating $f$ using $=_2$ and mating $\widehat f$ using $\neq_2$ respectively. 

With respect to mating gadgets, the following first order orthogonality was introduced. 
\begin{definition}[First order orthogonality  \cite{realodd}]\label{def-first-order}
       Let $f$ be a complex-valued signature of arity $n \geqslant 2$. It satisfies the \emph{first order orthogonality ({\sc 1st-Orth})} if there exists some $\mu\neq 0$ such that for all indices $i\in [n]$, the entries of $f$ satisfy the following equations
\begin{equation*}
    |{\bf f}_{i}^{0}|^2=|{\bf f}_{i}^{1}|^2=\mu, \text{ and } \langle{\bf f}_{i}^{0}, {\bf f}_{i}^{1}\rangle =0.
\end{equation*}
    \end{definition}
    
     \begin{remark}
      When $f$ is a real-valued signature,
    the  inner product is just the ordinary dot product which can be represented  by mating using $=_2$.
    Thus, $f$ satisfies {\sc 1st-Orth} iff there is some  real $\mu \neq 0$ such that for all indices $i$, $M(\mathfrak{m}_{i}f)= \mu I_2$.
     On the other hand,
    when $\widehat{f}$ is a signature with {\sc ars}, by (\ref{eqn:ars}), the complex inner product can be represented by mating  using $\neq_2$.
    Thus,
    $\widehat{f}$ satisfies {\sc 1st-Orth} iff there is some real $\mu \not =0$ such that for all $i$, $M(\widehat{\mathfrak{m}}_{i}\widehat{f})= \mu N_2$.
    Moreover, $f$ satisfies {\sc 1st-Orth} iff $\widehat{f}$ satisfies it.
\end{remark}

\begin{lemma}[\cite{realodd}]\label{lem-same-mu}
Let $f$ be a real-valued signature of arity $n$.
If for all indices $i\in [n]$, $M(\mathfrak{m}_{i}f)= \mu_i I_2$ for some real $\mu_i\neq 0$, then $f$ satisfies {\sc 1st-Orth} (i.e., all $\mu_i$ have the same value).
\end{lemma}

\subsection{Tractable signatures}
We give some known signature sets 
that
 define polynomial time computable (tractable) counting problems.

\begin{definition}
Let $\mathscr{T}$ denote the set of tensor products of unary and binary signatures. 
\end{definition}

\begin{definition}
\label{definition-product-2}
 A signature on a set of variables $X$
 is of \emph{product type} if it can be expressed as a
product of unary functions,
 binary equality functions $([1,0,1])$,
and binary disequality functions $([0,1,0])$, each on one or two
variables of $X$.
 We use $\mathscr{P}$ to denote the set of product-type functions.
\end{definition}
Note that the product in Definition~\ref{definition-product-2}
are ordinary products of functions (not tensor products); in particular
they may be applied on overlapping sets of variables.
\begin{definition}\label{definition-affine}
 A signature $f(x_1, \ldots, x_n)$ of arity $n$
is \emph{affine} if it has the form
 \[
  \lambda \cdot \chi_{A X = 0} \cdot {\mathfrak i} ^{Q(X)},
 \]
 where $\lambda \in \mathbb{C}$,
 $X = (x_1, x_2, \dotsc, x_n, 1)$,
 $A$ is a matrix over $\mathbb{Z}_2$,
 $Q(x_1, x_2, \ldots, x_n)\in \mathbb{Z}_4[x_1, x_2, \ldots, x_n]$
is a  multilinear polynomial with total degree $d(Q)\leqslant 2$ and
 the additional requirement that the coefficients of all
 cross terms are even, i.e., $Q$ has the form
 \[Q(x_1, x_2, \ldots, x_n)=a_0+\displaystyle\sum_{k=1}^na_kx_k+\displaystyle\sum_{1\leq i<j\leq n}2b_{ij}x_ix_j,\]
 and $\chi$ is a 0-1 indicator function
 such that $\chi_{AX = 0}$ is~$1$ iff $A X = 0$.
 We use $\mathscr{A}$ to denote the set of all affine signatures.
\end{definition}
If the support set $\mathscr{S}(f)$ is an affine linear subspace, then we say $f$ has affine support. Clearly, any affine signature has affine support.
Moreover,  we have that 
any signature of product type has affine support \cite{jcbook}.
When $\mathscr{S}(f)$ is affine, 
we can pick a set of free variables such that in  $\mathscr{S}(f)$, 
every variable is an affine linear combination of free variables.
Affine functions satisfy the following congruity or semi-congruity.
\begin{lemma}[\cite{jcbook}]\label{lem-congruity}
Let $f(x_1, \ldots, x_n)=(-1)^{Q(x_1, \ldots, x_n)}\in \mathscr{A}$,  and $y=x_n+L(x_1, \ldots, L_{n-1})$ be a linear combination  of variables $x_1, \ldots, x_{n}$ that involves $x_n$.
Define $$g(x_1, \ldots, x_{n-1})=\frac{f_{y=0}(x_1, \ldots, x_{n-1}, y+L)}{f_{y=1}(x_1, \ldots, x_{n-1}, y+L)}=(-1)^{Q(x_1, \ldots, x_{n-1}, L)+{Q(x_1, \ldots, x_{n-1}, L+1)}}.$$
Then, $g$ satisfies the following property.
\begin{itemize}
    \item (Congruity) $g\equiv1$ or $g\equiv-1$, or
    \item (Semi-congruity) $g(x_1, \ldots, x_{n-1})=(-1)^{L(x_1, \ldots, x_{n-1})}$ where $L(x_1, \ldots, x_{n-1})\in \mathbb{Z}_2[x_1, \ldots, x_{n-1}]$ is an affine linear polynomial (degree $d(L)= 1$). 
\end{itemize}{}
In particular, if $d(Q)=1$, then $g$ has congruity.
\end{lemma}{}

Let $T_{\alpha^s}=\left[\begin{smallmatrix}
1 & 0\\
0 & \alpha^s\\
\end{smallmatrix}\right]$ where $\alpha=\frac{1+\ii}{\sqrt{2}}$ and $s$ is an integer.

\begin{definition}
 A signature $f$ is local-affine if for each $\sigma=s_1s_2\ldots s_n \in \{0, 1\}^{n}$ in the support of $f$, $(T_{\alpha^{s_1}}\otimes T_{\alpha^{s_2}}\otimes\cdots \otimes T_{\alpha^{s_n}})f\in \mathscr{A}$. We use $\mathscr{L}$ to denote the set of local-affine signatures. 
\end{definition}

\begin{definition} \label{def:prelim:trans}
 We say a signature set $\mathcal{F}$ is $\mathscr{C}$-transformable
 if there exists a $T \in \rm{GL}_2(\mathbb{C})$ such that
 $(=_2)(T^{-1})^{\otimes 2} \in \mathscr{C}$ and $T\mathcal{F} \subseteq \mathscr{C}$. 
\end{definition}

This definition is important because if $\Holant(\mathscr{C})$ is tractable,
then $\Holant(\mathcal{F})$ is tractable for any $\mathscr{C}$-transformable set $\mathcal{F}$. Then, the following tractable result is known \cite{CLX-HOLANTC, Backens-Holant-c}.
\begin{theorem}\label{thm-main-thr}
Let $\mathcal{F}$ be a set of complex valued signatures.
Then $\Holant(\mathcal{F})$ is tractable if
    \begin{equation}\label{main-thr}
\mathcal{F}\subseteq \mathscr{T}, ~~~
 \mathcal{F} \text{ is }\mathscr{P}\text{-transformable,}~~~
 \mathcal{F} \text{ is } \mathscr{A}\text{-transformable,~~~or~} 
 \mathcal{F} \text{ is } \mathscr{L}\text{-transformable.}\tag{$\mathrm{ \textcolor{red}{T}}$}
 \end{equation}
\end{theorem}





\begin{lemma}[\cite{realodd}]\label{lem-hard-sign}
Let $\mathcal{F}$ be a set of real-valued signatures.
If $\mathcal{F}$ does not satisfy condition \rm{(\ref{main-thr})},  then for every $Q\in {\bf O}_2$,  $Q\mathcal{F}$ also does not satisfy condition \rm{(\ref{main-thr})}. Moreover,  $\widehat{\mathcal{F}}\not\subseteq \mathscr{P}$ and $\widehat{\mathcal{F}}\not\subseteq \mathscr{A}$.
\end{lemma}

\subsection{Hardness results and P-time reductions}
We give some known hardness results.  We state these results for our setting.




\begin{theorem}[\cite{CLX-HOLANTC, realodd}]\label{hard-result}
Let $\mathcal{F}$ be a set of real-valued signatures.
If $\mathcal{F}$ does not satisfy condition \rm{(\ref{main-thr})}. 
Then for every $Q\in {\bf O}_2$ and every $k\geqslant 2$, $\CSP_2(Q\mathcal{F})$ and $\holant{\neq_2}{=_k, \widehat{Q}\widehat{\mathcal F}}$ are \#P-hard.
\end{theorem}

\begin{theorem}[\cite{realodd}]\label{odd-dic}
Let $\mathcal{F}$ be a set of real-valued signatures containing a nonzero signature of odd arity.
If $\mathcal{F}$ does not satisfy condition \rm{(\ref{main-thr})},
then $\Holant(\mathcal{F})$ is \#P-hard.
\end{theorem}

The following reduction is obtained by
\emph{polynomial interpolation}. 

\begin{lemma}[\cite{jcbook}]\label{2by2-interpolation}
Let $f$ and $g$ be nonzero binary signatures with $M(f)=P^{-1}\left[\begin{smallmatrix}
\lambda_1 & 0\\
0 & \lambda_2\\
\end{smallmatrix}\right]P$ and
$M(g)=P^{-1}\left[\begin{smallmatrix}
1 & 0\\
0 & 0\\
\end{smallmatrix}\right]P$ for some invertible matrix $P$.
If $\lambda_1\neq 0$ and $|\frac{\lambda_2}{\lambda_1}|\neq 1$,
then $$\Holant(g, \mathcal{F})\leqslant_T\Holant(f, \mathcal{F})$$ for any signature set $\mathcal{F}$.
\end{lemma}

By Lemmas \ref{lin-wang} and \ref{2by2-interpolation}, we can always realize a nonzero unary signature from nonzero signatures not satisfying {\sc 1st-Orth}. We have the following \#P-hardness result.

\begin{lemma}\label{even-first-order}
Let $\mathcal{F}$ be a set of real-valued signatures containing a nonzero signature that does not satisfy {\sc 1st-Orth}. If $\mathcal{F}$ does not satisfy condition \rm{(\ref{main-thr})},
then $\Holant(\mathcal{F})$ is \#P-hard.
 \end{lemma}
 
 \begin{proof}
 Consider $M_{i}(f)$ for all indices $i$. 
 Clearly, $M(\mathfrak{m}_if)=M_{i}(f)M^{\tt T}_{i}(f)$ 
 is a real symmetric
positive semi-definite matrix, which is diagonalizable with two non-negative real eigenvalues $\lambda_i\geqslant\mu_i\geqslant 0$.
These two eigenvalues are not both zero since $f$ is real valued and $f \not\equiv 0$,
and so $M(\mathfrak{m}_if)\neq 0$. Thus, $\lambda_i\neq 0$.
Then, $|\frac{\mu_i}{\lambda_i}|=1$
iff $\lambda_i=\mu_i$.
In other words, 
$M(\mathfrak{m}_if)= \mu_i I_2$ for some real $\mu_i\neq 0$.

Since $f$ does not satisfy {\sc 1st-Orth}, by Lemma~\ref{lem-same-mu}, there is an index $i$ such that $M(\mathfrak{m}_if)\neq  \mu_i I_2$ for any real $\mu_i\neq 0$. 
Thus, $M(\mathfrak{m}_if)$ has two eigenvalues with different norms. By Lemma \ref{2by2-interpolation},
 we can realize a nonzero binary signature $g$ such that $M(g)$ is degenerate. This implies that  $g$ can be factorized as a tensor product of two nonzero unary signatures. 
 By Lemma \ref{lin-wang}, we can realize a nonzero unary signature and hence by Theorem \ref{odd-dic},  $\Holant(\mathcal{F})$ is \#P-hard.
 \end{proof}

We also need to use the results of \emph{eight-vertex models} and  \emph{Eulerian Orientation} (EO) problems.
\begin{theorem}[\cite{cai-fu-eight}]\label{eight-vertex}
Let $\widehat{f}$ be a signature with 
$M(\widehat{f})=\left[\begin{smallmatrix}
    c & 0 & 0 & a\\
     0 & d  & b & 0\\
      0 & \overline b & \overline d  & 0\\
      \overline  a & 0 & 0 &  \overline{c}\\
    \end{smallmatrix}\right]$. Then,  $\holant{\neq_2}{\widehat f}$ is \#P-hard in the following cases.
    \begin{itemize}
        \item $\widehat f$ has support $6$,
        \item $\widehat f$ has support $4$ and the nonzero entries of $M(\widehat f)$ do not have the same norm, or
        \item $\widehat f$ has support $8$, all nonzero entries of $M(\widehat f)$ are positive real numbers and are not all equal.
    \end{itemize}
\end{theorem}

\begin{theorem}[\cite{cai-fu-shao-eo}]\label{eo}
Let $\widehat{\mathcal{F}}$ be a set of \rm{EO} signatures (i.e., with half-weighted support)  satisfying {\sc ars}. Then $\Holant(\mathcal{DEQ}\mid \widehat{\mathcal{F}})$ is \#P-hard unless $\widehat{\mathcal{F}}\subseteq \mathscr{P}$ or $\widehat{\mathcal{F}}\subseteq \mathscr{A}$.
\end{theorem}

The following reduction states that we can realize all $\mathcal{EQ}_2$ once we have $=_4$ in $\Holant(\mathcal{F})$.

\begin{lemma}[\cite{jcbook}]\label{lem-=4-red}
$\CSP_2(\mathcal{F})\leqslant_T\Holant(=_4, \mathcal F)$.
\end{lemma}

The following reductions state that we can realize all $\mathcal{DEQ}$ once we have $\neq_4$ in $\holant{\neq_2}{\widehat{\mathcal{F}}}.$
\begin{lemma}[\cite{cai-fu-shao-eo}]\label{lem-neq4-red}
$\Holant(\mathcal{DEQ}\mid \widehat{\mathcal{F}})\leqslant_T \Holant(\neq_2\mid \mathcal{DEQ},  \widehat{\mathcal{F}})
\leqslant_T \Holant(\neq_2\mid \neq_4,  \widehat{\mathcal{F}}).$ 
\end{lemma}

\subsection{A summary of notations}
We use the following Table~\ref{tab:notation} to summarize notations given in this section. 
In the left column, we list notations in $\holant{=_2}{\mathcal{F}}$ where $\mathcal{F}$ is a set of real-valued signatures, and in the right column, we list corresponding notations in $\holant{\neq_2}{\widehat{\mathcal{F}}}$ where $\widehat{\mathcal{F}}=Z^{-1}\mathcal{F}$ is the set of complex-valued signatures with {\sc ars}.
Note that although $\mathcal{EO}$ also satisfies {\sc ars}, we will only use it in $\holant{=_2}{\mathcal{F}}$.
Similarly, we will only use $\mathcal{DEQ}$ and $\mathcal{D}$ in $\holant{\neq_2}{\widehat{\mathcal{F}}}$ although it is real-valued.

\begin{table}[!h]
\renewcommand{\arraystretch}{1.5}
    \centering
    \begin{tabular}{|c|c|}
    \hline
        $\holant{=_2}{\mathcal{F}}$  where $\mathcal{F}$ is real-valued & $\holant{\neq_2}{\widehat{\mathcal{F}}}$ where $\widehat{\mathcal{F}}$ satisfies {\sc ars}\\
        \hline
           $\mathcal{EQ}=\{=_1, =_2, \ldots, =_n,\ldots\}$ & N$/$A \\
       \hline
      N$/$A &  $\mathcal{DEQ}=\{\neq_2, \neq_4, \ldots, \neq_{2n}, \ldots\}$, $\mathcal{D}=\{\neq_2\}$\\
      \hline
     $\mathcal{O}=\{\text{binary orthogonal and zero signatures}\}$    & $\widehat{\mathcal{O}}=\{\text{binary signatures with {\sc ars} and parity}\}$\\
       \hline
       $\mathcal{B}=\{=_2, =_2^-, \neq_2, \neq_2^-\}$ & $\widehat{\mathcal{B}}=\{\neq_2, =_2, (-\ii)\cdot=^-_2, \ii\cdot\neq_2^-\}$\\
       \hline
    a holographic transformation $Q\mathcal{F}$ by $Q\in {\bf O}_2$ & a holographic transformation $\widehat{Q}\widehat{\mathcal{F}}$ by $\widehat{Q}\in \widehat{{\bf O}_2}$ \\
      \hline
      a merging gadget $\partial_{ij}f=f_{ij}^{00}+f_{ij}^{11}$ & 
      a merging gadget $\widehat\partial_{ij}\widehat{f}=\widehat{f}_{ij}^{01}+\widehat{f}_{ij}^{10}$ \\
      \hline
      extending gadgets $\{f\}^{\mathcal{B}}_{=_2}$ with ${\mathcal{B}}$ &  extending gadgets $\{\widehat{f}\}^{\widehat{\mathcal{B}}}_{\neq_2}$  with ${\widehat{\mathcal{B}}}$\\
      \hline
      a mating gadget $\mathfrak{m}_{ij}f=M_{ij}(f)I_2^{\otimes n-1}M^{\tt T}_{ij}(f)$ & a mating gadget $\widehat{\mathfrak{m}}_{ij}\widehat{f}=M_{ij}(\widehat f)N_2^{\otimes n-1}M^{\tt T}_{ij}(\widehat f)$ \\
      \hline
    \end{tabular}
    \caption{Comparisons of notations in $\holant{=_2}{\mathcal{F}}$ and $\holant{\neq_2}{\widehat{\mathcal{F}}}$}
    \label{tab:notation}
\end{table}

Recall that $\mathcal{F}^{\otimes}$ denotes the set $\{\lambda\bigotimes^k_{i=1}f_i\mid \lambda \in \mathbb{R}\backslash\{0\},  k\geqslant 1, f_i\in \mathcal{F}\}$ for any signature set $\mathcal{F}$.
We remark that both $\mathcal{O}^{\otimes}$ and $\widehat{\mathcal{O}}^{\otimes}$ contain all
zero signatures of even arity since the binary zero signature is in $\mathcal{O}$ and $\widehat{\mathcal{O}}$.
However, $\mathcal{B}^{\otimes}$, $\widehat{\mathcal{B}}^{\otimes}$, and $\mathcal{D}^{\otimes}$ do \emph{not} contain any zero signatures.

In the following, without other specifications, we use $f$ to denote a real-valued signature and $\mathcal{F}$ to denote a set of real-valued signatures. 
We use $\widehat{f}$ to denote a signature satisfying {\sc ars} and $\widehat{\mathcal{F}}$ to denote a set of such signatures. 
We use $Q$ to denote a matrix in ${\bf O}_2$, and $\widehat{Q}$ to denote a matrix in $\widehat{{\bf O}_2}$.
Clearly, if $\mathcal{F}$ is real-valued, then $Q\mathcal{F}$ is also real-valued.
Equivalently, if $\widehat{\mathcal{F}}$ satisfies {\sc ars}, then $\widehat{Q}\widehat{\mathcal{F}}=\widehat{Q\mathcal{F}}$ also satisfies {\sc ars}.
\section{Proof Organization}\label{sec:proof-outline}
By Theorem \ref{thm-main-thr}, if  $\mathcal F$ satisfies  condition \rm{(\ref{main-thr})}, then $\Holant(\mathcal F)$ is P-time
computable. 
So, we only need to prove that $\Holant(\mathcal F)$ is \#P-hard when $\mathcal F$ does not satisfy  condition \rm{(\ref{main-thr})}. 
If $\mathcal{F}$ contains a nonzero signature of odd arity,
then by Theorem~\ref{odd-dic}, we are done.
In the following without other specifications, when refer to a real-valued signature set $\mathcal{F}$ or a corresponding signature set $\widehat{\mathcal{F}}=Z^{-1}\mathcal{F}$ satisfying {\sc ars}, we always assume that  they  consist of signatures of even arity, and $\mathcal{F}$ does not satisfy  condition \rm{(\ref{main-thr})}. 

In Section \ref{sec-first-second}, we generalize the notion of first order orthogonality ({\sc 1st-Orth}) to second order orthogonality ({\sc 2nd-Orth}).
This property plays a key role in our proof.
We show that all irreducible signatures in $\mathcal{F}$ satisfy {\sc 2nd-Orth}, or else, we get  \#P-hardness based on results of \#CSP problems,  \#EO problems and eight-vertex models (Lemma \ref{second-ortho}).
We derive some consequences from  the condition {\sc 2nd-Orth}
 for signatures with {\sc ars}. 
These will be  used throughout in the proof.

In Section \ref{sec-induction}, we give the induction framework of the proof. 
Since $\mathcal{F}$ does not satisfy  condition \rm{(\ref{main-thr})}, $\mathcal{F}\not\subseteq \mathscr{T}$. 
Also since $\mathcal{O}^\otimes\subseteq\mathscr{T}$, and by Lemma~\ref{lem-2-ary}, we may assume that $\mathcal{F}$ contains a signature $f$ of arity $2n\geqslant 4$ where $f\notin \mathcal{O}^\otimes$.
We want to achieve a proof of  \#P-hardness by induction on $2n$.
When $2n=2$,
as a corollary of {\sc 1st-Orth}, we show that $\Holant(\mathcal{F})$ is \#P-hard  (Lemma~\ref{lem-2-ary}).
When $2n=4$, by {\sc 2nd-Orth}, we show that $\Holant(\mathcal{F})$ is \#P-hard   (Lemma \ref{lem-4-ary}).

In Sections \ref{sec-f6} and \ref{sec-holantb}, we handle the case of arity 6.
Let $f\notin \mathcal{O}^{\otimes}$ be a 6-ary signature in $\mathcal{F}$.
We show that $\Holant(\mathcal{F})$ is \#P-hard or the extraordinary signature which we named $f_6$ with the Bell property can be realized (Theorem \ref{lem-6-ary-f6}). 
By gadget construction, all four Bell signatures $\mathcal{B}$ can be realized from $f_6$.
Then we prove the \#P-hardness of $\Holantb(f_6, \mathcal{F})=\Holant(\mathcal{B}, f_6, \mathcal{F})$ (Theorem \ref{thm-holantb} and Lemma~\ref{lem-f6-b-hardness}).
Combining these two results,  we have $\Holant(\mathcal{F})$ is \#P-hard (Lemma~\ref{lem-6-ary}).

In Section \ref{sec-f8}, we handle the case of arity $8$.
Let $f\notin \mathcal{O}^{\otimes}$ be an 8-ary signature in $\mathcal{F}$.
We show that $\Holant(\mathcal{F})$ is \#P-hard or another extraordinary signature which we named  $f_8$ with the strong Bell property can be realized (Theorem \ref{thm-f8}). 
One can prove that $\mathcal{B}$ cannot be realized from $f_8$ by gadget construction. 
However, by introducing Holant problems with limited appearance and using the strong Bell property of $f_8$, we show $\Holantb( f_8, \mathcal{F})\leq_T \Holant(f_8, \mathcal{F})$ (Lemmas \ref{lem-0to1}). 
Then, we prove the \#P-hardness of $\Holantb(f_8, \mathcal{F})$.
Combining these results,  we have $\Holant(\mathcal{F})$ is \#P-hard (Lemma~\ref{lem-8-ary}).

In Section \ref{sec-10-ary}, we show that our induction proof works for signatures of arity $2n\geqslant 10$.
Let $f\notin \mathcal{O}^{\otimes}$ be a $2n$-ary $(2n\geqslant 10)$ signature in $\mathcal{F}$.
Then, $\Holant(\mathcal{F})$ is \#P-hard or we can realize a signature of arity less than $2n$ that is not in  $\mathcal{O}^\otimes$ (Lemma~\ref{lem-10-ary}).
Then, by a sequence of reductions of length independent of the problem instance size, we can eventually realize a signature of arity 
at most $8$ that is not in  $\mathcal{O}^\otimes$.
Finally, combining Lemmas \ref{lem-2-ary}, \ref{lem-4-ary}, \ref{lem-6-ary}, \ref{lem-8-ary} and \ref{lem-10-ary}, we finish the proof of Theorem \ref{main-theorem}.
In the actual proof, for convenience, many results are proved in the setting of $\holant{\neq}{\widehat{\mathcal{F}}}$ which is equivalent to $\Holant(\mathcal{F})$ under the $Z^{-1}$ transformation.

 \section{Second Order Orthogonality}\label{sec-first-second}
 In this section, we generalize the notion of first order orthogonality ({\sc 1st-Orth}) to second order orthogonality ({\sc 2nd-Orth}) (Definition \ref{def:second-order-othor}).
We show that for real-valued $\mathcal{F}$ that does not satisfy condition (\ref{main-thr}), every irreducible $f\in \mathcal{F}$ of arity at least $4$ satisfies {\sc 2nd-Orth}, or otherwise $\Holant(\mathcal{F})$ is \#P-hard (Lemma~\ref{second-ortho}).
Then, we derive some consequences from  the condition {\sc 2nd-Orth}
 for signatures with {\sc ars}. 
These will be  used throughout in the following proof.

  \begin{definition}[Second order orthogonality]\label{def:second-order-othor}
       Let $f$ be a complex-valued signature of arity $n \geqslant 4$. It satisfies the \emph{second order orthogonality ({\sc 2nd-Orth})} if there exists some $\lambda\neq 0$ such that for all pairs of indices $\{i, j\}\subseteq [n]$, the entries of $f$ satisfy 
\begin{equation*}
    |{\bf f}_{ij}^{00}|^2=|{\bf f}_{ij}^{01}|^2=|{\bf f}_{ij}^{10}|^2=|{\bf f}_{ij}^{11}|^2=\lambda, 
\text{ ~~~~\rm and }~~~~ \langle{\bf f}_{ij}^{ab}, {\bf f}_{ij}^{cd}\rangle =0 ~~~~\text{\rm  for  all } (a, b)\neq (c, d).
\end{equation*}
    \end{definition}
  \begin{remark}
  Similar to the remark of first order orthogonality (Definition~\ref{def-first-order}),
    $f$ satisfies {\sc 2nd-Orth} iff there is some $\lambda \not =0$ such that for all $(i, j)$, $M(\mathfrak{m}_{ij}f)= \lambda I_4=\lambda I_2^{\otimes 2}$, and
    $\widehat{f}$ satisfies {\sc 2nd-Orth} iff there is some $\lambda \not =0$ such that for all $(i, j)$, $M(\widehat{\mathfrak{m}}_{ij}\widehat{f})= \lambda N_4=\lambda N_2^{\otimes 2}$.
    Moreover, $f$ satisfies {\sc 2nd-Orth} iff $\widehat{f}$ satisfies it.
    Clearly, {\sc 2nd-Orth} implies {\sc 1st-Orth}.
    \end{remark}
  In the next, we will prove Lemma~\ref{second-ortho} based on dichotomies of \#CSP problems,  \#EO problems and eight-vertex models. 
Since \#EO problems and eight-vertex models are defined as special cases of the problem $\Holant(\neq_2 \mid \widehat{\mathcal F})$,
for convenience, 
we will consider the problem $\Holant(\neq_2 \mid \widehat{\mathcal F})$ which is equivalent to  $\Holant(\mathcal{F})$. 
Recall that $\widehat{\mathcal{F}}=Z^{-1}{\mathcal{F}}$ satisfies {\sc ars}, and we assumed that $\mathcal{F}$ does not satisfy condition (\ref{main-thr}). 
We first give the following lemma.

\begin{lemma}\label{lem-neq-4-hard}
$\Holant(\mathcal{DEQ} \mid \widehat{\mathcal{F}})$ is \#P-hard.
\end{lemma}
\begin{proof}
Since $\mathcal{F}$ does not satisfy condition (\ref{main-thr}), by Lemma~\ref{lem-hard-sign}, $\widehat{\mathcal{F}}\not\subseteq \mathscr{P}$ and $\widehat{\mathcal{F}}\not\subseteq \mathscr{A}$.
If $\widehat{\mathcal{F}}$ is a set of EO signatures, 
then by Theorems \ref{eo}, 
 $\Holant(\mathcal{DEQ} \mid \widehat{\mathcal{F}})$ is \#P-hard since $\widehat{\mathcal{F}}\not\subseteq\mathscr{P}$ and $\widehat{\mathcal{F}}\not\subseteq\mathscr{A}$.
Thus, we may assume that there is a signature $\widehat f \in \widehat{ \mathcal{F}}$ whose support is not half-weighted. 
Suppose that $\widehat f$ has arity $2n$.
Since $\mathscr{S}(\widehat{f})\not\subseteq \mathscr{H}_{2n}$, by {\sc ars}, there is an $\alpha \in \mathbb{Z}_2^{2n}$ with ${\rm wt}(\alpha)=k<n$ such that 
${\widehat f}(\alpha) \neq 0$.
We first show that we can realize a signature ${\widehat g}$ of arity $2n-2k$ such that ${\widehat g}({\vec{0}})\neq 0$. 
If ${\rm wt}(\alpha)=k=0$, then we are done. Otherwise, we have $n>k\geqslant 1$. Thus, $2n\geqslant 4$ and
$\alpha$ has length at least $4$. 
By Lemma \ref{lem-zero_2}, 
there is a pair of indices $\{i, j\}$ such that $\widehat{\partial}_{ij}\widehat{f}({\beta})\neq 0$ for some ${\rm wt}(\beta)=k-1$. 
Clearly, $\widehat{\partial}_{ij}\widehat{f}$ has arity $2n-2$.
Since $0 \leqslant k-1
< (2n-2)/2$,
$\widehat{\partial}_{ij}\widehat{f}$ is not an EO signature. 
Now we can continue this process, and by a chain of merging gadgets using $\neq_2$, 
we can realize a signature $\widehat g$ of arity $2m=2n-2k$ such that  ${\widehat g}({\vec{0}})
\neq 0$.
Denote by $a =
{\widehat g}({\vec{0}})$. 

Then, we connect all $2m$  variables of $\widehat g$ with $2m$  variables  of $\neq_{4m}$ that always take the same value in $\mathscr{S}(\neq_{4m})$  using $\neq_2$.
We get a signature $\widehat h$ of arity $2m$ where
$\widehat h(\vec{0})=a$, $\widehat h(\vec{1})=\bar a$ by {\sc ars}, and $\widehat{h}(\gamma)=0$ elsewhere. 
Then, consider the holographic transformation by $\widehat Q=
\left[\begin{smallmatrix}
\sqrt[2m]{\bar a} & 0\\
0 & \sqrt[2m]{ a}\\
\end{smallmatrix}
\right] \in \widehat{\rm{O}}_2$. 
It transforms $\widehat{h}$ to $\neq_{2m}$, but does not change $\mathcal{DEQ}$. 
Thus,
$$\Holant(\mathcal{DEQ} \mid \widehat{h}, \widehat{\mathcal{F}})\equiv_T \Holant(\mathcal{DEQ} \mid =_{2m}, \widehat Q\widehat{\mathcal{F}}).$$
Since $\widehat{\mathcal{F}}$ does not satisfy condition (\ref{main-thr}), by Theorem~\ref{hard-result}, $\Holant(\mathcal{DEQ} \mid =_{2m}, \widehat Q\widehat{\mathcal{F}})$ is \#P-hard.
Thus, $\Holant(\mathcal{DEQ} \mid \widehat{\mathcal{F}})$ is \#P-hard.
\end{proof}

  
  Then, we consider signatures $\widehat{\mathfrak{m}}_{ij}\widehat{f}$ realized by mating.   
  
 \begin{lemma}\label{lem-4.5}
 Let $\widehat f\in \widehat{\mathcal{F}}$ be a  signature of arity $2n\geqslant 4$. Then, 
 \begin{itemize}

 \item $\holant{\neq_2}{\widehat{\mathcal{F}}}$ is \#P-hard, or
      \item for all pairs of indices $\{i, j\}$, there exists a nonzero binary signature $\widehat{b}_{ij}\in \widehat{\mathcal{O}}$ such that $\widehat{b}_{ij}(x_i, x_j)\mid \widehat{f}$ or $M(\widehat{\mathfrak{m}}_{ij}\widehat{f})= \lambda_{ij} N_4$ for some real $\lambda_{ij}\neq 0$.
 \end{itemize}
 \end{lemma}
  
  \begin{proof}
  If $\widehat{f}\equiv 0$,  then the lemma holds trivially since for all $\{i, j\}$ and any $\widehat{b}_{ij}\neq 0$, $\widehat{b}_{ij}(x_i, x_j)\mid \widehat{f}$.
  Thus, we may assume that $f\not\equiv 0$.
  
  If $\widehat{f}$ does not satisfy {\sc 1st-Orth}, then $f$ does not satisfy it. By Lemma~\ref{even-first-order}, $\holant{\neq_2}{\widehat{\mathcal{F}}}\equiv_T\holant{=_2}{\mathcal{F}}$ is \#P-hard. 
Thus, we may assume that $\widehat{f}$ satisfies {\sc 1st-Orth}. 
Then, for all indices $i$, we have 
$$M(\widehat{\mathfrak{m}}_i\widehat{f})
=\left[\begin{matrix}
\langle \widehat{{\bf f}}_i^0, \widehat{{\bf f}}_i^1 \rangle & |\widehat{{\bf f}}_i^0|^2\\
|\widehat{{\bf f}}_i^1|^2 & \langle \widehat{{\bf f}}_i^1, \widehat{{\bf f}}_i^0 \rangle
\end{matrix}\right]=
\mu\left[\begin{matrix}
0 & 1 \\
 1 & 0
\end{matrix}\right].
$$
For any variable $x_i$, 
we may take another variable $x_j$ $(j\neq i)$ and partition the sum in the inner product $\langle \widehat{{\bf f}}_i^0, \widehat{{\bf f}}_i^1\rangle=0$ into two sums depending on whether $x_j=0$ or $1$.
Also, by {\sc ars} we have $$\langle \widehat{{\bf f}}_i^0, \widehat{{\bf f}}_i^1 \rangle=\langle \widehat{{\bf f}}_{ij}^{00}, \widehat{{\bf f}}_{ij}^{10} \rangle+\langle \widehat{{\bf f}}_{ij}^{01}, \widehat{{\bf f}}_{ij}^{11}\rangle = \langle \widehat{{\bf f}}_{ij}^{00}, \widehat{{\bf f}}_{ij}^{10} \rangle+\langle \overline{\widehat{{\bf f}}_{ij}^{10}}, \overline{\widehat{{\bf f}}_{ij}^{00}}\rangle= 2 \langle \widehat{{\bf f}}_{ij}^{00}, \widehat{{\bf f}}_{ij}^{10} \rangle =0.$$
Thus, for all pairs of  indices $\{i, j\}$,  $\langle \widehat{{\bf f}}_{ij}^{00}, \widehat{{\bf f}}_{ij}^{10} \rangle =0$ and $\langle \widehat{{\bf f}}_{ij}^{01}, \widehat{{\bf f}}_{ij}^{11} \rangle =0$. 
(Note that by exchanging $i$ and $j$ we also  have $\langle \widehat{{\bf f}}_{ij}^{00}, \widehat{{\bf f}}_{ij}^{01} \rangle =0$ and $\langle \widehat{{\bf f}}_{ij}^{10}, \widehat{{\bf f}}_{ij}^{11} \rangle =0$.)
Also by {\sc ars},
we have $|{\widehat{{\bf f}}}_{ij}^{00}|^2=|{\overline{\widehat{{\bf f}}_{ij}^{11}}}|^2=|{\widehat{{\bf f}}}_{ij}^{11}|^2$ and $|{\widehat{{\bf f}}}_{ij}^{01}|^2=|\overline{{\widehat{{\bf f}}}_{ij}^{10}}|^2=|{\widehat{{\bf f}}}_{ij}^{10}|^2.$

Now, consider $\widehat{\mathfrak{m}}_{ij}\widehat{f}$ for all pairs of indices $\{i, j\}$. 
$$M(\widehat{\mathfrak m}_{ij}\widehat f)=\begin{bmatrix}
{\widehat{{\bf f}}}_{ij}^{00}\\
{\widehat{{\bf f}}}_{ij}^{01}\\
{\widehat{{\bf f}}}_{ij}^{10}\\
{\widehat{{\bf f}}}_{ij}^{11}\\
\end{bmatrix}
\left[\begin{matrix}
{\overline{{\widehat{{\bf f}}}_{ij}^{11}}}^{\tt T} &{\overline{{\widehat{{\bf f}}}_{ij}^{10}}}^{\tt T}
&{\overline{{\widehat{{\bf f}}}_{ij}^{01}}}^{\tt T}&{\overline{{\widehat{{\bf f}}}_{ij}^{00}}}^{\tt T}
\end{matrix}\right]=
\left[\begin{matrix}
\langle{\widehat{{\bf f}}}_{ij}^{00}, {\widehat{{\bf f}}}_{ij}^{11}\rangle & 0 & 0& |{\widehat{{\bf f}}}_{ij}^{00}|^2\\
0 & \langle{\widehat{{\bf f}}}_{ij}^{01}, {\widehat{{\bf f}}}_{ij}^{10}\rangle & |{\widehat{{\bf f}}}_{ij}^{01}|^2 & 0\\
0 & |{\widehat{{\bf f}}}_{ij}^{10}|^2 & \langle{\widehat{{\bf f}}}_{ij}^{10}, {\widehat{{\bf f}}}_{ij}^{01}\rangle  & 0\\
|{\widehat{{\bf f}}}_{ij}^{11}|^2 & 0 & 0 & \langle{\widehat{{\bf f}}}_{ij}^{11}, {\widehat{{\bf f}}}_{ij}^{00}\rangle
\end{matrix}\right].$$
Note that $|\langle{\widehat{{\bf f}}}_{ij}^{00}, {\widehat{{\bf f}}}_{ij}^{11}\rangle|\leqslant |{\widehat{{\bf f}}}_{ij}^{00}| \cdot |{\widehat{{\bf f}}}_{ij}^{11}|$ by Cauchy-Schwarz inequality.
Clearly, $\widehat{\mathfrak{m}}_{ij}\widehat{f}$ has even parity, and thus it represents a signature of the eight-vertex model.
If there exists a pair of indices $\{i, j\}$ such that $\Holant(\neq_2 \mid \widehat{\mathfrak{m}}_{ij}\widehat f)$ is \#P-hard, then we are done
since $\Holant(\neq_2 \mid \widehat{\mathfrak{m}}_{ij}\widehat f)\leqslant_T \Holant(\neq_2 \mid \widehat{\mathcal{F}})$.
Thus, we may assume all
$\widehat{\mathfrak{m}}_{ij}\widehat f$ belong to the tractable family for eight-vertex models.
Clearly, by observing its
antidiagonal entries
of the matrix $M(\widehat{\mathfrak{m}}_{ij}\widehat{f})$, we have $\widehat{\mathfrak{m}}_{ij}\widehat{f} \not\equiv 0$ since $\widehat f\not \equiv 0$.
By Theorem \ref{eight-vertex}, there are three possible cases.

\begin{itemize}
    \item There exists a pair $\{i, j\}$ such that $\widehat{\mathfrak{m}}_{ij}\widehat f$ has support of size $2$. By Cauchy-Schwarz inequality,  $ M(\widehat{\mathfrak{m}}_{ij}\widehat f)$ is either of the form $\lambda_{ij} \left[\begin{smallmatrix}
    0 & 0 & 0 & 1\\
     0 & 0 & 0 & 0\\
      0 & 0 & 0 & 0\\
       1 & 0 & 0 & 0\\
    \end{smallmatrix}\right]$ where $\lambda_{ij} = |{\widehat{{\bf f}}}_{ij}^{00}|^2=|{\widehat{{\bf f}}}_{ij}^{11}|^2\neq 0$ or $\lambda_{ij} \left[\begin{smallmatrix}
    0 & 0 & 0 & 0\\
     0 & 0 & 1 & 0\\
      0 & 1 & 0 & 0\\
       0 & 0 & 0 & 0\\
    \end{smallmatrix}\right]$ where $ \lambda_{ij}= |{\widehat{{\bf f}}}_{ij}^{01}|=|{\widehat{{\bf f}}}_{ij}^{10}| \neq 0$.
     In both cases, $\neq_4$ is realizable since $\lambda_{ij} \not =0$.
    The form that $\langle\widehat{\bf f}_{ij}^{01}, \widehat{\bf f}_{ij}^{10}\rangle\neq 0$ while $|\widehat{\bf f}_{ij}^{01}|^2=|\widehat{\bf f}_{ij}^{10}|^2=0$  cannot occur since $|\langle\widehat{\bf f}_{ij}^{01}, \widehat{\bf f}_{ij}^{10}\rangle|\leqslant |\widehat{\bf f}_{ij}^{01}||\widehat{\bf f}_{ij}^{10}|$. Also, the form that $\langle\widehat{\bf f}_{ij}^{00}, \widehat{\bf f}_{ij}^{11}\rangle\neq 0$ while $|\widehat{\bf f}_{ij}^{00}|^2=|\widehat{\bf f}_{ij}^{11}|^2=0$  cannot occur.
    Since $\neq_4$ is available, by Lemma~\ref{lem-neq4-red}, $\holant{\mathcal{DEQ}}{\widehat{\mathcal{F}}}\leqslant_T \holant{=_2}{\widehat{\mathcal{F}}}.$
    By Lemma~\ref{lem-neq-4-hard}, $ \holant{=_2}{\widehat{\mathcal{F}}}$ is \#P-hard.
    \item There exists a pair $\{i, j\}$ such that $\widehat{\mathfrak{m}}_{ij}\widehat f$ has support of size $8$. 
    We can rename the four variables of $\widehat{\mathfrak{m}}_{ij}\widehat f$ in a cyclic permutation.
    We use $\widehat{g}$ to denote this signature.
    Then $M(\widehat g)= M_{12}(\widehat g)= \left[\begin{smallmatrix}
    c & 0 & 0 & d\\
     0 & a  & b & 0\\
      0 & b & a  & 0\\
       \bar d & 0 & 0 &  \bar{c}\\
    \end{smallmatrix}\right]$ where $a$ and $b$ are positive real numbers and $c$ and $d$ are nonzero complex numbers. 
    Consider the signature $\widehat{\mathfrak{m}}_{12}\widehat{g}$ realized by mating $\widehat{g}$.
    We denote it by $\widehat{h}$. Then,
$$M(\widehat{h})=M(\widehat g)N_4M^{\tt T}(\widehat g)= \left[\begin{matrix}
    2cd & 0 & 0 & |c|^2+|d|^2\\
     0 &  2ab & a^2+b^2 & 0\\
      0 & a^2+b^2 & 2ab  & 0\\
       |c|^2+|d|^2 & 0 & 0 &  2\bar c \bar d\\
    \end{matrix}\right]=\left[\begin{matrix}
    c' & 0 & 0 & d'\\
     0 &  a' & b' & 0\\
      0 & b' & a'  & 0\\
       d' & 0 & 0 & \bar {c'}\\
    \end{matrix}\right],$$
    where $a', b'$, and $d'$ are positive real numbers and $c'$ is a nonzero complex number.
    Suppose that the argument of $c'$ is $\theta$, i.e., $c'=|c'|e^{\ii\theta}.$

Consider the holographic transformation by $\widehat{Q}=\left[\begin{smallmatrix}
e^{-\ii\theta/4} & 0\\
0 & e^{\ii\theta/4} \\
\end{smallmatrix}\right]\in \widehat{{\bf O}_2}.$ 
Then,
$$\Holant(\neq_2 \mid \widehat{h}, \widehat{\mathcal{F}})\equiv_T \Holant(\neq_2 \mid \widehat{Q}\widehat{h}, \widehat{Q}\widehat{\mathcal{F}}).$$
Note that $M(\widehat{Q}\widehat{h})=\left[\begin{smallmatrix}
    |c'| & 0 & 0 & d'\\
     0 &  a' & b' & 0\\
      0 & b' & a'  & 0\\
       d' & 0 & 0 & |c'|\\
    \end{smallmatrix}\right]$ where all entries are positive real numbers.
    Notice that all weight 2 entries of $\widehat{h}$
    are unchanged in $\widehat{Q}\widehat{h}$.
    By Theorem \ref{eight-vertex}, $\holant{\neq_2}{\widehat{Q}\widehat{h} }$ is \#P-hard unless $a'=b'=|c'|=d'$. Thus, we may assume that 
    $M(\widehat{Q}\widehat{h})=\left[\begin{smallmatrix}
    1 & 0 & 0 & 1\\
     0 &  1 & 1 & 0\\
      0 & 1 & 1  & 0\\
       1 & 0 & 0 &  1\\
    \end{smallmatrix}\right]$ up to normalization. 
Notice that $M(Z(\widehat{Q}\widehat{{h}}))=Z^{\otimes 2}M(\widehat{Q}\widehat{h})(Z^{\tt T})^{\otimes 2}=\left[\begin{smallmatrix}
    1 & 0 & 0 & 0\\
     0 &  0 & 0 & 0\\
      0 & 0 & 0  & 0\\
       0 & 0 & 0 &  1\\
    \end{smallmatrix}\right]$,
    which is the arity-4 equality $(=_4)$.
    Consider the holographic transformation by $Z$ which transfers $\neq_2$ back to $=_2$.
    Remember that $\widehat{Q}=Z^{-1}{Q}Z$.
    Then, $Z(\widehat{Q}\widehat{F})=Z(Z^{-1}QZ)(Z^{-1}\mathcal{F})=Q\mathcal{F}$.
    Since $\widehat{Q}\in \widehat{{\bf O}_2}$, we have  $Q\in {\bf O}_2$.
    Thus,
$$\Holant(\neq_2 \mid \widehat{Q}\widehat{h}, \widehat{Q}\widehat{\mathcal{F}})\equiv_T \Holant(=_2 \mid =_4, Q\mathcal{F}).$$
By Lemma \ref{lem-=4-red}, $\CSP_2(Q\mathcal{F})\leqslant_T \Holant(=_2 \mid =_4, Q\mathcal{F})$.
Since $\mathcal{F}$ does not satisfy condition (\ref{main-thr}) and $Q\in {\bf O}_2$, 
by Theorem~\ref{hard-result}, $\CSP_2(Q\mathcal{F})$ is \#P-hard.
Thus, $\holant{\neq_2}{\widehat{\mathcal{F}}}$ is \#P-hard.
    
       \item For all $\{i, j\}$, $\widehat{\mathfrak{m}}_{ij}\widehat f$ has support of size $4$. 
       By Cauchy-Schwarz inequality, $M(\widehat{\mathfrak{m}}_{ij}\widehat f)$ is of the form 
    $\left[\begin{smallmatrix}
    b & 0 & 0 & a\\
     0 & 0 & 0 & 0\\
      0 & 0 & 0 & 0\\
       a & 0 & 0 & \bar b\\
    \end{smallmatrix}\right]$
    or  $\left[\begin{smallmatrix}
    0 & 0 & 0 & 0\\
     0 & b & a & 0\\
      0 & a & \bar b & 0\\
       0 & 0 & 0 & 0\\
    \end{smallmatrix}\right]$ where $a^2-|b|^2=0$, 
    or the form $\lambda_{ij} \left[\begin{smallmatrix}
    0 & 0 & 0 & 1\\
     0 & 0 & 1 & 0\\
      0 & 1 & 0 & 0\\
       1 & 0 & 0 & 0\\
    \end{smallmatrix}\right]$ where $\lambda_{ij} = |{\widehat{{\bf f}}}_{ij}^{00}|^2= |{\widehat{{\bf f}}}_{ij}^{01}| \neq 0$. 
    If $M(\mathfrak{m}_{ij}\widehat{f})=\lambda_{ij}N_4$, then we are done.
   Otherwise, $M(\widehat{\mathfrak{m}}_{ij}\widehat f)$ has rank one.
   Hence, $M_{ij}(\widehat f)$ also has rank one. 
    Then, by observing the form
    of $M(\widehat{\mathfrak{m}}_{ij}\widehat f)$ especially the  all zero rows,
    $\widehat f$ can be factorized as 
    $\widehat  b_{ij}(x_i, x_j) \otimes \widehat g$ where $\widehat b_{ij}\in\widehat{\mathcal{O}}$ and $\widehat g$ is a signature on the other $n-2$ variables.
    Thus, we are done.
\end{itemize}
The lemma is proved.
  \end{proof}
  \begin{remark}
  We give a restatement of Lemma~\ref{lem-4.5} in the setting of $\Holant(\mathcal{F})$.
Let $f\in \mathcal{F}$ be a signature of arity $2n\geqslant4$. Then, $\Holant(\mathcal{F})$ is \#P-hard, or for all pairs of indices $\{i, j\}$, 
there exists a nonzero binary signature $b_{ij}\in \mathcal{O}$ such that $b_{ij}(x_i, x_j)\mid f$ or $M(\mathfrak{m}_{ij}f)=\lambda_{ij}I_4$ for some real $\lambda_{ij}\neq 0$.
  \end{remark}

Now for an irreducible signature $\widehat{f}$ of arity $2n\geqslant 4$, we show that it satisfies {\sc 2nd-Orth} or we get  \#P-hardness. 
\begin{lemma}\label{second-ortho}
Let $\widehat f\in \widehat{\mathcal{F}}$ be an irreducible signature of arity $2n\geqslant 4$.
If $\widehat{f}$ does not satisfy {\sc 2nd-Orth}, then $\holant{\neq}{\widehat{\mathcal{F}}}$ is \#P-hard.
\end{lemma}
\begin{proof}
Since $\widehat{f}$ is irreducible, by Lemma \ref{lem-4.5}, 
$M(\widehat{\mathfrak{m}}_{ij}\widehat f)=\lambda_{ij}N_4$ for all $\{i, j\}$. 
     Now, we show all $\lambda_{ij}$ have the same value.   If we connect further the two respective pairs of variables of $\mathfrak m_{ij}f$, which totally connects two copies of $f$, we get a value $4\lambda_{ij}$. 
This value clearly does not depend on the particular indices $\{i, j\}$.
    We denote the value $\lambda_{ij}$ by $\lambda$.  This value is nonzero because $\widehat f$ is irreducible.
\end{proof}

We derive some consequences from  the condition {\sc 2nd-Orth}
 for signatures with {\sc ars}. 
Suppose that $\widehat{f}$ satisfies {\sc 2nd-Orth}.
First, by definition we have 
$|\widehat{{\bf f}}_{ij}^{ab}|^2=\lambda$
for any $(x_i, x_j)=(a, b)\in\{0, 1\}^2.$
   Given a vector $\widehat{{\bf f}}_{ij}^{ab}$, we can pick a third variable $x_k$ and partition $\widehat{{\bf f}}_{ij}^{ab}$ into two vectors $\widehat{{\bf f}}_{ijk}^{ab0}$ and $\widehat{{\bf f}}_{ijk}^{ab1}$ according to $x_k=0$ or $1$.
    By setting $(a, b)=(0, 0)$, we have
    \begin{equation}\label{e1}
        |\widehat{{\bf f}}_{ij}^{00}|^2=|\widehat{{\bf f}}_{ijk}^{000}|^2+|\widehat{{\bf f}}_{ijk}^{001}|^2=\lambda.
    \end{equation}
    Similarly, we consider the vector $\widehat{{\bf f}}_{ik}^{00}$ and partition it according to $x_j=0$ or $1$. We have
    \begin{equation}\label{e2}
        |\widehat{{\bf f}}_{ik}^{00}|^2=|\widehat{{\bf f}}_{ijk}^{000}|^2+|\widehat{{\bf f}}_{ijk}^{010}|^2=\lambda.
    \end{equation}
    Comparing equations (\ref{e1}) and (\ref{e2}), we have $|\widehat{{\bf f}}_{ijk}^{001}|^2=|\widehat{{\bf f}}_{ijk}^{010}|^2$. 
    Moreover, by {\sc ars}, we have $|\widehat{{\bf f}}_{ijk}^{010}|^2=|\widehat{{\bf f}}_{ijk}^{101}|^2.$ Thus, we have $|\widehat{{\bf f}}_{ijk}^{001}|^2=|\widehat{{\bf f}}_{ijk}^{101}|^2$.
    Note that the vector $\widehat{{\bf f}}_{jk}^{01}$ is partitioned into two vectors $\widehat{{\bf f}}_{ijk}^{001}$ and $\widehat{{\bf f}}_{ijk}^{101}$ 
    according to $x_i=0$ or $1$.
    That is 
    $$|\widehat{{\bf f}}_{jk}^{01}|^2=|\widehat{{\bf f}}_{ijk}^{001}|^2+|\widehat{{\bf f}}_{ijk}^{101}|^2=\lambda.$$ Thus, we have $|\widehat{{\bf f}}_{ijk}^{001}|^2=|\widehat{{\bf f}}_{ijk}^{101}|^2=\lambda/2$. 
     Then, by equation (\ref{e1}), we have $|\widehat{{\bf f}}_{ijk}^{000}|^2 =\lambda/2$, and again by {\sc ars}, we also have $|\widehat{{\bf f}}_{ijk}^{111}|^2 =|\widehat{{\bf f}}_{ijk}^{000}|^2=\lambda/2$.
     Note that indices $i, j, k$ are picked arbitrarily, by symmetry, we have 
\begin{equation}\label{eqn:arity>=6-basecase}
 |\widehat{{\bf f}}_{ijk}^{abc}|^2=\lambda/2
\end{equation} for all  $(x_i, x_j, x_k)=(a, b, c)\in\{0, 1\}^3.$ 

     Given a vector $\widehat{{\bf f}}_{ijk}^{abc}$, we can continue to pick a fourth variable $x_\ell$ and partition $\widehat{{\bf f}}_{ijk}^{abc}$ into two vectors $\widehat{{\bf f}}_{ijk\ell}^{abc0}$ and $\widehat{{\bf f}}_{ijk\ell}^{abc1}$ according to $x_\ell=0$ or $1$.
     By setting $(a, b, c)=(0, 0, 0)$, 
we have from (\ref{eqn:arity>=6-basecase})
     \begin{equation}\label{e3}
     |\widehat{{\bf f}}_{ijk}^{000}|^2=|\widehat{{\bf f}}_{ijk\ell}^{0000}|^2+|\widehat{{\bf f}}_{ijk\ell}^{0001}|^2=\lambda/2.
     \end{equation}
   Similarly, we consider the vector $\widehat{{\bf f}}_{ij\ell}^{001}$ and partition it according to $x_k=0$ or $1$. We have
   \begin{equation}\label{e4}
     |\widehat{{\bf f}}_{ij\ell}^{001}|^2=|\widehat{{\bf f}}_{ijk\ell}^{0001}|^2+|\widehat{{\bf f}}_{ijk\ell}^{0011}|^2=\lambda/2.
     \end{equation}
     Comparing equations (\ref{e3}) and (\ref{e4}), and also    by {\sc ars}, we have 
     \begin{equation}\label{e5}
    |\widehat{{\bf f}}_{ijk\ell}^{0000}|^2=|\widehat{{\bf f}}_{ijk\ell}^{0011}|^2=|\widehat{{\bf f}}_{ijk\ell}^{1100}|^2=|\widehat{{\bf f}}_{ijk\ell}^{1111}|^2
     \end{equation}
     for all indices $\{i, j, k, \ell\}$.
     Similarly, we can get 
       \begin{equation}\label{e6}
    |\widehat{{\bf f}}_{ijk\ell}^{0001}|^2=|\widehat{{\bf f}}_{ijk\ell}^{0010}|^2=|\widehat{{\bf f}}_{ijk\ell}^{1101}|^2=|\widehat{{\bf f}}_{ijk\ell}^{1110}|^2.
     \end{equation}
     By the definition of second order orthogonality, we also have 
     \begin{equation}\label{e7}
         \langle\widehat{{\bf f}}_{ij}^{ab}, \widehat{{\bf f}}_{ij}^{cd} \rangle=0
     \end{equation}
for all variables $x_i, x_j$ and $(a, b)\neq (c, d)$.   

Equations~(\ref{e5}), (\ref{e6}) and (\ref{e7}) will be used frequently in the analysis of signatures satisfying {\sc ars} and {\sc 2nd-Orth}.
This is also a reason why we consider the problem
in the setting under the $Z^{-1}$ transformation, $\holant{\neq_2}{\widehat{\mathcal{F}}}$, where we 
can express these consequences of {\sc 2nd-Orth} elegantly,  instead of $\Holant(\mathcal{F})$ which is logically equivalent. 
By combining {\sc 2nd-Orth}  and {\sc ars} of the signature $\widehat{f}$, we  get these
simply expressed, thus easily applicable,
conditions  in terms of  norms  and inner products.

     \section{The Induction Proof: Base Cases $2n \leqslant 4$}\label{sec-induction}
     In this section, we introduce the induction framework and handle the base cases (Lemmas~\ref{lem-2-ary} and \ref{lem-4-ary}).
Recall that $\widehat{\mathcal{O}}$ denotes the set of binary signatures with {\sc ars} and parity (including the binary zero signature), and $\widehat{\mathcal{O}}^{\otimes}$ denotes the set of tensor products of signatures in $\widehat{\mathcal{O}}$.
 Since ${\mathcal{F}}$ does not satisfy  condition \rm{(\ref{main-thr})}, $\widehat{\mathcal{F}}\not\subseteq \mathscr{T}$.
 Also, since  $\widehat{\mathcal{O}}^{\otimes}\subseteq \mathscr{T}$, $\widehat{\mathcal{F}}\not\subseteq \widehat{\mathcal{O}}^{\otimes}$.
Thus, there is a nonzero signature $\widehat{f}\in \widehat{\mathcal{F}}$ of arity $2n$ such that $\widehat{f}\notin \widehat{\mathcal{O}}^{\otimes}$.
  We want to achieve a proof of  \#P-hardness by induction on $2n$. 
  We first consider the base that $2n=2$.
Notice that a nonzero binary signature  $\widehat{f}$ satisfies {\sc 1st-Orth} iff its matrix form (as a 2-by-2 matrix) is orthogonal.
Thus, $\widehat{f}\notin\widehat{\mathcal{O}}$ implies that it does not satisfy {\sc 1st-Orth}.
Then, we have the following result.
 \begin{lemma}\label{lem-2-ary}
 Let $\mathcal F$ contain a binary signature $f\notin \mathcal{O}^{\otimes}$.
 Then, $\Holant(\mathcal{F})$ is \#P-hard.
 
 Equivalently, 
  let $\widehat{\mathcal{F}}$ contain a binary signature $\widehat{f}\notin \widehat{\mathcal{O}}^{\otimes}$.
 Then, $\Holant(\neq_2\mid \widehat{\mathcal{F}})$ is \#P-hard.
 \end{lemma}
 \begin{proof}
 We prove this lemma in the setting of $\Holant(\mathcal{F})$.
 Since ${\mathcal{O}}^{\otimes}$ contains the binary zero signature, $f\notin {\mathcal{O}}^{\otimes}$ implies that $f\not\equiv 0$.
 If $f$ is reducible, then it is a tensor product of two nonzero unary signatures. 
 By Lemma \ref{lin-wang}, we can realize a nonzero unary signature by factorization, and we are done by Theorem \ref{odd-dic}.
 Otherwise, $f$ is irreducible. 
 Since $f\notin \mathcal{O}^{\otimes}$, $f$ does not satisfy {\sc 1st-Orth}. 
 By Lemma \ref{even-first-order},  $\Holant(\mathcal{F})$ is \#P-hard.
 \end{proof}
 
  Then, the general induction framework is  that we start with a signature $\widehat f$ of arity $2n\geqslant 4$ that is not in $\widehat{\mathcal{O}}^{\otimes}$, 
and realize a signature $\widehat g$ of arity $2k \leqslant 2n-2$ that is also not in $\widehat{\mathcal{O}}^{\otimes}$, or otherwise we can directly show $\Holant(\neq_2 \mid \widehat{\mathcal F})$ is {\rm \#}P-hard.
If we can reduce the arity down to $2$ (by a sequence of reductions of length independent of the problem instance size), then we have a binary signature $\widehat{b}\notin\widehat{\mathcal{O}}$. 
By Lemma \ref{lem-2-ary}, 
we are done.

For the inductive step,
we first consider the case that $\widehat{f}$ is reducible.
Suppose that $\widehat{f}=\widehat{f_1} \otimes \widehat{f_2}$.
If $\widehat{f_1}$ 
or $\widehat{f_2}$ have odd arity, then we can realize a signature of odd arity by factorization and we are done.
Otherwise, $\widehat{f_1}$ and $\widehat{f_2}$ have even arity. Since $\widehat{f}\notin \widehat{\mathcal{O}}^{\otimes}$, we know $\widehat{f_1}$ and $\widehat{f_2}$ cannot both be in $\widehat{\mathcal{O}}^{\otimes}$. Then, we can realize a signature of lower arity that is not in $\widehat{\mathcal{O}}^{\otimes}$ by factorization. We are done.
Thus, in the following we may assume that $\widehat f$ is irreducible.
Then, we may further assume that $\widehat f$ satisfies {\sc 2nd-Orth}. Otherwise, we get  \#P-hardness by Lemma \ref{second-ortho}.
We use merging with $\neq_2$ to realize signatures of arity $2n-2$ from $\widehat f$.
 Consider $\widehat{\partial}_{ij}\widehat f$ for all pairs of indices $\{i, j\}$. If there exists a pair $\{i, j\}$ such that $\widehat{\partial}_{ij}\widehat f \notin \widehat{\mathcal{O}}^{\otimes}$, 
then  we can realize
$\widehat g=\widehat{\partial}_{ij}\widehat f$ which has arity $2n-2$, and we are done.
Thus, we may assume $\widehat{\partial}_{ij}\widehat f \in \widehat{\mathcal{O}}^{\otimes}$ for all $\{i, j\}$. 
 We denote this property 
 by $\widehat f\in \widehat{\int}\widehat{\mathcal{O}}^{\otimes}$. We want to achieve our induction proof based on these two properties:  {\sc 2nd-Orth}
 and $\widehat f\in \widehat{\int}\widehat{\mathcal{O}}^{\otimes}$. 
We consider the case that $2n=4$.

\begin{lemma}\label{lem-4-ary}
 Let $\widehat{\mathcal{F}}$ contain a $4$-ary signature $\widehat{f}\notin \widehat{\mathcal{O}}^{\otimes}$.
  Then, $\Holant(\neq_2 \mid \widehat{\mathcal F})$ is {\rm \#}P-hard.
  \end{lemma}
\begin{proof}
Since $\widehat{f}\notin\widehat{{\mathcal{O}}}^{\otimes}$, $f\not\equiv 0$.
First, we may assume that $\widehat f$ is irreducible. 
Otherwise, we can realize a nonzero unary signature or a binary signature that is not in $\widehat{\mathcal{O}}$.
Then, by Theorem \ref{odd-dic} and Lemma \ref{lem-2-ary}, we have  \#P-hardness.
Since $\widehat f$ is irreducible, we may further assume that $\widehat{f}$ satisfies {\sc 2nd-Orth}.
Otherwise, by Lemma \ref{second-ortho}, we get  \#P-hardness. 

We consider binary signatures $\widehat{\partial}_{ij}\widehat{f}$ realized from $\widehat{f}$ by merging using $\neq_2$.
Under the assumption that $\widehat{f}$  satisfies {\sc 2nd-Orth}, we will show that  there exits a pair $\{i, j\}$ such that $\widehat{\partial}_{ij}\widehat{f} \notin\widehat{\mathcal{O}}$.
Then by Lemma \ref{lem-2-ary}, we are done.
For a contradiction, suppose that $\widehat f\in \widehat{\int} \widehat{\mathcal{O}} $ i.e.,  $\widehat{\partial}_{ij}\widehat{f}\in \widehat{\mathcal{O}}$ for all pairs $\{i, j\}$. 
Since $\widehat{f}$ satisfies {\sc 2nd-Orth},
by equations (\ref{e5}) and (\ref{e6}), we have $|\widehat{{\bf f}}_{ijk\ell}^{0000}|=|\widehat{{\bf f}}_{ijk\ell}^{0011}|=|\widehat{{\bf f}}_{ijk\ell}^{1111}|$ 
and $ |\widehat{{\bf f}}_{ijk\ell}^{0001}|=|\widehat{{\bf f}}_{ijk\ell}^{1110}|$ respectively for any permutation $(i, j, k, \ell)$ of $(1, 2, 3, 4)$. Thus  all entries of $\widehat f$ on inputs of even weight $\{0, 2, 4\}$ have the same norm, and all entries of $\widehat f$ on inputs of odd weight $\{1, 3\}$  have the same norm. 
We denote by $\nu_0$ and $\nu_1$
the norm squares of entries on inputs of even weight and odd weight,   respectively.

Then, we consider the equation  $\langle\widehat{{\bf f}}_{12}^{01}, \widehat{{\bf f}}_{12}^{10} \rangle=0$ from (\ref{e7}) by taking $(i, j)=(1, 2)$. We have 
$$\langle\widehat{{\bf f}}_{12}^{01}, \widehat{{\bf f}}_{12}^{10} \rangle=
\widehat f^{0100}\overline{\widehat f^{1000}}+f^{0101}\overline{\widehat f^{1001}}+\widehat f^{0110}\overline{\widehat f^{1010}}+\widehat f^{0111}\overline{\widehat f^{1011}}=0.$$
(Here for clarity, we omitted the subscript $1234$ of $\widehat f_{1234}^{abcd}$.)
By {\sc ars}, we have $\widehat f^{0111}\overline{\widehat f^{1011}}=\overline{\widehat f^{1000}}\widehat f^{0100}$ and $\widehat f^{0110}\overline{\widehat f^{1010}}=\overline{\widehat f^{1001}}\widehat f^{0101}.$ 
Thus, we have 
\begin{equation}\label{e5.71}
\widehat f^{0100}\overline{\widehat f^{1000}}+\widehat f^{0101}\overline{\widehat f^{1001}}=0.
\end{equation}
Note that by taking  norm, $|\widehat f^{0100}\overline{\widehat f^{1000}}| = \nu_1$ and  $|\widehat f^{0101}\overline{\widehat f^{1001}}| = \nu_0$. Then, it follows that $\nu_0 = \nu_1$. Thus, all entries of $\widehat f$ have the same norm. We normalize the norm to be $1$ since $\widehat f\not\equiv 0$.

Consider $\widehat{\partial}_{12}\widehat{f}$. We have
$$\widehat{\partial}_{12}\widehat{f}=(\widehat f^{0100}+\widehat f^{1000}, ~~~\widehat f^{0101}+\widehat f^{1001}, ~~~\widehat f^{0110}+\widehat f^{1010}, ~~~\widehat f^{0111}+\widehat f^{1011}),$$
and by assumption $\widehat{\partial}_{12}\widehat{f} \in \widehat{\mathcal{O}}.$
Thus, at least one of the two entries $\widehat f^{0100}+\widehat f^{1000}$ and $\widehat f^{0101}+\widehat f^{1001}$ is equal to zero.
If $\widehat f^{0100}+\widehat f^{1000}=0$, then we have $$\widehat f^{0100}\overline{\widehat f^{1000}}=(-\widehat f^{1000})\overline{\widehat f^{1000}}=-|\widehat f^{1000}|^2=-1.$$
Then, by equation (\ref{e5.71}), we have $\widehat f^{0101}\overline{\widehat f^{1001}}=1$. Otherwise, $\widehat f^{0101}+\widehat f^{1001}=0$. Then, we have $\widehat f^{0101}\overline{\widehat f^{1001}}=-1$ while $\widehat f^{0100}\overline{\widehat f^{1000}}=1$. Thus, among these two products $\widehat f^{0100}\overline{\widehat f^{1000}}$ and  $\widehat f^{0101}\overline{\widehat f^{1001}}$, exactly one is equal to $1$, while the other is $-1$. Then, we have 
$$\widehat f^{0100}\overline{\widehat f^{1000}}\widehat f^{0101}\overline{\widehat f^{1001}}=-1.$$
Similarly, by considering $\widehat{\partial}_{23}\widehat{f}$ and $\widehat{\partial}_{31}\widehat{f}$ respectively, we have 
$$\widehat f^{0010}\overline{\widehat f^{0100}}\widehat f^{0011}\overline{\widehat f^{0101}}=-1, ~~~~\text{ and }~~~~ \widehat f^{1000}\overline{\widehat f^{0010}}\widehat f^{1001}\overline{\widehat f^{0011}}=-1.$$
Multiply these three products, we have 
$$|\widehat f^{0100}|^2|\widehat f^{0010}|^2|\widehat f^{1000}|^2|\widehat f^{0101}|^2|\widehat f^{0011}|^2|\widehat f^{1001}|^2=(-1)^3=-1.$$
A contradiction!
\end{proof}
\begin{remark}
In this proof, we showed that there is no irreducible 4-ary signature $\widehat{f}$ that satisfies both   {\sc 2nd-Orth} and $\widehat f\in \widehat{\int}\widehat{\mathcal{O}}^{\otimes}$.
\end{remark}

If Lemma~\ref{lem-4-ary} were to hold 
for signatures of arity $2n\geqslant 6$, i.e., there is no irreducible signature  $\widehat f$ of $2n\geqslant 6$ such that $\widehat f$ satisfies both {\sc 2nd-Orth} and $\widehat f\in \widehat{\int}\widehat{\mathcal{O}}^{\otimes}$, then the induction proof holds and we are done.
We show that this is true for signatures of arity $2n\geqslant 10$ in Section \ref{sec-10-ary}.
However, 
there are extraordinary signatures  of arity $6$ and $8$ with special closure properties (Bell properties) such that they  satisfy both {\sc 2nd-Orth} and $\widehat f\in \widehat{\int}\widehat{\mathcal{O}}^{\otimes}$.

\section{First Major Obstacle: 6-ary Signatures with Bell Property}\label{sec-f6}
We consider the following 6-ary signature $\widehat{f}_6$. We use $\chi_S$ to denote the indicator function on set $S$.
Let 
 $$\widehat{f_6}=\chi_S \cdot {(-1)^{x_1x_2+x_2x_3+x_1x_3+x_1x_4+x_2x_5+x_3x_6}}$$ where $S =\mathscr{S}(\widehat{f_6}) =\mathscr{E}_6= \{\alpha\in \mathbb{Z}_2^6\mid {\rm wt}(\alpha) \equiv 0 \bmod 2\}.$
 One can check that $\widehat{f_6}$ is irreducible, and
$\widehat{f_6}$ satisfies both {\sc 2nd-Orth} and  $\widehat f\in \widehat{\int}\widehat{\mathcal{O}}^{\otimes}$.
 $\widehat{f_6}$ has the following matrix form 
 \begin{equation}\label{eqn:definiton-f6}
 M_{123,456}(\widehat{f_6})=\left[\begin{matrix}
1 & 0 & 0 & 1 & 0 & 1 & 1 & 0\\
0 & -1 & 1 & 0 & 1 & 0 & 0 & -1\\
0 & 1 & -1 & 0 & 1 & 0 & 0 & -1\\
-1 & 0 & 0 & -1 & 0 & 1 & 1 & 0\\
0 & 1 & 1 & 0 & -1 & 0 & 0 & -1\\
-1 & 0 & 0 & 1 & 0 & -1 & 1 & 0\\
-1 & 0 & 0 & 1 & 0 & 1 & -1 & 0\\
0 & 1 & 1 & 0 & 1 & 0 & 0 & 1\\
\end{matrix}\right].
\end{equation}
We use Figure \ref{fig-arity-6} to visualize this matrix.  
A block with orange color denotes an entry $+1$ and a block with blue color denotes an entry $-1$. Other blank blocks are zeros.

\begin{figure}[!htp]
    \centering
    \includegraphics[height=2in]{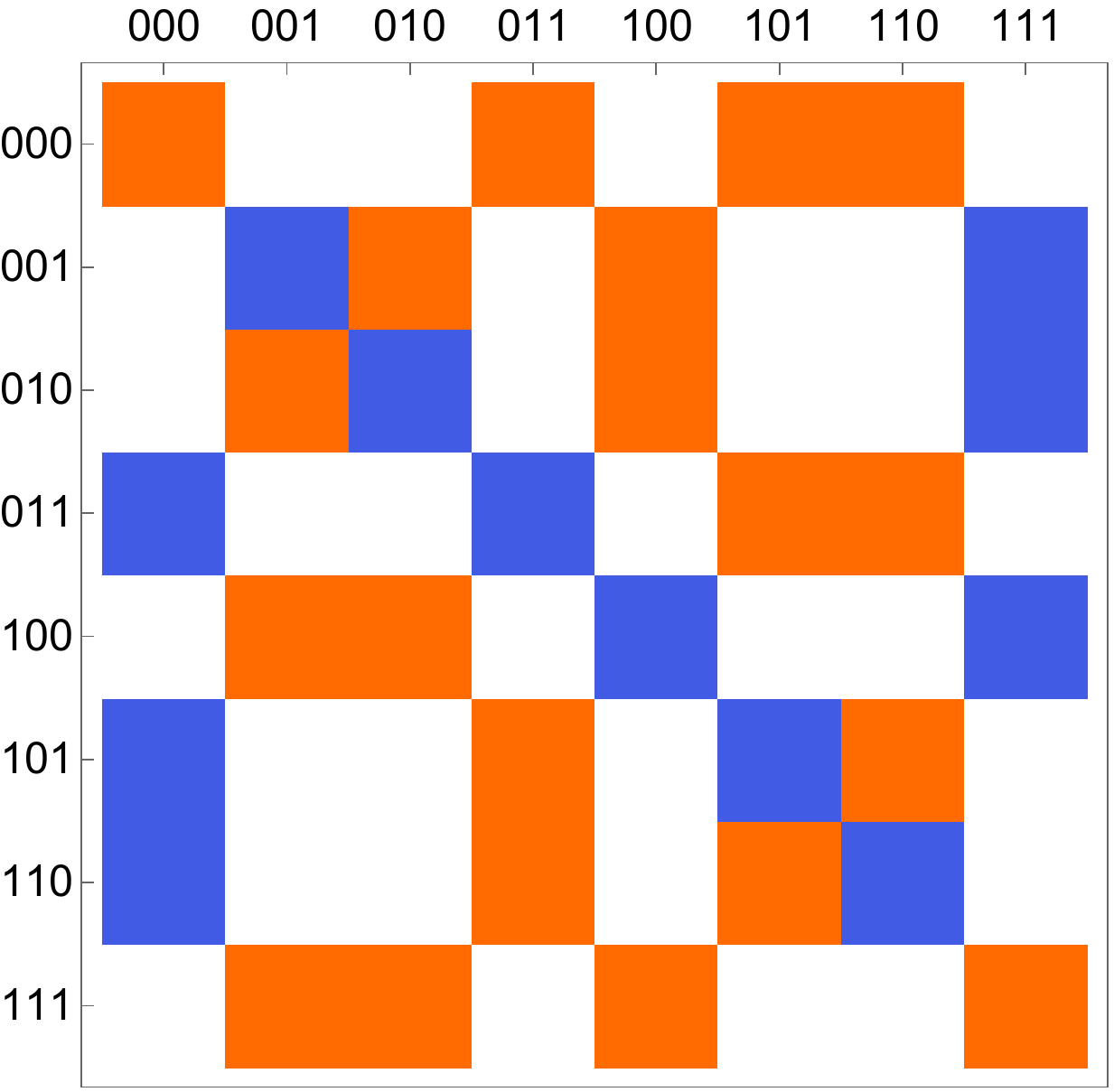}
    \caption{A visualization of $\widehat{f_6}$}
    \label{fig-arity-6}
\end{figure}

\subsection{The discovery of $\widehat{f_6}$}
In this subsection, we show how this extraordinary signature $\widehat{f_6}$ was discovered. 
We prove that if $\widehat{\mathcal{F}}$ contains a 6-ary signature $\widehat{f}$ where $\widehat{f}\notin\widehat{\mathcal{O}}{^\otimes}$, then $\holant{\neq_2}{\widehat{\mathcal{F}}}$ is \#P-hard or $\widehat{f_6}$ is realizable from $\widehat{f}$ after a holographic transformation by some $\widehat{Q}\in \widehat{{\bf O}_2}$ (Theorem~\ref{lem-6-ary-f6}).
The general strategy of this proof is to show that we can realize signatures with special properties from $\widehat{f}$ step by step (Lemmas~\ref{lem-arity6-1}, \ref{lem-arity6-2}, \ref{lem-arity6-3} and \ref{lem-arity6-4}), and finally we can realize $\widehat{f}_6$, or else we can realize signatures that lead to  \#P-hardness. So this $\widehat{f}_6$
emerges as essentially the unique
(and true) obstacle to our proof of \#P-hardness in this setting.

\begin{lemma}\label{lem-arity6-1}
Suppose that  $\widehat{\mathcal{F}}$ contains a $6$-ary signature $ \widehat{f}\notin \widehat{\mathcal{O}}{^\otimes}$.
Then, $\holant{\neq_2}{\widehat{\mathcal{F}}}$ is \#P-hard, or an irreducible 6-ary signature $\widehat{f'}$  is realizable from $\widehat{f}$, where  $\widehat{f'}(\alpha)=0$ for all $\alpha$ with ${\rm wt}(\alpha)=2$ or $4$.
Moreover,  $\widehat{f'}$ is realizable by extending variables of $\widehat{f}$ with nonzero binary signatures in $\widehat{\mathcal{O}}$ that are  realizable
 by factorization from $\widehat{\partial}_{12}\widehat{f}.$
\end{lemma}

\begin{proof}
Since $\widehat{f}\notin \widehat{\mathcal{O}}^{\otimes}$, $\widehat f\not\equiv 0$.
Again, we may assume that $\widehat{f}$ is irreducible.
Otherwise, by factorization, we can realize a nonzero signature of odd arity, or a signature of arity $2$ or $4$ that is not in $\widehat{\mathcal{O}}^{\otimes}$.
Then by Theorem~\ref{odd-dic}, or Lemmas~\ref{lem-2-ary} or \ref{lem-4-ary}, we get \#P-hardness.
 Under the assumption that $\widehat{f}$ is irreducible, we may further assume that $\widehat{f}$ satisfies {\sc 2nd-Orth} by Lemma~\ref{second-ortho}.
Also, we may assume that $\widehat{f}\in \int\widehat{\mathcal{O}}^{\otimes}$.
Otherwise, there is a pair of indices $\{i, j\}$ such that the 4-ary signature $\widehat{\partial}_{ij}\widehat{f}\notin \widehat{\mathcal{O}}^{\otimes}$. 
Then by Lemma \ref{lem-4-ary}, $\holant{\neq_2}{\widehat{\mathcal{F}}}$ is \#P-hard.

If for all pairs of indices $\{i, j\}$,  $\widehat{\partial}_{ij}\widehat{f}\equiv 0$, then by Lemma \ref{lem-zero_2}, we have ${\widehat f}(\alpha)=0$ for all $\alpha$ with ${\rm wt}(\alpha)\neq 0$ and $6$. 
Since $f\not\equiv 0$, clearly such a signature  does not satisfy {\sc 2nd-Orth}. Contradiction.
Otherwise, 
there is a pair of indices $\{i, j\}$ such that $\widehat{\partial}_{ij}\widehat{f}\not\equiv 0$. 
By renaming variables, without loss of generality, we assume that $\widehat{\partial}_{12}\widehat{f}\not\equiv 0$.
Since $\widehat{\partial}_{12}\widehat{f}\in \widehat{\mathcal{O}}^{\otimes}$,  in the UPF of $\widehat{\partial}_{12}\widehat{f}$, by renaming variables we assume that variables $x_3$ and $x_4$ appear in one  nonzero binary signature $\widehat{b_1}(x_3, x_4)\in \widehat{\mathcal{O}}^{\otimes}$, and variables $x_5$ and $x_6$ appear in the other nonzero binary  signature $\widehat{b_2}(x_5, x_6) \in \widehat{\mathcal{O}}^{\otimes}$.
Thus, we have 
\begin{equation*}
\widehat{\partial}_{12}\widehat{f}=  \widehat{b_1}(x_3, x_4)\otimes \widehat{b_2}(x_5, x_6)\not\equiv 0.
\end{equation*}

By Lemma \ref{lin-wang}, we know that  these two binary signatures $\widehat{b_1}$ and $\widehat{b_2}$ are realizable by factorization. 
Note that for a nonzero binary signature $\widehat{b_i}(x_{2i+1}, x_{2i+2})\in \widehat{\mathcal{O}}$ $(i\in \{1, 2\})$, if we connect the variable $x_{2i+1}$ of two copies of $\widehat{b_{i}}(x_{2i+1}, x_{2i+2})$ using $\neq_2$ (mating two binary signatures), then we get $\neq_2$ up to a nonzero scalar.
We consider the following gadget construction $G_1$ on $\widehat{f}$. 
Recall that in the setting of $\holant{\neq}{\widehat{\mathcal{F}}}$, variables are connected using $\neq_2$.
For $i\in \{1, 2\}$, by a slight abuse of variable names, we connect the variable $x_{2i+1}$ of $\widehat{f}$ 
with the variable $x_{2i+1}$ of $\widehat{b_{i}}(x_{2i+1}, x_{2i+2})$. 
We get a signature $\widehat{f'}$ of arity $6$. 
Such a gadget construction does not change the irreducibility of $f$. Thus, $\widehat{f'}$ is irreducible.
Again, we may assume that $\widehat{f'}\in \widehat{\int}\widehat{\mathcal{O}}^{\otimes}$ and  $\widehat{f'}$ satisfies {\sc 2nd-Orth}. Otherwise, we are done.

Consider $\widehat{\partial}_{12}\widehat{f'}$. 
Since the above gadget construction $G_1$ does not touch variables $x_1$ and $x_2$ of $f$, the operation of forming $G_1$
commutes with the merging operation
$\widehat{\partial}_{12}$.
Thus, $\widehat{\partial}_{12}\widehat{f'}$ can be realized by performing the gadget construction $G_1$ on $\widehat{\partial}_{12}\widehat{f}$, which connects each binary signature $\widehat{b_{i}}$ $(i\in \{1, 2\})$ of $\widehat{\partial}_{12}\widehat{f}$ with another copy of itself using $\neq_2$ (in the mating fashion). 
Then, each $\widehat{b_{i}}$ in $\widehat{\partial}_{12}\widehat{f}$ is changed to $\neq_2$ up to a nonzero real scalar. 
After normalization and renaming variables, we have 
\begin{equation*}
\widehat{\partial}_{12}\widehat{f'}= (\neq_2)(x_3, x_4)\otimes (\neq_2)(x_5, x_6).
\end{equation*}
Since $\widehat{\partial}_{12}\widehat{f'}\in {\mathcal{D}}^{\otimes}$, for any $\{i, j\}$ disjoint with $\{1, 2\}$  we have $\widehat{\partial}_{(ij)(12)}\widehat{f'}\in {\mathcal{D}}^{\otimes}$, and hence $\widehat{\partial}_{ij}\widehat{f'}\not\equiv 0$.

Now, we show that for all pairs of indices $\{i, j\}$, $\widehat{\partial}_{ij}\widehat{f'}$ has even parity. 
We first consider the case that $\{i, j\}$ is disjoint with $\{1,2\}$. 
Connect variables $x_i$ and $x_j$ of $\widehat{\partial}_{12}\widehat{f'}$ using $\neq_2$. 
Since $\widehat{\partial}_{12}\widehat{f'}$ has even parity, a merging gadget using $\neq_2$ will change the parity from even to odd. 
Thus,
 $\widehat\partial_{(ij)(12)}\widehat{f'}$ has odd parity. 
Consider $\widehat\partial_{ij}\widehat{f'}$. 
Remember that $\widehat\partial_{ij}\widehat{f'}\not\equiv 0$ since $\widehat{\partial}_{(ij)(12)}\widehat{f'}\not\equiv 0$.
Since  $\widehat{f'}\in \widehat{\int}\widehat{\mathcal{O}}^{\otimes}$,   We have $\widehat\partial_{ij}\widehat{f'} \in \mathcal{O}^{\otimes}$.
Thus, $\widehat\partial_{ij}\widehat{f'}$ has (either odd or even) parity. 
For a contradiction, suppose that it has odd parity.
Then,  $\widehat\partial_{(12)(ij)}\widehat{f'}$ has even parity since it is realized by merging using $\neq_2$.
A  signature that has  both even parity and odd parity
is identically zero. Thus $\widehat\partial_{(12)(ij)}\widehat{f'}$ is the  zero signature.
 However, since $\widehat{\partial}_{(ij)(12)}\widehat{f'} \in \mathcal{D}^{\otimes}$,
 it is not  the  zero signature. 
Contradiction. 
Therefore, $\widehat{\partial}_{ij}\widehat{f'}$ has even parity for all $\{i,j\}$ disjoint with $\{1,2\}$. 

Then, consider $\widehat{\partial}_{ij}\widehat{f'}$ for $\{i, j\}\cap\{1, 2\}\neq \emptyset$. 
If $\{1, 2\}=\{i, j\}$, then clearly, $\widehat\partial_{12}\widehat{f'}$ has even parity. 
Otherwise,
without loss of generality, 
we may assume that $i=1$ and $j\neq 2$.
Consider $\widehat\partial_{1j}\widehat{f'}$ for $3\leqslant j \leqslant 6$. 
If it is a zero signature, then it has even parity.
Otherwise, $\widehat\partial_{1j}\widehat{f'}\not\equiv 0$.
Since $\widehat{\partial}_{1j}\widehat{f'}\in \widehat{\mathcal{O}}^{\otimes}$, we assume that it has the following UPF
$$\widehat\partial_{1j}\widehat{f'}=\widehat{b'_1}(x_2, x_u)\otimes \widehat{b'_2}(x_v, x_w).$$
By connecting variables $x_u$ and $x_v$ of $\widehat{\partial}_{1j}\widehat{f'}$ using $\ne_2$, 
we get $\widehat\partial_{(uv)(1j)}\widehat{f'}$.
Since the merging gadget connects two nonzero binary signatures in $\widehat{\mathcal{O}}$, the resulting signature is a nonzero binary signature. 
Thus, $\widehat\partial_{(uv)(1j)}\widehat{f'}\not\equiv 0$.
Notice that $\{u, v\}$ is disjoint with $\{1, 2\}$. As showed above, $\widehat\partial_{uv}\widehat{f'}$ has even parity.
Then,  $\widehat\partial_{(1j)(uv)}\widehat{f'}$ has odd parity. 
For a contradiction, suppose that $\widehat\partial_{1j}\widehat{f'}$ has odd parity.
Then $\widehat\partial_{(uv)(1j)}\widehat{f'}$ has even parity. 
But a  nonzero signature $\widehat\partial_{(uv)(1j)}\widehat{f'}$ cannot have both even parity and odd parity. Contradiction. 
Thus,  $\widehat\partial_{1j}\widehat{f'}$ has even parity.

We have proved that $\widehat{\partial}_{ij}\widehat{f'}$ has even parity for all pairs of indices $\{i, j\}$.
In other words, for all pairs of indices $\{i, j\}$ and all $\beta\in \mathbb{Z}_2^4$ with ${\rm wt}(\beta)=1$ or $3$, we have
$(\widehat{\partial}_{ij}\widehat{f'})(\beta)=0$.
Then, by Lemma~\ref{lem-zero_2},
$\widehat{f'}(\alpha)=0$ for all $\alpha$ with ${\rm wt}(\alpha)=2$ or $4$. 
Clearly,
$\widehat{f'}$ is realized by extending  $\widehat{f}$ with nonzero binary signatures in $\widehat{\mathcal{O}}$ that are   realized
 by factorization from $\widehat{\partial}_{12}\widehat{f}.$
\end{proof}

\begin{lemma}\label{lem-arity6-2}
Suppose that $\widehat{\mathcal{F}}$ contains an irreducible $6$-ary  signature $\widehat{f'}$ where
$\widehat{f'}(\alpha)=0$ for all $\alpha$ with ${\rm wt}(\alpha)=2$ or $4$. Then,
  $\holant{\neq_2}{\widehat{\mathcal{F}}}$ is \#P-hard, or
 $\mathscr{S}(\widehat{f'})=\mathscr{O}_6=\{\alpha\in \mathbb{Z}^6_2\mid {\rm wt}(\alpha) \text{ is } \rm odd\}$ and all nonzero entries of $\widehat{f'}$ have the same norm.
\end{lemma}

\begin{proof}


Since $\widehat{f'}$ is irreducible, again  we may assume that $\widehat{f'}$ satisfies {\sc 2nd-Orth} and $\widehat{f'}\in 
\widehat\int \mathcal{\widehat{O}}^{\otimes}.$
Let $\{i, j, k, \ell\}$ be an arbitrarily chosen subset of indices from $\{1, \ldots, 6\}$, and $\{m, n\}$ be the other two indices.
Then by equation (\ref{e6}), and the condition that $\widehat{f'}$ 
vanishes at weight 2 and 4,
we have 
\begin{equation}\label{eqn-7.3-1}
 |\widehat{{\bf f'}}_{ijk\ell}^{0001}|^2=|\widehat{f'}^{000100}_{ijk\ell mn}|^2+|\widehat{f'}^{000111}_{ijk\ell mn}|^2=|\widehat{f'}^{001000}_{ijk\ell mn}|^2+|\widehat{f'}^{001011}_{ijk\ell mn}|^2=|\widehat{{\bf f'}}_{ijk\ell}^{0010}|^2.
 \end{equation}
Also, by considering indices $\{k, \ell, m, n\}$, we have 
\begin{equation}\label{eqn-7.3-2}
 |\widehat{{\bf f'}}_{k\ell mn}^{0100}|^2=|\widehat{f'}^{000100}_{ijk\ell mn}|^2+|\widehat{f'}^{110100}_{ijk\ell mn}|^2=|\widehat{f'}^{001000}_{ijk\ell mn}|^2+|\widehat{f'}^{111000}_{ijk\ell mn}|^2=|\widehat{{\bf f'}}_{k \ell mn}^{1000}|^2.
 \end{equation}
By {\sc ars}, we have 
\begin{equation}\label{eqn-7.3-3}
|\widehat{f'}^{000111}_{ijk\ell mn}|^2=|\widehat{f'}^{111000}_{ijk\ell mn}|^2,
\end{equation}
and 
\begin{equation}\label{eqn-7.3-4}
|\widehat{f'}^{001011}_{ijk\ell mn}|^2=|\widehat{f'}^{110100}_{ijk\ell mn}|^2.
\end{equation}
By calculating (\ref{eqn-7.3-1}) $+$ (\ref{eqn-7.3-2}) $-$ (\ref{eqn-7.3-3}) $-$ (\ref{eqn-7.3-4}), we have 
\begin{equation}\label{eqn-7.3-5}
|\widehat{f'}^{000100}_{ijk\ell mn}|^2=|\widehat{f'}^{001000}_{ijk\ell mn}|^2.
\end{equation}
By (\ref{eqn-7.3-1}) $-$ (\ref{eqn-7.3-5}), we have 
\begin{equation}\label{eqn-7.3-6}
|\widehat{f'}^{000111}_{ijk\ell mn}|^2=|\widehat{f'}^{001011}_{ijk\ell mn}|^2.
\end{equation}
From (\ref{eqn-7.3-5}),
since the indices $(i,j,k,\ell,m,n)$ can be an arbitrary permutation of $(1,2,3,4,5,6)$, 
for all $\alpha$, $\beta\in \mathbb{Z}_2^6$ with ${\rm wt}(\alpha)={\rm wt}(\beta)=1$, we have $|\widehat{f'}(\alpha)|=|\widehat{f'}(\beta)|$.
The same statement holds for
${\rm wt}(\alpha)={\rm wt}(\beta)=3$, by (\ref{eqn-7.3-6}).

Let $a=|\widehat{f'}(\vec 0^{6})|$;
by {\sc ars}, $a=|\widehat{f'}(\vec 1^{6})|$ as well. 
It is the norm of entries of $\widehat{f'}$ on input of Hamming weight $0$ and $6$.
We use $b$ to denote  the norm of entries of $\widehat{f'}$ on inputs of Hamming weight $1$.
By {\sc ars}, $b$ is also the norm of entries of $\widehat{f'}$ on inputs of Hamming weight $5$.
We use $c$ to denote the norm of entries of $\widehat{f'}$ on inputs of Hamming weight $3$.
Remember that by assumption, $|\widehat{f'}(\alpha)|=0$ if ${\rm wt}(\alpha)=2$ or $4$.

By equation (\ref{e5}), we have 
$$ |\widehat{{\bf f'}}_{1234}^{0000}|^2=a^2+2b^2=|\widehat{{\bf f'}}_{1234}^{0011}|^2=2c^2.$$
Clearly, we have $0 \leqslant a,b\leqslant c$. If $c=0$, then $a=b=0$ which implies that $\widehat{f'}$ is a zero signature. This is a contradiction since $\widehat{f'}$ is irreducible. 
Therefore $c\neq 0$. We  normalize  $c$ to $1$. Then
\[a^2 + 2b^2 = 2.\]
 We will show that $b=1$ and $a=0$. 
 This will finish the proof of the lemma.
 For a contradiction, suppose that $b<1$, then we also have  $a>0$.

Consider signatures $\widehat{{f'}}_{12}^{01}$, $\widehat{{f'}}_{12}^{10}$ and $\widehat{\partial}_{12}\widehat{f'}=\widehat{{f'}}_{12}^{01}+\widehat{{f'}}_{12}^{10}$.
Since $\widehat{f'}(\alpha)=0$ for all $\alpha$ with ${\rm wt}(\alpha)= 2 \text{ or } 4$, $\widehat{f'}^{01}_{12}(\beta)=0$ and $\widehat{f'}^{10}_{12}(\beta)=0$ for all $\beta$ with ${\rm wt}(\beta) = 1$ or $3$. 
Thus, $\widehat{{f'}}_{12}^{01}$ and $\widehat{{f'}}_{12}^{10}$ have even parity.
We also consider the complex inner product $\langle\widehat{{\bf f'}}_{12}^{01}, \widehat{{\bf f'}}_{12}^{10}\rangle$. 
First we build the following table.

\begin{table}[!h]
\renewcommand{\arraystretch}{2}
\centering
\begin{tabular}{ |c|c|c| c| c| c| c| c| c| } 
 \hline
$\widehat{{f'}}_{12}^{01}$ & $\widehat{f'}^{010000}$ & $\widehat{f'}^{010011}$ & $\widehat{f'}^{010101}$ & $\widehat{f'}^{010110}$ & $\widehat{f'}^{011001}$  & $\widehat{f'}^{011010}$ & $\widehat{f'}^{011100}$ & $\widehat{f'}^{011111}$\\ 
  \hline
$\widehat{{f'}}_{12}^{10}$ & $\widehat{f'}^{100000}$ & $\widehat{f'}^{100011}$ & $\widehat{f'}^{100101}$ & $\widehat{f'}^{100110}$ & $\widehat{f'}^{101001}$  & $\widehat{f'}^{101010}$ & $\widehat{f'}^{101100}$ & $\widehat{f'}^{101111}$\\ 
 \hline
 $\widehat{\partial}_{12}\widehat{f'}$ & $s_1$ & $s_2$ & $s_3$ & $s_4$ & $\overline{s_4}$ &  $\overline{s_3}$ &  $\overline{s_2}$ &  $\overline{s_1}$ \\ 
 \hline
 $\langle\widehat{{\bf f'}}_{12}^{01}, \widehat{{\bf f'}}_{12}^{10}\rangle$ & $p_1$ & $p_2$ & $p_3$ & $p_4$ & $p_4$ & $p_3$ & $p_2$ & $p_1$\\ 
 \hline
\end{tabular}
\caption{Entries of $\widehat{{f'}}_{12}^{01}$, $\widehat{{f'}}_{12}^{10}$, $\widehat{\partial}_{12}\widehat{f'}$ and pairwise product terms in $\langle\widehat{{\bf f'}}_{12}^{01}, \widehat{{\bf f'}}_{12}^{10}\rangle$ on even-weighed inputs}\label{table-12neq}
\end{table}


In Table \ref{table-12neq}, we call these four rows by Row 1, 2, 3 and 4 respectively and these nine columns by Column 0, 1, \ldots and 8 respectively. 
We use $T_{i,j}$ to denote the cell in Row $i$ and Column $j$.  Table \ref{table-12neq} is built as follows. 
\begin{itemize}
    \item 
In Row 1 and Row 2, we list the entries of signatures $\widehat{{ f'}}_{12}^{01}$ and $\widehat{{ f'}}_{12}^{10}$ that are on even-weighted inputs (excluding the first two bits that are pinned) respectively. 
Note that, those that did not appear are  0 entries  on odd-weighted inputs (excluding the first two bits that are pinned) of the signatures $\widehat{ f'}_{12}^{01}$ and $\widehat{{ f'}}_{12}^{10}$, since $\widehat{{ f'}}_{12}^{01}$ and $\widehat{{ f'}}_{12}^{10}$ have even parity. 
\item In Row 3, we list the corresponding entries of the signature $\widehat{\partial}_{12}\widehat{f'}=\widehat{{ f'}}_{12}^{01}+\widehat{{ f'}}_{12}^{10}$, i.e., $T_{3,j}=T_{1,j}+T_{2,j}$ for $1\leqslant j \leqslant 8$.
\item In Row 4, we list the corresponding  items in the complex inner product $\langle\widehat{{\bf f}'}_{12}^{01}, \widehat{{\bf f}'}_{12}^{10}\rangle$, i.e., $T_{4,j}=T_{1,j}\cdot \overline{T_{2,j}}$ for $1\leqslant j \leqslant 8$.
\end{itemize}
For $1\leqslant j\leqslant 8$, we consider the  entry in $T_{1,j}$ and  the entry in $T_{2, 9-j}$.
By {\sc ars}, we have $T_{1,j}=\overline{T_{2, 9-j}}$
because their corresponding inputs
are complement of each other.
Thus,  $$T_{3,j}=T_{1,j}+T_{2,j}=\overline{T_{2, 9-j}}+\overline{T_{1, 9-j}}=\overline{T_{3, 9-j}},$$ and 
$$T_{4,j}=T_{1,j}\cdot \overline{T_{2,j}}=\overline{T_{2,9-j}}\cdot T_{2, 9-j}=T_{4, 9-j}.$$
We use $s_1, \ldots, s_4$ to denote the values in $T_{3,1}, \ldots, T_{3, 4}$ and $p_1, \ldots, p_4$ to denote the values in $T_{4,1}, \ldots, T_{4, 4}.$
Correspondingly, the values in $T_{3,5}, \ldots, T_{3, 8}$ are $\overline{s_4}, \ldots, \overline{s_1}$ and the values in $T_{4,5}, \ldots, T_{4,8}$ are $p_4, \ldots, p_1$.
We also use $x_j$ and $y_j$ $(1\leqslant j\leqslant 8)$ to denote the entries in $T_{1,j}$ and $T_{2,j}$ respectively.

 By {\sc 2nd-Orth}, we have $\langle\widehat{{\bf f'}}_{12}^{01}, \widehat{{\bf f'}}_{12}^{10}\rangle=2(p_1+p_2+p_3+p_4)=0$. 
 Also we have $|p_1|=b^2$ and $|p_2|=|p_3|=|p_4|=1$. 
Notice the fact that if $x_i+y_i=0$, then $x_i\cdot \overline{y_i}=x_i\cdot \overline{-x_i}=-|x_i|^2=-|x_i\cdot \overline{y_i}|$.
Thus,  if $s_1=0$ then $p_1=-|p_1|=-b^2$ and  for any $i=2, 3, 4$, if $s_i=0$  then $p_i=-1$.
Note that $\widehat{\partial}_{12}\widehat{f'}(\beta)=\widehat{{ f'}}_{12}^{01}(\beta)+\widehat{{ f'}}_{12}^{10}(\beta)=0$ for all $\beta$ with ${\rm wt}(\beta)=1$ or $3$.
Among all $16$ entries of  $\widehat{\partial}_{12}\widehat{f'}$, $s_1, \ldots, s_4, \overline{s_4}, \ldots, \overline{s_1}$ are those that are possibly nonzero.
Since $\widehat{\partial}_{12}\widehat{f'}\in \widehat{\mathcal{O}}^{\otimes}$, it has support of size either $4$ or $0$. 
Thus, among $s_1, s_2, s_3$ and $s_4$, either  exactly two of them are zero or they are all zero. 
There are three possible cases.

\begin{itemize}
    \item $s_1=s_2=s_3=s_4=0$. Then $p_1+p_2+p_3+p_4=-b^2-3\leqslant -3 \ne 0$. Contradiction. 
    \item $s_1\neq 0$ and two of $s_2, s_3$ and $s_4$ are zero. Without loss of generality, we may assume that $s_2=s_3=0$. Then $p_2=p_3=-1$.
    Since $p_1+p_2+p_3+p_4=0$, we have $p_1+p_4=-p_2-p_3=2$.
    Then, $2=|p_1+p_4|\leqslant|p_1|+|p_4|=b^2+1<2$. Contradiction.
    \item $s_1=0$ and one of $s_2, s_3$ and $s_4$ is zero. Without loss of generality, we may assume that $s_2=0$. 
    Then $p_1=-b^2$ and $p_2=-1$. Thus, $p_3+p_4=-p_1-p_2=1+b^2<2$. 
   Let $\theta=\arccos{\frac{1+b^2}{2}}$. We know that $0<\theta<\frac{\pi}{2}$. 
   Recall that $|p_3|=|p_4|=1$. Thus, $p_3= e^{\pm\ii\theta}$ and $p_4= e^{\mp\ii\theta}$ (and $p_3=\overline{p_4}$).
\end{itemize}
Let $P=\{-1, e^{\ii\theta}, e^{-\ii\theta}\}$. 
Thus,  $p_2, p_3, p_4\in P$. 
Otherwise, we get a contradiction.

Now, we consider signatures  $\widehat{\partial}_{ij}\widehat{f'}$ for all pairs of indices $\{i, j\}$. 
By symmetry, the same conclusion holds.
In other words, let $\{i, j\}$ be an arbitrarily chosen pair of indices  from $\{1, \ldots, 6\}$ and $\{k, \ell, m, n\}$ be the other four indices, and 
let $\beta\in \mathbb{Z}_2^4$ be an assignment on variables $(x_k, x_\ell, x_m, x_n)$ with ${\rm wt}(\beta)=2$. 
Then, we have  $\widehat{f'}_{ijk\ell mn}^{01\beta}\cdot \overline{\widehat{f'}_{ijk\ell mn}^{10\beta}}\in P.$
Since the indices $(i,j,k,\ell,m,n)$ can be an arbitrary permutation of $(1,2,3,4,5,6)$, we have $\widehat{f'}(\alpha)\cdot\overline{\widehat{f'}({\alpha'})}\in P$ for any two assignments $\alpha$ and $\alpha'$ on the six variables 
where ${\rm wt}(\alpha)={\rm wt}(\alpha')=3$ and ${\rm wt}(\alpha\oplus\alpha')=2$,
because for any such two strings  $\alpha$ and $\alpha'$, there exist two bit positions on which $\alpha$ and $\alpha'$ take values
01 and 10 respectively.

We consider the following three inputs $\alpha_1=100011$, $\alpha_2=010011$ and $\alpha_3=001011$ of $\widehat{f'}$. We have   $\widehat{f'}({\alpha_1})\cdot\overline{\widehat{f'}({\alpha_2})}=q_{12}\in P$, $\widehat{f'}({\alpha_2})\cdot\overline{\widehat{f'}({\alpha_3})}=q_{23}\in P$ and $\widehat{f'}({\alpha_1})\cdot\overline{\widehat{f'}({\alpha_3})}=q_{13}\in P.$
Recall that  $|\widehat{f'}({\alpha_2})|=1$ since ${\rm wt}(\alpha_2)=3$.
Then, $$q_{12}\cdot q_{23}=\widehat{f'}({\alpha_1})\cdot\overline{\widehat{f'}({\alpha_2})}\cdot \widehat{f'}({\alpha_2})\cdot\overline{\widehat{f'}({\alpha_3})}=|\widehat{f'}({\alpha_2})|^2\cdot\widehat{f'}({\alpha_1})\cdot\overline{\widehat{f'}({\alpha_3})}=q_{13}\in P.$$
However, since $0<\theta<\frac{\pi}{2}$,  it is easy to check that for any two (not necessarily distinct) elements in $P$, their product is not in $P$. 
Thus, we get a contradiction.
This proves that
$b=c=1$ and $a=0$.

Therefore we have proved that, $\mathscr{S}(\widehat{f'})=\mathscr{O}_6$,
and all its nonzero entries have the same norm that is normalized to $1$.
\end{proof}
\begin{lemma}\label{lem-arity6-3}
Suppose that $\widehat{\mathcal{F}}$ contains an irreducible $6$-ary signature $\widehat{f'}$ where $\mathscr{S}(\widehat f')=\mathscr{O}_6$ and $|\widehat f'(\alpha)|=1$ for all $\alpha\in\mathscr{S}(\widehat f')$. Then, 
  $\holant{\neq_2}{\widehat{\mathcal{F}}}$ is \#P-hard, or
   after a holographic transformation by some $\widehat{Q}=\left[\begin{smallmatrix}
{\overline \rho} & 0\\
0 & \rho\\
\end{smallmatrix}\right] \in\widehat{{\bf O}_2}$ where $\rho=e^{\ii\delta}$ and $0\leqslant \delta <\pi/2$, an irreducible 6-ary signature $\widehat{f''}$ and $=_2$ are realizable from $\widehat{f'}$ where $\mathscr{S}(\widehat f'')=\mathscr{O}_6$ and there  exists $\lambda=1$ or $\ii$ such that for all $\alpha\in \mathscr{S}(\widehat{f''})$,  $\widehat{f''}(\alpha)=\pm \lambda$, i.e.,  $\holant{\neq_2}{=_2, \widehat{f''},  \widehat{Q}\widehat{\mathcal{F}}}\leqslant_T\holant{\neq_2}{\widehat{\mathcal{F}}}$ where 
   $\widehat{f''}=\widehat{Q}\widehat{f'}$.
   Moreover, 
the nonzero binary signature $(\rho^2, 0, 0, \overline{\rho^2}) \in \widehat{\mathcal{O}}$ is realizable from $\widehat{\partial}_{ij}\widehat{f'}$ for some $\{i, j\}$.
\end{lemma}

\begin{proof}
Again,  we may assume that $\widehat{f'}$ satisfies {\sc 2nd-Orth} and $\widehat{f'}\in 
\widehat\int \mathcal{\widehat{O}}^{\otimes}.$
We first show that there  exists $\lambda=1$ or $\ii$ such that
for all $\alpha\in \mathscr{S}(\widehat{f'})$ with ${\rm wt}(\alpha)=3$, $\widehat{f''}(\alpha)=\pm \lambda$, or else we get  \#P-hardness.

 Let's revisit  Table \ref{table-12neq}. 
 Now we have $|p_1|=|p_2|=|p_3|=|p_4|=1$. 
Recall that for $1\leqslant i\leqslant 4$,  $s_i=0$ implies that $p_i=-1$. 
 Since $\widehat\partial_{12}\widehat{f'}\in \widehat{\mathcal{O}}^{\otimes 2}$, it has support of size $4$ or 0. 
Thus, among $s_1, s_2, s_3$ and $s_4$, either  exactly two of them are zero or they are all zero. 
If they are all zero, then we have $p_1+p_2+p_3+p_4=-4\neq 0$. 
This is a contradiction to  our assumption that $\widehat{f'}$ satisfies {\sc 2nd-Orth}. 
Thus, exactly two of  $s_1$, $s_2$, $s_3$ and $s_4$ are zeros. 
Suppose that they are $s_i$ and $s_j$. 
Recall that we use $x_i$ and $y_i$ $(1\leqslant i \leqslant 8)$ to denote the entries in Row 1 and Row 2 of Table~\ref{table-12neq}.
Thus $|x_i| = |y_i| =1$,
for $1\leqslant i\leqslant 8$.
Since $s_i=x_i+y_i=0$ and $s_j=x_j+y_j=0$, we have $x_i=-y_i$, and $x_j=-y_j$.
Also, since $s_i=s_j=0$, we have $p_i=p_j=-1$. 
Let $\{\ell, k\}=\{1,2,3,4\}\backslash\{i, j\}$.
Then, by {\sc 2nd-Orth}, we have $p_\ell+p_k=-p_i-p_j=2$. Since $|p_\ell|=|p_k|=1$, we have $p_\ell=p_k=1$.
Note that $p_\ell=x_\ell\cdot \overline{y_\ell}=1$ and also $1= |y_\ell| = y_\ell\cdot\overline{y_\ell}$. Thus, we have $x_\ell=y_\ell$. Similarly, $x_k=y_k$.
Thus, for all $1\leqslant i\leqslant 8$, $x_i=\pm y_i$.
Consider $\widehat{\partial}_{ij}\widehat{f'}$ for all pairs of indices $\{i, j\}$. 
By symmetry, the same conclusion holds.
Thus,
  $\widehat{f}(\alpha)=\pm\widehat{f}({\alpha'})$ for any two inputs $\alpha$ and $\alpha'$ on the six variables
  where ${\rm wt}(\alpha)={\rm wt}(\alpha')=3$ and ${\rm wt}(\alpha\oplus\alpha')=2$.
In particular, we have 
\[\widehat{f'}^{000111}=\varepsilon_1 \widehat{f'}^{001011}=\varepsilon_2 \widehat{f'}^{011001}=\varepsilon_3 \widehat{f'}^{111000},\]
where $\varepsilon_1, \varepsilon_2, \varepsilon_3=\pm 1$ independently. 
By {\sc ars}, we have $\widehat{f'}^{000111}=\overline{\widehat{f'}^{111000}}.$
\begin{itemize} 
\item
If $\widehat{f'}^{000111}=\widehat{f'}^{111000}=\overline{\widehat{f'}^{111000}},$ then $\widehat{f'}^{111000}=\pm 1.$
\item
If $\widehat{f'}^{000111}=-\widehat{f'}^{111000}=\overline{\widehat{f'}^{111000}},$ then $\widehat{f'}^{111000}=\pm \ii.$
\end{itemize}{}
Thus, there  exists $\lambda=1$ or $\ii$ such that $\widehat{f'}^{000111}=\pm \lambda$ and $\widehat{f'}^{111000}=\pm \lambda$.
Consider any $\alpha\in \mathbb{Z}_2^6$ with ${\rm wt}(\alpha)=3$. If $\alpha\in\{000111, 111000\}$, then clearly, $\widehat{f'}(\alpha)=\pm \lambda$. 
Otherwise, either ${\rm wt}(\alpha\oplus000111)=2$ or ${\rm wt}(\alpha\oplus111000)=2$.
Then, $\widehat{f'}(\alpha)=\pm \lambda$.
Thus, there  exists $\lambda=1$ or $\ii$ such that for all $\alpha\in\mathbb{Z}_2^6$ with ${\rm wt}(\alpha)=3$,  $\widehat{f'}(\alpha)=\pm \lambda$.

Since $\widehat{f'}(\alpha)\neq 0$ for all $\alpha$ with ${\rm wt}(\alpha)=1$, by Lemma \ref{lem-zero_2}, there exists a pair of indices $\{i, j\}$ such that $(\widehat{\partial}_{ij}\widehat{f'})^{0000}\neq 0$. 
Since $\widehat{\partial}_{ij}\widehat{f'}\in \mathcal{O}^{\otimes}$, it is of the form $(a, 0, 0, \bar a)\otimes (b, 0, 0, \bar b)$, where $ab\neq 0$, since no other factorization form in $\mathcal{O}^{\otimes}$ has a nonzero value at 0000.
By Lemma \ref{lin-wang}, we can realize the signature $\widehat{g}=(a, 0, 0, \bar a)$.
Here, we can normalize $a$ to $e^{\ii\theta}$ where $0\leqslant\theta<\pi$.
Then, let $\rho=e^{\ii \theta/2}.$
Clearly, $ 0\leqslant \theta/2 < \pi/2$.
Consider a holographic transformation by $\widehat{Q}=\left[\begin{smallmatrix}
{\overline \rho} & 0\\
0 & {\rho}\\
\end{smallmatrix}\right].$
Note that $(\neq_2)(\widehat{Q}^{-1})^{\otimes 2}=(\neq_2)$ and $\widehat{Q}^{\otimes 2}\widehat{g}=(1, 0, 0, 1).$
 The holographic transformation by $\widehat{Q}$ does not change $\neq_2$, but transfers $\widehat{g}=(a, 0, 0, \bar a)$ to $(=_2)=(1, 0, 0, 1)$.
Thus, we have 
$$\holant{\neq_2}{\widehat{g}, \widehat{f'},  \widehat{\mathcal F}}\equiv_T \holant{\neq_2}{=_2, \widehat{Q}\widehat{f'},  \widehat{Q}\widehat{\mathcal F}}.$$
We denote $\widehat{Q}\widehat{f'}$ by $\widehat{f''}$.
Note that $\widehat{Q}$ does not change those entries of $\widehat{f'}$ that are on half-weighted inputs. 
Thus, for all $\alpha$ with ${\rm wt}(\alpha)=3$, we have $\widehat{f''}({\alpha})=\pm \lambda$ for some $\lambda=1$ or $\ii$.
Also,  $\widehat{Q}$ does not change the parity and irreducibility of $\widehat{f'}$. 
Thus $\widehat{f''}$ has odd parity and  $\widehat{f''}$ is irreducible.
Again, we may assume that $\widehat{f''}$ satisfies {\sc 2nd-Orth} and $\widehat{f''}\in \widehat{\int}\widehat{\mathcal{O}}^{\otimes}$. Otherwise, we are done. 

In the problem $\holant{\neq_2}{=_2, \widehat{f''},  \widehat{Q}\widehat{\mathcal F}}$, 
we can connect two $\neq_2$ on the LHS using $=_2$ on the RHS, and then we can realize $=_2$ on the LHS. 
Thus, we can use $=_2$ to merge variables of $\widehat{f''}$. 
Therefore, we may further assume  $\widehat{f''}\in {\int}\widehat{\mathcal{O}}^{\otimes}$, i.e., $\partial_{ij}\widehat{f''}\in \widehat{\mathcal{O}}^{\otimes}$ for all pairs of indices $\{i, j\}$; 
otherwise, there exist two variables of $\widehat{f''}$ such that by merging these two variables using $=_2$, we can realize a 4-ary signature that is not in $\widehat{\mathcal{O}}^{\otimes}$, and then 
 by Lemma \ref{lem-4-ary} we are done.

Consider the signature $\partial_{12}\widehat{f''}=\widehat{{f''}}_{12}^{00}+\widehat{{f''}}_{12}^{11}$ and the inner product $\langle\widehat{{\bf f''}}_{12}^{00}, \widehat{{\bf f''}}_{12}^{11}\rangle$. Same as Table~\ref{table-12neq}, we build the following Table~\ref{table-12eq}.

\begin{table}[!hbtp]
\renewcommand{\arraystretch}{2}
\centering
\begin{tabular}{ |c|c|c| c| c| c| c| c| c| } 
 \hline
$\widehat{{f''}}_{12}^{00}$ & $\widehat{f''}^{000001}$ & $\widehat{f''}^{000010}$ & $\widehat{f''}^{000100}$ & $\widehat{f''}^{000111}$ & $\widehat{f''}^{001000}$  & $\widehat{f''}^{001011}$ & $\widehat{f''}^{001101}$ & $\widehat{f''}^{001110}$\\ 
  \hline
$\widehat{{f''}}_{12}^{11}$ & $\widehat{f''}^{110001}$ & $\widehat{f''}^{110010}$ & $\widehat{f''}^{110100}$ & $\widehat{f''}^{110111}$ & $\widehat{f''}^{111000}$  & $\widehat{f''}^{111011}$ & $\widehat{f''}^{111101}$ & $\widehat{f''}^{111110}$\\ 
 \hline
 ${\partial}_{12}\widehat{f''}$ & $t_1$ & $t_2$ & $t_3$ & $t_4$ & $\overline{t_4}$ &  $\overline{t_3}$ &  $\overline{t_2}$ &  $\overline{t_1}$ \\ 
 \hline
 $\langle\widehat{{\bf f''}}_{12}^{00}, \widehat{{\bf f''}}_{12}^{11}\rangle$ & $q_1$ & $q_2$ & $q_3$ & $q_4$ & $q_4$ & $q_3$ & $q_2$ & $q_1$\\ 
 \hline
\end{tabular}
\caption{Entries of $\widehat{{f''}}_{12}^{00}$, $\widehat{{f''}}_{12}^{11}$, ${\partial}_{12}\widehat{f''}$ and pair-wise product terms in $\langle\widehat{{\bf f''}}_{12}^{00}, \widehat{{\bf f''}}_{12}^{11}\rangle$ on odd-weighed inputs}\label{table-12eq}
\end{table}
Same as the proof of $x_i=\pm y_i$ for Table~\ref{table-12neq}, we have  $\widehat{f''}^{000001}=\pm \widehat{f''}^{110001}.$
Since $\widehat{f''}^{110001} =\pm \lambda$, $\widehat{f''}^{000001} = \pm \lambda$, (here  $\pm$ can be either $\pm$ or $\mp$).
Consider $\partial_{ij}\widehat{f''}$ for all pairs of indices $\{i, j\}$. By symmetry, the same conclusion holds.
Thus, for every $\alpha\in \mathbb{Z}_2^6$ with ${\rm wt}(\alpha)=1$, 
$\widehat{f''}(\alpha)=\pm \lambda.$
Therefore, using {\sc ars}, there  exists $\lambda=1$ or $\ii$ such that for all $\alpha\in \mathscr{S}(\widehat{f''})$, $\widehat{f''}(\alpha)=\pm \lambda$, and we have the reduction $$\holant{\neq_2}{=_2, \widehat{f''}, \widehat{Q}\widehat{\mathcal{F}}}\leqslant_T\holant{\neq_2}{\widehat{\mathcal{F}}}$$ for some $\widehat{Q}\in\widehat{{\bf O}_2}$.
Clearly, $\widehat{f''}=\widehat{Q}\widehat{f'}$ where $\widehat{Q}=\left[\begin{smallmatrix}
{\overline \rho} & 0\\
0 & \rho\\
\end{smallmatrix}\right] \in\widehat{{\bf O}_2}$, and the nonzero binary signature $(\rho^2, 0, 0, \overline{\rho^2}) \in \widehat{\mathcal{O}}$ is realizable from $\widehat{\partial}_{ij}\widehat{f'}$ for some $\{i, j\}$.
\end{proof}

Finally, we go for the kill in the next lemma.
Recall the signature $\widehat{f_6}$
defined in (\ref{eqn:definiton-f6}).
This \emph{Lord of Intransigence} at arity 6 makes its appearance
in Lemma~\ref{lem-arity6-4}.
\begin{lemma}\label{lem-arity6-4}
Suppose that $\widehat{\mathcal{F}}$ contains an irreducible $6$-ary signature $\widehat{f''}$ where $\mathscr{S}(\widehat{f''})=\mathscr{O}_6$, and 
there  exists $\lambda=1$ or $\ii$ such that for all $\alpha\in \mathscr{S}(\widehat{f''})$, $\widehat{f''}(\alpha)=\pm \lambda$.
Then, 
 $\holant{\neq_2}{=_2, \widehat{\mathcal{F}}}$ is \#P-hard, or
  $\widehat{f_6}$ is realizable from $\widehat{f''}$ and $=_2$,
   i.e., $\holant{\neq_2}{ \widehat{f_6}, \widehat{\mathcal{F}}}\leqslant_T\holant{\neq_2}{=_2,  \widehat{\mathcal{F}}}$.
   Moreover, $\widehat{f_6}$ is realizable by extending variables of $\widehat{f''}$ with binary signatures in ${\widehat{\mathcal{B}}}$,
   i.e., $\widehat{f_6}\in \{\widehat{f''}\}_{\neq_2}^{\widehat{\mathcal{B}}}.$ 
\end{lemma}
\begin{proof}
Again,  we may assume that $\widehat{f''}$ satisfies {\sc 2nd-Orth} and $\widehat{f''}\in 
\widehat\int \mathcal{\widehat{O}}^{\otimes}.$
Since  $=_2$ is available
on the RHS, given any signature $\widehat{f}
\in \widehat{\mathcal{F}}$, 
we can extend any variable $x_i$ of $\widehat{f}$  with  $=_2\in \widehat{\mathcal{B}}$ using
$\ne_2$.
This gives a signature $\widehat g$ where $\widehat g_i^0=\widehat f_i^1$ and $\widehat g_i^1=\widehat f_i^0$.
We call this extending gadget construction the flipping operation on variable $x_i$. 
Clearly, it does not change the reducibility or irreducibility   of $\widehat{f}$. But it changes the parity of $\widehat{f}$ if $\widehat{f}$ has parity.
Once a signature  $\widehat{f}$ is realizable, we can modify it by flipping some of its variables.

We first show that we can realize a signature $\widehat{f^\ast}$ from $\widehat{f''}$ having support $\mathscr{S}(\widehat{f^\ast})= \mathscr{E}_6 = \{\alpha\in\mathbb{Z}_2^6\mid {\rm wt}(\alpha)\equiv0 \mod 2\}$, and $\widehat{f^\ast}(\alpha)=\pm 1$ for all $\alpha\in \mathscr{S}(\widehat{f^\ast}).$
Remember that $=_2$ is available.
If we connect $=_2$ with an arbitrary variable of $\widehat{f''}$ using $\neq_2$, we 
 will change the parity of $\widehat{f''}$ from odd to even.
If $\widehat{f''}(\alpha)=\pm 1$ for all $\alpha \in \mathscr{S}(\widehat{f''})$, then $\widehat{f^\ast}$ can be realized by  flipping an arbitrary variable of $\widehat{f''}$.
Otherwise, $\widehat{f''}(\alpha)=\pm \ii$ for all $\alpha \in \mathscr{S}(\widehat{f''})$.
Consider $\widehat{\partial}_{12}\widehat{f''}.$
Look at Table \ref{table-12eq}. 
We use $x_i$ and $y_i$ $(1\leqslant i\leqslant 8)$ to denote entries in Row 1 and 2.
As we have showed, $x_i=\pm y_i$.
Thus, $t_i=\pm 2\ii$ or $0$ for $1\leqslant i \leqslant 4$.
Remember that if $t_i=0$ (i.e., $x_i=-y_i$), then $q_i=x_i\cdot \overline{y_i}=- x_i\cdot \overline{x_i} = -|x_i|^2=-1$.
If $t_i=0$ for all $1\leqslant i \leqslant 4$, then $$\langle\widehat{{\bf f''}}_{12}^{00}, \widehat{{\bf f''}}_{12}^{11}\rangle=2(q_1+q_2+q_3+q_4)=-4\neq 0.$$
This contradicts with our assumption that $\widehat{f''}$ satisfies {\sc 2nd-Orth}.
Thus, $t_i$ ($1\leqslant i \leqslant 4$) are not all zeros.
Then $(\widehat{\partial}_{12}\widehat{f''})\not\equiv 0$.
Thus, $\mathscr{S}({\widehat{\partial}_{12}\widehat{f''}})\neq\emptyset$ and $(\widehat{\partial}_{12}\widehat{f''})(\alpha)=\pm 2i$ for all $\alpha \in \mathscr{S}({\widehat{\partial}_{12}\widehat{f''}})$.

Since $\widehat{\partial}_{12}\widehat{f''}\in
\mathcal{\widehat{O}}^{\otimes}$ and it has even parity, 
$\widehat{\partial}_{12}\widehat{f''}$ is of the form $2\cdot(a, 0, 0, \bar a)\otimes (b, 0, 0, \bar b)$ or $2\cdot(0, a, \bar a, 0)\otimes(0, b, \bar b, 0)$, where   the norms of $a$ and $b$ are normalized to $1$.
In both cases, we have $ab, \bar ab, a\bar b, \bar a \bar b\in \{\ii, -\ii\}$.
Thus, $ab\cdot \bar a b=(a\bar a) b^2=b^2=\pm 1$.
Then, $b=\pm 1$ or $\pm \ii$. 
If $b=\pm 1$, then $a=a\bar b\cdot b=\pm \ii$.
Similarly, if $b=\pm \ii$, then $a=a\bar b\cdot b=\pm 1$.
Thus, among $a$ and $b$, exactly one is $\pm \ii$.
Thus, by factorization we can realize the binary signature $\widehat{g}=(\ii, 0, 0, -\ii)$ or $(0, \ii, -\ii, 0)$ up to a scalar $-1$.
Connecting an arbitrary variable of $\widehat{f}$ with a variable of $\widehat{g}$, 
 we can get a signature which has parity and all its nonzero entries have value $\pm 1$.
If the resulting signature has even parity, then we get the desired $\widehat{f^\ast}$.
If it has odd parity, then we can flip one of its variables to change the parity.  
Thus, we can realize a signature $\widehat{f^\ast}$ by extending variables of $\widehat{f''}$ with binary signatures in $\widehat{\mathcal{B}}^{\otimes}$ such that $\mathscr{S}(\widehat{f^\ast})=
\mathscr{E}_6$,
and $\widehat{f^\ast}(\alpha)=\pm 1$ for all $\alpha\in \mathscr{S}(\widehat{f^\ast}).$

Consider the following 16 entries of $\widehat{f^\ast}$.
In Table \ref{tab:16-entry}, we list 16 entries of $\widehat{f^\ast}$  with  $x_1x_2x_3=000, 011, 101, 110$ as the row index and  $x_4x_5x_6=000, 011, 101, 110$ as the column index.
We also view these 16 entries in Table \ref{tab:16-entry} as a 4-by-4 matrix denoted by $M_r(\widehat{f^\ast})$, and we call it the  representative matrix of $\widehat{f^\ast}$. 
Note that for any $\alpha\in \mathscr{S}(\widehat{f^\ast})$ such that the entry $\widehat{f^\ast}(\alpha)$ does not appear in $M_r(\widehat{f^\ast})$, 
$\widehat{f^\ast}(\overline{\alpha})$  appears in $M_r(\widehat{f^\ast})$. 
Since $\widehat{f^\ast}({\alpha})=\pm 1\in \mathbb{R}$, 
$\overline{\widehat{f^\ast}({\alpha})}={\widehat{f^\ast}({\alpha})}$.
By {\sc ars}, ${\widehat{f^\ast}(\overline{\alpha})}=\overline{\widehat{f^\ast}({\alpha})}={\widehat{f^\ast}({\alpha})}$.
Thus, the 16 entries of the matrix
$M_r(\widehat{f^\ast})$  listed
in  Table \ref{tab:16-entry} gives a complete account for all the
 32 nonzero entries of $\widehat{f^\ast}$.
\begin{table}[!h]
\renewcommand{\arraystretch}{2}
    \centering
    \begin{tabular}{|l|c|c|c|c|}
    \hline
    \diagbox[width=1.5in, height=7ex]{$x_1x_2x_3$}{$x_4x_5x_6$} &  000 (Col 1) &   011 (Col 2)&  101 (Col 3)&  110 (Col 4)\\
    \hline
   000 (Row 1)& $\widehat{f^\ast}^{000000}$      & $\widehat{f^\ast}^{000011}$ & $\widehat{f^\ast}^{000101}$ & $\widehat{f^\ast}^{000110}$ \\
            \hline
     011 (Row 2)& $\widehat{f^\ast}^{011000}$      & $\widehat{f^\ast}^{011011}$ & $\widehat{f^\ast}^{011101}$ & $\widehat{f^\ast}^{011110}$ \\            \hline
   101 (Row 3)& $\widehat{f^\ast}^{101000}$      & $\widehat{f^\ast}^{101011}$ & $\widehat{f^\ast}^{101101}$ & $\widehat{f^\ast}^{101110}$ \\
               \hline
     110 (Row 4)& $\widehat{f^\ast}^{110000}$      & $\widehat{f^\ast}^{110011}$ & $\widehat{f^\ast}^{110101}$ & $\widehat{f^\ast}^{110110}$ \\
   \hline
    \end{tabular}
    \caption{Representative entries of $\widehat{f^\ast}$}
    \label{tab:16-entry}
\end{table}

We use $(m_{ij})_{i,j=1}^4$ to denote the 16 entries of $M_r(\widehat{f^\ast})$.
We claim that  any two rows of $M_r(\widehat{f^\ast})$ are orthogonal;
this follows from the fact that $\widehat{f^\ast}$ satisfies {\sc 2nd-Orth} and {\sc ars}.
For example, consider the first two rows of $M_r(\widehat{f^\ast})$.
By {\sc 2nd-Orth}, the inner product  $\langle\widehat{{\bf f^\ast}}_{23}^{00}, \widehat{{\bf f^\ast}}_{23}^{11}\rangle$ for the real-valued $\widehat{f^\ast}$
is 
\[\sum_{(x_1, x_4, x_5, x_6) \in \mathbb{Z}_2^4} \widehat{f^\ast}^{x_1 00 x_4 x_5 x_6}
\widehat{f^\ast}^{x_1 11 x_4 x_5 x_6} =0,\]
where the sum has 8 nonzero product terms. 
The first 4 terms given by $x_1=0$ are the pairwise products $m_{1j}m_{2j}$,
for $1 \leqslant j \leqslant 4$. The second 4 terms  are, by {\sc ars},
the pairwise products $m_{2j}m_{1j}$ in the reversal order of $1 \leqslant j \leqslant 4$, where we exchange row 1 with row 2
on the account of flipping the summation index $x_1$   from 0 to 1, and simultaneously flipping both $x_2$ and $x_3$.
This shows that $\sum_{j=1}^4 m_{1j}m_{2j} = 0$.
Similarly  any two columns of  $M_r(\widehat{f^\ast})$ are orthogonal.

Also, we consider  the inner product $\langle\widehat{{\bf f^\ast}}_{14}^{00}, \widehat{{\bf f^\ast}}_{14}^{11}\rangle=0$. 
It is computed using the following $16$ entries  in $M_r(\widehat{f^\ast})$, listed in Table~\ref{tab:my_label}.

\begin{table}[!h]
    \centering
    \begin{tabular}{|c|c|c|c|c|c|c|c|}
\hline
  \rule{0pt}{20pt}   $\widehat{f^\ast}^{000000}$ &$\widehat{f^\ast}^{000011}$ 
     & $\widehat{f^\ast}^{010010}$
     & $\widehat{f^\ast}^{010001}$
     & $\widehat{f^\ast}^{001010}$
     & $\widehat{f^\ast}^{001001}$
     & $\widehat{f^\ast}^{011000}$
     &$\widehat{f^\ast}^{011011}$    \\
     $=m_{11}$ &  $=m_{12}$ &  $=m_{33}$ &  $=m_{34}$ &  $=m_{43}$ &  $=m_{44}$ &   $=m_{21}$ &  $=m_{22}$ \\
     \hline
   \rule{0pt}{20pt}  $\widehat{f^\ast}^{100100}$ &
     $\widehat{f^\ast}^{100111}$ &
     $\widehat{f^\ast}^{110110}$ &
     $\widehat{f^\ast}^{110101}$ &
     $\widehat{f^\ast}^{101110}$ &
       $\widehat{f^\ast}^{101101}$&
       $\widehat{f^\ast}^{111100}$&
       $\widehat{f^\ast}^{111111}$ \\
       $=m_{22}$ &  $=m_{21}$ &  $=m_{44}$ &  $=m_{43}$ &  $=m_{34}$ &  $=m_{33}$ &   $=m_{12}$ &  $=m_{11}$\\ \hline

    \end{tabular}
    \caption{Pair-wise product terms in $\langle\widehat{{\bf f^\ast}}_{14}^{00}, \widehat{{\bf f^\ast}}_{14}^{11}\rangle$ on even-weighed inputs}
    \label{tab:my_label}
\end{table}

Let $M_r(\widehat{f^\ast})_{[1,2]}$ be the 2-by-2 submatrix of $M_r(\widehat{f^\ast})$ by picking the first two rows and the first two columns, and  $M_r(\widehat{f^\ast})_{[3,4]}$ be the 2-by-2 submatrix of $M_r(\widehat{f^\ast})$ by picking the last two rows and the last two columns.
Indeed, 
\begin{equation*}
    \begin{aligned}
\langle\widehat{{\bf f^\ast}}_{14}^{00}, \widehat{{\bf f^\ast}}_{14}^{11}\rangle&=2({\rm perm}(M_r(\widehat{f^\ast})_{[1,2]})+{\rm perm}(M_r(\widehat{f^\ast})_{[3,4]}))\\
&=2(m_{11}m_{22}+m_{12}m_{21}+m_{33}m_{44}+m_{34}m_{43})=0.
    \end{aligned}
    \end{equation*}

Then,
we show that by renaming or flipping variables of $\widehat{f^\ast}$, 
we may modify $\widehat{f^\ast}$ to realize a signature whose representative matrix is obtained by performing row permutation, column permutation, or matrix transpose on $M_r(\widehat{f^\ast})$.
First, if we exchange the names of variables $(x_1, x_2, x_3)$ with variables $(x_4, x_5, x_6)$, then the representative matrix  $M_r(\widehat{f^\ast})$ will be transposed. 
Next, consider the group $\mathfrak G$ of permutations on the rows $\{1,2,3,4\}$ effected by any sequence of   operations of renaming and flipping  variables in $\{x_1, x_2, x_3\}$. By renaming variables in $\{x_1, x_2, x_3\}$, we can switch any two rows among Row 2, 3 and 4.
Thus $S_3$ on $\{2,3,4\}$ is contained in
$\mathfrak G$.
Also, if we flip both variables $x_2$ and $x_3$ of $\widehat{f^\ast}$, 
then for the realized signature, its representative matrix can be obtained by  switching both the pair Row 1 and Row 2, and 
the pair Row 3 and Row 4 of $M_r(\widehat{f^\ast})$.
Thus, 
the permutation $(12)(34) \in \mathfrak G$.
It follows that
${\mathfrak G} = S_4$.
Thus, by renaming or flipping  variables of $\widehat{f^\ast}$, 
we can permute any two rows or any two columns of $M_r(\widehat{f^\ast})$, or transpose $M_r(\widehat{f^\ast})$.
For the resulting signature, 
we may assume that its representative matrix $A$ also satisfy ${\rm perm}(A_{[1,2]})+{\rm perm}(A_{[3,4]})=0$, and any two rows of $A$ are orthogonal and any two columns of $A$ are orthogonal. 
Otherwise, we get \#P-hardness.
In the following, without loss of generality, we may modify $M_r(\widehat{f^\ast})$ by permuting any two rows or any two columns, or taking transpose.
We show that it will give $M_r(\widehat{f_6})$, after a normalization by $\pm 1$. 
In other words, $\widehat{f_6}$ is realizable from $\widehat{f^\ast}$ by renaming or flipping  variables, up to a normalization by $\pm 1$.

Consider any two rows, Row $i$ and Row $j$,  of $M_r(\widehat{f^\ast})$. 
Recall that every entry of $M_r(\widehat{f^\ast})$ is $\pm 1$.
We say that Row $i$ and Row $j$ differ in Column $k$ if $m_{ik}\neq m_{jk}$, which implies that $m_{ik}=-m_{jk}$;
otherwise, they are equal $m_{ik}=m_{jk}$.
In the former case, $m_{ik}\cdot m_{jk}=-1$, and in the latter case
$m_{ik}\cdot m_{jk}=1$.
Since Row $i$ and Row $j$ are orthogonal, they differ in exactly two columns and 
are equal in the other two columns.
Similarly,  for  any two columns  of $M_r(\widehat{f^\ast})$,
 they  differ in exactly two rows and  are equal in the other two rows.
Depending on the number of $-1$ entries in each row and column of $M_r(\widehat{f^\ast})$, we consider the following two cases. 
\begin{itemize}
\item
Every row and column of $M_r(\widehat{f^\ast})$ has an odd number of $-1$ entries.

Consider Row 1. 
It has either exactly three $-1$ entries or exactly one $-1$ entry.
If it has three $-1$ entries, then we modify $M_r(\widehat{f^\ast})$ by multiplying the matrix  with $-1$.
This does not change the parity of the number of $-1$ entries in each row and each column. 
By such a modification, Row 1 has exactly one $-1$ entry.
By permuting columns, we may assume that 
Row 1 is $(-1, 1, 1, 1).$
Consider the number of $-1$ entries in  Rows 2, 3 and  4.
\begin{itemize}
    \item 
If they all have exactly one $-1$ entry,  by orthogonality,
the unique column locations of   the $-1$ entry 
in each row must be pairwise distinct.
Then, by possibly permuting rows 
2, 3 and 4 
we may assume that 
the matrix $M_r(\widehat{f^\ast})$ has the following form $$M_r(\widehat{f^\ast})=\begin{bmatrix}
-1& 1 &1 & 1\\
1 & -1 & 1 & 1\\
1 & 1 & -1 & 1\\
1 & 1 & 1 & -1\\
\end{bmatrix}.$$
Then, ${\rm perm}(M_r(\widehat{f^\ast})_{[1, 2]})+{\rm perm}(M_r(\widehat{f^\ast})_{[3, 4]})=2+2=4\neq 0.$ 
Contradiction.
  
  \item
  Otherwise, among Rows 2, 3 and 4, there is one that has three $-1$ entries. 
  By permuting rows, we may assume that Row 2 has three $-1$ entries.
  Since Row 2 and Row 1 differ in two columns, the only $+1$ entry in Row 2 is not in  Column 1.
  By possibly permuting Columns 2,  3 and 4, without loss of generality, we may assume that Row 2 is $(-1, 1, -1, -1)$.
  Then, we consider Column 3 and Column 4.
  Since every column has an odd number of $-1$ entries and $m_{13}=1$ and $m_{23}=-1$, we have $m_{33}=m_{43}$,
  both $+1$ or $-1$.
  Similarly, $m_{34}=m_{44}$.
  Also, since Column 3 and Column 4 differ in exactly two rows, and $m_{13}=m_{14}$ and  $m_{23}=m_{24}$, we have $m_{33}=-m_{34}$ and $m_{43}=-m_{44}$.
  Thus, $M_r(\widehat{f^\ast})_{[3,4]}= \pm \left[\begin{smallmatrix}
  1 & -1\\
  1 & -1\\
  \end{smallmatrix}\right]$.
  In both cases, we have ${\rm perm}(M_r(\widehat{f^\ast})_{[1,2]})=-2$.
  Notice that  $M_r(\widehat{f^\ast})_{[1,2]}=\left[\begin{smallmatrix}
  -1 & 1\\
  -1 & 1\\
  \end{smallmatrix}\right]$.
  Thus, ${\rm perm}(M_r(\widehat{f^\ast})_{[1,2]})+{\rm perm}(M_r(\widehat{f^\ast})_{[3,4]})=-4\neq 0$. Contradiction.
  
  \end{itemize}
    \item There is a row or a column of $M_r(\widehat{f^\ast})$ such that it has an even number of $-1$ entries.
    By transposing $M_r(\widehat{f^\ast})$, we may assume that it is a row, say Row $i$. 
    For any other Row $j$, it differs with Row $i$ in exactly two columns.
    Thus, Row $j$ also has an even number of $-1$ entries.
  If all  four rows of $M_r(\widehat{f^\ast})$ have exactly two $-1$ entries, then one can check that there are two rows such that one row is a scalar $(\pm 1)$ multiple of the other, thus 
  not orthogonal; this is a contradiction.
  Thus, there exists a row in which the number of $-1$ entries is $0$ or $4$. 
  By permuting rows, we may assume that it is Row 1.
  Also, by possibly multiplying $M_r(\widehat{f^\ast})$ with $-1$, we may assume that all entries of Row 1 are $+1$.
  Thus, Row 1 is $(1, 1, 1, 1)$.
  
  By orthogonality, all other rows have exactly two $-1$ entries.
  By permuting  columns (which does not change Row 1), we may assume that Row 2 is $(-1, -1, 1, 1)$.
  Then, consider Row 3.
  It also has exactly two $-1$ entries.
  Moreover, since Row 2 and Row 3 differ in 2 columns, among $m_{31}$ and $m_{32}$, exactly one is $-1$.
  By permuting Column 1 and Column 2 (which does not change Row 1 and Row 2), we may assume that $m_{31}=-1$.
  Also, among $m_{33}$ and $m_{34}$, exactly one is $-1$.
  By permuting Column 3 and Column 4 (still this will not change  Row 1 and Row 2), we may assume that $m_{33}=-1$.
  Thus, Row 3 is $(-1, 1, -1, 1)$.
  Finally, consider Row 4. 
  It also has two $-1$ entries.
  One can easily check that Row 4 has two possible forms, 
  $(-1, 1, 1, -1)$ or $(1, -1, -1, 1).$
  If Row 4 is $(1, -1, -1, 1)$,
  then, $$M_r(\widehat{f^\ast})=\begin{bmatrix}
1& 1 &1 & 1\\
-1 & -1 & 1 & 1\\
-1 & 1 & -1 & 1\\
1 & -1 & -1 & 1\\
\end{bmatrix}.$$
Thus, ${\rm perm}(M_r(\widehat{f^\ast})_{[12]})+{\rm perm}(M_r(\widehat{f^\ast})_{[34]})=-4\neq 0.$
Contradiction.

Thus, 
  Row 4 is $(-1, 1, 1, -1)$.
  Then $$M_r(\widehat{f^\ast})=\begin{bmatrix}
1& 1 &1 & 1\\
-1 & -1 & 1 & 1\\
-1 & 1 & -1 & 1\\
-1 & 1 & 1 & -1\\
\end{bmatrix}.$$
This gives the desired $M_r(\widehat{f_6})$.
\end{itemize}

Therefore, $\widehat{f_6}$ is realizable from $\widehat{f^\ast}$.

Since $\widehat{f_6}$ is realized from $\widehat{f^\ast}$ by flipping (and permuting) variables, i.e., extending some variables of $\widehat{f^\ast}$ with $=_2$ (using $\neq_2$), we have $\widehat{f_6}\in  \{\widehat{f^\ast}\}_{\neq_2}^{\widehat{\mathcal{B}}}.$
Since $\widehat{f^\ast}$ is realized from $\widehat{f''}$ by extending some variables of $\widehat{f''}$ with signatures in $\widehat{\mathcal{B}}$, we have $\widehat{f^{\ast}}\in  \{\widehat{f''}\}_{\neq_2}^{\widehat{\mathcal{B}}}.$
By Lemma \ref{lem-extending}, we have $\widehat{f_6}\in  \{\widehat{f''}\}_{\neq_2}^{\widehat{\mathcal{B}}}.$
\end{proof}

\begin{theorem}
\label{lem-6-ary-f6}
Suppose that $\widehat{\mathcal{F}}$ contains a $6$-ary signature $ \widehat{f}\notin \widehat{\mathcal{O}}{^\otimes}$.
Then,
\begin{itemize}
   \item  $\holant{\neq_2}{\widehat{\mathcal{F}}}$ is \#P-hard, or
   \item there exists some $\widehat{Q}\in \widehat{{\bf O}_2}$ such that $\holant{\neq_2}{\widehat{f}_6, \widehat Q\widehat{\mathcal{F}}}\leqslant_T\holant{\neq_2}{\widehat{\mathcal{F}}}.$
\end{itemize}
\end{theorem}

\begin{proof}
By Lemmas~\ref{lem-arity6-1}, \ref{lem-arity6-2} and \ref{lem-arity6-3}, $\holant{\neq_2}{\widehat{\mathcal{F}}}$ is \#P-hard, or  $\holant{\neq_2}{=_2, \widehat{f''}, \widehat{Q}\widehat{\mathcal{F}}}\leqslant_T\holant{\neq_2}{\widehat{\mathcal{F}}}$ for some $\widehat{Q}$ where $Q \in  \widehat{{\bf O}_2}$, and some irreducible 6-ary signature $\widehat{f''}$ where $\mathscr{S}(\widehat f'')=\mathscr{E}_{6}$ and there  exists $\lambda=1$ or $\ii$ such that for all $\alpha\in \mathscr{S}(\widehat{f''})$, $\widehat{f''}(\alpha)=\pm \lambda$.
Remember  that $\widehat{Q}\widehat{\mathcal{F}}=\widehat{Q\mathcal{F}}$ where $Q=Z\widehat{Q}Z^{-1}\in {\bf O}_2$.
Clearly, $Q\mathcal{F}$ is a set of real-valued signatures of even arity.
Since $\mathcal{F}$ does not satisfy condition (\ref{main-thr}), 
by Lemma~\ref{lem-hard-sign}, $Q\mathcal{F}$ also does not satisfy it.
Then, by Lemma~\ref{lem-arity6-4},
$\holant{\neq_2}{=_2, \widehat{f''}, \widehat{Q}\widehat{\mathcal{F}}}$ is
\#P-hard, or $\holant{\neq_2}{\widehat{f_6}, \widehat{Q}\widehat{\mathcal{F}}}\leqslant_T\holant{\neq_2}{=_2, \widehat{f''}, \widehat{Q}\widehat{\mathcal{F}}}.$
Thus, $\holant{\neq_2}{\widehat{\mathcal{F}}}$ is \#P-hard, or $\holant{\neq_2}{\widehat{f}_6, \widehat Q\widehat{\mathcal{F}}}\leqslant_T\holant{\neq_2}{\widehat{\mathcal{F}}}.$
\end{proof}
\begin{remark}
Theorem~\ref{lem-6-ary-f6} can be more succinctly stated as simply that a reduction \[\holant{\neq_2}{\widehat{f}_6, \widehat Q\widehat{\mathcal{F}}}\leqslant_T\holant{\neq_2}{\widehat{\mathcal{F}}}\]
exists, because 
 when  $\holant{\neq_2}{\widehat{\mathcal{F}}}$ is \#P-hard, the reduction
 exists trivially. 
However in keeping with the cadence of the other lemmas and theorems in this subsection, we list them as two cases.
\end{remark}

Now, we want to show that $\holant{\neq_2}{\widehat{f}_6, \widehat Q\widehat{\mathcal{F}}}$ is \#P-hard for all $\widehat{Q}\in \widehat{{\bf O}_2}$ and all $\widehat{\mathcal{F}}$ where $\mathcal{F}=Z\widehat{\mathcal{F}}$ is a real-valued signature set that does not satisfy condition (\ref{main-thr}). 
If so, then we are done.
Recall that for all $\widehat{Q}\in \widehat{{\bf O}_2}$, $\widehat{Q}\widehat{\mathcal{F}}=\widehat{Q\mathcal{F}}$ for some $Q\in {\bf O}_2$.
Moreover, for all $Q\in {\bf O}_2$, and all real-valued $\mathcal{F}$ that 
does not satisfy condition (\ref{main-thr}), 
$Q\mathcal{F}$ is also a real-valued signature set 
that does not satisfy condition (\ref{main-thr}).
Thus, it suffices for us to show that $\holant{\neq_2}{\widehat{f}_6, \widehat{\mathcal{F}}}$ is \#P-hard for all real-valued $\mathcal{F}$ that 
does not satisfy condition (\ref{main-thr}).

\subsection{\#P-hardness conditions and two properties of $\widehat{f_6}$}
In this subsection,
we give three conditions (Lemmas~\ref{lem-2-notb}, \ref{lem-4-notb} and \ref{lem-6-notb}) 
 which can quite straightforwardly lead to  the \#P-hardness of $\holant{\neq_2}{\widehat{f}_6, \widehat{\mathcal{F}}}$.
 We will extract two properties from $\widehat{f_6}$, the non-$\widehat{\mathcal{B}}$ hardness (Definition~\ref{def:non-B-hard}) and the realizability of $\widehat{\mathcal{B}}$ (Lemma~\ref{lem-b-realize-f6}).
 Later, 
we will prove the \#P-hardness of $\holant{\neq_2}{\widehat{f_6}, \widehat{ \mathcal{F}}}$ based on these two properties.

\begin{lemma}\label{lem-2-notb}
 $\holant{\neq_2}{\widehat{f_6}, \widehat{\mathcal{F}}}$ is \#P-hard 
if $\widehat{\mathcal{F}}$ contains a 
nonzero binary signature $\widehat{b}\notin \widehat{\mathcal{B}}^{\otimes}$.
\end{lemma}
\begin{proof}
If $\widehat{b}\notin \widehat{\mathcal{O}}^{\otimes}$, then by Lemma \ref{lem-2-ary}, we are done. 
Otherwise, $\widehat{b}\in \widehat{\mathcal{O}}^{\otimes}$. 
Thus, $\widehat{b}=(a, 0, 0, \bar a)$ or  $\widehat{b}=(0, a, \bar a, 0)$.
Since $\widehat{b}\not\equiv0$, $a\neq 0$.
We normalize the norm of $a$  to $1$.
Since $\widehat{b}\notin \widehat{\mathcal{B}}^{\otimes}$, $a\neq \pm 1$ or $\pm \ii$.
We first consider the case that $\widehat{b}(y_1, y_2)=(0, a, \bar a, 0)$.
Connecting variables $x_1$ and $x_2$ of $\widehat{f_6}$ with variables $y_2$ and $y_1$ of $\widehat{b}$ using $\neq_2$, we get a 4-ary signature $\widehat{g}$.
We list the truth table of $\widehat{g}$ indexed by the assignments of variables $(x_3, x_4, x_5, x_6)$ from $0000$ to $1111$.
$$\widehat{g}=(0, a+\bar a, -a+\bar a, 0, a-\bar a, 0, 0, -a-\bar a, -a-\bar a, 0, 0, -a+\bar a, 0, a-\bar a, a+\bar a, 0).$$
Since $a$ has norm 1, and $a\neq \pm 1$ or $\pm \ii$, 
$|a\pm \bar a|\neq 0$.
Thus, $|\mathscr{S}(\widehat g)|=8$.
Clearly, every 4-ary signature that is in $\widehat{\mathcal{O}}^{\otimes}$ has support of size $0$ or $4$.
Thus, $\widehat{g}\notin\widehat{\mathcal{O}}^{\otimes}$. 
By Lemma~\ref{lem-4-ary},  $\holant{\neq_2}{\widehat{f_6}, \widehat{\mathcal{F}}}$ is \#P-hard. 
We prove the  case  $\widehat{b}(y_1, y_2)=(a, 0, 0, \bar a)$
similarly. By connecting variables $x_1$ and $x_2$ of $\widehat{f_6}$ with variables $y_1$ and $y_2$ of $\widehat{b}$ using $\neq_2$, we also get a 4-ary signature that is not in $\widehat{\mathcal{O}}^{\otimes}$. 
The lemma is proved.
\end{proof}

\begin{definition}\label{def:non-B-hard}
We say a signature set
$\widehat{\mathcal{F}}$ is  non-$\widehat{\mathcal{B}}$ hard, if for any nonzero binary signature $\widehat{b}\notin \widehat{\mathcal{B}}^{\otimes }$, the problem
$\holant{\neq_2}{\widehat{b}, \widehat{\mathcal{F}}}$ is \#P-hard.
Correspondingly, we say that a signature set $\mathcal{F}$ is non-${\mathcal{B}}$ hard, if for any nonzero binary signature $b\notin \mathcal{B}^{\otimes}$, the problem $\Holant(b, \mathcal{F})$ is \#P-hard.
\end{definition}
Clearly,  Lemma~\ref{lem-2-notb} says that $\{\widehat{f_6}\}\cup \widehat{\mathcal{F}}$ is non-$\widehat{\mathcal{B}}$ hard for any $\widehat{\mathcal{F}}$ (where $\mathcal{F}=Z\widehat{\mathcal{F}}$ is a real-valued signature set that does not satisfy condition (\ref{main-thr})).
Before we give the other two \#P-hardness conditions, 
we first explain why we introduce the notion of non-$\widehat{\mathcal{B}}$ hardness.
We will extract two properties from $\widehat{f_6}$ to prove the \#P-hardness of $\holant{\neq_2}{\widehat{f_6}, \widehat{ \mathcal{F}}}$.
These are the non-$\widehat{\mathcal{B}}$ hardness and the realizability of $\widehat{\mathcal{B}}$. 
From
Lemma~\ref{lem-b-realize-f6}\footnote{This lemma and the following Theorem~\ref{thm-holantb} are stated and proved in the setting of $\Holant(\mathcal{F})$.} we  get
the redutcion 
$\holant{\neq_2}{\widehat{f_6}, \widehat{\mathcal{B}}\cup\widehat{\mathcal{F}}}
\leqslant\holant{\neq_2}{\widehat{f_6}, \widehat{\mathcal{F}}}$.
We will show that for any non-$\widehat{\mathcal{B}}$ hard set $\widehat{\mathcal{F}}$ where $\mathcal{F}$ does not satisfy condition (\ref{main-thr}),
 $\holant{\neq_2}{\widehat{\mathcal{B}}\cup\widehat{\mathcal{F}}}$ is \#P-hard (Theorem~\ref{thm-holantb}).
 This directly implies that  $\holant{\neq_2}{\widehat{f_6}, \widehat{ \mathcal{F}}}$ is \#P-hard when $\mathcal{F}$ does not satisfy condition (\ref{main-thr}). 
 This slightly more general Theorem~\ref{thm-holantb}  will also be used when dealing with signatures of arity $8$.
 Now, let us continue to give two more \#P-hardness conditions without assuming the availability of $\mathcal{B}$ (Lemma~\ref{lem-4-notb} and \ref{lem-6-notb}).


\begin{lemma}\label{lem-4-notb}
Suppose that $\widehat{\mathcal{F}}$ is non-$\widehat{\mathcal{B}}$ hard.
Then
$\holant{\neq_2}{\widehat{\mathcal{F}}}$ is \#P-hard if $\widehat{\mathcal{F}}$ contains a nonzero 4-ary signature $\widehat{f}\notin \widehat{\mathcal{B}}^{\otimes}$.
\end{lemma}
\begin{proof}
If $\widehat{f}\notin \widehat{\mathcal{O}}^{\otimes}$, then by Lemma~\ref{lem-4-ary}, we are done.
Otherwise, $\widehat{f}=\widehat{b}_1\otimes\widehat{b_2}$,
where the binary signatures $\widehat{b_1}, \widehat{b_2}
\in \widehat{\mathcal{O}}^{\otimes}$.
Since $\widehat{f}\notin \widehat{\mathcal{B}}^{\otimes}$, $\widehat{b}_1$ and $\widehat{b_2}$ are not both in $\widehat{\mathcal{B}}^{\otimes}$.
Then, we can realize a binary signature that is not in $\widehat{\mathcal{B}}^{\otimes}$ by factorization.
Since $\widehat{\mathcal{F}}$ is non-$\widehat{\mathcal{B}}$ hard,
we are done.
\end{proof}

Let $H=\frac{1}{\sqrt{2}}\left[\begin{smallmatrix}
1 & 1\\
-1 & 1\\
\end{smallmatrix}\right]$.
Then $\widehat{H}=Z^{-1}HZ=\left[\begin{smallmatrix}
\frac{(1+\ii)}{\sqrt{2}} & 0\\
0 & \frac{(1-\ii)}{\sqrt{2}}\\
\end{smallmatrix}\right]=\left[\begin{smallmatrix}
e^{\ii\pi/4} & 0\\
0 & e^{-\ii\pi/4}\\
\end{smallmatrix}\right].$
Let $\widehat{f_6^H}=\widehat{H}\widehat{f_6}$.
Let $\widehat{\mathcal{F}_6}=\{\widehat{f_6}\}^{\widehat{B}}_{\neq_2}$ be the set of signature realizable by extending variables of $\widehat{f_6}$ with binary signatures in $\widehat{\mathcal{B}}$ using $\neq_2$, and  $\widehat{\mathcal{F}^H_6}=\{\widehat{f^H_6}\}^{\widehat{B}}_{\neq_2}$ be the set of signature realizable by extending variables of $\widehat{f^H_6}$ with binary signatures in $\widehat{\mathcal{B}}$ using $\neq_2$.
One can check that $\widehat{\mathcal{F}^H_6}=\widehat{H}\widehat{\mathcal{F}_6}\neq \widehat{\mathcal{F}_6}$.


\begin{lemma}\label{lem-6-notb}
Suppose that $\widehat{\mathcal{F}}$ is non-$\widehat{\mathcal{B}}$ hard.
Then,
$\holant{\neq_2}{\widehat{\mathcal{F}}}$ is \#P-hard if $\widehat{\mathcal{F}}$ contains a nonzero 6-ary signature $\widehat{f}\notin \widehat{\mathcal{B}}^{\otimes}\cup\widehat{\mathcal{F}_6}\cup\widehat{\mathcal{F}^H_6}$.
\end{lemma}
\begin{proof}
If $\widehat{f}$ is reducible,
since $\widehat{f}\notin \widehat{\mathcal{B}}^{\otimes}$,
then by factorization, we can realize a nonzero signature of odd arity or a nonzero signature of arity $2$ or $4$ that is not in $\widehat{\mathcal{B}}^{\otimes}$.
If we have a nonzero signature of odd arity, then we are done  by Theorem~\ref{odd-dic}.
If we have a nonzero signature of 2, then we are done  because  $\widehat{\mathcal{F}}$
is non-$\widehat{\mathcal{B}}$ hard.
If we have a nonzero signature of 4, then we are done by 
Lemma~\ref{lem-4-notb}.
Now we assume that $\widehat{f}$ is irreducible.
In particular, being irreducible,
$\widehat{f} \not \in \widehat{\mathcal{O}}{^\otimes}$.
For a contradiction, suppose that $\holant{\neq_2}{ \widehat{\mathcal{F}}}$ is not \#P-hard.
Then, by Theorem~\ref{lem-6-ary-f6}, 
$\widehat{f_6}$ is realizable from $\widehat{f}$.
Remember that we realize $\widehat{f_6}$ from $\widehat{f}$ by realizing $\widehat{f'}$, $\widehat{f''}$ and $\widehat{f^\ast}$ (Lemmas~\ref{lem-arity6-1},  \ref{lem-arity6-3} and \ref{lem-arity6-4}).
We will trace back this process and show that they are all in $\widehat{\mathcal{F}_6}\cup \widehat{\mathcal{F}^H_6}$, which contradicts with the condition that  $\widehat{f}\notin \widehat{\mathcal{F}_6}\cup\widehat{\mathcal{F}^H_6}.$

\begin{enumerate}
    \item 
First, by  Lemma~\ref{lem-arity6-4}, $\widehat{f_6}\in \{\widehat{f''}\}^{\widehat{\mathcal{B}}}_{\neq_2}.$
Then, by Lemma~\ref{lem-extending},
$\widehat{f''}\in \{\widehat{f_6}\}^{\widehat{\mathcal{B}}}_{\neq_2}=\widehat{\mathcal{F}_6}$.
\item Then, by Lemma~\ref{lem-arity6-3}, $\widehat{f''}=\widehat{Q}\widehat{f'}$ for some $\widehat{Q}=\left[\begin{smallmatrix}
e^{-\ii\delta} & 0\\
0 & e^{\ii\delta}\\
\end{smallmatrix}\right]\in \widehat{{\bf O}_2}$ where $0\leqslant \delta<\pi/2$, and the binary signature $\widehat{b}=(e^{\ii2\delta}, 0, 0, e^{-\ii2\delta})$ is realizable from $\widehat{f'}$ where $\widehat{f'}$ is realizable from $\widehat{f}$. 
Thus, $\widehat{b}$ is realizable from $\widehat{\mathcal{F}}$.
If $e^{\ii2\delta}\neq \pm 1$ or $\pm \ii$, then $\widehat{b}\notin\widehat{\mathcal{B}}^{\otimes}$.
Since $\widehat{\mathcal{F}}$ is non-$\widehat{\mathcal{B}}$ hard, we get \#P-harness. Contradiction.
Otherwise, since $0\leqslant \delta <\pi/2$, 
$e^{\ii2\delta}=1$ or $\ii$ and then, $\delta=0$ or $\pi/4$.
If $\delta=0$, then 
$e^{\ii\delta}=e^{-\ii\delta}=1$ and $\widehat{f''}=\widehat{Q}\widehat{f'}= \widehat{f'}$.
Thus, $\widehat{f'}\in \widehat{\mathcal{F}_6}$.
If $\delta=\pi/4$, then  $\widehat{f'}=\widehat{Q}^{-1}\widehat{f''}$ where $\widehat{Q}^{-1}= \left[\begin{smallmatrix}
e^{\ii\pi/4} & 0\\
0 & e^{-\ii\pi/4}\\
\end{smallmatrix}\right]= \widehat{H}.$
Since $\widehat{f''}\in \widehat{\mathcal{F}_6}$, 
$\widehat{f'}= \widehat{H}\widehat{f''}\in \widehat{H}\widehat{\mathcal{F}_6}=\widehat{\mathcal{F}_6^H}.$
\item Finally, by Lemma~\ref{lem-arity6-1}, 
 $\widehat{f'}$ is realized by extending variables of $\widehat{f}$ with nonzero binary signatures realized from $\widehat{\partial}_{12}\widehat{f}.$
 If we can realize a nonzero binary signature that is not in $\widehat{\mathcal{B}}^{\otimes 1}$ from $\widehat{\partial}_{12}\widehat{f}$ by factorization, then 
 since $\widehat{\mathcal{F}}$ is non-$\widehat{\mathcal{B}}$ hard,
 we get  \#P-hardness. Contradiction.
 Thus, we may assume that all nonzero binary signatures realizable from $\widehat{\partial}_{12}\widehat{f}$ are in $\widehat{\mathcal{B}}^{\otimes 1}$.
 Then, $\widehat{f'}$ is realized by extending variables of $\widehat{f}$ with nonzero binary signatures in $\widehat{\mathcal{B}}^{\otimes 1}$.
 Thus, $\widehat{f'}\in\{\widehat{f}\}^{\widehat{\mathcal{B}}}_{\neq_2}$.
 By Lemma~\ref{lem-extending}, $\widehat{f}\in\{\widehat{f'}\}^{\widehat{\mathcal{B}}}_{\neq_2}.$
 Since $\widehat{f'}\in \widehat{\mathcal{F}_6}$ or $\widehat{\mathcal{F}^H_6}$,  $\widehat{f}\in \widehat{\mathcal{F}_6}$ or $\widehat{\mathcal{F}^H_6}$. Contradiction.
\end{enumerate}
Thus, $\holant{\neq_2}{ \widehat{\mathcal{F}}}$ is \#P-hard if $\widehat{\mathcal{F}}$ contains a nonzero 6-ary signature $\widehat{f}\notin \widehat{\mathcal{B}}^{\otimes}\cup\widehat{\mathcal{F}_6}\cup\widehat{\mathcal{F}^H_6}$.
\end{proof}

 We go back to real-valued Holant problems under the $Z$-transformation. 
 Consider the problem $\Holant(f_6, \mathcal{F})$ where $$f_6=Z\widehat{f_6}={\chi_S \cdot (-1)^{x_1+x_2+x_3+x_1x_2+x_2x_3+x_1x_3+x_1x_4+x_2x_5+x_3x_6}}$$ and $S = \mathscr{S}(f_6)= \mathscr{E}_6.$
 The signature $f_6$ has a quite similar matrix form to $\widehat{f_6}$.
$$M_{123,456}({f_6})=\left[\begin{matrix}1 & 0 & 0 & 1 & 0 & 1 & 1 & 0\\
0 & 1 & -1 & 0 & -1 & 0 & 0 & 1\\
0 & -1 & 1 & 0 & -1 & 0 & 0 & 1\\
-1 & 0 & 0 & -1 & 0 & 1 & 1 & 0\\
0 & -1 & -1 & 0 & 1 & 0 & 0 & 1\\
-1 & 0 & 0 & 1 & 0 & -1 & 1 & 0\\
-1 & 0 & 0 & 1 & 0 & 1 & -1 & 0\\
0 & -1 & -1 & 0 & -1 & 0 & 0 & -1\\
\end{matrix}\right].$$

Since $\widehat{f^H_6}=\widehat{H}\widehat{f_6}=\widehat{Hf_6}$, $f_6^H=Z\widehat{f^H_6}=Hf_6$.
Also, since $\widehat{\mathcal{F}_6}=\{\widehat{f_6}\}_{\neq_2}^{\widehat{\mathcal{B}}}$,
$\mathcal{F}_6=Z\widehat{\mathcal{F}_6}=\{f_6\}^\mathcal{B}_{=_2}$ is the set of signatures realizable by extending variables of $f_6$ with binary signatures in $\mathcal{B}$ using $=_2$.
Similarly, since $\widehat{\mathcal{F}^H_6}=\{\widehat{f^H_6}\}_{\neq_2}^{\widehat{\mathcal{B}}}$, $\mathcal{F}^H_6=Z\widehat{\mathcal{F}^H_6}=\{f_6\}^\mathcal{B}_{=_2}$ is the set of signatures realizable by extending variables of $f^H_6$ with binary signatures in $\mathcal{B}$ using $=_2$. 
Notice that $f_6\in \mathscr{A}$ and $\mathcal{B}\subseteq\mathscr{A}$. Thus, $\mathcal{F}_6 \subseteq \mathscr{A}$.
Also, the binary signature $(1, 1, -1, 1)$ with a signature matrix $H$ is in $\mathscr{A}$.
Thus, $f_6^H\in \mathscr{A}$ and then  $\mathcal{F}^H_6 \subseteq \mathscr{A}$.
Also, $\mathscr{S}(f_6)=\mathscr{E}_6$ and one can check that $\mathscr{S}(f^H_6)=\mathscr{O}_6$.
Thus, for every $f\in \mathcal{F}_6\cup \mathcal{F}^H_6$, $\mathscr{S}(f)=\mathscr{E}_6$ or $\mathscr{O}_6$.
Since $f_6$ and $f^H_6$ satisfy {\sc 2nd-Orth}, one can easily check that every $f\in \mathcal{F}_6\cup \mathcal{F}^H_6$ satisfies {\sc 2nd-Orth}.

We want to show that $\Holant(f_6, \mathcal{F})\equiv_T\holant{\neq_2}{\widehat{f_6}, \widehat{\mathcal{F}}}$ is \#P-hard for all real-valued $\mathcal{F}$ that does not satisfy condition (\ref{main-thr}).
By Lemma~\ref{lem-2-notb}, $\{f_6\}\cup\mathcal{F}$ is non-${\mathcal{B}}$ hard.
We restate Lemmas~\ref{lem-4-notb} and \ref{lem-6-notb} in the setting of $\Holant(\mathcal{F})$ for non-${\mathcal{B}}$ hard $\mathcal{F}$.
\begin{corollary}\label{lem-non-f6}
Suppose that $\mathcal{F}$ is non-${\mathcal{B}}$ hard.
Then,
$\Holant(\mathcal{F})$ is \#P-hard if $\mathcal{F}$ contains  a nonzero signature $f$ of arity at most  $6$ where $f\notin \mathcal{B}^{\otimes}\cup\mathcal{F}_6\cup\mathcal{F}_6^H$.
\end{corollary}
\begin{remark}
Notice that $\mathcal{B}^{\otimes}\cup\mathcal{F}_6\cup\mathcal{F}_6^H\subseteq\mathscr{A}$. 
Thus, for any non-${\mathcal{B}}$ hard set $\mathcal{F}$,
$\Holant(\mathcal{F})$ is \#P-hard if $\mathcal{F}$ contains  a nonzero signature $f$ of arity at most  $6$ where $f\notin \mathscr{A}$.
\end{remark}

Now, we show that all four binary signatures in $\mathcal{B}$ are realizable from $f_6$. 
\begin{lemma}\label{lem-b-realize-f6}
$\Holant(\mathcal{B}, f_6, \mathcal{F})\leqslant\Holant(f_6, \mathcal{F})$.
\end{lemma}
\begin{proof}
Consider $\partial_{12}f_6$.
Notice that 
\setcounter{MaxMatrixCols}{20}
$$\begin{bmatrix}
{\bf f_6}_{12}^{00}\\
{\bf f_6}_{12}^{11}\\
\end{bmatrix}=
\begin{bmatrix}
1 & 0 & 0 & 1 & 0 & 1 & 1 & 0 & 0 & 1 & -1 & 0 & -1 & 0 & 0 & 1\\
-1 & 0 & 0 & 1 & 0 & 1 & -1 & 0 & 0 & -1 & -1 & 0 & -1 & 0 & 0 & -1\\
\end{bmatrix}.$$
Thus, $\partial_{12}f_6(x_3, x_4, x_5, x_6)={f_6}_{12}^{00}+{f_6}_{12}^{11}$ has the truth table $(0, 0, 0, 1, 0, 1, 0, 0, 0, 0, -1, 0, -1, 0, 0, 0).$ 
In other words, $\partial_{12}f_6(0011)=1$, $\partial_{12}f_6(0101)=1$, $\partial_{12}f_6(1010)=-1$, $\partial_{12}f_6(1100)=-1$, and $\partial_{12}f_6=0$ elsewhere.
Then, $$\mathscr{S}(\partial_{12}f_6)=\{(x_3, x_4, x_5,x_6)\in\mathbb{Z}_2^4\mid x_3\neq x_6 \land x_4\neq x_5\},$$ and $$\partial_{12}f_6(x_3, x_4, x_5, x_6)=
(\neq^-_2)(x_3, x_6)\otimes (\neq_2)(x_4, x_5).$$
Thus, by factorization we can realize $\neq_2^-$ and $\neq_2$.
Then connecting a variable of $\neq_2^-$ with a variable of $\neq_2$ (using $=_2$), we will get $=_2^-$.
Thus, $\mathcal{B}$ is realizable from $f_6$.
\end{proof}

We define  the problem $\Holantb(\mathcal{F})$ to be $\Holant(\mathcal{B} \cup \mathcal{F}).$
For all $\{i, j\}$ and every $b\in \mathcal{B}$,
 consider signatures $\partial^{b}_{ij}f_6$ (i.e., $\partial^{+}_{ij}f_6$, $\partial^{\widehat +}_{ij}f_6$, $\partial^{-}_{ij}f_6$ and  $\partial^{\widehat -}_{ij}f_6$) realized by merging variables $x_i$ and $x_j$ of $f_6$ using the binary signature $b$.
If there were one that is not in $\mathcal{B}^{\otimes 2}$, then by Corollary \ref{lem-non-f6}, we would be done.
However, it is observed in \cite{realodd} that $f_6$ satisfies the following Bell property. 
\begin{definition}[Bell property]
An irreducible signature $f$ satisfies the Bell property if for all pairs of indices $\{i, j\}$ and every $b\in \mathcal{B}$, $\partial^b_{ij}f\in \mathcal{B}^{\otimes}.$
\end{definition}

It can be directly checked that
\begin{lemma}
Every signature in $\mathcal{F}_6\cup\mathcal{F}^H_6$ satisfies the Bell property.
\end{lemma}

Now consider all possible  gadget constructions. 
If we could realize a signature of arity at most $6$ that is not in $\mathcal{B}^{\otimes}\cup\mathcal{F}_6\cup\mathcal{F}_6^H$ from $\mathcal{B}$ and $f_6$ by any possible gadget, 
then by Corollary \ref{lem-non-f6} there would be a somewhat more straightforward proof to our dichotomy theorem for the case of arity 6.
However, after many failed attempts, we believe there is
a more intrinsic reason why this approach cannot work.
The following conjecture  formulates this difficulty. This truly makes $f_6$ the
\emph{Lord of Intransigence} at arity 6.
\begin{conjecture}
All nonzero signatures of arity at most $6$ realizable from $\mathcal{B}\cup\{f_6\}$ are in  $\mathcal{B}^{\otimes}\cup\mathcal{F}_6$.
Also, all signatures of arity at most $6$ realizable from $\mathcal{B}\cup\{f^H_6\}$ are in  $\mathcal{B}^{\otimes}\cup\mathcal{F}^H_6.$
\end{conjecture}
So to prove the \#P-hardness of $\Holantb(f_6, \mathcal{F})$, we have to make additional use of   $\mathcal{F}$. In particular, we need to use a non-affine signature in $\mathcal{F}$.

\section{The \#P-hardness of $\Holantb(\mathcal{F})$}\label{sec-holantb}
In this section, we  prove that for all real-valued non-$\mathcal{B}$ hard set $\mathcal{F}$ that does not satisfy condition (\ref{main-thr}), $\Holantb(\mathcal{F})$ is \#P-hard  (Theorem~\ref{thm-holantb}).
For any real-valued  set $\mathcal{F}$ that does not satisfy condition (\ref{main-thr}), the set
$\{f_6\}\cup\mathcal{F}$ is non-$\mathcal{B}$ hard, and
since  $\mathcal{B}$ is realizable from $f_6$,
 $\Holant(f_6, \mathcal{F})$ is \#P-hard
by Theorem~\ref{thm-holantb}.
Combining with Theorem~\ref{lem-6-ary-f6}, we show that $\holant{\neq_6}{\widehat{\mathcal{F}}}$ is \#P-hard if $\widehat{\mathcal{F}}$ contains a 6-ary signature that is not in $\widehat{\mathcal{O}}^{\otimes}$ (Lemma~\ref{lem-6-ary}).

Since $\mathcal{F}$ does not satisfy condition (\ref{main-thr}), $\mathcal{F}\not\subseteq\mathscr{A}$.
Thus, it contains a signature $f$ of arity $2n$ that is not in $\mathscr{A}$.
In the following, we will prove the \#P-hardness of $\Holantb(\mathcal{F})$ where $\mathcal{F}$ is 
non-$\mathcal{B}$ hard by induction on $2n \geqslant 2$.
For the base cases $2n\leqslant 6$,
 by  Corollary~\ref{lem-non-f6} and the Remark after that,  $\Holantb( \mathcal{F})$ is \#P-hard.
 Then, starting with a signature  of arity $2n\geqslant 8$ that is not in $\mathscr{A}$, 
 we want to  realize a signature of lower arity $2k\leqslant 2n-2$ that is also not in $\mathscr{A}$, or else we get  \#P-hardness directly.
 If we can reduce the arity down to at most $6$, then we are done. 
 
 Let $f\notin \mathscr{A}$ be a nonzero signature of arity $2n\geqslant 8$.
We first show that  if $f$ does not have parity, then we get \#P-hardness (Lemma~\ref{lem-parity}). 
 Then, suppose that $f$ has parity.
 If $f$ is reducible, since $f$ has
 even arity (as we assumed so starting from Section~\ref{sec:proof-outline}),  it is a tensor product of two signatures of  odd arity, or a tensor product of two signatures of  even arity  which are not both in $\mathscr{A}$ since $f\notin \mathscr{A}$.
 Thus, by factorization, we can realize a nonzero signature of odd arity and we get  \#P-hardness by Theorem~\ref{odd-dic}, or we can realize a signature of lower even arity that is not in $\mathscr{A}$.
 Thus, we may assume that $f$ is irreducible. Then by Lemma~\ref{second-ortho} and
 the Remark after Definition~\ref{def:second-order-othor} we may assume  $f$ satisfies {\sc 2nd-Orth}.

Consider signatures $\partial^b_{ij}f$ (i.e., $\partial^+_{ij}f$, $\partial^-_{ij}f$,  $\partial^{\widehat+}_{ij}f$ and $\partial^{\widehat-}_{ij}f$) realized by merging variables $x_i$ and $x_j$ of $f$ using $b\in \mathcal{B}$ for all pairs of indices $\{i, j\}$ and every $b\in \mathcal{B}$.
If there is one signature that is not in $\mathscr{A}$, then we have realized a signature of arity $2n-2$ that is not in $\mathscr{A}$.
Otherwise, 
$\partial^b_{ij}f\in \mathscr{A}$ for all $\{i, j\}$ and every $b\in \mathcal{B}$. 
We denote this property by $f\in \int_{\mathcal{B}}\mathscr{A}.$ 
Now, assuming that $f$ has parity, $f$ satisfies {\sc 2nd-Orth} and $f\in \int_{\mathcal{B}}\mathscr{A},$  
we would like to reach a contradiction by showing that this would force $f$ itself to belong to $\mathscr{A}$.
However, quite amazingly, there do exist non-affine signatures that  satisfy these stringent conditions.
We will show how they are discovered and handled (Lemmas~\ref{lem-affine-norm}, \ref{lem-affine-support} and \ref{lem-in-affine}).

In this section, all signatures are real-valued.
When we say an entry of a signature has norm $a$, we mean it takes value $\pm a$. 
Since $\mathcal{B}$ is available in $\Holantb( \mathcal{F})$, if a signature $f$ is realizable in $\Holantb( \mathcal{F})$, then we can realize all signatures in $\{f\}^{\mathcal{B}}_{=_2}$ that are realizable by extending  $f$ with $\mathcal{B}^{\otimes 1}$ (using $=_2$). 
If we extend the variable $x_i$ of $f$ with $\neq_2$, 
then we will get a signature $g$ where ${g}_i^0=f_i^1$ and ${g}_i^1=f_i^0$. 
This is a flipping operation on the variable $x_i$.
If we extend the variable $x_i$ of $f$ with $=_2^-$, 
then we will get a signature $g$ where ${g}_i^0=f_i^0$ and ${g}_i^1=-f_i^1$. 
We call this a negating operation on the variable $x_i$.
In the following, once $f$ is realizable in $\Holantb( \mathcal{F})$, 
we may modify it by flipping or negating.
This will not change the complexity of the problem. 

\subsection{Parity  condition}
We first show that if $\mathcal{F}$ contains a signature that does not have  parity, then we can get  \#P-hardness.

\begin{lemma}\label{lem-parity}
Suppose that $\mathcal{F}$ is non-$\mathcal{B}$ hard and $\mathcal{F}$ contains a signature $f$ of arity $2n$. If $f$  does not have parity, then $\Holantb( \mathcal{F})$ is \#P-hard.
\end{lemma}
\begin{proof}
We prove this lemma by induction on $2n$.
We first consider the base case that $2n=2$.
Since $f$ has no parity, $f\notin \mathcal{B}$.
Since $\mathcal{F}$ is non-$\mathcal{B}$ hard, $\Holantb(\mathcal{F})$ is \#P-hard.

Now, suppose that $\Holantb(\mathcal{F})$ is \#P-hard when $2n=2k\geqslant 2$. 
Consider the case that $2n=2k+2 \geqslant 4$.
We will show that we can realize a signature $g$ of arity $2k$ with no parity from $f$, i.e., $\Holantb(g, \mathcal{F})\leqslant_T \Holantb( \mathcal{F})$.
Then by the induction hypothesis, we have $\Holantb( \mathcal{F})$ is \#P-hard when $2n=2k+2.$

Since $f$ has no parity,
$f\not\equiv 0$.
It has at least a nonzero entry. 
By flipping variables of $f$, we may assume that $f(\vec{0}^{2n})=x\neq 0$.
We denote $\vec{0}^{2n}$ by $\alpha=000\delta$ where $\delta=\vec{0}^{2n-3}.$
Since $f$ has no parity and $f(\vec{0}^{2n})\neq 0$, 
there exists an input $\alpha'$ with ${\rm wt}(\alpha')\equiv 1 \pmod 2$ such that $f(\alpha')=x'\neq 0$. 
Since $2n\geqslant 4$, we can find three bits of $\alpha'$ such that on these three bits, the values of $\alpha'$ are the same.
By renaming variables of $f$ which gives a permutation of $\alpha'$, without loss of generality,  we may assume that these are the first three bits, i.e, $\alpha'_1=\alpha'_2=\alpha'_3$.

We first consider the case that $\alpha'_1\alpha'_2\alpha'_3=000$.
Then, $\alpha'=000\delta'$ for some $\delta'\in \mathbb{Z}_2^{2n-3}$ where ${\rm wt}(\delta')={\rm wt}(000\delta')={\rm wt}(\alpha')\equiv 1 \pmod 2.$
We consider the following six entries of $f$.
$$x=f(000\delta), x'=f(000\delta'), y=f(011\delta), y'=f(011\delta'), z=f(101\delta), z'=f(101\delta').$$

Consider signatures $\partial^+_{23}f$ and $\partial^{-}_{23}f$ realized by connecting variables $x_2$ and $x_3$ of $f$ using $=^+_2$ and ${=}^{-}_2$ respectively.
Clearly, $\partial^+_{23}f$ and $\partial^{-}_{23}f$ have arity $2n-2$. 
If one of them has no parity, then we are done.
Thus, we may assume that $\partial^+_{23}f$ and $\partial^{-}_{23}f$ both have parity.
Note that $x+y$ and $x'+y'$ are entries of the signature $\partial^+_{23}f$ on inputs $0\delta$ and $0\delta'$ respectively. 
Clearly, ${\rm wt}(0\delta)=0$ and ${\rm wt}(0\delta')\equiv1 \pmod 2.$
Since $\partial^+_{23}f$ has parity, at least one of $x+y$ and $x'+y'$ is zero. 
Thus, we have $(x+y)(x'+y')=0$.
Also, note that $x-y$ and $x'-y'$ are entries of the signature $\partial^{-}_{23}f$ on inputs $0\delta$ and $0\delta'$ respectively. 
Then, since $\partial^-_{23}f$ has parity, similarly we have $(x-y)(x'-y')=0$.
Thus, 
\begin{equation}\label{equ-xxyy}
    (x+y)(x'+y')+(x-y)(x'-y')=2(xx'+yy')=0.
\end{equation}

Consider signatures $\partial^+_{13}f$ and $\partial^{-}_{13}f$ realized by connecting variables $x_1$ and $x_3$ of $f$ using $=_2$ and ${=}^{-}_2$ respectively. 
Again if one of them has no parity, then we are done.
Suppose that $\partial^+_{13}f$ and $\partial^{-}_{13}f$ both have parity.
Then, $(x+z)(x'+z')=0$ since $x+z$ and $x'+z'$ are entries of $\partial^+_{13}f$ on  inputs $0\delta$ and $0\delta'$ respectively.  
Similarly, $(x-z)(x'-z')=0$.
Thus, 
\begin{equation}\label{equ-xxzz}
    (x+z)(x'+z')+(x-z)(x'-z')=2(xx'+zz')=0.
\end{equation}

Consider signatures $\partial^{\widehat+}_{12}f$ and ${\partial}^{\widehat-}_{12}f$ realized by connecting variables $x_1$ and $x_2$ of $f$ using $\neq_2$ and ${\neq}^{-}_2$ respectively. 
Again if one of them has no parity, then we are done.
Suppose that  $\partial^{\widehat+}_{12}f$ and ${\partial}^{\widehat-}_{12}f$ both have parity. 
Then, $(y+z)(y'+z')=0$ since $y+z$ and $y'+z'$ are entries of $\partial^{\widehat+}_{12}f$ on inputs $1\delta$ and $1\delta'$  respectively, and ${\rm wt}(1\delta)=1$ and ${\rm wt}(1\delta')\equiv 0 (\bmod 2)$.  
Similarly, $(y-z)(y'-z')=0$.
Thus, 
\begin{equation}\label{equ-yyzz}
   (y+z)(y'+z')+(y-z)(y'-z')=2(yy'+zz')=0. 
\end{equation}
Then, consider (\ref{equ-xxyy}) $+$ (\ref{equ-xxzz}) $-$ (\ref{equ-yyzz}). We have $xx'=0.$
However, since $x=f(\vec{0}^{2n})\neq 0$ and $x'=f(\alpha')\neq 0$, $xx'\neq$ 0. Contradiction.

For the case that $\alpha'_1\alpha'_2\alpha'_3=111$, we have $\alpha'=111\delta'$ for some $\delta'\in \mathbb{Z}_2^{2n-3}$ where ${\rm wt}(\delta')={\rm wt}(111\delta')-3={\rm wt}(\alpha')-3\equiv 0 \pmod 2.$
We consider the following six entries of $f$.
$$x=f(000\delta), x'=f(111\delta'), y=f(011\delta), y'=f(100\delta'), z=f(101\delta), z'=f(010\delta').$$
We still consider signatures $\partial^+_{23}f$, $\partial^{-}_{23}f$, $\partial^+_{13}f$, $\partial^{-}_{13}f$, $\partial^{\widehat+}_{12}f$ and ${\partial}^{\widehat-}_{12}f$ and suppose that they all have parity.
Then, similar to the above proof of the case $\alpha'_1\alpha'_2\alpha'_3=000$, we can show that 
$xx'=0$. Contradiction. 

Thus, among $\partial^+_{23}f$, $\partial^{-}_{23}f$, $\partial^+_{13}f$, $\partial^{-}_{13}f$, $\partial^{\widehat+}_{12}f$ and ${\partial}^{\widehat-}_{12}f$, at least one does not have parity. 
Thus, we realized a $2k$-ary signature with no parity. 
By our induction hypothesis, we are done. 
\end{proof}

\subsection{Norm condition}
 Under the assumptions that $f$ has parity, $f$ satisfies {\sc 2nd-Orth} and $f\in \int_{\mathcal{B}}\mathscr{A},$ 
 we consider whether all nonzero entries of $f$ have the same norm.
In  Lemma~\ref{lem-affine-norm}, 
we will show that the answer is yes, but only for signatures of arity $2n\geqslant 10$ (this lemma does not require $\mathcal{F}$ to be non-$\mathcal{B}$ hard). 
For a signature $f$ of arity $2n=8$, we show that either all nonzero entries of $f$ have the same norm, or one of the following signatures $g_8$ or $g'_8 $ is realizable. These two signatures are defined by
$g_8 = \chi_S - 4 \cdot f_8$ and $g'_8=q_8-   4\cdot f_8$,  where 
$$S = \mathscr{S}(q_8)=\mathscr{E}_8,~~~~
q_8=\chi_S (-1)^{\sum_{1\leqslant i<j\leqslant 8} x_ix_j}  ~~~~\text{ and }$$ 
\begin{equation}\label{equ-T-polynomial}
 {\small \begin{aligned}
f_8=\chi_T  \text { with } T = \mathscr{S}(f_8)=\{(x_1, x_2, \ldots, x_8)\in \mathbb{Z}_2^{8} \mid~
&x_1+x_2+x_3+x_4= 0, ~ x_5+x_6+x_7+x_8= 0,\\
&x_1+x_2+x_5+x_6= 0, ~ x_1+x_3+x_5+x_7= 0
\}.
\end{aligned}}
\end{equation}
It is here the function $f_8$ makes its first appearance, we dub it 
the \emph{Queen of the Night}.
Clearly, $g_8,g'_8\notin \mathscr{A}$ since their nonzero entries have two different norms 1 and 3. 
One can check that $g_8$ and $g'_8$ have parity, $g_8$ and $g'_8$ satisfy {\sc 2nd-Orth} and $g_8, g'_8\in \int_{\mathcal{B}}\mathscr{A}.$
Thus, one cannot get a non-affine signature by connecting two variables of $g_8$ or $g'_8$ using signatures in $\mathcal{B}$.
However, fortunately by merging two arbitrary variables of $g_8$ using $=_2$ and  two arbitrary variables of $g'_8$ using $=^-_2$, 
we can get 6-ary irreducible signatures that do not satisfy {\sc 2nd-Orth}.
Thus, we get \#P-hardness.

The following Lemma \ref{lem-indenpent-set} regarding the independence number of a family of special graphs is at the heart of the discovery of the signature $g_8$. 
It should be of independent interest. 

\begin{definition}\label{def-g2n-h2n}
Define the graphs $G_{2n}$ and $H_{2n}$ as follows.
The vertex set $V(G_{2n})$ is the set $\mathscr{E}_{2n}$ of all even weighted points 
in $\mathbb{Z}_2^{2n}$. 
The vertex set $V(H_{2n})$ is the set $\mathscr{O}_{2n}$ of all odd weighted points 
in $\mathbb{Z}_2^{2n}$. 
Two points $u, v\in \mathscr{E}_{2n}$ (or $\mathscr{O}_{2n}$)
are connected by an edge iff ${\rm wt}(u \oplus v)=2.$ 
\end{definition}

Let $\alpha(G_{2n})$ be the independence number of  $G_{2n}$ i.e, the size of a maximum independent set of $G_{2n}$, and $\alpha(H_{2n})$ be the independence number of  $H_{2n}.$
Let $S\subseteq [2n]$.
We define $\varphi_S$ be a mapping that flips the values on  bits in $S$ for all $u\in \mathscr{E}_{2n}.$
In other words, suppose that $u'=\varphi_S(u)$. Then,
$u'_i=\overline{u_i}$ if $i\in S$ and $u'_i=u_i$ if $i \notin S$ where $u'_i$ and $u_i$ are values of $u$ and $u$ on bit $i$ respectively.
For all $S$,
clearly ${\rm wt}(u\oplus v)=2$ iff ${\rm wt}(\varphi_S(u)\oplus\varphi_S(v))=2.$
When $|S|$ is odd, $\varphi_S( \mathscr{E}_{2n})= \mathscr{O}_{2n}$. 
One can easily check that $\varphi_S$ gives an isomorphism between $G_{2n}$ and $H_{2n}$.
When $|S|$ is even, $\varphi_S( \mathscr{E}_{2n})= \mathscr{E}_{2n}$. 
Then, $\varphi_S$ gives an automorphism of $G_{2n}$.
Also, by permuting these $2n$ bits, we can get an automorphism of $G_{2n}$.
In fact, the automorphism group of $G_{2n}$ is generated by these operations. 

\begin{lemma}\label{authomorphism}
Let $2n \geqslant 6$.
Every  automorphism $\psi$ of $G_{2n}$
is a product $ \varphi_S  \circ \pi$ for some  automorphism $\pi$ induced by a permutation of $2n$ bits, and an  automorphism $\varphi_S$ given by flipping the values on some bits in a set $S$ of even cardinality.
\end{lemma}

\begin{proof}
Let $\psi$ be an arbitrary
automorphism  of $G_{2n}$.
Suppose $\psi(\vec{0}^{2n})
= u$. 
Let $S\subseteq [2n]$ be the index
set where $u_i=1$.
Then  $|S| = {\rm wt}(u)$ is even, and 
$\psi' =\varphi_S  \circ \psi$
fixes $\vec{0}^{2n}$.
Consider $\psi'(v)$ for all
$v \in  \mathscr{E}_{2n}$ of
${\rm wt}(v) =2$.
Since $\psi'$ is an automorphism
fixing $\vec{0}^{2n}$,
$\psi'(v)$ has weight 2.
We denote by $e_{ij}$ the $2n$-bit string with ${\rm wt}(e_{ij}) = 2$
having 1's on bits $i$ and $j$, for $1  \leqslant i < j \leqslant 2n$. Then,  $e_{12}
= 11 \vec{0}^{2n-2}$.
By a suitable permutation $\pi$ of the variables, we have
$\pi \circ \psi'(e_{12})
= e_{12}$,
while still fixing $\vec{0}^{2n}$.
We will show that $\pi \circ \psi'=\pi \circ \varphi_S \circ \psi$ is the identity mapping, i.e., $\pi \circ \varphi_S \circ \psi= 1_{G_{2n}}.$
Then, $\psi=\varphi^{-1}_S \circ \pi^{-1}.$ We are done.

For simplicity of notations, we reuse $\psi$ to denote $\pi \circ \psi'$ in the following and we show that $\psi=1_{G_{2n}}$.
Consider
$e_{1i}$,
for $3 \leqslant i \leqslant 2n$.
Note that $\psi(e_{1i})$ is some $e_{st}$ and must have Hamming distance
2 to $e_{12}$. It is easy to see that the only possibilities are $s \in \{1,2\}$ and $t >2$, i.e.,
from $e_{12}$
we flip exactly one bit  in  $\{1, 2\}$ and another bit in $\{3, \ldots, 2n\}$.
Suppose there are $i \not = i'$ $(i, i'\geqslant 3)$ such that  $\psi(e_{1i}) = e_{1t}$
and  $\psi(e_{1i'}) = e_{2t'}$. Since ${\rm wt}(e_{1i} \oplus e_{1i'}) = 2$,
we must have $t=t'$. 
Since $2n\geqslant 6$, we can pick another $i''\geqslant 3$ such that $i''\neq i$ and $i'$.
Then, this leads to a contradiction 
since $e_{1i''}$
must either be mapped to $e_{1t}$ if $\psi(e_{1i''}) = e_{1t''}$, 
or be mapped to $e_{2t}$ if $\psi(e_{1i''}) = e_{2t''}$;  neither is possible.  Thus either
$\psi(e_{1i})$ is some $e_{1t}$ for all $3\leqslant i\geqslant 2n$, or is some $e_{2t}$ for all $3\leqslant i\geqslant 2n$.
By a permutation of $\{1,2\}$ (which maintains the property that $\psi$ fixes $\vec{0}^{2n}$ and $e_{12}$)
we may assume it is the former. Then the mapping $i \mapsto t$ given by  $\psi(e_{1i})= e_{1t}$ for $3\leqslant i\geqslant 2n$ defines
a permutation of the variables for $3\leqslant i\geqslant 2n$ (which again 
maintains the property that $\psi$ fixes $\vec{0}^{2n}$ and $e_{12}$) and, after a permutation of the variables
we may now assume that $\psi$ fixes $\vec{0}^{2n}$ and all $e_{1i}$.
For any $1 \leqslant i < j \leqslant 2n$,
we have ${\rm wt}(\psi(e_{ij})) =2$ and  $\psi(e_{ij})$ has distance
2 from both $\psi(e_{1i}) = e_{1i}$  and $\psi(e_{1j}) = e_{1j}$.
Then $\psi(e_{ij})$ must be obtained from $e_{1i}$ by flipping exactly one bit in
$\{1, i\}$ and another bit out of $\{1, i\}$. 
However, it cannot flip bit $i$ which would result in
some $e_{1t}$ for some $t>2$, because $\psi$ already fixed $e_{1t}$.
Thus, it flips bit 1 but not bit $i$.
Similarly in view of $e_{1j}$, we must flip bit $1$ but not bit $j$.
Hence $\psi(e_{ij}) = e_{ij}$, and therefore $\psi$ fixes 
all $v$ with Hamming weight ${\rm wt}(v) \leqslant 2$.

Inductively assume $\psi$ fixes all $v$ of ${\rm wt}(v) \leqslant 2k$,
for some $k \geqslant 1$.
If $k < n$ we prove that $\psi$ also fixes all $v$ of ${\rm wt}(v) = 2k+2$.
For notational simplicity we may assume  $v = \vec{1}^{2k+2} \vec{0}^{2n-2k-2}$.
As $2k+2 \geqslant 4$, we can choose
$u = \vec{1}^{2k} 00 \vec{0}^{2n-2k-2}$ and $w = 00 \vec{1}^{2k} \vec{0}^{2n-2k-2}$,
and the two 00 in $u$ and $w$ among  the first $2k+2$ bits are in
disjoint bit positions.
Clearly ${\rm wt}(\psi(v)) \geqslant 2k+2$ since all
strings of ${\rm wt} \leqslant 2k$ are fixed. 
Also since $\psi(v)$ has Hamming distance 2 from
$\psi(u) =u$ and $\psi(w) = w$,
it has weight exactly $2k+2$, and is obtained from $u$
by flipping two bits from 00 to 11 in positions $>2k$,
as well as  obtained from $w$
by flipping two bits from 00 to 11 in positions in $\{1,2\}  \cup\{t \mid t>2k+2\}$.
In particular, it is 1 in  positions 1 to $2k$ (in view of $u$),
and  it is also 1 in  positions 3 to $2k+2$.
But together these positions cover all bits 1 to $2k+2$.
Thus $\psi(v) = v$.  This completes the induction, and  proves the lemma for all $2n \geqslant 6$.
%
\end{proof}

\begin{remark}
The condition $2n \geqslant 6$ in
Lemma~\ref{authomorphism} is necessary.
Here is a counter example for $2n=4$:
$\psi$ fixes $0000$ and $1111$,
and it maps $\alpha$ to $\overline{\alpha}$
for all $\alpha \in \{0, 1\}^4$ with
${\rm wt}(\alpha) =2$.
If $\psi$ were to be expressible as
$\varphi_S \circ \pi$, then since
$\psi(0000) = 0000$, we have
$S= \emptyset$. Then by $\psi(0011) = 1100$ and $\psi(0101) = 1010$, the permutation $\pi$ must
map variable $x_1$ to $x_4$. However this
violates $\psi(1001) = 0110$.
\end{remark}

\begin{lemma}\label{lem-indenpent-set}
Let $\{G_{2n}\}$ be the sequence of graphs defined above.
\begin{itemize}
\item If $2n=8$, then $\alpha(G_{8})=\frac{1}{8}|\mathscr{E}_{8}| =  2^{4},$ and the maximum independent set $I_8$ of $G_8$ is unique up to an automorphism, 
where 
\begin{equation*}
\begin{aligned}
I_8=&\{00000000, 00001111, 00110011, 00111100, 01010101, 01011010, 01100110, 01101001,\\
& ~10010110,  10011001, 10100101, 10101010, 11000011, 11001100, 11110000, 11111111\}.
\end{aligned}
\end{equation*}
\item  If $2n\geqslant 10$, then $\alpha(G_{2n})<\frac{1}{8}|\mathscr{E}_{2n}|= 2^{2n-4}$.
\end{itemize}
\end{lemma}

\begin{proof}

We first consider the case  $2n=6$. One can check that the following set $$I_6=\{000000, 001111, 110011, 111100\}$$ is an independent set of $G_6$.
Thus, $\alpha(G_6)\geqslant 4$.
Next, we show that $\alpha(G_6)= 4$ and $I_6$ is the unique maximum independent set of $\alpha(G_6)$ up to an  automorphism.

 Let $J_6$ be an maximum independent set of $G_6$.
 Clearly, $|J_6|\geqslant 4$.
After an automorphism of $G_6$ by flipping some bits, we may assume that $\vec{0}^6\in J_6$.
Then for any $u\in \mathscr{E}_6$ with ${\rm wt}(u)=2$, $u\notin J_6$.
If $\vec{1}^6\in J_6$, then for any $u\in \mathscr{E}_6$ with ${\rm wt}(u)=4$, $u\notin J_6$.
Thus, $J_6$ is maximal with $|J_6|=2<4$, a contradiction.
Thus, we have $\vec{1}^6\notin J_6$.
Then all vertices in $J_6$, except $\vec{0}^6$ have hamming weight $4$. 
After an automorphism by permuting bits (this will not change $\vec{0}^6$), we may assume that $u=001111\in  J_6$.
Consider some other $v\in J_6$ with ${\rm wt}(v)=4$.
If $v_1v_2=01$ or $10$, then ${\rm wt}(v_3v_4v_5v_6)=3$.
Thus, ${\rm wt}(u\oplus v)={\rm wt}(00\oplus v_1v_2)+{\rm wt}(1111\oplus v_3v_4v_5v_6)=1+1=2$.
Contradiction.
The only  $v\in J_6$ with ${\rm wt}(v)=4$, and
$v_1v_2 = 00$ is $v= 001111=u$.
Thus, $v_1v_2=11$, i.e., both bits of $v$ are 1
where $u$ is $00$.
After an automorphism by permuting bits in $\{3, 4, 5, 6\}$ (this will not change $\vec{0}^6$ and $u$), we may assume that $v=110011\in  J_6$.
For any other $w\in J_6$ with ${\rm wt}(w)=4$, 
we must have 
$w_1w_2=11$ (by the same
proof for the
pair $(u,v)$ applied to  $(u,w)$),
and also $w_3w_4=11$
(by the same
proof for the
pair $(u,v)$ applied to  $(v,w)$).
Thus, $w=111100$.
Then, $J_6=\{\vec{0}^6, u, v, w\}=I_6$ is maximal. 
Thus, $\alpha(G_6)= 4$ and $I_6$ is the unique maximum independent set of $\alpha(G_6)$ up to an  automorphism.

Consider $2n=8$. One can check that $I_8$ is an independent set of $G_8$.
Thus, $\alpha(G_8)\geqslant 16$.
We use $G_8^{ab}$ to denote the subgraph of $G_8$  induced by vertices $\{u\in \mathscr{E}_8\mid u_1u_2=ab\}$ for $(a, b)\in \mathbb{Z}_2^2$.
Clearly, $G_8^{00}$ and $G_8^{11}$ are isomorphic to $G_6$, and  $G_8^{01}$ and $G_8^{10}$ are isomorphic to $H_6$.
Since $H_6$ is isomorphic to $G_6$, $G_8^{01}$ and $G_8^{10}$ are also isomorphic to $G_6$.
Let $J_8$ be a maximum independent set of $G_8$. Clearly, $|J_8|\geqslant |I_8|= 16$.
Also, we use $J_8^{ab}$ to denote the subset $\{u\in J_8\mid u_1u_2=ab\}$ for $(a, b)\in \mathbb{Z}_2^2$.
Similarly, we can define $I^{ab}_8$.
Since $J_8$ is an independent set of $G_8$, clearly, for every $(a, b)\in \mathbb{Z}_2^2$, $J_8^{ab}$ is an independent set of $G_8^{ab}$. 
Since $G_8^{ab}$ is isomorphic to $G_6$ and  $\alpha(G_6)=4$, thus $|J_8^{ab}|\leqslant 4$.
Then $|J_8|\leqslant 16$.
Thus,    $|J_8|=16$, and $|J_8^{ab}|= 4$ for every $(a, b)\in \mathbb{Z}_2^2$.
Since the maximum independent set of $G_6$ is unique up to an automorphism of $G_6$,
which can be extended to an
automorphism of $G_8$ by fixing the first two bits, we may assume that 
$$J_8^{00}=I^{00}_8=\{00000000, 00001111, 00110011, 00111100\}$$ after an automorphism of 
$G_8$.


Then, consider $J_8^{01}$. 
 We show that for any $u\in J_8^{01}$, $u_3\neq u_4$, $u_5\neq u_6$ and $u_7\neq u_8$.
Otherwise, by switching the pair of bits $\{3, 4\}$ with $\{5, 6\}$ or $\{7, 8\}$ (this will not change $J^{00}_8$), we may assume that $u_3=u_4$. 
Then ${\rm wt}(u_1u_2u_3u_4)$ is odd.
Since ${\rm wt}(u)$ is even, ${\rm wt}(u_5u_6u_7u_8)$ is odd.
Thus, either $u_5=u_6$ or $u_7=u_8$.
Still by switching the pair $\{5, 6\}$ with $\{7, 8\}$
(again this will not change $J^{00}_8$), we may assume that $u_5=u_6$. 
Then since ${\rm wt}(u_5u_6u_7u_8)$ is odd, we have $u_7\neq u_8$.
Then, one can check that there exists some  $v\in J_8^{00}$ such that $v_3v_4v_5v_6=u_3u_4u_5u_6$.
Since $v_1=v_2$ and $u_1\neq u_2$, ${\rm wt}(u_1u_2\oplus v_1v_2)=1$.
Also since  $v_7=v_8$ and $u_7\neq u_8$, ${\rm wt}(u_7u_8\oplus v_7v_8)=1$.
Thus, ${\rm wt}(u\oplus v)={\rm wt}(u_1u_2\oplus v_1v_2)+{\rm wt}(u_7u_8\oplus v_7v_8)=2.$ 
This means that the vertices $u$ and $v$ are connected in the graph $G_8$, a 
contradiction.
Thus, for any $u\in J_8^{01}$, $u_3\neq u_4$, $u_5\neq u_6$ and $u_7\neq u_8$.
By permuting bit $3$ with bit $4$, bit $5$ with bit $6$, and bit $7$ with bit $8$ (this will not change $J^{00}_{8})$, we may assume that $01010101\in J^{01}_8.$
Consider some other $w\in J^{01}_8$.
Since $w_{2i+1}\neq w_{2i+2}$ for any $i=1, 2$ or $3$, the pair $w_{2i+1}w_{2i+2}=01$ or $10$.
Among these three pairs, let $k$ denote the number of pairs that are $01$.
If $k=3$, then $w=01010101$. Contraction.
If $k=2$, then ${\rm wt}(01010101\oplus w)=2$. Contradiction.
If $k=0$, then $w=01101010$.
One can check that $\{01010101, 01101010\}$ is already a maximal independent set of $G^{01}_8$ and it has size $2<4$. Contradiction.
Thus, $k=1$. Then, $w$ can take ${3\choose 1}$ possible values.
Thus, $$J^{01}_8\subseteq I^{01}_8=\{01010101, 01011010, 01100110, 01101001\}.$$
Since, $|J^{01}_8|=4$, $J^{01}_8=I^{01}_8.$

 Consider some $u\in J_8^{10}$.
 Similar to the proof of $ J_8^{01}$, we can show that $u_3\neq u_4$, $u_5\neq u_6$ and $u_7\neq u_8$.
Thus, $u$ can take $2^3$ possible values. 
Moreover, for any $01u'\in J^{01}_8$, $10u'\notin J^{10}_8.$
Thus, there are only four remaining values that $u$ can take. 
Then, 
$$J^{10}_8\subseteq I^{10}_8=\{10010110, 10011001, 10100101, 10101010\}.$$
Since $|J^{10}_8|=4$, $J^{10}_8=I^{10}_8$.

Finally, consider $J_{8}^{11}$. 
We show that for any $u\in J_{8}^{11}$, $u_3=u_4$, $u_5=u_6$ and $u_7=u_8$. 
Otherwise, by permuting the pair of bits $\{3, 4\}$ with $\{5, 6\}$ or $\{7, 8\}$ (one can check that this will not change $J^{01}_8$ and $J^{10}_8$), we may assume that $u_3\neq u_4$.
Since ${\rm wt}(u)$ is even, between ${\rm wt}(u_5u_6)$ and ${\rm wt}(u_7u_8)$, exactly one is even and the other is odd.
By permuting the pair of bits $\{5, 6\}$ with $\{7, 8\}$, we may further assume that $u_5\neq u_6$ and $u_7= u_8$.
Then, one can check that there exists some $v\in J_8^{01}$ such that $u_3u_4u_5u_6=v_3v_4v_5v_6$. 
Since  $u_1=u_2$ and $v_1\neq v_2$, 
 ${\rm wt}(u_1u_2\oplus v_1v_2)=1$.
Also since  $u_7= u_8$ and $v_7\neq v_8$, ${\rm wt}(u_7u_8\oplus v_7v_8)=1$.
Thus, ${\rm wt}(u\oplus v)={\rm wt}(u_1u_2\oplus v_1v_2)+{\rm wt}(u_7u_8\oplus v_7v_8)=2.$ 
Contradiction.
Thus, for any $u\in J_8^{11}$, it can take $2^3$ possible values.
Moreover, for any $00u'\in J_{8}^{00}$, we have $11u'\notin J_{8}^{11}$.
Thus, there are only four remaining values that $u$ can take.
Then, 
$$J_8^{11}\subseteq I_8^{11}=\{11000011, 11001100, 11110000, 11111111\}.$$
Thus,  after an automorphism, $J_8=I_8$. In other words, $I_8$ is the unique maximum independent set of $G_8$ up to an  automorphism.

Now, we consider the case  $2n\geqslant 10$.
For every $(a, b)\in \mathbb{Z}_2^2$, we define $G_{2n}^{ab}$ to be the subgraph of $G_{2n}$ induced by $\{u\in G_{2n}\mid u_1u_2=ab\}$, and  it is isomorphic to $G_{2n-2}.$
Thus, $$\alpha(G_{2n})\leqslant \alpha(G_{2n}^{00}) +\alpha(G_{2n}^{01})+\alpha(G_{2n}^{10})+\alpha(G_{2n}^{11})=4\alpha(G_{2n-2}).$$
Then, $\alpha(G_{2n-2})<2^{(2n-2)-4}$ will imply that $\alpha(G_{2n})<2^{2n-4}.$
Thus, in order to prove $\alpha(G_{2n})<2^{2n-4}$ for all $2n \geqslant 10$, it suffices to prove that $\alpha(G_{10})<2^{10-4}.$ 
For a contradiction, suppose that $\alpha(G_{10})\geqslant 2^{10-4}.$
Let $I$ be a maximum independent set of $G_{10}$.
Then, $|I|\geqslant 2^6$.
We define $I^{ab}=\{u\in I\mid u_1u_2=ab\}$ for every $(a, b)\in \mathbb{Z}_2.$
Since $G^{ab}_{10}$ is isomorphic to $G_8$ and $\alpha(G_8)= 2^4$, $|I^{ab}|\leqslant 2^4$ for every $(a, b)\in \mathbb{Z}_2^2$.
Then, $|I|\leqslant 4\cdot 2^4$. 
Thus, $|I|=2^6$ and $|I^{ab}|= 2^4$ for every $(a, b)\in \mathbb{Z}_2^2$.
Since the maximum independent set of $G_8$ is unique up to an automorphism of $G_8$ which can be extended to an automorphism of $G_{10}$ by fixing the first two bits, we may assume that $I^{00}=\{00u\mid u\in I_8\}$.

Consider $I^{01}$.
Since $|I^{01}|\neq 0$, there exists some $01v\in I^{01}$.
Since ${\rm wt}(v)$ is odd, among  ${\rm wt}(v_3v_4)$, ${\rm wt}(v_5v_6)$, ${\rm wt}(v_7v_8)$ and ${\rm wt}(v_9v_{10})$, there is an odd number (either one  or three) of pairs such that ${\rm wt}(v_{2i+1}v_{2i+2})$ $(1 \leqslant i \leqslant 4)$ is odd, i.e., $v_{2i+1}\neq v_{2i+2}.$
In other words, there are exactly three pairs among $v_3v_4, v_5v_6, v_7v_8$ and $v_9v_{10}$ such that the values inside each pair are all equal  with  each other or all distinct with each other.
By permuting these pairs of bits $\{3, 4\}$, $\{5, 6\}$, $\{7, 8\}$ and $\{9, 10\}$ (this will not change $I^{00}$), we may assume that either $v_3=v_4$, $v_5=v_6$, $v_7=v_8$ and $v_9\neq v_{10}$, or $v_3\neq v_4$, $v_5\neq v_6$, $v_7\neq v_8$ and $v_9=v_{10}$.
In both cases, one can check that there exists some $00u\in I^{00}$ such that $u_i=v_i$ for $i\in \{3, \ldots, 8\}.$
Moreover, $u_{9}=u_{10}$ if $v_9\neq v_{10}$, and  $u_{9}\neq u_{10}$ if $v_9= v_{10}.$
Then, ${\rm wt}(00u\oplus 01v)={\rm wt}(00\oplus 01)+{\rm wt}(u_9u_{10}\oplus v_9v_{10})=2$.
This contradiction proves that
$\alpha(G_{10})<2^{10-4}$, and the lemma is proved.
\end{proof}
\begin{remark}
We remark that $I_8=\mathscr{S}(f_8)$.
Later, we will see that the signature $f_8$, this Queen of the Night, and its support $\mathscr{S}(f_8)$ have even more extraordinary properties.
\end{remark}

We consider a particular gadget construction that will be used in  our proof. 
Let  $h_4$  be a 4-ary signature with signature matrix   $M_{12,34}(h_4)=H_4=\left[\begin{smallmatrix}
1 & 0 & 0 & 1\\
0 & 1 & 1 & 0\\
0 & 1 & -1 & 0\\
1 & 0 & 0 & -1\\
\end{smallmatrix}\right].$
Notice that $H_4H^{\tt T}_4=H_4H_4=2I_4$, and $h_4$ is an affine signature. 
The following is
 called an $H_4$ gadget construction on $f$,
 denoted by $^{H_4}_{ij}f$.
 This is the signature obtained by
 connecting variables $x_3$ and $x_4$ of $h_4$ with variables $x_i$ and $x_j$ of $f$ using $=_2$, respectively.
Note that $^{H_4}_{ij}f$ is not necessarily realizable from $f$ since  $h_4$ may not be available. 
However, we will analyze  the structure of $f$ by analyzing $^{H_4}_{ij}f$.
For convenience, we consider $(i, j)=(1, 2)$ and we use  $\widetilde f$ to denote $^{H_4}_{12}f$. 
  The following results (Lemmas~\ref{lem-tilde-1} and \ref{lem-norm-f-tilde}) about $\widetilde f= {^{H_4}_{12}}f$ hold for all $^{H_4}_{ij}f$ by replacing $\{1, 2\}$ with $\{i, j\}$.
Note that 
$\widetilde{f}$ has the following signature matrix 
$$M_{12}(\widetilde{f})=
\begin{bmatrix}
{\bf \widetilde f}^{00}_{12}\\
{\bf \widetilde f}^{01}_{12}\\
{\bf \widetilde f}^{10}_{12}\\
{\bf \widetilde f}^{11}_{12}\\
\end{bmatrix}
=H_4M_{12}({f})=
\begin{bmatrix}
{\bf f}_{12}^{00}+{\bf f}_{12}^{11}\\
{\bf f}_{12}^{01}+{\bf f}_{12}^{10}\\
{\bf f}_{12}^{01}-{\bf f}_{12}^{10}\\
{\bf f}_{12}^{00}-{\bf f}_{12}^{11}\\
\end{bmatrix}=
\begin{bmatrix}
\partial^+_{12}{\bf f}\\
{\partial}^{\widehat+}_{12}{\bf f}\\
{\partial}^{\widehat-}_{12}{\bf f}\\
\partial^-_{12}{\bf f}\\
\end{bmatrix}.$$
We give the following relations between $f$ and $\widetilde{f}$.
\begin{lemma}\label{lem-tilde-1}
\begin{enumerate}
    \item 
If $f$ has even parity then
$\widetilde f$ also has even parity. 
\item If $f\in \mathscr{A}$, then $\widetilde{f}\in \mathscr{A}$.
\item If $M(\mathfrak{m}_{12}f)=\lambda I_4$ for some real $\lambda\neq 0$, then $M(\mathfrak{m}_{12}\widetilde{f})=2\lambda I_4$.
 \item 
  If $\partial_{12}^{+}f, \partial_{12}^{-}f,  \partial_{12}^{\widehat{+}}f, \partial_{12}^{\widehat{-}}f\in \mathscr{A}$, then  $\widetilde{f}^{00}_{12}$, $\widetilde{f}^{01}_{12}$, $\widetilde{f}^{10}_{12}$, $\widetilde{f}^{11}_{12} \in \mathscr{A}$. 
  \item 
    For $\{u, v\}$ disjoint with $\{1, 2\}$ and $b\in \mathcal{B}$, if $\partial_{uv}^{b}f\in \mathscr{A}$, then $\partial_{uv}^{b}\widetilde f\in \mathscr{A}$.

\end{enumerate}
\end{lemma}
\begin{proof}
Since $h_4$ has even parity and $h_4\in \mathscr{A}$,  the first and second propositions hold.

 If $M(\mathfrak{m}_{12}f)=\lambda I_4$, then
$M(\mathfrak{m}_{12}\widetilde{f})=M_{12}(\widetilde f)M^{\tt T}_{12}(\widetilde f)=H_4M_{12}(f)M^{\tt T}_{12}(f)H^{\tt T}_4=\lambda H_4I_4H^{\tt T}_4=2\lambda I_4.$
The third proposition holds.

By the matrix form $M_{12}(\widetilde f)$, the fourth proposition holds. 

Since the $H_4$ gadget construction only touches variables $x_1$ and $x_2$ of $f$, it commutes with merging gadgets on variables other than $x_1$ and $x_2$. Thus $\partial^{b}_{ij}\widetilde{f}=\widetilde{\partial^{b}_{ij}f}$.
For all $b\in \mathcal{B}$ and all $\{i, j\}$ disjoint with $\{1, 2\}$, if $\partial^b_{ij}f\in \mathscr{A}$ where $\partial^b_{ij}{f}$ are signatures realized by connecting variables $x_i$ and $x_j$ of 
$f$ using $b$, 
then $\partial^b_{ij}\widetilde f=\widetilde {\partial^{b}_{ij}{f}}\in \mathscr{A}$.
The last proposition holds. 
\end{proof}

Clearly, if $f\in \int_{\mathcal{B}}\mathscr{A}$, then  $\widetilde{f}^{00}_{12}$, $\widetilde{f}^{01}_{12}$, $\widetilde{f}^{10}_{12}$, $\widetilde{f}^{11}_{12} \in \mathscr{A}$.
Thus, for every $(a, b)\in \mathbb{Z}^2_2$, if  $\widetilde f_{12}^{ab}\not\equiv 0$, then its nonzero entries  have the same norm, denoted by $n_{ab}.$
Let $n_{ab}=0$ if $\widetilde f_{12}^{ab}\equiv 0$.
We have the following results regarding these norms  $n_{ab}.$

\begin{lemma}\label{lem-norm-f-tilde}
 Let $f$ be an irreducible signature of arity $2n\geqslant 6$.
 Suppose that $f$ has even parity, $f$ satisfies 
{\sc 2nd-Orth} and $f\in \int_{\mathcal{B}}\mathscr{A}$.
 \begin{enumerate}
     \item For any $(a, b), (c, d)\in \mathbb{Z}^2_2$, there exists some $k\in \mathbb{Z}$ such that
 $n_{ab}=\sqrt{2}^{k}n_{cd}\neq 0$, and $n_{ab}=n_{cd}$ iff $|\mathscr{S}(\widetilde{f}^{ab}_{12})|=|\mathscr{S}(\widetilde{f}^{cd}_{12})|.$
 \item Furthermore, if $\widetilde{f}_{12}^{00}(\vec{0}^{2n-2})\neq 0$ and $n_{00}>n_{11}$, 
then $\mathscr{S}(\widetilde{f}^{11}_{12})=\mathscr{E}_{2n-2}$\footnote{Here, $\mathscr{E}_{2n-2}=\{(x_3, \ldots, x_8)\in \mathbb{Z}^6_2\mid x_3+\cdots+x_8=0\}$. 
When context is clear, we do not specify the variables of $\mathscr{E}_{2n-2}$ and also $\mathscr{O}_{2n-2}$.}
and $n_{ab}=n_{11}$ or $2n_{11}$ for every $(a, b)\in \mathbb{Z}^2_2$;
in particular, $n_{00}=2n_{11}$.
Symmetrically,  if $\widetilde{f}_{12}^{11}(\vec{0}^{2n-2})\neq 0$ and $n_{00}<n_{11}$, 
then $\mathscr{S}(\widetilde{f}^{00}_{12})=\mathscr{E}_{2n-2}$
and $n_{ab}=n_{00}$ or $2n_{00}$ for every $(a, b)\in \mathbb{Z}^2_2$, and $n_{11} = 2 n_{00}$. 
 \end{enumerate}
 
\end{lemma}
\begin{proof}
Since $f$ satisfies {\sc 2nd-Orth},  $M(\mathfrak{m}_{12}f)=\lambda I_4$ for some real $\lambda \neq 0$.
Then, by Lemma~\ref{lem-tilde-1}, $M(\mathfrak{m}_{12}\widetilde{f})=2\lambda I_4\neq 0$.
Thus, $|\widetilde{f}_{12}^{ab}|^2=2\lambda\neq 0$ for every $(a, b)\in \mathbb{Z}^2_2$.
Also, since  $f\in \int_{\mathcal{B}}\mathscr{A}$, by Lemma~\ref{lem-tilde-1},  for every $(a, b)\in \mathbb{Z}^2_2$, $\widetilde{f}^{ab}_{12} \in \mathscr{A}$.
Thus, $\mathscr{S}(\widetilde{f}_{12}^{ab})$ is affine and $|\mathscr{S}(\widetilde{f}_{12}^{ab})|=2^{k_{ab}}$ for some integer $k_{ab} \geqslant 0$.
Note that $$|\widetilde{f}_{12}^{ab}|^2=n^2_{ab}\cdot |\mathscr{S}(\widetilde{f}_{12}^{ab})|=n^2_{ab}\cdot 2^{k_{ab}}\neq 0.$$
Thus, for any $(a, b), (c, d)\in \mathbb{Z}^2_2$, 
$n^2_{ab}\cdot 2^{k_{ab}}=n^2_{cd}\cdot 2^{k_{cd}}\neq 0$.
Then, $n_{ab}= \sqrt{2}^{k}n_{cd}\neq 0$ where $k=k_{cd}-k_{ab}\in\mathbb{Z}$.
Clearly, $k=0$ iff $|\mathscr{S}(\widetilde{f}_{12}^{ab})|=2^{k_{ab}}=2^{k_{cd}}=|\mathscr{S}(\widetilde{f}_{12}^{cd})|.$

Now we prove the second part of this lemma.
We give the proof for the case that $\widetilde{f}_{12}^{00}(\vec{0}^{2n-2})\neq 0$ and $n_{00}>n_{11}$.
The proof of the case that $\widetilde{f}_{12}^{11}(\vec{0}^{2n-2})\neq 0$ and $n_{00}<n_{11}$ is symmetric.
We first show that $\mathscr{S}(\widetilde{f}^{11}_{12})=\mathscr{E}_{2n-2}$.
For a contradiction, suppose that $\mathscr{S}(\widetilde{f}^{11}_{12})\neq \mathscr{E}_{2n-2}$.
Since $f$ has even parity, by Lemma~\ref{lem-tilde-1}, $\widetilde f$ has even parity.
Then,
$\widetilde{f}^{11}_{12}$ also has even parity. 
Thus, $\mathscr{S}(\widetilde{f}^{11}_{12})\subsetneq \mathscr{E}_{2n-2}$.
There exists $\theta \in \mathscr{E}_{2n-2}$ such that $\theta \notin \mathscr{S}(\widetilde{f}^{11}_{12})$.
Also, since $n_{11}\neq 0$, $\widetilde{f}^{11}_{12}\not\equiv 0$.
Then,  $\mathscr{S}(\widetilde{f}^{11}_{12})\neq \emptyset$ and there exists $\delta \in \mathscr{E}_{2n-2}$ such that $\delta \in \mathscr{S}(\widetilde{f}^{11}_{12})$.
Then, we can find a pair $\alpha, \beta\in \mathscr{E}_{2n-2}$ where ${\rm wt}(\alpha\oplus\beta)=2$
such that one is in $ \mathscr{S}(\widetilde{f}^{11}_{12})$ and the other one is not in $ \mathscr{S}(\widetilde{f}^{11}_{12})$.
\begin{itemize}
    \item If ${\rm wt}(\alpha)\neq {\rm wt}(\beta)$, then clearly the difference between their Hamming weights is $2$ since ${\rm wt}(\alpha\oplus\beta)=2$.
Thus, $\alpha$ and $\beta$ differ in two bits $i, j$  
on which one takes value $00$ and the other takes value $11$. 
\item
If  ${\rm wt}(\alpha)= {\rm wt}(\beta)$, then they  differ in two bits $i, j$ on which one takes value $01$ and the other takes value $10$.
Without loss of generality, we assume that $\alpha_i\alpha_j=01$ and $\beta_i\beta_j=10$.
They take the same value on other bits.
Since $\alpha, \beta \in \mathscr{E}_{2n-2}$ and $2n\geqslant 6$, they have even Hamming weight and length at least $4$.
Thus, there is another bit $k$ such that on this bit, $\alpha_k=\beta_k=1$.
Consider $\gamma\in \mathscr{E}_{2n-2}$ where $\gamma_i\gamma_j\gamma_k=000$  and $\gamma$ takes the same value as $\alpha$ and $\beta$ on other bits. 
Clearly, $\rm{wt}(\gamma)+2=\rm{wt}(\alpha)=\rm{wt}(\beta)$.
If $\gamma \in \mathscr{S}(\widetilde{f}^{11}_{12})$, then between $\alpha$ and $\beta$, we pick the one that is not in $\mathscr{S}(\widetilde{f}^{11}_{12})$.
Otherwise, we pick the one that is in $\mathscr{S}(\widetilde{f}^{11}_{12})$.
In both cases, we can get a pair of inputs in
$\mathscr{E}_{2n-2}$ such that one is in  $\mathscr{S}(\widetilde{f}^{11}_{12})$ and the other is not in  $\mathscr{S}(\widetilde{f}^{11}_{12})$, and 
they have Hamming distance 2 as well as different Hamming weights.
\end{itemize}

Thus, we can always find a pair $\alpha, \beta\in \mathscr{E}_{2n-2}$ where ${\rm wt}(\alpha\oplus\beta)=2$ and $\alpha, \beta$ differ in two bits $i, j$  
on which one takes value $00$ and the other takes value $11$, such that one is in $ \mathscr{S}(\widetilde{f}^{11}_{12})$ and the other is not in $ \mathscr{S}(\widetilde{f}^{11}_{12})$.
Clearly, $\{i, j\}$ is disjoint with $\{1, 2\}$.

Consider signatures $\partial^+_{ij}\widetilde f$ and
$\partial^-_{ij}\widetilde f$.
Then, 
$\widetilde f(11\alpha)+\widetilde f(11\beta)$
is an entry of $\partial^+_{ij}\widetilde f$,
and $\widetilde f(11\alpha)-\widetilde f(11\beta)$ is an entry of $\partial^-_{ij}\widetilde f$.
Since between $\widetilde f(11\alpha)$ and $\widetilde f(11\beta)$, exactly one is nonzero and it has norm $n_{11}$, we have 
$$|\widetilde f(11\alpha)+\widetilde f(11\beta)|=|\widetilde f(11\alpha)-\widetilde f(11\beta)|=n_{11}.$$
Thus, both $\partial^+_{ij}\widetilde f$ and
$\partial^-_{ij}\widetilde f$ have an entry with norm $n_{11}$.
Let $\delta\in\mathscr{E}_{2n}$ where $\delta_i\delta_j=11$ and $\delta$ takes value $0$ on other bits.
Then, clearly, $\widetilde f(\vec{0}^{2n})+\widetilde f(\delta)$ is an entry of $\partial^+_{ij}\widetilde f$, 
and $\widetilde f(\vec{0}^{2n})-\widetilde f(\delta)$ is an entry of $\partial^-_{ij}\widetilde f$.
\begin{itemize}
\item If $\widetilde  f(\delta)\neq 0$, then $|\widetilde f(\delta)|=n_{00}$ since $\delta_{1}\delta_2=00$.
Since $\widetilde f(\vec{0}^{2n})\neq 0$, $|\widetilde f(\vec{0}^{2n})|=n_{00}$.
Thus, between $\widetilde f(\vec{0}^{2n})+\widetilde f(\delta)$ and $\widetilde f(\vec{0}^{2n})-\widetilde f(\delta)$, one has norm $2n_{00}$ and the other is zero.
Therefore, between $\partial^+_{ij}\widetilde f$ and
$\partial^-_{ij}\widetilde f$, one signature has an entry with norm $2n_{00}$.
Remember that both $\partial^+_{ij}\widetilde f$ and
$\partial^-_{ij}\widetilde f$ have an entry with norm $n_{11}$.
Clearly, $2n_{00}>n_{11}$.
Thus, between $\partial^+_{ij}\widetilde f$ and
$\partial^-_{ij}\widetilde f$,
there is a signature that has two entries with different norms.
Clearly, such a signature is not in $\mathscr{A}$.
However, since $f\in \int_{\mathcal{B}}\mathscr{A}$, 
by Lemma \ref{lem-tilde-1}, $\partial^+_{ij}\widetilde f, \partial^-_{ij}\widetilde f\in \mathscr{A}$. 
Contradiction.
    \item 
If $\widetilde f(\delta)=0$, then  $|\widetilde f(\vec{0}^{2n})+\widetilde f(\delta)|=|\widetilde f(\vec{0}^{2n})|=n_{00}>n_{11}$.
Thus, $\partial^{+}_{ij}\widetilde f$ has two nonzero entries with different norms $n_{00}$ and $n_{11}$.
Then, $\partial^{+}_{ij}\widetilde f\notin \mathscr{A}$.
Contradiction.
\end{itemize}
Thus, $\mathscr{S}(\widetilde f_{12}^{11})=\mathscr{E}_{2n-2}$.

Then, we show that $n_{ab}=n_{11}$ or $2n_{11}$ for any $(a, b)\in \mathbb{Z}_2^2$. 
Clearly, we may assume that $(a, b)\neq (1, 1)$.
For a contradiction, suppose that $n_{ab}\neq n_{11}$ and $2n_{11}$.
First, we show that  $|\mathscr{S}(\widetilde f_{12}^{ab})|<2^{2n-3}$ and $n_{ab}>n_{11}$.
Since $f$ has parity, $\widetilde f_{12}^{ab}$ also has parity (either even or odd depending on whether ${\rm wt}(ab)=0$ or $1$).
Thus $|\mathscr{S}(\widetilde f_{12}^{ab})|\leqslant|\mathscr{E}_{2n-2}|=|\mathscr{O}_{2n-2}|=2^{2n-3}$.
If the equality holds, then $n_{ab}=n_{11}$ since  $n^2_{ab}|\mathscr{S}(\widetilde f_{12}^{ab})|=n^2_{11}|\mathscr{S}(\widetilde f_{12}^{11})|$ and $|\mathscr{S}(\widetilde f_{12}^{11})|=2^{2n-3}$. Contradiction.
Thus, $|\mathscr{S}(\widetilde f_{12}^{ab})|<2^{2n-3}$ and also $n_{ab}>n_{11}.$

Depending on whether $f_{12}^{ab}$ has even parity or odd parity, 
we can pick a pair of inputs $\alpha, \beta$ with ${\rm wt}(\alpha\oplus\beta)=2$ from $\mathscr{E}_{2n-2}$ or $\mathscr{O}_{2n-2}$ such that exactly one is in $\mathscr{S}(f_{12}^{ab})$ and the other is not in $\mathscr{S}(f_{12}^{ab})$.
Suppose that $\alpha$ and $\beta$ differ in bits $i ,j$.
Depending on whether $\alpha_i=\alpha_j$ or $\alpha_i\neq\alpha_j$, 
we can connect variables $x_i$ and $x_j$ of $\widetilde f$ using $=_2^+$ and $=_2^-$, or $\neq_{2}^{+}$ and $\neq_{2}^{-}$.
We will get two signatures $\partial^+_{ij}\widetilde f$ and $\partial^-_{ij}\widetilde f$, or $\partial^{\widehat+}_{ij}\widetilde f$ and $\partial^{\widehat-}_{ij}\widetilde f$.
We  consider the case that 
$\alpha_i=\alpha_j$.
For the case that $\alpha_i\neq \alpha_j$, the analysis is the same by replacing $\partial^+_{ij}\widetilde f$ and $\partial^-_{ij}\widetilde f$ with $\partial^{\widehat+}_{ij}\widetilde f$ and $\partial^{\widehat-}_{ij}\widetilde f$ respectively.

Consider signatures $\partial^+_{ij}\widetilde f$ and $\partial^-_{ij}\widetilde f$.
Then, $\widetilde f(ab\alpha)+ \widetilde f(ab\beta)$ is an entry of  $\partial^+_{ij}\widetilde f$, and $\widetilde f(ab\alpha)- \widetilde f(ab\beta)$ is an entry of  $\partial^-_{ij}\widetilde f$.
Since between $\alpha$ and $\beta$, exactly one is in $\mathscr{S}(f_{12}^{ab})$,
between $\widetilde f(ab\alpha)$ and $\widetilde f(ab\beta)$, exactly one is nonzero and it has norm $n_{ab}$.
Thus, $$|\widetilde f(ab\alpha)+ \widetilde f(ab\beta)|=|\widetilde f(ab\alpha)- \widetilde f(ab\beta)|=n_{ab}.$$
Both $\partial^+_{ij}\widetilde f$ and $\partial^-_{ij}\widetilde f$ have an entry with norm $n_{ab}$.

Let $\alpha', \beta' \in \mathscr{E}_{2n-2}$ where $\alpha'_i\alpha'_j=\alpha_i\alpha_j$, $\alpha'_k=\alpha'_i\oplus \alpha'_j$ for some $k\neq i, j$\footnote{Since $2n-2\geqslant 4$, such a $k$ exists. Here, $\alpha'_k=0$ since $\alpha_i=\alpha_j$  in this case under discussion.
For the case that $\alpha_i\neq \alpha_j$, we have $\alpha'_k=1$.} and $\alpha'$ takes value $0$ on other bits, and $\beta'_i\beta'_j=\beta_i\beta_j$, $\beta'_k=\beta'_i\oplus\beta'_j$ for the same $k\neq i, j$ and $\beta'$ takes value $0$ on other bits.
Clearly, $\alpha'$ and $\beta'$ differ in bits $i$ and $j$ and they differ in the same way as $\alpha$ and $\beta$. 
Then, $\widetilde f(11\alpha')+ \widetilde f(11\beta')$ is an  entry of   $\partial^{+}_{ij}\widetilde f$, and  $\widetilde f(11\alpha')- \widetilde f(11\beta')$ is an  entry of   $\partial^{-}_{ij}\widetilde f$.
Since $\mathscr{S}(\widetilde f^{11}_{12})=\mathscr{E}_{2n-2}$, 
both $\widetilde f(11\alpha')$ and $\widetilde f(11\beta')$ are nonzero and they have norm $n_{11}$.
Thus, between $\widetilde f(11\alpha')+ \widetilde f(11\beta')$ and 
$\widetilde f(11\alpha')- \widetilde f(11\beta')$, exactly one is zero and the other has norm $2n_{11}$.
Thus, between signatures $\partial_{ij}\widetilde f$ and $\partial^{-}_{ij}\widetilde f$, 
there is a signature that has two entries with different norms  $2n_{11}$ and $n_{ab}$.
Such a signature is not in $\mathscr{A}$.
However, since $f\in \int_{\mathcal{B}}\mathscr{A}$, 
by Lemma \ref{lem-tilde-1}, $\partial_{ij}\widetilde f, \partial^-_{ij}\widetilde f\in \mathscr{A}$. 
Contradiction.
Thus, $n_{ab}=n_{11}$ or $2n_{11}$ for any $(a, b)\in \mathbb{Z}_2^2$.
\end{proof}

We also give the following results about multilinear  polynomials $F(x_1, \ldots, x_n)\in\mathbb{Z}_2[x_1, \ldots, x_n]$.
We use $d(F)$ to denote the total degree of  $F$.  
For $\{i, j\}\subseteq \{1, \ldots, n\}=[n]$, we use $F^{ab}_{ij}\in Z_2[\{x_1, \ldots, x_n\}\backslash\{x_i, x_j\}]$ to denote the polynomial obtained by setting $(x_i, x_j)=(a, b)$ in $F$.

\begin{definition}
Let $F(x_1, \ldots, x_n)\in\mathbb{Z}_2[x_1, \ldots, x_n]$ be a  multilinear polynomial.
We say $F$ is a complete quadratic  polynomial if  $d(F) = 2$ and for all $\{i, j\}\subseteq [n]$, the quadratic term $x_ix_j$ appears in $F$.
We say $F$ is a complete cubic  polynomial if $d(F) = 3$ and for all $\{i, j, k\}\subseteq [n]$, the cubic term $x_ix_jx_k$ appears in $F$.
\end{definition}{}

\begin{lemma}\label{lem-mutilinear-poly}
Let $F(x_1, \ldots, x_n)\in\mathbb{Z}_2[x_1, \ldots, x_n]$ be a  multilinear polynomial. 
\begin{enumerate}
    \item If for all $\{i, j\}\subseteq [n]$, $F_{ij}^{00}+F_{ij}^{11}\equiv 0$ or $1$, and $F_{ij}^{01}+F_{ij}^{10}\equiv 0$ or $1$, then $d(F)\leqslant 2$. Moreover, if $d(F)=2$, then $F$ is a complete quadratic  polynomial.
    \item If for all $\{i, j\}\subseteq [n]$, $d(F_{ij}^{00}+F_{ij}^{11})\leqslant 1$, and $d(F_{ij}^{01}+F_{ij}^{10})\leqslant 1$, then $d(F)\leqslant 3$. Moreover, if $d(F)=3$, then $F$ is a complete cubic  polynomial.
\end{enumerate}
\end{lemma}

\begin{proof}
We prove the first part. The proof for the second part is similar which we omit here.

For all $\{i, j\}\subseteq[n]$, we write $F\in \mathbb{Z}_2[x_1, \ldots, x_n]$ as a polynomial of variables $x_i$ and $x_j$.
$$F=X_{ij}x_ix_j+Y_{ij}x_i+Z_{ij}x_j+W_{ij}$$ where $X_{ij}, Y_{ij}, Z_{ij}, W_{ij}\in  \mathbb{Z}_2[\{x_3, \ldots, x_n\}\backslash\{x_i, x_j\}]$.
Then, $$F_{ij}^{00}=W_{ij} ~~\text{ and  }~~ F_{ij}^{11}=X_{ij}+Y_{ij}+Z_{ij}+W_{ij}.$$
Thus, $X_{ij}+Y_{ij}+Z_{ij} = F_{ij}^{00}+F_{ij}^{11}
\equiv 0 \text{ or }1.$
Also, $$F_{ij}^{01}=Z_{ij}+W_{ij} ~~\text{ and }~~ F_{ij}^{10}=Y_{ij}+W_{ij}.$$
Thus, $Y_{ij}+Z_{ij}= F_{ij}^{01}+F_{ij}^{10}\equiv 0 \text{ or }1.$
Then, $X_{ij} \equiv  0$ or $1$ for all $\{i, j\}$.
Thus, $d(F)\leqslant 2$. 

Suppose that $d(F)=2$. then $F$ has at least a quadratic term $x_ux_v$ $(u\neq v)$.
   Without loss of generality, we assume that the term $x_1x_2$ appears in $F$. 
  We first show that for all $2\leqslant j \leqslant n$, the quadratic term $x_1x_j$ appears in $F$.
   Since $x_1x_2$ is already in $F$, we may assume that $3\leqslant j$.
   We write $F$ as a polynomial of variables $x_2$ and $x_j$. 
   $$F=X_{2j}x_2x_j+Y_{2j}x_2+Z_{2j}x_j+W_{2j},$$
   where $X_{2j}, Y_{2j}, Z_{2j}, W_{2j}$ do not involve
   $x_2$ and $x_j$.
   Since $x_1x_2$ appears in $F$, $x_1$ appears in $Y_{2j}$.
     As we have proved above, $Y_{2j}+Z_{2j} \equiv 0$ or $1$.
     Thus, $x_1$ also appears in $Z_{2j}$, which means that $x_1x_j$ appears in $F$.
     Then, for all $2\leqslant j\leqslant n$, $x_1x_j$ appears in $F$.
     
     Then, for all $2\leqslant i < j\leqslant n$, we write $F$ as a polynomial of variables $x_1$ and $x_i$.
   $$F=X_{1i}x_1x_i+Y_{1i}x_1+Z_{1i}x_i+W_{1i},$$
   where $X_{1i}, Y_{1i}, Z_{1i}, W_{1i}$ do not involve
   $x_1$ and $x_i$.
     Since $x_1x_j$ appears in $F$, $x_j$ appears in $Y_{1i}$.
     Since $Y_{1i}+Z_{1i}\equiv 0$ or $1$, $x_j$ also appears in $Z_{1i}$.
     Thus, $x_ix_j$ appears in $F$.
     Then, for all $2\leqslant i < j\leqslant n$, the quadratic term $x_ix_j$ appears in $F$.
     Thus,  for all $\{i, j\}\subseteq [n]$,  $x_ix_j$ appears in $F$.
\end{proof}

Now, we are ready to take a major
step towards Theorem~\ref{thm-holantb}.
\begin{lemma}\label{lem-affine-norm}
Let $2n\geqslant 8$ and let $f\in \mathcal{F}$ be a $2n$-ary  irreducible signature with parity. Then,
\begin{itemize}
\item  $\Holantb(\mathcal{F})$ is \#P-hard, or
    \item there is  a signature $g\notin \mathscr{A}$ of arity $2k<2n$  that is realizable from $f$ and $\mathcal{B}$, or
    \item after normalization, $f(\alpha)=\pm 1$ for all $\alpha \in \mathscr{S}(f)$.
\end{itemize}
\end{lemma}
\begin{proof}
Since $f$ is irreducible, we may assume that it satisfies {\sc 2nd-Orth}.
Otherwise, we get  \#P-hardness by Lemma~\ref{second-ortho}.
Also, we may assume that $f\in \int_{\mathcal{B}}\mathscr{A}$.
Otherwise, we can realize a signature of arity $2n-2$ that is not in $\mathscr{A}$ by merging $f$ using  some $b\in \mathcal{B}$.

For any four entries $x, y, z, w$ of $f$ on inputs $\alpha, \beta, \gamma, \delta\in \mathbb{Z}_2^{2n}$ written in the form of a 2-by-2 matrix $\left[\begin{smallmatrix} x & y\\ z & w
\end{smallmatrix}
\right]=\left[\begin{smallmatrix} f(\alpha) & f(\beta)\\ f(\gamma) & f(\delta)
\end{smallmatrix}
\right]$,  we say that such a matrix is a \emph{distance-2 square} if there exist four bits $i,j,k, \ell$ such that $\alpha_i\alpha_j=\beta_i\beta_j=\overline{\gamma_i\gamma_j}=\overline{\delta_i\delta_j}$, $\alpha_k\alpha_{\ell}=\gamma_k\gamma_\ell=\overline{\beta_k}\overline{\beta_\ell}=\overline{\delta_k\delta_\ell}$ and  $\alpha$, $\beta$, $\gamma$ and $\delta$ take the same values on other bits. 
An equivalent description is that 
\begin{equation}\label{eqn:distance-2 square}
\delta=\alpha\oplus\beta\oplus\gamma, ~~
{\rm wt}(\alpha\oplus\beta)=2, ~~{\rm wt}(\alpha\oplus\gamma)=2 ~~\mbox{and}~~ {\rm wt}(\alpha\oplus\delta)=4.
\end{equation}
Indeed  (\ref{eqn:distance-2 square})
is clearly satisfied by any distance-2 square.
Conversely, suppose (\ref{eqn:distance-2 square}) holds.
If we flip any bit $i$ in all $\alpha, \beta, \gamma$ and $\delta$, both (\ref{eqn:distance-2 square}) and the
bitwise description are invariant, and thus we may assume
$\alpha = \vec{0}^{2n}$. 
By ${\rm wt}(\alpha\oplus\gamma) =2$,
there exist two bits $i, j$ such that 
$\gamma_i\gamma_j=11$, and $\gamma$ takes $0$ on other bits. 
By
${\rm wt}(\alpha\oplus\beta)=2$, 
there exits two bits $k, \ell$  such that 
$\beta_k\beta_\ell=11$, and $\beta$ takes $0$ on other bits. 
Since $\delta=\alpha\oplus\beta\oplus\gamma$,
${\rm wt}(\beta\oplus\gamma)={\rm wt}(\alpha\oplus\delta)=4.$
Thus,
the bits $i, j, k, \ell$ are distinct four bits.
Then, $\delta_i\delta_j\delta_k\delta_{\ell}=1111$ and $\delta$ takes $0$ on other bits.
Thus, $\alpha$, $\beta$, $\gamma$ and $\delta$ satisfy the bitwise description of distance-2 squares.

We give an example of such a distance-2 square. 
Let 
$$\left[\begin{matrix} x & y\\ z & w
\end{matrix}
\right]=\left[\begin{matrix} f(\alpha) & f(\beta)\\ f(\gamma) & f(\delta)
\end{matrix}
\right]=\left[\begin{matrix}
f(0001\theta) & f(0010\theta)\\
f(1101\theta) & f(1110\theta)\\
\end{matrix}\right]$$
where $\theta\in \mathbb{Z}_2^{2n-4}$ is an arbitrary binary string of length $2n-4$. 
In this example, $(i, j)=(1, 2)$ and $(k, \ell)=(3, 4)$. 
We show next that such a distance-2 square 
$\left[\begin{smallmatrix} x & y\\ z & w
\end{smallmatrix}
\right]$
has the
property described in
(\ref{equ-square1}) 
$\sim$
(\ref{equ-square4}).

By connecting variables $x_1$ and $x_2$ of $f$ using $=^+_2$ and $=_2^{-}$ respectively, 
we get signatures $\partial^+_{12}f$ and $\partial^-_{12}f$. 
By our assumption, $\partial^+_{12}f$ and $\partial^-_{12}f$ are affine signatures. 
Note that, $x+z$ and $y+w$ are entries of $\partial^+_{12}f$ on inputs $01\theta$ and $10\theta \in \mathbb{Z}_2^{2n-2}$.
Since $\partial^+_{12}f\in \mathscr{A}$, if $x+z$ and $y+w$ are both nonzero, 
then they have the same norm.
Thus, we have $(x+z)(y+w)=0$ or $(x+z)^2=(y+w)^2$.
Similarly, $x-z$ and $y-w$ are entries of $\partial^-_{12}f\in \mathscr{A}$.
Thus, we have $(x-z)(y-w)=0$ or $(x-z)^2=(y-w)^2$.

Also, by connecting variables  $x_3$ and $x_4$ of $f$ using $\neq_2$ and $\neq_2^{-}$ respectively, 
we get signatures $\partial^{\widehat+}_{34}f$ and $\partial^{\widehat-}_{34}f$ that are affine signatures.
Note that, $x+y$ and $z+w$ are entries of $\partial^{\widehat+}_{34}f$ on inputs $00\theta$ and $11\theta$.
Since $\partial^{\widehat+}_{34}f\in \mathscr{A}$,
we have  $(x+y)(z+w)=0$ or $(x+y)^2=(z+w)^2$.
Similarly, $x-y$ and $z-w$ are entries of $\widehat\partial^{-}_{34}f$.
Then, we have $(x-y)(z-w)=0$ or $(x-y)^2=(z-w)^2$.

Now, consider an arbitrary distance-2 square $\left[\begin{smallmatrix} x & y\\ z & w
\end{smallmatrix}
\right]=\left[\begin{smallmatrix} f(\alpha) & f(\beta)\\ f(\gamma) & f(\delta)
\end{smallmatrix}
\right]$.
Depending on whether $\alpha_i=\alpha_j$ or $\alpha_i\neq \alpha_j$,
 we can use $=_2^{+}$ and $=_2^-$, or $\neq^{+}_2$ and $\neq^-_2$ respectively,  to connect variables $x_i$ and $x_j$ of $f$
 to produce two signatures $\partial^+_{ij}f$  and  $\partial^-_{ij}f$, or  $\partial^{\widehat +}_{ij}f$  and  $\partial^{\widehat -}_{ij}f$  in either case, such that $x\pm z$ and $y\pm w$ are both entries of the resulting two signatures. 
Since the two resulting signatures are in affine,  we have 
\begin{equation}\label{equ-square1}
(x+z)(y+w)=0 ~~\text{ or }~~ (x+z)^2=(y+w)^2,
\end{equation} and 
\begin{equation}\label{equ-square2}
(x-z)(y-w)=0 ~~\text{ or }~~ (x-z)^2=(y-w)^2.
\end{equation}
Similarly, by connecting variables $x_k$ and $x_\ell$ of $f$ using either $=_2^{\pm}$ or $\neq^{\pm}_2$, we have  \begin{equation}\label{equ-square3}
    (x+y)(z+w)=0 ~~\text{ or }~~ (x+y)^2=(z+w)^2
    \end{equation}
    and 
    \begin{equation}\label{equ-square4}
 (x-y)(z-w)=0 ~~\text{ or }~~ (x-y)^2=(z-w)^2.
    \end{equation}
    
Now, we show that by solving equations (\ref{equ-square1}) $\sim$ (\ref{equ-square4}), 
every distance-2 square has one of the following forms (after normalization, row or column permutation, multiplying a $-1$ scalar of one row or one column, and taking transpose)
$$\underbrace{\left[\begin{matrix} 0 & 0\\ 0 & 0
\end{matrix}
\right], \left[\begin{matrix} 1 & 0\\ 0 & 0
\end{matrix}
\right],
\left[\begin{matrix} 1 & 1\\ 0 & 0
\end{matrix}
\right],
\left[\begin{matrix} 1 & 0\\ 0 & 1
\end{matrix}
\right],
\left[\begin{matrix} 1 & 1\\ 1 & 1
\end{matrix}
\right],
\left[\begin{matrix} 1 & 1\\ 1 & -1
\end{matrix}
\right],}_{\text{ type \Rmnum{1}}}
\text{ ~~~ }
\underbrace{
\left[\begin{matrix} 1 & a\\ a & 1
\end{matrix}
\right] (a>1),}_{\text{ type \Rmnum{2}}}
\text{ ~~or~~ }
\underbrace{\left[\begin{matrix} 1 & 1\\ 3 & -1
\end{matrix}
\right]}_{\text{ type \Rmnum{3}}}.
$$
We say that the first six forms are type \Rmnum{1}, and the other two are type \Rmnum{2} and type \Rmnum{3} respectively. 
These forms listed above are canonical forms of each type.

Let $\left[\begin{smallmatrix}
x & y\\
z & w\\
\end{smallmatrix}\right]$ be a distance-2 square.
Consider $$p=(x+y)(z+w)(x+z)(y+w)(x-y)(z-w)(x-z)(y-w).$$
\begin{itemize}
    \item If $p=0$, then among its  eight factors (four sums and four differences), at least one factor is zero. 
    By taking transpose and row permutation, we may assume that $x+y=0$ or $x-y=0$.
    If $x+y=0$, then by multiplying the column $\left[\begin{smallmatrix}
 y\\
 w\\
\end{smallmatrix}\right]$ with $-1$, we can modify this distance-2 square to get $x-y=0$.
Thus, we may assume that $x-y=0$.
If $x=y=0$, then by (\ref{equ-square1}),
we have $z=0$ or $w=0$, or $z=\pm w$.
Thus, after normalizing operations of row and column permutation and multiplication by $-1$, we reach the following
canonical forms $\left[\begin{smallmatrix}
0 & 0\\
0 & 0\\
\end{smallmatrix}\right]$,
$\left[\begin{smallmatrix}
1 & 0\\
0 & 0\\
\end{smallmatrix}\right]$ or $\left[\begin{smallmatrix}
1 & 1\\
0 & 0\\
\end{smallmatrix}\right].$
Otherwise, $x=y\neq 0$. 
    Consider $q=(x+z)(y+w)(x-z)(y-w)$.
\begin{itemize}

\item   If $q=0$, then among its four factors (two sums and two differences), at least one is zero.
By column permutation on the matrix $\left[\begin{smallmatrix}
x & y\\
z & w\\
\end{smallmatrix}\right]$  and multiplying the row $ (z, w)$ with $-1$ (which does not change the values of $x$ and $y$), we may assume that $x-z=0$.
    Thus, $x=y=z\neq 0$.
    We normalize their values to $1$.
    Then by (\ref{equ-square1}), $1+w=0$ or $1+w=
    \pm 2$.
    Thus, $w=-1, 1$ or $-3$.
    If $w=\pm 1$, then $\left[\begin{smallmatrix}
x & y\\
z & w\\
\end{smallmatrix}\right]$ has the canonical form $\left[\begin{smallmatrix}
1 & 1\\
1 & 1\\
\end{smallmatrix}\right]$ or $\left[\begin{smallmatrix}
1 & 1\\
1 & -1\\
\end{smallmatrix}\right].$
If $w=-3$, then  $\left[\begin{smallmatrix}
x & y\\
z & w\\
\end{smallmatrix}\right]= \left[\begin{smallmatrix}
1 & 1\\
1 & -3\\
\end{smallmatrix}\right]$ which has the  canonical form 
$\left[\begin{smallmatrix}
1 & 1\\
3 & -1\\
\end{smallmatrix}\right]$ (Type \Rmnum{3}).
    \item  If $q\neq 0$,
    then $(x+z)(y+w)\neq 0$ and $(x-z)(y-w)\neq 0$.
 By equations (\ref{equ-square1}) and (\ref{equ-square2}), $(x+z)^2=(y+w)^2$ and $(x-z)^2=(y-w)^2.$
 Thus, $xz=yw$.
 Since $x=y\neq 0$, $z=w$.
 If $z=w=0$, then this gives the canonical form  $\left[\begin{smallmatrix}
1 & 1\\
0 & 0\\
\end{smallmatrix}\right]$.
Otherwise, $z=w\neq 0$. Then $z+w\neq 0$ and hence  by (\ref{equ-square3}), $z+w=\pm (x+y)$.
Since $z=w$ and $x=y$,
we get $z=\pm x$. 
Thus, $x+z=0$ or $x-z=0$. Contradiction.
\end{itemize}
\item If $p\neq 0$, then all its eight factors are nonzero. 
Thus by (\ref{equ-square1}) $\sim$ (\ref{equ-square4}), $(x+z)^2=(y+w)^2$, $(x-z)^2=(y-w)^2$,  $(x+y)^2=(z+w)^2$ and $(x-y)^2=(z-w)^2.$
By solving these equations, we have $x^2=w^2$, $y^2=z^2$, and $xy=zw$. 
If $x=y=z=w=0$, then it gives the canonical form $\left[\begin{smallmatrix}
0 & 0\\
0 & 0\\
\end{smallmatrix}\right]$.
Otherwise, by permuting rows and columns, we may assume that $x\neq 0$ and $|x|$ is the smallest among the norms of nonzero entries in  $\left[\begin{smallmatrix}
x & y\\
z & w\\
\end{smallmatrix}\right]$. 
We normalize $x$ to $1$.
Since $x^2=w^2$, we get $w=\pm 1$.
By multiplying the row $(z, w)$ with $-1$ (which does not change $xy=zw$), we may assume that $w=1$.
Then, $xy=zw$ implies that $y=z$.
If $y=z=0$, then $\left[\begin{smallmatrix}
x & y\\
z & w\\
\end{smallmatrix}\right]$ has the canonical form  $\left[\begin{smallmatrix}
1 & 0\\
0 & 1\\
\end{smallmatrix}\right]$.
Otherwise, since $|x|=1$ is the smallest norm among nonzero entries, $y=z=\pm a$ where $a\geqslant 1.$
If $a=1$ (i.e., $y=z=\pm 1$), then  $\left[\begin{smallmatrix}
x & y\\
z & w\\
\end{smallmatrix}\right]$ has the canonical form  $\left[\begin{smallmatrix}
1 & 1\\
1 & 1\\
\end{smallmatrix}\right]$.
If $a>1$, then  $\left[\begin{smallmatrix}
x & y\\
z & w\\
\end{smallmatrix}\right]$ has the canonical form of Type \Rmnum{2}.

\end{itemize}
Thus, every distance-2 square has a canonical form of Type \Rmnum{1}, \Rmnum{2} or \Rmnum{3}.

Note that given a particular distance-2 square of $f$, by normalization, and renaming or flipping or negating variables of $f$,
we can always modify this distance-2 square to get its canonical form.
Clearly, for signatures of arity at least $4$, distance-2 squares exist. 
We consider the following two cases according to which types of distance-2 squares appear in $f$.

\vspace{2ex}
{\noindent \bf Case 1.} All distance-2 squares in $f$ are of type \Rmnum{1}. 
\vspace{1.5ex}

We show that (after normalization) $f(\alpha)=\pm 1$  for all $\alpha \in \mathscr{S}(f)$.
Since $f\not\equiv0$, it has at least one nonzero entry. 
By normalization, we may assume that $1$ is the smallest norm of all nonzero entries of $f$. 
Then by flipping variables of $f$, 
we may assume that $f(\vec{0}^{2n})=1$.
For a contradiction, suppose that there is some $\beta\in \mathscr{S}(f)$ such that $f(\beta)\neq \pm 1$. 
Then by our assumption that $1$ is the smallest norm and $|f(\beta)|\neq 1$, we have $|f(\beta)|>1$.
Also, since $f$ has parity and $\vec{0}^{2n}\in \mathscr{S}(f)$, $f$ has even parity.
Thus,
${\rm wt}(\beta)\equiv 0 \pmod 2.$
By renaming  variables of $f$, 
we may assume that $\beta=\vec{1}^{2m}\vec{0}^{2n-2m},$
for some $m \geqslant 1$.
(This does not affect the normalization $f(\vec{0}^{2n})=1$). 
Then, we show that for all $\alpha=\delta\vec{0}^{2n-2m}$ where $\delta\in \mathbb{Z}_2^{2m}$, $f(\alpha)= \pm 1$.
We prove this by induction on ${\rm wt}(\delta)$. 
This will 
lead to a contradiction when ${\rm wt}(\delta)=2m$, since $|f(\beta)|=|f(\vec{1}^{2m}\vec{0}^{2n-2m})|\neq 1$.

Since  
$f(\vec{0}^{2n})=1$,
we may assume ${\rm wt}(\delta) \geqslant 2$.
We first consider the base case that ${\rm wt}(\delta)=2$.
By renaming the first $2m$  variables, without loss of generality, we may assume that $\delta=11\vec{0}^{2m-2}$ and then 
$\alpha = 11\vec{0}^{2n-2}=1100\vec{0}^{2n-4}.$ This renaming will not change $\beta$.
Consider the following distance-2 square $$\left[\begin{matrix} x & y\\ z & w\end{matrix}\right]=\left[\begin{matrix} f(0000\vec{0}^{2n-4}) & f(1100\vec{0}^{2n-4})\\ f(0011\vec{0}^{2n-4}) & f(1111\vec{0}^{2n-4})\end{matrix}\right].$$
Recall our assumption that every distance-2 square is of type \Rmnum{1}. 
Here $x = f(\vec{0}^{2n})$, and $y = f(\alpha)$.
Since $x=1$, $\left[\begin{smallmatrix} x & y\\z & w
\end{smallmatrix}\right]$ being of type \Rmnum{1} implies that  $y=0$ or $\pm 1$
(the normalization steps include possibly multiplying a row or a column by $-1$). 
We want to show that $|y|=1$; for a contradiction, suppose that $y=0$.
We consider the following two extra entries of $f$,
where $\overline{\delta}=00\vec{1}^{2m-2}.$
$$x'=f(\overline{\delta}\vec{0}^{2n-2m})=f(00\vec{1}^{2m-2}\vec{0}^{2n-2m}) ~~\text{ and }~~ y'=f(\beta)=f(11\vec{1}^{2m-2}\vec{0}^{2n-2m}).
$$
By connecting variables $x_1$ and $x_2$ of $f$ using $=_2$ and $=_2^-$, we get signatures $\partial_{12}f$ and $\partial^{-}_{12}f$ respectively. 
Note that both $x+y$ and $x'+y'$ are entries of $\partial_{12}f$.
 Since $\partial_{12}f\in \mathscr{A}$, 
 we have $(x+y)(x'+y')=0$ or $(x+y)^2=(x'+y')^2$.
 We can also consider $\partial^{-}_{12}f$
 and get
 $(x-y)(x'-y')=0$ or $(x-y)^2=(x'-y')^2$.
 Since $x=1$ and $y=0$, we have 
 \[\Big[x'+y'=0 ~~\mbox{or}~~ (x'+y')^2= 1\Big] ~~~~\mbox{and}~~~~ \Big[x'-y'=0 ~~\mbox{or}~~ (x'-y')^2= 1\Big].\]
 Recall that $|y'|=|f(\beta)|>1$.
Clearly $x'+y'=0$ and  $x'-y'=0$ cannot be both true,
otherwise $y'=0$.
Suppose one of them is true, then $x'=\pm y'$.
And at least one of
$(x'+y')^2= 1$
or
$(x'-y')^2= 1$ holds.
So either $|x'+y'|=1$ or $|x'-y'|=1$. Substituting
$x'=\pm y'$ we reach a contradiction to
$|y'|>1$.
So 
neither $x'+y'=0$ nor  $x'-y'=0$ holds.
Then
$(x'+y')^2= 1$
and
$(x'-y')^2= 1$.
Subtracting them,  we get $x'y'=0$, and since $y' \ne 0$, we get $x'=0$. But then
this contradicts $|y'|>1$ and $(x'+y')^2= 1$.
Therefore, $y\neq 0$. Then, $y=\pm 1$.
Thus, $y=f(\delta\vec{0}^{2n-2m})=\pm 1$ for all $\delta$ with ${\rm wt}(\delta)=2.$

If $2m=2$, then the induction is finished. Otherwise, $2m>2$.
Inductively 
for some $2k\geqslant 2$, we assume  that
$f(\theta\vec{0}^{2n-2m})=\pm 1$ for all $\theta\in \mathbb{Z}_2^{2m}$ with ${\rm wt}(\theta)\leqslant 2k < 2m$.
Let  $\delta$ be such that  ${\rm wt}(\delta)=2k+2\leqslant 2m$ and we show that $f(\delta\vec{0}^{2n-2m})=\pm 1$.
Since ${\rm wt}(\delta)=2k+2\geqslant 4$, we can find four bits of $\delta$ such that the values of $\delta$ are $1$ on these four bits.
Without loss of generality, we assume that they are the first four bits, i.e. $\delta=1111\delta'$ where $\delta'\in \mathbb{Z}_2^{2m-4}.$
Consider the following distance-2 square 
$$\left[\begin{matrix} x & y\\ z & w\end{matrix}\right]=
\left[\begin{matrix} f(0000\delta'\vec{0}^{2n-2m}) & f(0011\delta'\vec{0}^{2n-2m})\\ f(1100\delta'\vec{0}^{2n-2m}) & f(1111\delta'\vec{0}^{2n-2m})\end{matrix}\right].$$
Clearly, three entries in this distance-2 square have
input strings of weight at most $2k$,
namely ${\rm wt}(0000\delta'\vec{0}^{2n-2m})=2k-2$, and ${\rm wt}(0011\delta'\vec{0}^{2n-2m})={\rm wt}(1100\delta'\vec{0}^{2n-2m})=2k$.
By our induction hypothesis, $x, y, z\in \{1, -1\}$.
Then, since the distance-2 square $\left[\begin{smallmatrix} x & y\\ z & w\end{smallmatrix}\right]$ is of type \Rmnum{1}, we have $w=f(\delta\vec{0}^{2n-2m})=\pm 1$.
The induction is complete.
This finishes the proof of
{Case 1}.

\vspace{2ex}
\noindent{\bf Case 2.} 
There is a type \Rmnum{2} or type \Rmnum{3} distance-2 square in $f$.
\vspace{1.5ex}

This is the case where signatures $g_8$ and $g'_8$ appear. We handle this case in two steps.

\vspace{1.5ex}
{\bf Step 1.} We show that after flipping variables of $f$, $\mathscr{S}(f)=\mathscr{E}_{2n}$, and after normalization $f(\alpha)=\pm 1$ or $\pm 3$ for all $\alpha\in \mathscr{S}(f)$.
 Let $\mathscr{S}_3(f)=\{\alpha\in\mathscr{S}(f)\mid f(\alpha)=\pm 3\}$.
We also show that $|\mathscr{S}_3(f)|=2^{2n-4}=\frac{1}{8}|\mathscr{S}(f)|$, and for any distinct $\alpha, \beta \in \mathscr{S}_3(f)$, ${\rm wt}(\alpha\oplus\beta)\geqslant 4$.
\vspace{1ex}



We first consider the case that there is a Type \Rmnum{2} distance-2 square in $f$.
We show that 
the only possible Type \Rmnum{2} distance-2 square in $f$ has the canonical form  $\left[\begin{smallmatrix} 1 & 3\\ 3 & 1
\end{smallmatrix}
\right]$. 
Suppose that a distance-2 square of Type \Rmnum{2} appears in $f$. 
By flipping and negating variables, we modify $f$ such that this distance-2 square is in its canonical form 
$\left[\begin{smallmatrix} 1 & a\\ a & 1
\end{smallmatrix}
\right] (a>1).$
Also, by flipping variables and renaming variables, we may assume that this distance-2 square appears on inputs $\alpha$, $\beta$, $\gamma$ and $\delta$ where
$$\left[\begin{matrix}
f(\alpha) & f(\beta)\\
f(\gamma) & f(\delta)\\
\end{matrix}\right]=
\left[\begin{matrix}
f(0000\vec{0}^{2n-4}) & f(0011\vec{0}^{2n-4})\\
f(1100\vec{0}^{2n-4}) & f(1111\vec{0}^{2n-4})\\
\end{matrix}
\right]=\left[\begin{matrix} 1 & a\\ a & 1
\end{matrix}\right].$$
Then, we  consider the entries of $\widetilde f$ on inputs $\alpha$, $\beta$, $\gamma$ and $\delta$. We have 
$$
\left[\begin{matrix}
\widetilde f(\alpha) & \widetilde f(\beta)\\
\widetilde f(\gamma) & \widetilde f(\delta)\\
\end{matrix}\right]=
\left[\begin{matrix}
f(\alpha)+f(\gamma) & f(\beta)+f(\delta)\\
f(\alpha)-f(\gamma) & f(\beta)-f(\delta)\\
\end{matrix}\right]=
\left[\begin{matrix}
1+a & 1+a\\
1-a & a-1\\
\end{matrix}
\right].$$
Since $a>1$, clearly $1+a\neq 0$, $1-a\neq 0$ and $|1+a|>|1-a|$.
Since $f$ has parity and $f(\vec{0}^{2n})=1$, $f$ has even parity.
By Lemma~\ref{lem-norm-f-tilde}(2), $\mathscr{S}(\widetilde f^{11}_{12})=\mathscr{E}_{2n-2}$ and $|1+a|=2|1-a|$.
Since $a>1$, we have 
$1+a=2(a-1)$. Then, $a=3$.
Thus, the only possible Type \Rmnum{2} distance-2 square in $f$ has the canonical form  $\left[\begin{smallmatrix} 1 & 3\\ 3 & 1
\end{smallmatrix}
\right]$. 

Under the assumption that a Type \Rmnum{2} distance-2 square appears in $f$ and
$\left[\begin{smallmatrix}
f(\alpha) & f(\beta)\\
f(\gamma) & f(\delta)\\
\end{smallmatrix}\right]
=\left[\begin{smallmatrix} 1 & 3\\ 3 & 1
\end{smallmatrix}\right]$,  we have
$\left[\begin{smallmatrix}
\widetilde f(\alpha) & \widetilde f(\beta)\\
\widetilde f(\gamma) & \widetilde f(\delta)\\
\end{smallmatrix}\right]=
\left[\begin{smallmatrix}
4 & 4\\
-2 & 2\\
\end{smallmatrix}
\right].$
As showed above, by Lemma~\ref{lem-norm-f-tilde}(2), $\mathscr{S}(\widetilde f^{11}_{12})=\mathscr{E}_{2n-2}$
and $n_{01}, n_{10}=2$ or $4$.
 We first prove 
 
 \begin{quote}
     {\bf Claim 1.} \emph{$\mathscr{S}(f_{12}^{00})=\mathscr{S}( f_{12}^{11})=\mathscr{E}_{2n-2}$, $f_{12}^{00}(\theta), f_{12}^{11}(\theta)=\pm 3$ or $\pm 1$ for all $\theta\in \mathscr{E}_{2n-2}$, and $|\mathscr{S}_3(f_{12}^{00})|+|\mathscr{S}_3(f_{12}^{11})|=2^{2n-5}.$}
 \end{quote}

Remember that  $\widetilde f^{00}_{12}, \widetilde f^{11}_{12} \in \mathscr{A}$. For any of them, its nonzero entries have the same norm.   
Since $\widetilde f(\alpha)=\widetilde f(00\vec{0}^{2n-2})=1+3=4$ and $\mathscr{S}(\widetilde f^{00}_{12})\subseteq\mathscr{E}_{2n-2}$,
for every $\theta\in \mathscr{E}_{2n-2}$, 
$\widetilde f(00\theta)=\pm 4$ or $0$.
Also, since
$\widetilde f(\gamma)=\widetilde f(11\vec{0}^{2n-2})=1-3=-2$, and $\mathscr{S}(\widetilde f^{11}_{12})=\mathscr{E}_{2n-2}$,
for every $\theta\in \mathscr{E}_{2n-2}$, 
$\widetilde f(11\theta)=\pm 2$.
Then, $$f(00\theta)=\frac{\widetilde f(00\theta)+\widetilde f(11\theta)}{2}=\frac{(\pm 4)+ (\pm 2)}{2} ~~\text{ or }~~ \frac{0+ (\pm 2)}{2}.$$
Thus, $f(00\theta)=\pm 3$ or $\pm 1$ for every $\theta\in \mathscr{E}_{2n-2}$.
Also, $$f(11\theta)=\frac{\widetilde f(00\theta)-\widetilde f(11\theta)}{2}=\frac{(\pm 4)- (\pm 2)}{2} ~\text{ or }~ \frac{0- (\pm 2)}{2}.$$
Thus, $f(11\theta)=\pm 3$ or $\pm 1$ for every $\theta\in \mathscr{E}_{2n-2}$.
Additionally note that, for any $\theta\in \mathscr{E}_{2n-2}$ if $\widetilde f(00\theta)=\pm 4$, 
then of the two values $f(00\theta)$ and $f(11\theta)$, 
exactly one is $\pm 3$ and the other one is $\pm 1$;
if $\widetilde f(00\theta)=0$, then $f(00\theta)=\pm 1$ and $f(11\theta)=\pm 1$.
Since $$|\widetilde f_{12}^{00}|^2=4^2\cdot|\mathscr{S}(\widetilde f_{12}^{00})|=|\widetilde f_{12}^{11}|^2=2^2\cdot|\mathscr{S}(\widetilde f_{12}^{11})|=2^2\cdot|\mathscr{E}_{2n-2}|,$$
 we have $|\mathscr{S}(\widetilde f_{12}^{00})|=|\mathscr{E}_{2n-2}|/4=2^{2n-5}$.
Thus, there are exactly $2^{2n-5}$ entries of $\widetilde f_{12}^{00}$ having value $\pm 4$, which give arise to exactly $2^{2n-5}$ many entries of value $\pm 3$
 among all entries of $f^{00}_{12}$ and $f^{11}_{12}$.
Claim 1 has been proved.

Next, we prove 

\begin{quote}
     {\bf Claim 2.} \emph{$\mathscr{S}(f_{12}^{01})=\mathscr{S}(f_{12}^{10})=\mathscr{O}_{2n-2}$, $f_{12}^{01}(\theta), f_{12}^{10}(\theta)=\pm 3$ or $\pm 1$ for all $\theta\in \mathscr{O}_{2n-2}$, and $|\mathscr{S}_3(f_{12}^{01})|+|\mathscr{S}_3(f_{12}^{10})|=2^{2n-5}.$}
     \end{quote}
     
We have $\widetilde f(\vec{0}^{2n})=  4$.
We have $n_{00}=4$ and $n_{11} = 2$.
Also recall that we have showed that   $n_{01}, n_{10}=2$ or $4$,
by Lemma~\ref{lem-norm-f-tilde}(2). There are three cases.
\begin{itemize}
    \item $n_{01}=n_{10}=2$. 
    Since $n_{11}=n_{01}=2$  and $$|\widetilde f_{12}^{11}|^2=n^2_{11}\cdot|\mathscr{S}(\widetilde f_{12}^{11})|=n^2_{01}\cdot|\mathscr{S}(\widetilde f_{12}^{01})|=|\widetilde f_{12}^{01}|^2,$$
    we have
    $$|\mathscr{S}(\widetilde f_{12}^{01})|=|\mathscr{S}(\widetilde f_{12}^{11})|=|\mathscr{E}_{2n-2}|=2^{2n-3}.$$
    Since $\widetilde f$ has even parity, $\mathscr{S}(\widetilde f_{12}^{01})\subseteq \mathscr{O}_{2n-2}$. As $|\mathscr{O}_{2n-2}|=2^{2n-3}$,
    we get $\mathscr{S}(\widetilde f_{12}^{01}) = \mathscr{O}_{2n-2}$.
   
    Similarly, we can show that $\mathscr{S}(\widetilde f_{12}^{10})= \mathscr{O}_{2n-2}$.
    Let $\zeta=0110\vec{0}^{2n-4}$ and $\eta=1010\vec{0}^{2n-4}$.
    Then, $\widetilde f(\zeta)=\pm 2$ and $\widetilde f(\eta)=\pm 2.$
    Note that $$f(\zeta)=\frac{\widetilde f(\zeta)+\widetilde f(\eta)}{2} ~\text{~ and ~}~ f(\eta)=\frac{\widetilde f(\zeta)-\widetilde f(\eta)}{2}.$$
    If $\widetilde f(\zeta)=\widetilde f(\eta)$, then $f(\zeta)=\pm 2$ and $f(\eta)=0$.
    If $\widetilde f(\zeta)=-\widetilde f(\eta)$, then $f(\zeta)=0$ and $f(\eta)=\pm 2$.
    We first consider the case that $f(\zeta)=\pm 2$.
    Let $\xi=1001\vec{0}^{2n-4}$.
    Consider the following distance-2 square. 
    $$\left[\begin{matrix}
f(\alpha) & f(\zeta)\\
f(\xi) & f(\delta)\\
\end{matrix}\right]=
\left[\begin{matrix}
f(0000\vec{0}^{2n-4}) & f(0110\vec{0}^{2n-4})\\
f(1001\vec{0}^{2n-4}) & f(1111\vec{0}^{2n-4})\\
\end{matrix}
\right]=\left[\begin{matrix} 1 & \pm 2\\ \ast & 1
\end{matrix}\right].$$
Clearly, it is not of type \Rmnum{1} nor type \Rmnum{3}.
Also, it is not of type \Rmnum{2} with the  canonical form  $\left[\begin{smallmatrix} 1 & 3\\ 3 & 1
\end{smallmatrix}
\right]$.  Contradiction.
If $f(\eta)=\pm 2$, then similarly by considering the distance-2 square $\left[\begin{smallmatrix}
f(\alpha) & f(\eta)\\
f(\tau) & f(\delta)\\
\end{smallmatrix}\right]$ where $\tau=0101\vec{0}^{2n-4}$, we get a contradiction.

\item $n_{01}=n_{10}=4$.
We still consider $$f(\zeta)=\frac{\widetilde f(\zeta)+\widetilde f(\eta)}{2} ~~\text{ and }~~ f(\eta)=\frac{\widetilde f(\zeta)-\widetilde f(\eta)}{2}, ~~\text{ where } \zeta=0110\vec{0}^{2n-4} ~\text{ and }~ \eta=1010\vec{0}^{2n-4}.$$
    Then, as
    $\zeta$ has leading bits 01
    and $\eta$ has leading bits 10,
    $$f(\zeta)=\frac{(\pm 4)+(\pm 4)}{2}, \frac{(\pm 4) +0}{2} \text{ or } \frac{0+0}{2} ~~\text{ and }~~ f(\eta)=\frac{(\pm 4)-(\pm 4)}{2}, \pm \frac{(\pm 4)-0}{2}  \text{ or } \frac{0-0}{2}.$$
    Thus, $f(\zeta), f(\eta)=\pm 4, \pm 2$ or $0$.
    If $f(\zeta)$ or $f(\eta) =\pm 4, \pm 2$ , then by considering the distance-2 square
    $\left[\begin{smallmatrix}
f(\alpha) & f(\zeta)\\
f(\xi) & f(\delta)\\
\end{smallmatrix}\right]$
 or $\left[\begin{smallmatrix}
f(\alpha) & f(\eta)\\
f(\tau) & f(\delta)\\
\end{smallmatrix}\right]$, we still get a contradiction.
    Thus we have $f(\zeta)=f(\eta)=0$.
    Then, consider the signature $^{H_4}_{23}f$, denoted by $\widetilde{f'}$.
    Since $f$ has even parity, $f$ satisfies {\sc 2nd-Orth} and $f\in \int_{\mathcal{B}}\mathscr{A}$,
    $\widetilde{f'}$ has even parity, $\widetilde{f'}_{23}^{00}, \widetilde{f'}_{23}^{01}, \widetilde{f'}_{23}^{10}, \widetilde{f'}_{23}^{11}\in \mathscr{A}$.
    Let $n'_{00}, n'_{01}, n'_{10}$ and $n'_{11}$ denote the norms of nonzero entries in $\widetilde{f'}_{23}^{00}, \widetilde{f'}_{23}^{01}, \widetilde{f'}_{23}^{10},$ and $\widetilde{f'}_{23}^{11}$ respectively.
    Notice that $$\widetilde{f'}(\alpha)=\widetilde{f'}(\vec{0}^{2n})=f(0000\vec{0}^{2n-4})+f(0110\vec{0}^{2n-4})=f(\alpha)+f(\zeta)=1+0=1.$$
    Thus, $n'_{00}=1$.
    Also, notice that $$\widetilde{f'}(\gamma)=\widetilde{f'}(1100\vec{0}^{2n-4})=f(1010\vec{0}^{2n-4})-f(1100\vec{0}^{2n-4})=f(\eta)-f(\gamma)=0-3=-3.$$
    Thus, $n'_{10}=3$.
    But by Lemma~\ref{lem-norm-f-tilde}(1), $n'_{00}=\sqrt{2}^k n'_{10}$ for some $k\in \mathbb{Z}$.
    However, clearly, $3\neq \sqrt{2}^k$ for any $k\in \mathbb{Z}$.
    Contradiction.
    \item Therefore exactly one of $n_{01}$ and $n_{10}$ is $2$ and the other is $4$.
    Let $(a, b)=(0, 1)$ or $(1, 0)$ be such that $n_{ab}=2$.
Since $n_{11}=2$ and $|\mathscr{S}(\widetilde{f}^{11}_{12})|=|\mathscr{E}_{2n-2}|$=$2^{2n-3}$, we have $|\mathscr{S}(\widetilde{f}^{ab}_{12})|=2^{2n-3}$.
Since $\widetilde{f}$ has even parity, $\widetilde{f}^{ab}_{12}$ has odd parity, thus $\mathscr{S}(\widetilde{f}^{ab}_{12})=\mathscr{O}_{2n-2}$.
Then, similar to the proof of $f_{12}^{00}$ and $f_{12}^{11}$, we can show that
for every $\theta\in \mathscr{O}_{2n-2}$, 
$f_{12}^{01}(\theta), f_{12}^{10}(\theta)=\pm 3$ or $\pm 1$. 
Also, among $f^{01}_{12}$ and $f^{10}_{12}$, exactly $2^{2n-5}$ many entries are $\pm 3$.
\end{itemize}
This completes the proof of Claim 2.

Thus, combining Claim 1 and Claim 2,
$\mathscr{S}(f)=\mathscr{E}_{2n}$, 
$f(\alpha)=\pm 1$ or $\pm 3$ for all $\alpha\in \mathscr{S}(f)$, and  $|\mathscr{S}_3(f)|=2^{2n-4}=\frac{1}{8}|\mathscr{S}(f)|$.
Also remember that by our assumption, $f(\vec{0}^{2n})=1$.

Now, we show that 
for any distinct $\alpha, \beta \in \mathscr{S}_3(f)$, ${\rm wt}(\alpha\oplus\beta)\geqslant 4$.
For a contradiction, suppose that $\alpha, \beta \in \mathscr{S}_3(f)$ and ${\rm wt}(\alpha\oplus\beta)=2$, and they differ at bits $i$ and $j$.
By renaming variables, without loss of generality, we may assume that $\{i, j\}=\{1, 2\}$.
This renaming does not change the value of $f(\vec{0}^{2n})=1$.
Since $f(11\vec{0}^{2n-2})=\pm 1$ or $\pm 3$,
of the values $f(00\vec{0}^{2n-2})+f(11\vec{0}^{2n-2})$  and
$f(00\vec{0}^{2n-2})-f(11\vec{0}^{2n-2})$, which are respectively 
an entry of $\widetilde f_{12}^{00}$ and an entry of
 $\widetilde f_{12}^{11}$,
at least one has norm $2$. 
Thus, among $n_{00}$ and $n_{11}$, at least one is $2$.
Since $f(\alpha)=\pm 3$ and $f(\beta) =\pm 3$, 
among $f(\alpha)+f(\beta)$ and $f(\alpha)-f(\beta)$, exactly one has norm $6$ and the other has norm $0$.
Clearly, $f(\alpha)+f(\beta)$ and $f(\alpha)-f(\beta)$ are entries of $\widetilde f$ since $\alpha$ and $\beta$ differ at bits $1$ and $2$.
Thus, among $n_{00}$, $n_{01}$, $n_{10}$ and $n_{11}$, one has norm $6$.
By Lemma~\ref{lem-norm-f-tilde}(1), $2=\sqrt{2}^k \cdot 6$ for some $k\in \mathbb{N}$. Contradiction.
This proves that
for any distinct $\alpha, \beta \in \mathscr{S}_3(f)$, ${\rm wt}(\alpha\oplus\beta)\geqslant 4$.

We have established the goal laid out in {Step 1} of {Case 2} under the assumption  
 that there is a Type II distance-2 square in
 $f$.
 

Finally, within {Step 1} of {Case 2}, we consider the case that a type \Rmnum{3} distance-2 square appears in $f$.
By flipping and negating variables, we modify $f$ such that this distance-2 square is in its canonical form 
$\left[\begin{smallmatrix} 1 & 3\\ 1 & -1
\end{smallmatrix}
\right].$
Also, by flipping variables and renaming variables, still we may assume that this distance-2 square appears on inputs $\alpha$, $\beta$, $\gamma$ and $\delta$ where
$$\left[\begin{matrix}
f(\alpha) & f(\beta)\\
f(\gamma) & f(\delta)\\
\end{matrix}\right]=
\left[\begin{matrix}
f(0000\vec{0}^{2n-4}) & f(0011\vec{0}^{2n-4})\\
f(1100\vec{0}^{2n-4}) & f(1111\vec{0}^{2n-4})\\
\end{matrix}
\right]=\left[\begin{matrix} 1 & 1\\ 3 & -1
\end{matrix}\right].$$
Then, we  consider the entries of $\widetilde f$ on inputs $\alpha$, $\beta$, $\gamma$ and $\delta$. We have 
$$
\left[\begin{matrix}
\widetilde f(\alpha) & \widetilde f(\beta)\\
\widetilde f(\gamma) & \widetilde f(\delta)\\
\end{matrix}\right]=
\left[\begin{matrix}
f(\alpha)+f(\gamma) & f(\beta)+f(\delta)\\
f(\alpha)-f(\gamma) & f(\beta)-f(\delta)\\
\end{matrix}\right]=
\left[\begin{matrix}
4 & 0\\
-2 & 2\\
\end{matrix}
\right].$$
Then exactly in the same way as the above proof when $\left[\begin{smallmatrix}
\widetilde f(\alpha) & \widetilde f(\beta)\\
\widetilde f(\gamma) & \widetilde f(\delta)\\
\end{smallmatrix}\right]=
\left[\begin{smallmatrix}
4 & 4\\
-2 & 2\\
\end{smallmatrix}
\right]$, we can show that the same result holds.
Thus, $\mathscr{S}(f)=\mathscr{E}_{2n}$, $f(\alpha)=\pm 1$ or $\pm 3$ for all $\alpha\in \mathscr{S}(f)$,  $|\mathscr{S}_3(f)|=2^{2n-4}=\frac{1}{8}|\mathscr{S}(f)|$, and for any distinct $\alpha, \beta \in \mathscr{S}(f)$ with ${\rm wt}(\alpha\oplus\beta)=2$, $\alpha$ and $\beta$ cannot be both in $\mathscr{S}_3(f)$.

This finishes the proof of {Step 1} of {Case 2}.

\vspace{1.5ex}
{\bf Step 2.}
Now we show that either $g_8$ or $g'_8$ is realizable from $f$. 
We will show that they are both irreducible and do not satisfy {\sc 2nd-Orth}, which gives  \#P-hardness.
 \vspace{1ex}

We define a graph $G_{2n}$ with vertex set
$\mathscr{E}_{2n}$, and  there is an edge between $\alpha$ and $\beta$ if ${\rm wt}(\alpha\oplus\beta)=2$. I.e.,
we view every $\alpha \in \mathscr{E}_{2n}$ as a vertex, and  the edges are
distance 2 neighbors in Hamming distance.
Then, $\mathscr{S}_3(f)$ is an independent set of $G_{2n}$.
Remember that $2n \geqslant 8$ by the hypothesis of the lemma.
If $2n\geqslant 10$, 
then by Lemma~\ref{lem-indenpent-set},  $|\mathscr{S}_3(f)|<\frac{1}{8}|\mathscr{S}(f)|$.
Contradiction.
 Thus, $2n=8$. 
After renaming and flipping variables, we may assume that $\mathscr{S}_3(f)=I_8=\mathscr{S}(f_8)$.
For brevity of notation,
let $S=\mathscr{E}_{8}$ and $T=\mathscr{S}(f_8)$.
We can pick $(x_1, \ldots, x_7)$ as a set of free variables of $S=\mathscr{E}_{8}$.
Then, there exists a multilinear  polynomial  $F(x_1, \ldots, x_{7})\in \mathbb{Z}_2[x_1, \ldots, x_7]$, and a multilinear  polynomial 
$G(x_1, \ldots, x_8)\in \mathbb{Z}_2[x_1, \ldots, x_8]$ that is viewed as a representative for its image in the quotient algebra
$\mathbb{Z}_2[x_1, \ldots, x_8]/(P_1, P_2, P_3, P_4)$ where $P_1, P_2, P_3, P_4$ are the four linear polynomials in (\ref{equ-T-polynomial}) such that $T$ is decided by $P_1=P_2=P_3=P_4=0$, such that 
$$f=\chi_S(-1)^{F(x_1, \ldots, x_{7})}+4\chi_T(-1)^{G(x_1, \ldots, x_{8})}.$$

We note that such  multilinear  polynomials  $F(x_1, \ldots, x_{7})$ and 
$G(x_1, \ldots, x_8)$ exist: For any point in $S \setminus T$ we can choose a unique value $s  \in \mathbb{Z}_2$ which represents the $\pm 1$ value
of $f$ as $(-1)^s$, and for any point in $T \subseteq S$ we can choose unique values $t \in \mathbb{Z}_2$
and $s' \in \mathbb{Z}_2$
such that $(-1)^{s'} + 4 (-1)^t$ 
represents the $\pm 3$ value
of $f$.

For $\{i, j\}\subseteq [7]=\{1, \ldots, 7\}$, remember that $F^{ab}_{ij}\in Z_2[\{x_1, \ldots, x_7\}\backslash\{x_i, x_j\}]$ is the function obtained by setting $(x_i, x_j)=(a, b)$ in $F$.
Similarly, we can define $G^{ab}_{ij}$ with respect to $P_1=P_2=P_3=P_4=0$ (any assignment of $(x_i, x_j)=(a, b)$ is consistent with $P_1=P_2=P_3=P_4=0$ which defines
$T$).
We make the following claim about $F^{ab}_{ij}$.

\begin{quote}
{\bf Claim 3.}
\emph{For all $\{i, j\}\subseteq[7]$, 
$F^{00}_{ij}+F^{11}_{ij}\equiv 0$ or $1$,  and also $F^{01}_{ij}+F^{10}_{ij}\equiv 0$ or $1$.}
\end{quote}

We first show how this claim will let us realize $g_8$ or $g'_8$, and lead to  \#P-hardness. Then, we give a proof of {Claim 3}.
By Claim 3 and Lemma~\ref{lem-mutilinear-poly}, the degree $d(F)\leqslant 2$.
\begin{itemize}
    \item  If $d(F)\leqslant 1$,  
    then $F$ is an affine linear combination of variables $x_1, \ldots, x_7$, i.e., $F=\lambda_0+\sum^7_{i=1}\lambda_ix_i$ where $\lambda_i\in \mathbb{Z}_2$ for $0\leqslant i \leqslant 7$. 
    Notice that if we negate the variable $x_i$ of $f$, we will get a signature $f'(x_1, \ldots, x_8)=(-1)^{x_i}f(x_1, \ldots, x_8).$
    For every $x_i$ appearing in $F$ (i.e., $\lambda_i=1$), we negate the variable $x_i$ of $f$.
    Also, if $\lambda_0=1$, then we normalize $f$ by a scalar $-1$.
    Then, we get a signature 
  $$f'=\chi_S\cdot 1+4\chi_T(-1)^{G'(x_1, \ldots, x_{8})}.$$
  This will not change the support of $f$ and also norms of entries of $f$.
  Thus, $f'(\alpha)=\pm 3$ or $\pm 1$ for all $\alpha \in \mathscr{S}(f')=\mathscr{E}_8$.
  Then, for every $\alpha\in T$,
  $f'(\alpha)=1+4(-1)^{G'(\alpha)}=\pm 3$, 
  which implies that $(-1)^{G'(\alpha)}=-1$ and $f'(\alpha)=-3$, because $1+4=5$ cannot be an entry of $f'$. 
  Therefore, $f'=\chi_S-4\chi_T=g_8.$
  Thus, $g_8$ is realizable from $f$.

  By merging variables $x_1$ and $x_5$ of $g_8$ using $=_2$, we can get a 6-ary signature $h$.
  We rename variables $x_2, x_3, x_4$ to $x_1, x_2, x_3$ and variables $x_6, x_7, x_8$ to $x_4, x_5, x_6$ (The choice of merging $x_1$ and $x_5$ is just for a simple renaming of variables). 
  Then after normalization by a scalar $1/2$, $h$ has the following signature matrix
  $$M_{123,456}({h})=A=\left[\begin{matrix}
-1 & 0 & 0 & 1 & 0 & 1 & 1 & 0\\
0 & -1 & 1 & 0 & 1 & 0 & 0 & 1\\
0 & 1 & -1 & 0 & 1 & 0 & 0 & 1\\
1 & 0 & 0 & -1 & 0 & 1 & 1 & 0\\
0 & 1 & 1 & 0 & -1 & 0 & 0 & 1\\
1 & 0 & 0 & 1 & 0 & -1 & 1 & 0\\
1 & 0 & 0 & 1 & 0 & 1 & -1 & 0\\
0 & 1 & 1 & 0 & 1 & 0 & 0 & -1\\
\end{matrix}\right].$$

  Consider the inner product $\langle {\bf h}_{14}^{00},  {\bf h}_{14}^{11}  \rangle$.
  One can check that 
  $$\langle {\bf h}_{14}^{00},  {\bf h}_{14}^{11}  \rangle=\sum_{1\leqslant i,j\leqslant 4}A_{i,j}\cdot A_{i+4, j+4}=8\neq 0.$$
  (This is the sum of pairwise products of every entry
  in the upper left $4\times 4$ submatrix of $A$
  with the corresponding entry of the lower right
   $4\times 4$ submatrix of $A$.)
   In fact, notice that $h(\overline{\alpha})=\overline{h(\alpha)}=h(\alpha)$.
   By considering the representative matrix $M_{r}(h)$ of $h$ (see Table~\ref{tab:16-entry}), we have 
   $$M_{r}(h)=\left[
   \begin{matrix}
    -1 & 1 & 1 & 1\\
    1 & -1 & 1 & 1\\
      1 & 1 & -1 & 1\\
        1 & 1 & 1 & -1\\
   \end{matrix}\right].$$
   Then, $$\langle {\bf h}_{14}^{00},  {\bf h}_{14}^{11}  \rangle=2({\rm perm}(M_{r}(h)_{[1, 2]})+{\rm perm}(M_{r}(h)_{[3, 4]}))=2(2+2)=8\neq 0.$$
  Also, since $\mathscr{S}(h)=\mathscr{E}_6$,
  it is easy to see that $h$ is  irreducible.
  Since $h$ does not satisfy {\sc 2nd-Orth},
  we get  \#P-hardness.
  \item If $d(F)=2$, then by Lemma~\ref{lem-mutilinear-poly}, for all $\{i, j\}\subseteq[7]$, $x_ix_j$ appears in $F$.
     Then, $F=\sum_{1\leqslant i<j \leqslant 7}x_ix_j+L$ where $L$ is an affine linear combination of variables $x_1, \ldots, x_7$.
     Since on the support $\mathscr{S}(f)=\mathscr{E}_8$, $x_1+\cdots+x_8=0$, and on Boolean inputs $x_8^2 = x_8$,
     we can substitute $F$ by $F'=F+x_8(x_1+\cdots+x_8) - (x_8^2 - x_8)=\sum_{1\leqslant i<j \leqslant 8}x_ix_j+L+x_8$  (all arithmetic mod 2).
     This will not change the signature $f$.
    Then, by negating variables of $f$ that appear as linear terms in $F'$ and normalization with a scalar $\pm 1$, 
 %
 we get a signature 
 $$f'=\chi_S (-1)^{\sum_{1\leqslant i<j\leqslant 8} x_ix_j}+4\chi_T (-1)^{G'(x_1, \ldots, x_8)}=q_8+4\chi_T (-1)^{G'(x_1, \ldots, x_8)}.$$
 where $q_8 = \chi_S (-1)^{\sum_{1\leqslant i<j\leqslant 8} x_ix_j}$  (see form (\ref{equ-T-polynomial})).
For every $\alpha \in T$,
since ${\rm wt}(\alpha)=0, 4$ or $8$,
it is easy to see that
$q_8(\alpha)=(-1)^{{\rm wt}(\alpha)\choose 2}=1$. 
Thus, $(-1)^{G'(\alpha)}$ must be $-1$ in order to get $1-4=-3$, of norm 3 for $f'$.
The other choice would give $1+4=5$ to be an entry of $f'$, a contradiction.
 Therefore,  $f'(\alpha)=q_8-4\chi_T=g'_8.$
 Thus, $g'_8$ is realizable from $f$.

  By merging variables $x_1$ and $x_5$ of $g'_8$ using $=_2^-$, we can get a 6-ary signature $h'$.
  After renaming variables  (same as we did for $h$) and normalization by a scalar $-1/2$,
  we have 
  $$M_{123,456}({h'})=B=\left[\begin{matrix}
1 & 0 & 0 & 1 & 0 & 1 & 1 & 0\\
0 & -1 & 1 & 0 & 1 & 0 & 0 & -1\\
0 & 1 & -1 & 0 & 1 & 0 & 0 & -1\\
1 & 0 & 0 & 1 & 0 & -1 & -1 & 0\\
0 & 1 & 1 & 0 & -1 & 0 & 0 & -1\\
1 & 0 & 0 & -1 & 0 & 1 & -1 & 0\\
1 & 0 & 0 & -1 & 0 & -1 & 1 & 0\\
0 & -1 & -1 & 0 & -1 & 0 & 0 & -1\\
\end{matrix}\right].$$
  
Consider the inner product $\langle {\bf h'}_{14}^{00},  {\bf h'}_{14}^{11}  \rangle$.
  One can check that 
  $$\langle {\bf h'}_{14}^{00},  {\bf h'}_{14}^{11}  \rangle=\sum_{1\leqslant i,j\leqslant 4}B_{i,j}\cdot B_{i+4, j+4}=-8\neq 0.$$
  Also, since $\mathscr{S}(h')=\mathscr{E}_6$,
  it is easy to see that $h'$ is  irreducible.
  Since $h'$ does not satisfy {\sc 2nd-Orth},
  we get  \#P-hardness.
\end{itemize}

This completes the proof of
{Step 2}, and the proof
of the lemma, modulo {Claim 3}.

\vspace{1ex}
Now, we prove {Claim 3}
that for all $\{i, j\}\subseteq [7]$, 
$F^{00}_{ij}+F^{11}_{ij}\equiv 0$ or $1$ and $F^{01}_{ij}+F^{10}_{ij}\equiv 0$ or $1$.
For simplicity of notation, we prove this for $\{i, j\}=\{1, 2\}$. The proof for arbitrary $\{i, j\}$ is the same by replacing  $\{1, 2\}$ by $\{i, j\}$.
Since $f\in \int_{\mathcal{B}}\mathscr{A}$, $\widetilde f_{12}^{00}, \widetilde f_{12}^{01}, \widetilde f_{12}^{10}, \widetilde f_{12}^{11}\in \mathscr{A}$.
Remember all nonzero entries in  $\widetilde f_{12}^{ab}$
have the same norm, denoted by  $n_{ab}$.
We first show that between $\widetilde f_{12}^{00}$ and $\widetilde f_{12}^{11}$, 
exactly one has support $\mathscr{E}_{2n-2}$ and its nonzero entries have norm $2$ and the other has nonzero entries of norm $4$, and
between $\widetilde f_{12}^{01}$ and $\widetilde f_{12}^{10}$, 
exactly one has support $\mathscr{O}_{2n-2}$ and its nonzero entries have norm $2$ and the other has  nonzero entries of norm $4$.
(This is not what we have proved in {Step 1} where $\{1, 2\}$ is a pair of particularly chosen indices. 
Here $\{1, 2\}$ means an arbitrary pair $\{i, j\}$.)

Consider $f^{00}_{12}(\vec{0}^{6})$ and  $f^{11}_{12}(\vec{0}^{6})$. 
By {Step 1} of {Case 2} and Lemma~\ref{lem-indenpent-set},
 we may assume that $\mathscr{S}_3(f)=\mathscr{S}(f_8)$ (after flipping and renaming variables). We have $00\vec{0}^{6}\in \mathscr{S}_3(f)$ and $11\vec{0}^{6}\notin\mathscr{S}_3(f)$.
Thus, $f^{00}_{12}(\vec{0}^{6})=\pm 3$  and $f^{11}_{12}(\vec{0}^{6})=\pm 1$.
(This is true when replacing $\{1, 2\}$ by an arbitrary pair of indices $\{i, j\}$.)
Thus, between $$\widetilde f^{00}_{12}(\vec{0}^{6})=f^{00}_{12}(\vec{0}^{6})+f^{11}_{12}(\vec{0}^{6}) \text{~~ and ~~} \widetilde f^{11}_{12}(\vec{0}^{6})=f^{00}_{12}(\vec{0}^{6})-f^{11}_{12}(\vec{0}^{6}),$$ one has norm $2$ and the other has norm $4$.
They are both nonzero.
Then, between $n_{00}$ and $n_{11}$, one is $2$ and the other is $4$.
By Lemma~\ref{lem-norm-f-tilde}(2), between $\widetilde f_{12}^{00}$ and $\widetilde f_{12}^{11}$, the one whose nonzero entries have norm $2$ has support $\mathscr{E}_{6}$, and moreover
$n_{01}$ and $n_{10}=2$ or $4$.
Since there exists $(a, b)=(0, 0)$ or $(1, 1)$ such that $$|\widetilde f_{12}^{ab}|^2=n^2_{ab}\cdot|\mathscr{S}(\widetilde f_{12}^{ab})|=2^2\cdot |\mathscr{E}_{6}|,$$ 
for $\widetilde f_{12}^{cd}$ where $(c, d)=(0, 1)$ or $(1, 0)$,  if $n_{cd}=2$, then $|\mathscr{S}(\widetilde f_{12}^{cd})|=|\mathscr{E}_{6}|=|\mathscr{O}_{6}|$.
Since $\widetilde f_{12}^{cd}$ has odd parity, $\mathscr{S}(\widetilde f_{12}^{cd})\subseteq\mathscr{O}_{6}.$
Thus, $|\mathscr{S}(\widetilde f_{12}^{cd})|=2^{2n-3}$ implies that 
$\mathscr{S}(\widetilde f_{12}^{cd})=\mathscr{O}_{6}.$
\begin{itemize}
    \item 
If $n_{01}=n_{10}=2$, then $\mathscr{S}(\widetilde f_{12}^{01})=\mathscr{S}(\widetilde f_{12}^{10})=\mathscr{O}_{6}$.
For an arbitrary $\theta\in \mathscr{O}_{6}$, $$f(01\theta)=\frac{\widetilde f(01\theta)+\widetilde f(10\theta)}{2}=\frac{(\pm 2)+(\pm 2)}{2} \text{~ and ~} f(10\theta)=\frac{\widetilde f(01\theta)-\widetilde f(10\theta)}{2}=\frac{(\pm 2)-(\pm 2)}{2}.$$
Thus, between $f(01\theta)$ and $f(10\theta)$, exactly one has norm $2$ and the other has norm $0$.
This gives a contradiction since every nonzero entry of $f$ has norm $1$ or $3$.
\item
If $n_{01}=n_{10}=4$, then still consider $f(01\theta)$ and $f(10\theta)$ for an arbitrary $\theta\in \mathscr{O}_{6}$. We know that $f(01\theta), f(10\theta)=\pm 4, \pm 2$ or $0$.
The case that $f(01\theta)=0$ or $f(10\theta)=0$ cannot occur since $\mathscr{S}(f)=\mathscr{E}_{2n}$ and clearly, $01\theta, 10\theta \in \mathscr{E}_{2n}$.
Thus, $f(01\theta), f(10\theta)=\pm 4, \pm 2$.
Still, we get a contradiction since every nonzero entry of $f$ has norm $1$ or $3$.
\item Thus, between $n_{01}$ and $n_{10}$, one is $2$ and the other is $4$.
\end{itemize}
Then, between $\widetilde f_{12}^{01}$ and $\widetilde f_{12}^{10}$, exactly one has support $\mathscr{O}_{6}$ and its nonzero entries have norm $2$, and the other has  nonzero entries of norm $4$.

Now, we show that $F_{12}^{00}+F_{12}^{11} \equiv 0$ or $1$.
We first consider the case that 
 between $\widetilde f_{12}^{00}$ and $\widetilde f_{12}^{11}$,
$\widetilde f_{12}^{11}=f_{12}^{00}-f_{12}^{11}$ is the signature whose support is $\mathscr{E}_{6}$ and nonzero entries have norm $2$; the case where it is $\widetilde f_{12}^{00}$ will be addressed shortly.
Let $S_0$ be the subspace in $\mathbb{Z}_2^6$ obtained by 
setting $x_1=x_2=0$ in $S=\mathscr{S}(f)=\mathscr{E}_{8}$, 
and $S_1$ be the subspace in $\mathbb{Z}_2^6$ obtained by 
setting  $x_1=x_2=1$.
Similarly, we can define $T_0$ and $T_1$,
replacing $S$ in the definition by $T=\mathscr{S}_3(f)=I_8$.
Clearly, $S_0=S_1=\{(x_3, \ldots, x_8)\in \mathbb{Z}_2^6\mid x_3+\cdots x_8=0\}=\mathscr{E}_{6}$.
Also, one can check that $T_0$ is disjoint with $T_1$.
Then $$f_{12}^{00}=\chi_{S_0}(-1)^{F_{12}^{00}(x_3, \ldots, x_{7})}+4\chi_{T_0}(-1)^{G_{12}^{00}(x_3, \ldots, x_{8})},$$ and 
$$f_{12}^{11}=\chi_{S_1}(-1)^{F_{12}^{11}(x_3, \ldots, x_{7})}+4\chi_{T_1}(-1)^{G_{12}^{11}(x_3, \ldots,  x_{8})}.$$
Thus, 
$$\widetilde f_{12}^{11}=\chi_{\mathscr{E}_6}((-1)^{F_{12}^{00}(x_3, \ldots, x_{7})}-(-1)^{F_{12}^{11}(x_3, \ldots, x_{7})})+4\chi_{T_0}(-1)^{G_{12}^{00}(x_3, \ldots, x_{8})}-4\chi_{T_1}(-1)^{G_{12}^{11}(x_3, \ldots, x_{8})}.$$
Since $\mathscr{S}(\widetilde f_{12}^{11})=\mathscr{E}_{6}$ and $n_{11}=2$, $\widetilde f_{12}^{11}(\theta)=\pm 2$ for every $\theta \in \mathscr{E}_{6}$.
If $\theta \notin T_0\cup T_1$, then $$\widetilde f_{12}^{11}(\theta)=(-1)^{F_{12}^{00}(\theta)}-(-1)^{F_{12}^{11}(\theta)}=\pm 2.$$
If $\theta \in T_0\cup T_1$, then it belongs to exactly one
of $T_0$ or $T_1$,
 $$\widetilde f_{12}^{11}(\theta)=(-1)^{F_{12}^{00}(\theta)}-(-1)^{F_{12}^{11}(\theta)}+4a=\pm 2,$$
 where $a= \pm 1$.
 In this case, the sum of the first two terms is still
 $(-1)^{F_{12}^{00}(\theta)}-(-1)^{F_{12}^{11}(\theta)}=\pm 2$, because the only other possible  value
 for $(\pm 1) - (\pm 1)$ is 
 $0$ and then we would have $4a=\pm 2$, a contradiction.
 Thus, for  every  $(x_3, \ldots, x_{7}) \in \mathbb{Z}^5_{2}$ which decides every $(x_3, \ldots, x_{8}) \in \mathscr{E}_{6}$ by $x_8=x_3+\cdots+x_7$,
$$(-1)^{F_{12}^{00}(x_3, \ldots, x_{7})}-(-1)^{F_{12}^{11}(x_3, \ldots, x_{7})}=\pm 2.$$
This implies that $$(-1)^{F_{12}^{00}(x_3, \ldots, x_{7})}=-(-1)^{F_{12}^{11}(x_3, \ldots, x_{7})}.$$
Thus, $$(-1)^{F_{12}^{00}(x_3, \ldots, x_{7})+F_{12}^{11}(x_3, \ldots, x_{7})}=-1.$$
Then, $F_{12}^{00}+F_{12}^{11}\equiv 1$.

Now we address the case that (between $\widetilde f_{12}^{00}$ and $\widetilde f_{12}^{11}$) it is $\widetilde f_{12}^{00}=f_{12}^{00}+f_{12}^{11}$  the signature whose support is $\mathscr{E}_{6}$ and nonzero entries have norm $2$.  Then similarly
for  every  $(x_3, \ldots, x_{7}) \in \mathbb{Z}^5_{2}$,
which determines every
$(x_3, \ldots, x_{8}) \in \mathscr{E}_{6}$,
$$(-1)^{F_{12}^{00}(x_3, \ldots, x_{7})}+(-1)^{F_{12}^{11}(x_3, \ldots, x_{7})}=\pm 2.$$
This implies that $$(-1)^{F_{12}^{00}(x_3, \ldots, x_{7})}=(-1)^{F_{12}^{11}(x_3, \ldots, x_{7})}.$$
Thus, $$(-1)^{F_{12}^{00}(x_3, \ldots, x_{7})+F_{12}^{11}(x_3, \ldots, x_{7})}=1$$
Then, $F_{12}^{00}+F_{12}^{11}\equiv 0.$

We have proved that, $F_{12}^{00}+F_{12}^{11}\equiv 0$ or $1$.

Also, consider $\widetilde f_{12}^{01}$ and $\widetilde f_{12}^{10}$.
One of them is a signature whose support is $\mathscr{O}_{2n-2}$ and nonzero entries have norm $2$.
Then similarly, for every $(x_3, \ldots, x_{7}) \in \mathbb{Z}^5$  which decides every $(x_3, \ldots, x_{8}) \in \mathscr{O}_{6}$ by $x_8=1+x_3+\cdots+x_7$,
$$(-1)^{F_{12}^{01}(x_3, \ldots, x_{7})}+(-1)^{F_{12}^{10}(x_3, \ldots, x_{7})}=\pm 2,$$ or 
$$(-1)^{F_{12}^{01}(x_3, \ldots, x_{7})}-(-1)^{F_{12}^{10}(x_3, \ldots, x_{7})}=\pm 2.$$
Then, $F_{12}^{01}+F_{12}^{10}\equiv 0$ or  $F_{12}^{01}+F_{12}^{10}\equiv 1$.
The above proof holds for all $\{i, j\}\subseteq [7]$.
Thus, for all $\{i, j\}\subseteq [7]$, 
$F_{ij}^{00}+F_{ij}^{11}\equiv 0$ or $1$, and $F_{ij}^{01}+F_{ij}^{10}\equiv 0$ or $1$.
\end{proof}
\begin{remark}
 The above proof does not require $\mathcal{F}$ to be non-$\mathcal{B}$ hard.
\end{remark}




\subsection{Support condition}
 Then, by further assuming that nonzero entries of $f$ have the same norm, we show that $f$ has affine support or we can get the \#P-hardness for non-$\mathcal{B}$ hard set $\mathcal{F}$ (Lemma~\ref{lem-affine-support}).
Here, we do require  $\mathcal{F}$ to be non-$\mathcal{B}$ hard.

We first give one more result about $\widetilde f$.
Remember that if $f\in \int_{\mathcal{B}}\mathscr{A}$, then  $\widetilde{f}^{00}_{12}$, $\widetilde{f}^{01}_{12}$, $\widetilde{f}^{10}_{12}$, $\widetilde{f}^{11}_{12} \in \mathscr{A}$, and $n_{ab}$ denotes the norm of nonzero entries of $\widetilde{f}^{ab}_{12}$.
Let $\widetilde{\mathcal{B}}=\left\{\widetilde{=^+_2}, \widetilde{=^-_2}, \widetilde{\neq^+_2}, \widetilde{\neq^-_2} \right\}$ where $\widetilde{=^+_2}=(2, 0, 0, 0)$, 
$\widetilde{=^-_2}=(0, 0, 0, 2)$, $\widetilde{\neq^+_2}=(0, 2, 0, 0)$ and $\widetilde{\neq^-_2}=(0, 0, 2, 0)$.
Signatures in 
$\widetilde{\mathcal{B}}$ are obtained by performing the $H_4$ gadget construction on binary signatures  in $\mathcal{B}$. 

\begin{lemma}\label{lem-tilde-norm-2}
 Let $f$ be an irreducible signature of arity $2n\geqslant 6$ with the following properties.
 \begin{enumerate}
     \item  $f$ has even parity, $f$ satisfies 
{\sc 2nd-Orth}, and $f\in \int_{\mathcal{B}}\mathscr{A}$;
\item for all $\{i, j\}$ disjoint with $\{1, 2\}$ and every $b\in \mathcal{B}$,
either $M(\mathfrak{m}_{12}(\partial^b_{ij}f))=\lambda^b_{ij}I_4$ for some real $\lambda^b_{ij}\neq 0$, or there exists a nonzero binary signature $g^b_{ij}\in \mathcal{B}$ such that $g^b_{ij}(x_1, x_2)\mid \partial^b_{ij}f$.
 \end{enumerate}{}
If $\mathscr{S}(\widetilde f^{01}_{12})=\mathscr{S}(\widetilde f^{10}_{12})$, $n_{00}>n_{01}>0$, 
then $\mathscr{S}(\widetilde f^{01}_{12})=\mathscr{O}_{2n-2}$. 
%
\end{lemma}
\begin{proof}
We first analyze the second property of $f$, i.e., the property about $\partial^b_{ij}f$.
\begin{itemize}
    \item 
If $M(\mathfrak{m}_{12}(\partial^b_{ij}f))=\lambda^b_{ij}I_4$, by Lemma~\ref{lem-tilde-1}, then $M(\mathfrak{m}_{12}(\widetilde{\partial^b_{ij}f)})=2\lambda^b_{ij}I_4$.
Since $\{i, j\}$ is disjoint with $\{1, 2\}$, the $H_4$ gadget on variables $x_1$ and $x_2$ commutes with the merging gadget on variables $x_i$ and $x_j$.
Thus, $\widetilde{\partial^b_{ij}f}={\partial^b_{ij}\widetilde f}$.
 Let $(\partial^b_{ij}\widetilde f)_{12}^{ab}$ be the signature obtained by setting variables $x_1$ and $x_2$ of $\partial^b_{ij}\widetilde f$ to $a$ and $b$, and $\partial^b_{ij}(\widetilde f_{12}^{ab})$ be the signature obtained by merging variables $x_i$ and $x_j$ of $\widetilde f_{12}^{ab}$.
 Again, since $\{1, 2\}$ and $\{i, j\}$ are disjoint, 
 $(\partial^b_{ij}\widetilde f)_{12}^{ab}=\partial^b_{ij}(\widetilde f_{12}^{ab})$.
 We denote them by $\partial^b_{ij}\widetilde f_{12}^{ab}$.
 Then, since $M(\mathfrak{m}_{12}(\widetilde{\partial^b_{ij}f)})=M(\mathfrak{m}_{12}({\partial^b_{ij}\widetilde f)})=2\lambda^b_{ij}I_4$,
 $$|\partial^b_{ij} {\bf \widetilde f}_{12}^{00}|^2=|\partial^b_{ij}\widetilde {\bf f}_{12}^{01}|^2=|\partial^b_{ij}\widetilde {\bf f}_{12}^{10}|^2=|\partial^b_{ij}\widetilde {\bf f}_{12}^{11}|^2=2\lambda^b_{ij}\neq 0.$$
 
 \item If $g^b_{ij}(x_1, x_2)\mid \partial^b_{ij}f$, i.e, $\partial^b_{ij}f=g^b_{ij}(x_1, x_2)\otimes h$, then $\widetilde {\partial^b_{ij}f}=\partial^b_{ij}\widetilde f=\widetilde{g^b_{ij}}(x_1, x_2)\otimes h$.
 Since $g^b_{ij}\in \mathcal{B}$, $\widetilde{g^b_{ij}}\in \widetilde{\mathcal{B}}$.
 By the form of signatures in $\widetilde{\mathcal{B}}$, 
 among $\partial^b_{ij} {\widetilde f}_{12}^{00}$, $\partial^b_{ij} {\widetilde f}_{12}^{01}$, $\partial^b_{ij} {\widetilde f}_{12}^{10}$ and 
 $\partial^b_{ij} {\widetilde f}_{12}^{11}$, at most one is a nonzero signature. 
 \end{itemize}
 
Combining  the above two cases
 we have that, among $\partial^b_{ij} {\widetilde f}_{12}^{00}$, $\partial^b_{ij} {\widetilde f}_{12}^{01}$, $\partial^b_{ij} {\widetilde f}_{12}^{10}$ and 
 $\partial^b_{ij} {\widetilde f}_{12}^{11}$, if at least two of them are nonzero signatures then they are all nonzero signatures.
 
 Now, we show that $\mathscr{S}(\widetilde f^{01}_{12})=\mathscr{O}_{2n-2}$.
 Since $f$ has even parity, $\widetilde f$ also has even parity.
 Then, $\widetilde f^{01}_{12}$ has odd parity, i.e.,
 $\mathscr{S}(\widetilde f^{01}_{12})\subseteq\mathscr{O}_{2n-2}$.
 For a contradiction, suppose that $\mathscr{S}(\widetilde f^{01}_{12})\subsetneq\mathscr{O}_{2n-2}$.
 Since $n_{01}>0$, $\mathscr{S}(\widetilde f^{01}_{12})\neq \emptyset$.
 Then, we can pick a pair of inputs $\alpha, \beta\in \mathscr{O}_{2n-2}$ with ${\rm wt}(\alpha\oplus\beta)=2$ such that
 $\alpha \in \mathscr{S}(\widetilde f^{01}_{12})$ and $\beta \notin \mathscr{S}(\widetilde f^{01}_{12}).$
 Also, since $\mathscr{S}(\widetilde f^{01}_{12})=\mathscr{S}(\widetilde f^{10}_{12})$, $\alpha \in \mathscr{S}(\widetilde f^{10}_{12})$ and $\beta \notin \mathscr{S}(\widetilde f^{10}_{12}).$
 Thus, 
  $|\widetilde f^{01}_{12}(\alpha)|=n_{01}$ and $|\widetilde f^{01}_{12}(\beta)|=0$, and   $|\widetilde f^{10}_{12}(\alpha)|=n_{10}$ and $|\widetilde f^{10}_{12}(\beta)|=0.$
  Suppose that $\alpha$ and $\beta$ differ in bits $i$ and $j$.
  Clearly, $\{i, j\}$ is disjoint with $\{1, 2\}.$
  Depending whether $\alpha_i=\alpha_j$ or $\alpha_i\neq \alpha_j$, 
  we connect variables $x_i$ and $x_j$ of $\widetilde f$ using $=_2^{+}$ 
  or $\neq_2^+$. 
  We get signatures $\partial^+_{ij} \widetilde f$ 
  or  $\partial^{\widehat+}_{ij} \widetilde f$
  respectively.
  We  consider the case that 
$\alpha_i=\alpha_j$; in this case $\{\alpha_i\alpha_j, \beta_i\beta_j\} = \{00, 11\}$.
For the case that $\alpha_i\neq \alpha_j$, the analysis is the same by replacing $\partial^+_{ij}\widetilde f$
with $\partial^{\widehat+}_{ij}\widetilde f$.

Consider $\partial^+_{ij}\widetilde f$.
Then,
because $\{\alpha_i\alpha_j, \beta_i\beta_j\} = \{00, 11\}$,
$\widetilde f^{01}_{12}(\alpha)+\widetilde f^{01}_{12}(\beta)$ and $\widetilde f^{10}_{12}(\alpha)+\widetilde f^{10}_{12}(\beta)$
are entries of $\partial^+_{ij}\widetilde f$;
more precisely, they are entries of $\partial^+_{ij}\widetilde f_{12}^{01}$ and $\partial^+_{ij}\widetilde f_{12}^{10}$ respectively.
Since $\widetilde f^{01}_{12}(\beta)=\widetilde f^{10}_{12}(\beta)=0$, we have
$$|\widetilde f^{01}_{12}(\alpha)+\widetilde f^{01}_{12}(\beta)|=|\widetilde f^{01}_{12}(\alpha)|=n_{01}\neq 0, \text{~~and~~}
|\widetilde f^{10}_{12}(\alpha)+\widetilde f^{10}_{12}(\beta)|=|\widetilde f^{10}_{12}(\alpha)|=n_{10}\neq 0.$$
Thus, $\partial^+_{ij}\widetilde f_{12}^{01}$ has a nonzero entry with norm $n_{01}$, and then
$\partial^+_{ij}\widetilde f_{12}^{01}\not\equiv 0$.
Also, we have $\partial^+_{ij}\widetilde f_{12}^{10}\not\equiv 0$.
Thus at least two
 among $\partial^+_{ij} {\widetilde f}_{12}^{00}$, $\partial^+_{ij} {\widetilde f}_{12}^{01}$, $\partial^+_{ij} {\widetilde f}_{12}^{10}$ and 
 $\partial^+_{ij} {\widetilde f}_{12}^{10}$ are nonzero, it follows that all of them are
 nonzero signatures.

Then $\partial^+_{ij}\widetilde f_{12}^{00}\not\equiv 0$.
Let $\partial^+_{ij}\widetilde f_{12}^{00}(\gamma)$ be a nonzero entry of $\partial^+_{ij}\widetilde f_{12}^{00}$.
Then, $\partial^+_{ij}\widetilde f_{12}^{00}(\gamma)=\widetilde f_{12ij}^{0000}(\gamma)+\widetilde f_{12ij}^{0011}(\gamma)\neq 0$.\footnote{For the case that $\alpha_i\neq \alpha_j$, $\partial^+_{ij}\widetilde f_{12}^{00}(\gamma)=\widetilde f_{12ij}^{0000}(\gamma)+\widetilde f_{12ij}^{0011}(\gamma)$ will be replced by  $\partial^{\widehat+}_{ij}\widetilde f_{12}^{00}(\gamma)=\widetilde f_{12ij}^{0001}(\gamma)+\widetilde f_{12ij}^{0010}(\gamma)$.}
Clearly,
$\widetilde f_{12ij}^{0000}(\gamma)$ and $\widetilde f_{12ij}^{0011}(\gamma)$ are entries of  $\widetilde f_{12}^{00}$, and they have norm $n_{00}$ or $0$.
Thus,  $\partial^+_{ij}\widetilde f_{12}^{00}(\gamma)$ has norm $2n_{00}$ or $n_{00}$.
Also, $\partial^+_{ij}\widetilde f_{12}^{00}(\gamma)$ is an entry of $\partial^+_{ij}\widetilde f$ on the input $00\gamma$.
Thus, $\partial^+_{ij}\widetilde f$ has a nonzero entry with norm $2n_{00}$ or $n_{00}$.
Since $n_{00}>n_{01}$, both $2n_{00}$ and $n_{00}$ are not equal to $n_{01}$.
Thus, $\partial^+_{ij}\widetilde f$ has two nonzero entries with different norms.
Such a signature is not in $\mathscr{A}$.
However,  since $f\in \int_{\mathcal{B}}\mathscr{A}$, 
by Lemma \ref{lem-tilde-1}, $\partial^+_{ij}\widetilde f \in \mathscr{A}$. 
Contradiction.
Thus,  $\mathscr{S}(\widetilde f^{01}_{12})=\mathscr{O}_{2n-2}$.
\end{proof}

We also give a result about the edge partition of complete graphs into two  complete tripartite subgraphs.
This result should also be of independent interest. We say   a graph $G = (V, E)$
is tripartite if 
 $V = V_1 \sqcup  V_2 \sqcup V_3$
and every $e \in E$ is between distinct
$V_i$ and $V_j$. Here $\sqcup$ 
denotes disjoint union. The parts
$V_i$ are allowed to be empty. It is a complete 
tripartite graph if every
pair between distinct
$V_i$ and $V_j$ is an edge.

\begin{definition}
Let $K_n$ be the complete graph on $n$ vertices.
We say $K_n$ has a \emph{tripartite 2-partition} if there exist  complete tripartite subgraphs $T_1$ and 
$T_2$ 
such that $\{E(T_1),   E(T_2)\}$ is a partition of $E(K_n)$, i.e., $E(K_n)=E(T_1)\sqcup E(T_2)$.
We say $T_1$ and $T_2$ are witnesses of a tripartite 2-partition of $K_n$.
\end{definition}

\begin{lemma}\label{lem-tripartite-partition}
$K_n$ has a tripartite 2-partition
iff $n \leqslant 5$.
For $n=5$, up to 
an automorphism of $K_5$,
there is a  unique
tripartite 2-partition where
$T_1$ is a triangle on 
$\{v_1, v_2, v_3\}$
and $T_2$ is the complete tripartite
graph with parts
$\{v_1, v_2, v_3\}$,  
$\{v_4\}$ and  $\{v_5\}$.
\end{lemma}

\begin{proof}

Let $T$ be a complete tripartite graph. 
Let $G_{2,1}$ be the union of $K_2$ and an isolated vertex. 
We first prove the following two claims.
\begin{quote}
  {\bf Claim 1.} \emph{$G_{2,1}$ is not an
induced subgraph of $T$.} 
\end{quote}

For a contradiction, suppose $G_{2,1}
= (V, E)$ is
an induced subgraph of  $T$, where
$V=\{v_1, v_2, v_3\}$, and 
$E = \{(v_1, v_2)\}$.
Then, $v_1$ and $v_2$ belong to distinct parts of $T$.
Since $(v_1, v_3), (v_2, v_3)\notin E(T)$, $v_1$ and $v_3$ belong to the same part of $T$, and so are  $v_2$ and $v_3$.
Thus, $v_1$ and $v_2$ belong to the same part of $T$. This contradiction proves Claim 1. 

\begin{quote}
   {\bf Claim 2.} \emph{$K_4$
   is not an
induced subgraph of $T$.} 
\end{quote}

For a contradiction, suppose $K_4$ 
on  $V=\{v_1, v_2, v_3, v_4\}$ is
an induced subgraph of  $T$. Then,
 for any two distinct vertices $v_i, v_j\in V$, the edge $(v_i, v_j)\in K_4$
 shows that $v_i$ and $v_j$ belong to distinct parts in $T$.
 But $T$ has at most three distinct
 nonempty parts. 
This contradiction proves Claim 2. 
 \vspace{.1in}
 
Now, we prove this lemma.
The cases  $n=1, 2, 3$ are trivial. 
When $n=4$, we have the following two tripartite 2-partitions of $K_4$, 
with $V(T_1)=\{v_1\}\sqcup\{v_2\}\sqcup\{v_3\}$ and $V(T_2)=\{v_1, v_2, v_3\}\sqcup\{v_4\}\sqcup\emptyset$, or alternatively with  $V({T'_1})=\{v_1\}\sqcup\{v_2\}\sqcup\emptyset$ and $V({T'_2})=\{v_1, v_2\}\sqcup\{v_3\}\sqcup\{x_4\}$.

We consider $n\geqslant 5$.
Suppose 
 $K_n$ has a tripartite 2-partition with  complete tripartite subgraphs $T_1 = (V_1, E_1)$ and 
$T_2 = (V_2, E_2)$. We write $(A_i, B_i, C_i)$ for the three parts of $T_i$, $i=1,2$.

Clearly $V = V_1 \cup V_2$, as all vertices of $V$ must appear
in either $T_1$ or $T_2$, for otherwise any edge incident to $v \in V \setminus (V_1 \cup V_2)$
is not in $E_1 \cup E_2$. 
If all parts of both $T_1$ and $T_2$ have size at most 1, then $|E_1 \sqcup E_2| \leqslant 6 < |K_5| \leqslant |K_n|$, a contradiction.
So at least one part, say $A_1$, has size
$|A_1| \geqslant 2$, and we let $a, a' \in A_1$. 
Then, $(a, a')\notin E_1$.
Thus, $(a, a')\in E_2$ and $a, a'\in V_2$.

We show that $(V_1  \setminus A_1) \cap (V_2  \setminus A_1) = \emptyset$.
Otherwise, there exists $v\in (V_1  \setminus A_1) \cap (V_2  \setminus A_1).$
Then, edges $(v, a), (v, a')\in E_1$.
Thus, among edges $(v, a), (v, a')$ and $(a, a')$ of $K_n$, $(a, a')$ is the only one in $T_2$.
Since $v, a, a'\in V_2$, $G_{2,1}$ is an induced subgraph of $T_2$. A violation of Claim 1. 

If both $V_1  \setminus A_1$ and $V_2  \setminus A_1$ are nonempty, then an edge in $K_n$ between $u \in V_1  \setminus A_1$
and $v \in V_2  \setminus A_1$ is in neither
$E_1$ nor $E_2$, since $u \not \in V_2$ and $v \not \in V_1$. This is a contradiction.
If $V_1 \setminus A_1  = \emptyset$, then
 $E_1 = \emptyset$, and then all edges of
$K_n$ belong to $T_2$, which violates  Claim 2.
So $V_2  \setminus A_1 = \emptyset$.
Since $V=V_1\cup V_2$, $V_2  \setminus A_1 = \emptyset$ implies that $V=V_1$.

Clearly $V_1 \setminus A_1  = B_1 \sqcup C_1$.
If $|B_1| \geqslant 2$, then there exists some $\{u, v\} \subseteq B_1 \subseteq V_1 \setminus A_1$,
which is disjoint from $V_2$. Thus
$\{u, v\}  \not \in E_1 \sqcup E_2$, a contradiction. Hence $|B_1| \leqslant 1$.
Similarly $|C_1| \leqslant 1$.
Finally, if $|A_1| \geqslant 4$, then there is a $K_4$ inside $A_1$
which must be an induced subgraph of $T_2$, 
a violation of  Claim 2. Thus $|A_1| \leqslant 3$.
It follows that $n \leqslant 5$ since $V=V_1=A_1\sqcup B_1\sqcup C_1$. 
If $n=5$, then $|A_1|=3$ and $|B_1| = |C_1| = 1$.
After relabeling vertices, we may assume that $A_1=\{v_1, v_2, v_3\}$, $B_1=\{v_4\}$ and $C_1=\{v_5\}$.
Then, we have $A_2=\{v_1\}$, $B_2=\{v_2\}$ and $C_2=\{v_3\}.$
This gives the unique tripartite 2-partition of $K_5$.
\end{proof}

We will apply Lemma~\ref{lem-tripartite-partition} to  multilinear $\mathbb{Z}_2$-polynomials. 
Remember that we take the reduction
of polynomials in $\mathbb{Z}_2[x_1, \ldots, x_n]$ modulo the ideal
generated by $\{x_i^2 - x_i \mid i \in [n]\}$ replacing any $F$
by its unique multilinear representative.
\begin{definition}
Let $F(x_1, \ldots, x_n)\in\mathbb{Z}_2[x_1, \ldots, x_n]$ be a  complete quadratic polynomial.
We say $F$ has a twice-linear 2-partition if there exist  $L_1, L_2, L_3, L_4\in \mathbb{Z}_2[x_1, \ldots, x_n]$ where 
$d(L_1)=d(L_2)=d(L_3)=d(L_4)\leqslant 1$ such that $F=L_1\cdot L_2+L_3\cdot L_4$.
\end{definition}

Lemma~\ref{lem-tripartite-partition} gives the following result about multilinear $\mathbb{Z}_2$-polynomials.

\begin{lemma}\label{lem-poly-partition}
Let $F(x_1, \ldots, x_n)\in\mathbb{Z}_2[x_1, \ldots, x_n]$ be a complete quadratic polynomial.
For $n\geqslant 6$, $F$  does not have a twice-linear 2-partition.
For $n=5$, $F$ has a twice-linear 2-partition $F=L_1\cdot L_2+L_3\cdot L_4$
iff (after renaming variables)
the cross terms of $L_1\cdot L_2$ and $L_3\cdot L_4$ correspond to
the unique tripartite 2-partition of $K_5$, and
we have $L_1 \cdot L_2 = (x_1+x_2+a)(x_2+x_3+b)$ and $L_3\cdot L_4 = (x_1+x_2+x_3+x_4+c)(x_1+x_2+x_3+x_5+d)$ 
for some $a, b, c, d\in \mathbb{Z}_2$.
\end{lemma}
\begin{proof}
We first analyze the quadratic terms that appear in a product of two linear polynomials.
We use $x_i\in L$ to denote that a linear term $x_i$ appears in a linear polynomial $L$.
Let $L_1$ and $L_2$ be two linear polynomials.

Let $U_1=\{x_i\mid x_i \in L_1, x_i \notin L_2\}$, 
$U_2=\{x_i\mid x_i \in L_1, x_i \in L_2\}$, and 
$U_3=\{x_i\mid x_i \notin L_1, x_i \in L_2\}.$
Then, $$L_1=\sum_{x_i\in U_1}x_i+\sum_{x_j\in U_2}x_j+a, \text{~~and~~} 
L_2=\sum_{x_j\in U_2}x_j+\sum_{x_k\in U_3}x_k+b$$ for some $a, b\in \mathbb{Z}_2^2$.
The quadratic terms in $L_1\cdot L_2$ are
from
$$(\sum_{x_i\in U_1}x_i+\sum_{x_j\in U_2}x_j)\cdot(\sum_{x_j\in U_2}x_j+\sum_{x_k\in U_3}x_k)$$
which are enumerated in
$$\sum_{x_i\in U_1, x_j\in U_2}x_ix_j+\sum_{x_i\in U_1, x_k\in U_3}x_ix_k+\sum_{x_j\in U_2, x_k\in U_3}x_jx_k.$$
Note that each term
$x_i^2$ for $i \in U_2$
is replaced by $x_i$ (thus no longer counted as a quadratic term) as
we calculate  modulo the ideal
generated by $\{x_i^2 - x_i \mid i \in [n]\}$,
and every  pairwise cross product term
$x_i x_j$ for $i, j \in U_2$ and $i \not = j$  disappears since it appears exactly twice.

If we view variables $x_1, \ldots, x_{n}$ as $n$ vertices and each quadratic term $x_ix_j$ as an edge between vertices $x_i$ and $x_j$,
then the quadratic terms in $L_1\cdot L_2$ are the edges of a complete tripartite subgraph $T$ of $K_n$ (the parts of a tripartite graph could be empty)  and $V(T)=U_1\sqcup U_2 \sqcup U_3$.
Therefore, 
$L_1\cdot L_2$ is one  of the two  terms of
a twice-linear 2-partition of a complete quadratic polynomial over $n$ variables 
iff $T$ is one tripartite  complete graph in
a tripartite 2-partition of the complete graph $K_n$.
By Lemma~\ref{lem-tripartite-partition}, a tripartite 2-partition does not exist for $K_n$ when $n\geqslant 6$.
Thus, $F$ does not have twice-linear partition when $n\geqslant 6$.
When $n=5$, the tripartite 2-partition of $K_5$ is unique up to relabeling.  
One tripartite complete
graph of  this tripartite 2-partition is a triangle, and we may assume it is on 
$\{x_1, x_2, x_3\}$.
Then, we  take $L_1\cdot L_2=(x_1+x_2+a)(x_2+x_3+b)$ for some $a, b\in \mathbb{Z}_2^2$,
and $L_3\cdot L_4=(x_1+x_2+x_3+x_4+c)(x_1+x_2+x_3+x_5+d)$ for some $c, d\in \mathbb{Z}_2^2$.
Thus,
a complete quadratic polynomial $F(x_1, \ldots, x_5)$ over 5 variables has  a twice-linear 2-partition iff (after renaming variables) 
$F=L_1\cdot L_2+L_3\cdot L_4$.
\end{proof}{}

Now, we are ready to make a further  major step towards Theorem~\ref{thm-holantb}. We first give a preliminary result.

\begin{lemma}\label{lem-disjoint-support}
Let $f$ be a $2n$-ary signature, where $2n \geqslant 4$. 
If $f\in \int_{\mathcal{B}}\mathscr{A}$ and $|f(\alpha)|= 1$  for all $\alpha\in \mathscr{S}(f)$, 
then for all $\{i, j\}\subseteq [2n]$,  $\mathscr{S}(f_{ij}^{00})=\mathscr{S}(f_{ij}^{11})$ or $\mathscr{S}(f_{ij}^{00})\cap \mathscr{S}(f_{ij}^{11})=\emptyset$, and $\mathscr{S}(f_{ij}^{01})=\mathscr{S}(f_{ij}^{10})$ or $\mathscr{S}(f_{ij}^{01})\cap \mathscr{S}(f_{ij}^{10})=\emptyset$.
\end{lemma}

\begin{proof}
We first prove that for all $\{i, j\}\subseteq [2n]$, $\mathscr{S}(f_{ij}^{00})=\mathscr{S}(f_{ij}^{11})$ or $\mathscr{S}(f_{ij}^{00})\cap \mathscr{S}(f_{ij}^{11})=\emptyset$.
For a contradiction, suppose that there exist $\alpha, \beta\in\mathbb{Z}_2^{2n-2}$ such that $\alpha\in \mathscr{S}(f_{ij}^{00})\cap \mathscr{S}(f_{ij}^{11})$ 
and $\beta \in \mathscr{S}(f_{ij}^{00})\Delta \mathscr{S}(f_{ij}^{11})$, where $\Delta$ denotes the symmetric difference between two sets. 
Consider signatures $\partial^+_{ij}f$ and $\partial^-_{ij}f$.
Then,
$f_{ij}^{00}(\alpha)+f_{ij}^{11}(\alpha)$ and $f_{ij}^{00}(\beta)+f_{ij}^{11}(\beta)$ are entries of $\partial^+_{ij}f$, and $f_{ij}^{00}(\alpha)-f_{ij}^{11}(\alpha)$ and $f_{ij}^{00}(\beta)-f_{ij}^{11}(\beta)$ are entries of $\partial^-_{ij}f$.
Since $\alpha \in \mathscr{S}(f_{ij}^{00})\cap \mathscr{S}(f_{ij}^{11})$, $f_{ij}^{00}(\alpha)=\pm 1$ and $f_{ij}^{11}(\alpha)=\pm 1$. 
Then between $f_{ij}^{00}(\alpha)+f_{ij}^{11}(\alpha)$ and $f_{ij}^{00}(\alpha)-f_{ij}^{11}(\alpha)$, exactly one has norm $2$ and the other is $0$. 
However, since $\beta \in \mathscr{S}(f_{ij}^{00})\Delta \mathscr{S}(f_{ij}^{11})$, between $f_{ij}^{00}(\beta)$  and $f_{ij}^{11}(\beta)$, exactly one is $0$ and the other has norm $1$.
Thus, 
$|f_{ij}^{00}(\beta)+f_{ij}^{11}(\beta)|=|f_{ij}^{00}(\beta)-f_{ij}^{11}(\beta)|=1.$
Then, between $\partial^+_{ij}f$ and $\partial^-_{ij}f$, there is a signature that has an entry of norm $1$ and an entry of norm $2$. 
Clearly, such a signature is not in $\mathscr{A}$. 
However, since $f\in \int_\mathcal{B}\mathscr{A}$, we have
$\partial^+_{ij}f$, $\partial^-_{ij}f\in \mathscr{A}$.
Contradiction.

By considering signatures $\partial^{\widehat+}_{ij}f$ and  $\partial^{\widehat-}_{ij}f$, similarly we can show that $\mathscr{S}(f_{ij}^{01})=\mathscr{S}(f_{ij}^{10})$ or $\mathscr{S}(f_{ij}^{01})\cap \mathscr{S}(f_{ij}^{10})=\emptyset$.
\end{proof}

The next lemma is a major step.
\begin{lemma}\label{lem-affine-support}
Suppose that $\mathcal{F}$ is non-$\mathcal{B}$ hard.
Let $f\in \mathcal{F}$ be an irreducible $2n$-ary $(2n\geqslant 8)$ signature with parity. 
Then,
\begin{itemize}
    \item  $\Holantb(\mathcal{F})$ is \#P-hard, or
    \item there is  a signature $g\notin \mathscr{A}$ of arity $2k<2n$  that is realizable from $f$ and $\mathcal{B}$, or
    \item $f$ has affine support.
\end{itemize}
\end{lemma}
\begin{proof}
Again, we may assume that $f$ satisfies {\sc 2nd-Orth} and $f\in \int_{\mathcal{B}}\mathscr{A}.$
Also, by Lemma \ref{lem-affine-norm}, we may assume that  $f(\alpha)=\pm 1$ for all $\alpha \in \mathscr{S}(f)$ after normalization.


For any four distinct binary strings $\alpha, \beta, \gamma, \delta \in \mathbb{Z}_2^{2n}$ with $\alpha\oplus\beta\oplus\gamma=\delta$, we define a \emph{score} $T(\alpha, \beta, \gamma, \delta)=(t_1, t_2, t_3)$ which are the values of ${\rm wt}(\alpha\oplus\beta)={\rm wt}(\gamma\oplus\delta),{\rm wt}(\alpha\oplus\gamma)={\rm wt}(\beta\oplus\delta)$ and ${\rm wt}(\alpha\oplus\delta)={\rm wt}(\beta\oplus\gamma)$ ordered  from the smallest to the largest.
We order the scores lexicographically, i.e.,
we say $T=(t_1, t_2, t_3) < T'=(t'_1, t'_2, t'_3)$ if $t_1<t'_1$, or $t_2<t'_2$ when $t_1=t'_1$, or $t_3<t'_3$ when $t_1=t'_1$ and $t_2=t'_2$.
Note that since $\alpha, \beta, \gamma, \delta$ are distinct, the smallest value of the  score $T$ is $(2, 2, 2)$.
We say that $(\alpha, \beta, \gamma, \delta)$ where $\alpha\oplus\beta\oplus\gamma=\delta$ forms a 
\emph{non-affine quadrilateral}
of $f$ if exactly three of them are in $\mathscr{S}(f)$ and the fourth is not.


For a contradiction, suppose that $\mathscr{S}(f)$ is not affine.
Then, $f$ has at least a non-affine quadrilateral.
Among all non-affine quadrilaterals of $f$, 
we pick the one $(\alpha, \beta, \gamma, \delta)$ with the minimum score $T_{\min}=T(\alpha, \beta, \gamma, \delta)=(t_1, t_2, t_3).$
Without loss of generality, we may assume that
among $\alpha, \beta, \gamma$ and $\delta$,
$\delta$ is the one that is not in $\mathscr{S}(f)$.

We first consider the case that $(2, 2, 2)<T_{\min}$.
We  prove that we can realize a non-affine signature from $f$ by merging. 
Depending on the values of $T_{\min}$, there are three cases.



\begin{itemize}
    \item 

$t_1\geqslant 4$.
 Without loss of generality, we may assume that $t_1={\rm wt}(\alpha\oplus\beta)$.
 (Note that even though we have named $\delta$ as the one not belonging to  $\mathscr{S}(f)$, since
 $\alpha\oplus\beta\oplus\gamma\oplus \delta =0$, we can  name them so that  $t_1={\rm wt}(\alpha\oplus\beta)$.)
 Then, there are at least four bits on which $\alpha$ and $\beta$ differ. 
Among these four bits, there are at least two bits on which $\gamma$ is identical to $\alpha$ or $\beta$.
Without loss of generality, we assume that these are the first two bits 
and $\gamma_1\gamma_2=\alpha_1\alpha_2$.
We have $\beta_1\beta_2=\overline{\alpha_1}\overline{\alpha_2}$, and
as  $\delta = \alpha\oplus\beta\oplus\gamma$, we have
$\delta_1\delta_2=\overline{\alpha_1}\overline{\alpha_2}$.
 Also by flipping variables, we may assume that $\alpha=\vec{0}^{2n}=00\vec{0}^{2n-2}$.
 Then, $\beta=11\beta^\ast$, $\gamma=00\gamma^\ast$ and $\delta=11\delta^\ast$ where $\beta^\ast, \gamma^\ast, \delta^\ast \in \mathbb{Z}_2^{2n-2}$ and $\delta^\ast=\beta^\ast\oplus \gamma^\ast$.
 We consider the following eight inputs of $f$.
 $$\begin{matrix}
 \alpha= 00\alpha^\ast & \alpha'=11\alpha^\ast & \beta'=00\beta^\ast & \beta=11\beta^\ast\\
 \gamma=00\gamma^\ast & \gamma'=11\gamma^\ast & \delta'=00\delta^\ast & \delta=11\delta^\ast
 \end{matrix}$$
 
 Note that $\gamma'=\alpha\oplus\alpha'\oplus\gamma$, and ${\rm wt}(\alpha\oplus\alpha')=2<t_1$. 
 Then, $$T(\alpha, \alpha', \gamma, \gamma')<T(\alpha, \beta, \gamma, \delta).$$
 By our assumption that  $T(\alpha, \beta, \gamma, \delta)$ is the minimum score among all non-affine quadrilaterals of $f$,  
  $(\alpha, \alpha', \gamma, \gamma')$ is not a non-affine quadrilateral of $f$. 
 Since $\alpha, \gamma\in \mathscr{S}(f)$, $\alpha'$ and $\gamma'$ are either both in $\mathscr{S}(f)$ or both not in $\mathscr{S}(f)$. 
 Also, note that $\gamma'=\alpha'\oplus\beta\oplus\delta$, and  ${\rm wt}(\alpha'\oplus\beta)={\rm wt}(\alpha\oplus\beta)-2=t_1-2<t_1$.
 Then, $$T(\alpha', \beta, \gamma', \delta)<T(\alpha, \beta, \gamma, \delta).$$
 Again since $T(\alpha, \beta, \gamma, \delta)$ is the minimum  score among all non-affine quadrilaterals of $f$,
 $(\alpha', \beta, \gamma', \delta)$ is not a non-affine quadrilateral. 
 Since $\beta \in \mathscr{S}(f)$ and $\delta\notin \mathscr{S}(f)$, $\alpha'$ and $\gamma'$ are not both in $\mathscr{S}(f)$. 
 Thus, $\alpha', \gamma' \notin \mathscr{S}(f)$.
 Similarly, 
 $(\beta', \beta, \delta', \delta)$ and $(\alpha, \beta', \gamma, \delta')$ 
are not non-affine quadrilaterals of $f$,
since their scores are  less than $T(\alpha, \beta, \gamma, \delta)$.
 Since  $\beta \in \mathscr{S}(f)$ and $\delta\notin \mathscr{S}(f)$,
 we cannot have both  $\beta', \delta' \in \mathscr{S}(f)$
 from considering $(\beta', \beta, \delta', \delta)$, and then
 from $(\alpha, \beta', \gamma, \delta')$, we cannot have exactly one of   $\beta', \delta'$ is in $\mathscr{S}(f)$.
 Thus, both $\beta', \delta' \notin \mathscr{S}(f)$.
 In other words, we have $f(\alpha')=f(\beta')=f(\gamma')=f(\delta')=0$.
 
 Consider the signature $\partial_{12}f$. 
 Then, $f(\alpha)+f({\alpha'})$, $f(\beta)+f({\beta'})$, $f(\gamma)+f({\gamma'})$ and $f(\delta)+f({\delta'})$ 
 are entries of $\partial_{12}f$ on inputs $\alpha^\ast$, $\beta^\ast$, $\gamma^\ast$ and $\delta^\ast$ respectively.
 Since $f(\alpha)+f({\alpha'})=f({\alpha})\neq 0$, $f(\beta)+f({\beta'})=f({\beta})\neq 0$ and $f(\gamma)+f({\gamma'})=f({\gamma})\neq 0$, $\alpha^\ast, \beta^\ast, \gamma^\ast\in \mathscr{S}(\partial_{12}f).$
 Meanwhile we have $f(\delta)+f(\delta')=0+0=0$, thus $\delta^\ast \notin \mathscr{S}(\partial_{12}f)$.
 Thus, $(\alpha^\ast, \beta^\ast, \gamma^\ast, \delta^\ast)$ is a non-affine quadrilateral of $\partial_{12}f$.
 Then, $\partial_{12}f$ is a non-affine signature of arity $2n-2$. Contradiction.

\item $t_1=2$ and $t_2 \geqslant 4$.
 Without loss of generality, we assume that ${\rm wt}(\alpha\oplus\gamma)=2$ and ${\rm wt}(\alpha\oplus\beta)=t_2\geqslant 4$.
 (Again, using $\alpha\oplus\beta\oplus\gamma\oplus \delta =0$,  a moment reflection shows that this is indeed without loss of generality, even though we have named $\delta \not \in  \mathscr{S}(f)$.)
 Again by flipping variables, we may assume that $\alpha=\vec{0}^{2n}.$
 Then, ${\rm wt}(\gamma)=2$ and ${\rm wt}(\beta)\geqslant 4$.
 Take four bits where $\beta_i =1$, for at most two of these we can have $\gamma_i =1$,
 thus 
 there exist two other bits of these four bits (we may assume that they are the first two bits) 
 such that $\gamma_1\gamma_2=00$  and $\beta_1\beta_2=11$. 
 Then,  $\alpha=00\alpha^\ast$, $\beta=11\beta^\ast$, $\gamma=00\gamma^\ast$,  and   $\delta=11\delta^\ast$ by  $\delta = \alpha\oplus\beta\oplus\gamma$, where $\beta^\ast, \gamma^\ast, \delta^\ast \in \mathbb{Z}_2^{2n-2}$, ${\rm wt}(\beta^\ast)\geqslant 2$, ${\rm wt}(\gamma^\ast)=2$ and $\delta^\ast=\beta^\ast\oplus\gamma^\ast.$
 Still, we consider the following eight inputs of $f$.
  $$\begin{matrix}
 \alpha= 00\alpha^\ast & \alpha'=11\alpha^\ast & \beta'=00\beta^\ast & \beta=11\beta^\ast\\
 \gamma=00\gamma^\ast & \gamma'=11\gamma^\ast & \delta'=00\delta^\ast & \delta=11\delta^\ast
 \end{matrix}$$
 Note that ${\rm wt}(\alpha\oplus\gamma)=2$ and ${\rm wt}(\alpha\oplus\alpha')=2<t_2$.
 Then,  $$T(\alpha, \alpha', \gamma, \gamma')<T(\alpha, \beta, \gamma, \delta).$$
 Then similarly since $T(\alpha, \beta, \gamma, \delta)$ is the minimum, $(\alpha, \alpha', \gamma, \gamma')$ is not a non-affine quadrilateral.
 Since $\alpha, \gamma \in \mathscr{S}(f)$, $\alpha'$ and $\gamma'$ are either both in $\mathscr{S}(f)$ or both not in it. 
 Also, note that ${\rm wt}(\alpha'\oplus\gamma')=2$  and  ${\rm wt}(\alpha'\oplus\beta)={\rm wt}(\alpha\oplus\beta)-2=t_2-2<t_2$.
 Then,
 $$T(\alpha', \beta, \gamma', \delta)<T(\alpha, \beta, \gamma, \delta).$$
 Thus, $(\alpha', \beta, \gamma', \delta)$ is not a non-affine quadrilateral. 
Since $\beta \in \mathscr{S}(f)$ and $\delta\notin \mathscr{S}(f)$, $\alpha'$ and $\gamma'$ are not both in $\mathscr{S}(f)$. 
 Thus, $\alpha', \gamma' \notin \mathscr{S}(f)$.
 Similarly, by considering $(\beta', \beta, \delta', \delta)$ and $(\alpha, \beta', \gamma, \delta')$, 
 we know that they are not non-affine quadrilaterals. 
 Thus, $\beta', \delta' \notin \mathscr{S}(f)$.
 In other words, we have $f(\alpha')=f(\beta')=f(\gamma')=f(\delta')=0$.
 Still consider the signature $\partial_{12}f$. We have $\partial_{12}f \notin \mathscr{A}$. 
 Contradiction.
 \item $t_1=2$, $t_2=2$ and $t_3=4$. 
 In this case,
 by the definition of distance-2 squares (equation (\ref{eqn:distance-2 square})),
 $\left[\begin{smallmatrix}
 f(\alpha) &f(\beta)\\
  f(\gamma) & f(\delta)\\
 \end{smallmatrix} \right]$ forms a distance-2 square.
 Clearly, it is not of type \Rmnum{1}, \Rmnum{2} or \Rmnum{3} since exactly one entry of this square is zero.
 As proved in Lemma~\ref{lem-affine-norm}, since $f$ has a distance-2 square that is not type \Rmnum{1}, \Rmnum{2} or \Rmnum{3}, then we can realize a non-affine signature by merging. 
 Contradiction.
 \end{itemize}
 
 Now, we consider the case that $T_{\min} = (2, 2, 2)$.

Then, we show that $|\mathscr{S}(f)|=2^{2n-2}$.
We consider the non-affine quadrilateral $(\alpha, \beta, \gamma, \delta)$ with score $T=(2, 2, 2)$.
By renaming and flipping variables, without loss of generality, we may assume that 
$$\begin{bmatrix}
\alpha & \beta\\
\gamma & \delta\\
\end{bmatrix}=
\begin{bmatrix}
000\vec{0}^{2n-3} & 011\vec{0}^{2n-3}\\
110\vec{0}^{2n-3} & 101\vec{0}^{2n-3}\\
\end{bmatrix},$$
and $\delta$ is the only one among four that is not in $\mathscr{S}(f)$.
By normalization, we may assume that $f(\alpha)=1$.
If $f(\gamma)=-1$, then we negate the variable $x_1$ of $f$.
This keeps $f_1^0$ unchanged but changes $f_1^1$ to $-f_1^1$, so this does not change the value of $f(\alpha)$, but changes the value of $f(\gamma)$ to $1$.
Thus, without loss of generality, we may assume that $f(\alpha)=f(\gamma)=1$.
Clearly, $f$ has even parity.
Consider the signature $\widetilde f$
by the $H_4$ gadget applied on variables $x_1$ and $x_2$ of $f$.
We have $\widetilde f_{12}^{00}(\vec{0}^{2n-2})=f(\alpha)+f(\gamma)=2$ and $\widetilde f_{12}^{01}(1\vec{0}^{2n-3})=f(\beta)+f(\delta)=f(\beta)$ since $f(\delta)=0$.
Remember that since $f\in \int_\mathcal{B}\mathscr{A}$, by Lemma~\ref{lem-tilde-1}, for all $(a, b)\in \mathbb{Z}^2_2$, $\widetilde f_{12}^{ab}\in \mathscr{A}$ and we use $n_{ab}$ to denote the norm of its nonzero entries. 
Thus, $n_{00}=2$ and $n_{01}=1$.
Also, we have $f(\beta) \not =0$ which is the same as  $1\vec{0}^{2n-3}\in \mathscr{S}(f_{12}^{01})$, and
$f(\delta) =0$ which is the same as $1\vec{0}^{2n-3}\notin \mathscr{S}(f_{12}^{10})$.
By Lemma~\ref{lem-disjoint-support},
$\mathscr{S}(f_{12}^{01})\cap\mathscr{S}(f_{12}^{10})=\emptyset$.
Remember that $\widetilde f_{12}^{01}=f_{12}^{01}+f_{12}^{10}$ and $\widetilde f_{12}^{10}=f_{12}^{01}-f_{12}^{10}$.
Then,
$$\mathscr{S}(\widetilde f_{12}^{01})=\mathscr{S}(f_{12}^{01})\cup\mathscr{S}(f_{12}^{10})=\mathscr{S}(\widetilde f_{12}^{10}).$$

Consider signatures $\partial^b_{ij}f$ for all $\{i, j\}$ disjoint with $\{1, 2\}$ and every $b\in \mathcal{B}$.
By Lemma~\ref{lem-4.5} and its remark, we may assume that either $M(\mathfrak{m}_{12}(\partial^b_{ij}f))=\lambda^b_{ij}I_4$ for some real $\lambda^b_{ij}\neq 0$, or there exists a nonzero binary signature $g^b_{ij}\in \mathcal{O}$ such that $g^b_{ij}(x_1, x_2)\mid \partial^b_{ij}f$. Otherwise, we get  \#P-hardness.

Consider the case that  $g^b_{ij}(x_1, x_2)\mid \partial^b_{ij}f$.
If $\partial^b_{ij}f\equiv0$, then we can let $g^b_{ij}\in \mathcal{B}$ since a zero signature can be divided by any nonzero binary signature.
If $\partial^b_{ij}f \not\equiv 0$, we can realize $g^b_{ij}$ by factorization. If $g^b_{ij}\notin \mathcal{B}^{\otimes 1}$, then we get  \#P-hardness since $\mathcal{F}$ is non-$\mathcal{B}$ hard.
Thus, we may assume that $g^b_{ij}\in \mathcal{B}$ after normalization. 
Therefore, for all $\{i, j\}$ disjoint with $\{1, 2\}$ and every $b\in \mathcal{B}$, we may assume that either $M(\mathfrak{m}_{12}(\partial^b_{ij}f))=\lambda^b_{ij}I_4$ for some real $\lambda^b_{ij}\neq 0$, or there exists a nonzero binary signature $g^b_{ij}\in \mathcal{B}$ such that $g^b_{ij}(x_1, x_2)\mid \partial^b_{ij}f$.
Then, by Lemma~\ref{lem-tilde-norm-2},  $\mathscr{S}(\widetilde f^{01}_{12})=\mathscr{O}_{2n-2}$.
Thus, $|\mathscr{S}(\widetilde f^{01}_{12})|=2^{2n-3}$.

Now consider again  the signature $f$.
Since $f$ satisfies {\sc 2nd-Orth}, and all its nonzero entries have norm $1$, for any $(a, b)\in \mathbb{Z}_2^2$,
$|{\bf f}^{ab}_{12}|^2=|\mathscr{S}(f^{ab}_{12})|$.
Then, $$|\mathscr{S}(f^{00}_{12})|=|\mathscr{S}({f^{01}_{12}})|=|\mathscr{S}(f^{10}_{12})|=|\mathscr{S}(f^{11}_{12})|.$$
Remember that  $\mathscr{S}(f^{01}_{12})\cap\mathscr{S}(f^{10}_{12})=\emptyset$, and $\mathscr{S}(\widetilde f^{01}_{12})=\mathscr{S}(f^{01}_{12})\cup\mathscr{S}(f^{10}_{12})$.
Then, $\mathscr{S}(f^{00}_{12})$ and $\mathscr{S}({f^{01}_{12}})$ form an equal size partition of $\mathscr{S}(\widetilde f^{01}_{12})$.
Thus, $|\mathscr{S}(f^{01}_{12})|=|\mathscr{S}(f^{10}_{12})|=\frac{1}{2}|\mathscr{S}(\widetilde f^{01}_{12})|=2^{2n-4}$.
Also, $|\mathscr{S}(f^{00}_{12})|=|\mathscr{S}(f^{11}_{12})|=2^{2n-4}.$
Therefore, 
$$|\mathscr{S}(f)|=|\mathscr{S}(f^{00}_{12})|+|\mathscr{S}({f^{01}_{12}})|+|\mathscr{S}(f^{10}_{12})|+|\mathscr{S}(f^{11}_{12})|
=4\cdot 2^{2n-4}=2^{2n-2}.$$
Since all nonzero entries of $f$ have norm $1$, $|{\bf f}|^2=|\mathscr{S}(f)|=2^{2n-2}.$
Also, since $f$ satisfies {\sc 2nd-Orth}, for all $\{i, j\}\in [2n]$ and all $(a, b)\in \mathbb{Z}_2^2$, $|{\bf f}_{ij}^{ab}|=\frac{1}{4}|{\bf f}|^2=2^{2n-4}$.

We denote $\mathscr{S}(f)$ by $S$.
Since $f$ has even parity, 
for every $(x_1, \ldots, x_{2n})\in S$, $x_1+\cdots+x_{2n}=0$, i.e.,  $S\subseteq \mathscr{E}_{2n}$.
Let $F(x_1, \ldots, x_{2n-1})\in \mathbb{Z}_2[x_1, \ldots, x_{2n-1}]$ be the 
multilinear polynomial such that 
$$F(x_1, \ldots, x_{2n-1})=\left\{
\begin{aligned}
1,&~ & (x_1, \ldots, x_{2n-1}, x_{2n})\in S \\
0,&~ & (x_1, \ldots, x_{2n-1}, x_{2n})\notin S 
\end{aligned} \text{ ~~where~~ } x_{2n}=\sum^{2n-1}_{i=1}x_i.
\right.$$
Then, $S=\{(x_1, \ldots, x_{2n})\in \mathscr{E}_{2n}\mid  
F(x_1, \ldots, x_{2n-1})=1\}.$

Now, we show that for all $\{i, j\}\subseteq[2n-1]$,
$F^{00}_{ij}+F^{11}_{ij} \equiv 0$ or $1$, and also  $F^{01}_{ij}+F^{10}_{ij} \equiv 0$ or $1$.
For simplicity of notations, we prove this for $\{i, j\}=\{1, 2\}$. The proof  for arbitrary $\{i, j\}$ is the same by replacing  $\{1, 2\}$ by $\{i, j\}$.
Consider $$S_0=\mathscr{S}(f_{12}^{00})=\{(x_3, \ldots, x_{2n})\in \mathscr{E}_{2n-2}\mid  F^{00}_{12}(x_3, \ldots,x_{2n-1})=1\},$$
and $$S_1=\mathscr{S}(f_{12}^{11})=\{(x_3, \ldots, x_{2n})\in \mathscr{E}_{2n-2}\mid  F^{11}_{12}(x_3, \ldots,x_{2n-1})=1\}.$$
Then, $$S_0\cap S_1=\{(x_3, \ldots, x_{2n})\in \mathscr{E}_{2n-2}\mid  F^{00}_{12}\cdot F^{11}_{12}=1\},$$
and $$S_0\cup S_1=\{(x_3, \ldots, x_{2n})\in \mathscr{E}_{2n-2}\mid  F^{00}_{12}+F^{11}_{12}+ F^{00}_{12}\cdot F^{11}_{12}=1\}.$$
By Lemma~\ref{lem-disjoint-support}, $S_0=S_1$ or $S_0\cap
S_1=0$.
\begin{itemize}
    \item 
If $S_0=S_1$, then for  every  $(x_3, \ldots, x_{2n-1}) \in \mathbb{Z}^{2n-3}_{2}$ which decides every $(x_3, \ldots, x_{2n}) \in \mathscr{E}_{2n-2}$ by $x_{2n}=x_3+\cdots+x_{2n-1}$,
$$F_{12}^{00}(x_3, \ldots, x_{2n-1})= F_{12}^{11}(x_3, \ldots, x_{2n-1}).$$
Thus, $F_{12}^{00}+F_{12}^{11}\equiv 0$.
\item 
If $S_0\cap S_1=\emptyset$, then
since $|S_0|=|S_1|=2^{2n-4}$ (which is still true when replacing $\{1, 2\}$ by an arbitrary $\{i, j\}$), $|S_0\cup S_1|=|S_0|+|S_1|=2^{2n-3}$.
Since $S_0\cup S_1\subseteq \mathscr{E}_{2n-2}$ and $|\mathscr{E}_{2n-2}|=2^{2n-3}$, 
$S_0\cup S_1=\mathscr{E}_{2n-2}$.
Thus, for every $(x_3, \ldots, x_{2n-1})\in \mathbb{Z}_{2}^{2n-3}$ which decides every $(x_3, \ldots, x_{2n}) \in \mathscr{E}_{2n-2}$ by $x_{2n}=x_3+\cdots+x_{2n-1}$,
$$ F^{00}_{12}(x_3, \ldots, x_{2n-1})\cdot F^{11}_{12}(x_3, \ldots, x_{2n-1})=0,$$ and $$F^{00}_{12}(x_3, \ldots, x_{2n-1})+ F^{11}_{12}(x_3, \ldots, x_{2n-1})+F^{00}_{12}\cdot F^{11}_{12}(x_3, \ldots, x_{2n-1})=1.$$
Thus, $F^{00}_{12}+F^{11}_{12}\equiv 1$.
\end{itemize}
Similarly, we can show that $F^{01}_{12}+F^{10}_{12} \equiv  0$ or $1$. 
Therefore, for all $\{i, j\}\subseteq[2n-1]$,
$F^{00}_{ij}+F^{11}_{ij} \equiv 0$ or $1$ and $F^{01}_{ij}+F^{10}_{ij} \equiv 0$ or $1$.
By Lemma~\ref{lem-mutilinear-poly}, $d(F)\leqslant 2$. 

If $d(F)\leqslant 1$, then clearly, $S=\{(x_1, \ldots, x_{2n})\in \mathscr{E}_{2n}\mid  
F(x_1, \ldots, x_{2n-1})=1\}$  is an affine linear space.
Thus, $f$ has affine support.
Otherwise, $d(F)= 2$.
By Lemma~\ref{lem-mutilinear-poly}, $F$ is a complete quadratic polynomial.
Consider signatures $f_{12}^{00}$ and $f_{12}^{11}$.
Remember that $f(000\vec{0}^{2n-3})=f(110\vec{0}^{2n-3})=1$.
Thus, $\vec{0}^{2n-2}\in S_0\cap S_1\neq \emptyset.$
Then, $S_0=S_1$.
Let $$S_{+}=\{\alpha\in S_0 \mid f_{12}^{00}(\alpha)=f_{12}^{11}(\alpha)\} \text{~~and~~} S_{-}=\{\alpha\in S_0 \mid f_{12}^{00}(\alpha)=-f_{12}^{11}(\alpha)\}.$$
Then, as $f$ takes $\pm 1$ values
on its support, $S_{+}=\mathscr{S}(\partial^+_{12}f)$ and $S_{-}=\mathscr{S}(\partial^-_{12}f).$
Since $\partial^+_{12}f, \partial^-_{12}f\in \mathscr{A}$, $S_{+}$ and $S_-$ are affine linear subspaces of $\mathscr{E}_{2n-2}$.
Also, by {\sc 2nd-Orth}, $\langle{\bf f}_{12}^{00}, {\bf f}_{12}^{11}\rangle=|S_+|-|S_-|=0$.
Thus, $|S_+|=|S_-|=\frac{1}{2}|S_0|=2^{2n-5}$.
Since $|\mathscr{E}_{2n-2}|=2^{2n-3}$, $S_+$ is a an affine linear subspaces of $\mathscr{E}_{2n-2}$ decided by two affine linear constraints $L^+_1=1$ and $L^+_2=1$.
(Here both $L^+_1$ and $L^+_2$ are \emph{affine}
linear forms which may 
have nonzero constant terms, but we write the constraints as $L^+_1=1$ and $L^+_2=1$.)
In other words, $$S_+=\{(x_3, \ldots, x_{2n})\in \mathscr{E}_{2n-2}\mid L^+_1=1 \text{~and~} L^+_2=1\}=\{(x_3, \ldots, x_{2n})\in \mathscr{E}_{2n-2}\mid L^+_1\cdot L^+_2=1\}.$$
Since for every $(x_3, \ldots, x_{2n})\in \mathscr{E}_{2n-2}$, $x_3+\cdots+x_{2n}=0$, we may substitute the appearance of $x_{2n}$ in $L^+_1$ and $L^+_2$ by $x_3+\cdot+x_{2n-1}$.
Thus, we may assume that $L^+_1, L^+_2\in \mathbb{Z}_2[x_3, \ldots, x_{2n-1}]$, and $d(L^+_1)=d(L^+_2)=1$.
Similarly, there exist $L^-_1, L^-_2\in \mathbb{Z}_2[x_3, \ldots, x_{2n-1}]$ with $d(L^-_1)=d(L^-_2)=1$ such that 
$$S_-=\{(x_3, \ldots, x_{2n})\in \mathscr{E}_{2n-2}\mid L^-_1=1 \text{~and~} L^-_2=1\}=\{(x_3, \ldots, x_{2n})\in \mathscr{E}_{2n-2}\mid L^-_1\cdot L^-_2=1\}.$$
Clearly, $S_+\cap S_-=\emptyset$. Then
$$S_+\cup S_-=\{(x_3, \ldots, x_{2n})\in \mathscr{E}_{2n-2}\mid L^+_1\cdot L^+_2 + L^-_1\cdot L^-_2=1\}.$$
Remember that 
$$S_{0}=S_+\cup S_-=\{(x_3, \ldots, x_{2n})\in \mathscr{E}_{2n-2}\mid F^{00}_{12}=1\}.$$
Thus, $L^+_1\cdot L^+_2 + L^-_1\cdot L^-_2=F^{00}_{12}$.
Since for all $1\leqslant i<j \leqslant2n-1$, the quadratic term $x_ix_j$ appears in $F$,  for all $3\leqslant i<j \leqslant2n-1$, the quadratic term $x_ix_j$ appears in $F^{00}_{12}.$
Thus, $F^{00}_{12}\in \mathbb{Z}_2[x_3, \ldots, x_{2n-1}]$ is a complete quadratic polynomial over $2n-3$ variables and it has a twice-linear 2-partition.
Since $2n\geqslant 8$, $2n-3\geqslant 5$.
By Lemma~\ref{lem-poly-partition}, we have $2n-3= 5$, 
and after renaming variables, $$F=(x_3+x_4+a)(x_4+x_5+b)+(x_3+x_4+x_5+x_6+c)(x_3+x_4+x_5+x_7+d)$$ where $a, b, c, d\in \mathbb{Z}_2$.
Without loss of generality, we may assume that $L_1^+\cdot L_2^+=(x_3+x_4+a)(x_4+x_5+b).$
Then, $$S_{+}=\mathscr{S}(\partial^+_{12}f)=\{(x_3, \ldots, x_{8})\in \mathscr{E}_{2n-2}\mid x_3= {x_4}+a
\text{ and }  x_4= {x_5}+b\},$$
for some $a, b\in \mathbb{Z}_2$.

Clearly $\partial^+_{12}f$ is a 6-ary signature and $|\mathscr{S}(\partial^+_{12}f)|=2^{5-2}=2^3$.
We show that $\partial^+_{12}f\notin \mathcal{B}^{\otimes 3}\cup\mathcal{F}_6\cup\mathcal{F}^H_6$.
Then, by Corollary~\ref{lem-non-f6}, we get  \#P-hardness.
Since the support of a signature in $\mathcal{F}_6\cup\mathcal{F}^H_6$ is either $\mathscr{E}_6$ or $\mathscr{O}_6$ whose sizes are both $2^5$.
Thus, $\partial^+_{12}f\notin \mathcal{F}_6\cup\mathcal{F}^H_6$.
For any 6-ary signature $g$ in $\mathcal{B}^{\otimes 3}$,
its 6 variables can be divided into three independent pairs such that on the support $\mathscr{S}(g)$, 
the values of variables inside each pair do not rely on the values of variables of other pairs.
Thus, if we pick any three variables in $\mathscr{S}(g)$, the degree of freedom of them is at least $2$; more precisely, there are at least 4 assignments on these three variables which
can be extended to an input in $\mathscr{S}(g)$.
However, in $\mathscr{S}(\partial^+_{12}f)$, the degree of freedom of variables $x_3, x_4, x_5$ is only $1$, namely
there are only two assignments on $x_3, x_4, x_5$ that can be extended to an input in $\mathscr{S}(\partial^+_{12}f)$.
Thus, $\partial^+_{12}f\notin \mathcal{B}^{\otimes 3}$.
This completes the proof of Lemma~\ref{lem-affine-support}.
\end{proof}

\subsection{Affine signature condition}
 Finally, by further assuming that $f$ has affine support, we consider whether $f$ itself is an affine signature. 
We prove that this is true only for signature of arity $2n\geqslant 10$.
For signature $f$ of arity $2n=8$, we show that either $f\in \mathscr{A}$ or the following signature is realizable. 
$$h_8=\chi_T \cdot (-1)^{x_1x_2x_3+x_1x_2x_5+x_1x_3x_5+x_2x_3x_5}, \text{ where } T = \mathscr{S}(h_8)=\mathscr{S}(f_8).$$
Note that in the support $\mathscr{S}(f_8)$
(see its definition (\ref{equ-T-polynomial}) for this
\emph{Queen of the Night} $f_8$), 
by taking  $x_1, x_2, x_3, x_5$ as free variables,  the remaining 4 variables are mod 2 sums of ${4 \choose 3}$ subsets of $\{x_1, x_2, x_3, x_5\}$.
Clearly, $h_8$ is not affine, but it has affine support and all its nonzero entries have the same norm.
One can check that $h_8$ satisfies {\sc 2nd-Orth} and $h_8\in \int_{\mathcal{B}}\mathscr{A}.$
But fortunately, we show that by merging $h_8$, we can realize a 6-ary signature that is not in 
$\mathcal{B}^{\otimes}\cup\mathcal{F}_6\cup\mathcal{F}_6^H$.
By Corollary \ref{lem-non-f6}, we are done.

After we give one more result about multilinear boolean polynomials, we  make our final step towards Theorem~\ref{thm-holantb}. 

\begin{lemma}\label{lem-cubic-poly}
Let $F(x_1, \ldots, x_n)\in \mathbb{Z}_2[x_1, \ldots, x_n]$ be a complete cubic polynomial, 
 $L(x_2,\ldots, x_n)\in \mathbb{Z}_2[x_2, \ldots, x_n]$ and $d(L)\leqslant 1$.
If we substitute  $x_1$ by 
$x_{n+1}+L(x_2,\ldots, x_n)$ 
in $F$ to get $F'$, and suppose 
$F'(x_2, \ldots, x_{n+1})=F(x_{n+1}+L, x_2,\ldots, x_n)\in \mathbb{Z}_2[x_2, \ldots, x_{n+1}]$ 
 is also a complete cubic polynomial, then 
\begin{itemize}
    \item If $n\geqslant 5$, then
    $L$ must be a constant $\epsilon = 0 ~\mbox{or}~ 1$.
    \item If $n=4$, then
    $L$ must be either $ \epsilon$, or
    of the form $x_i+x_j+ \epsilon$, for some $\epsilon = 0 ~\mbox{or}~ 1$, for some $\{i, j\}\in \{2, 3, 4\}.$
\end{itemize}
\end{lemma}{}
\begin{proof}
Since $F(x_1, \ldots, x_n)$ is a complete cubic polynomial, we can write it  as 
$$F(x_1, \ldots, x_n)=x_1\cdot \sum_{2\leqslant i<j \leqslant n}x_ix_j+\sum_{2\leqslant i<j<k \leqslant n}x_ix_jx_k+G(x_1, \ldots, x_n)$$
where $d(G)\leqslant 2$.
Then, 
\begin{equation*}
    \begin{aligned}
F'(x_2, \ldots, x_n, x_{n+1})=(x_{n+1}+L)\cdot \sum_{2\leqslant i<j \leqslant n}x_ix_j+\sum_{2\leqslant i<j<k \leqslant n}x_ix_jx_k+G(x_{n+1}+L, \ldots, x_n).
    \end{aligned}
\end{equation*}
Let $G'(x_2, \ldots, x_n, x_{n+1})=G(x_{n+1}+L, \ldots, x_n).$ Since $d(L)\leqslant 1$ and $d(G)\leqslant 2$, $d(G')\leqslant 2$.
Then, there is no cubic term in $G'(x_2, \ldots, x_n, x_{n+1})$.
Since $F'(x_2, \ldots, x_n, x_{n+1})$ is a complete cubic polynomial over variables $(x_2, \ldots, x_n, x_{n+1})$ and $x_{n+1}\cdot \sum_{2\leqslant i<j \leqslant n}x_ix_j+\sum_{2\leqslant i<j<k \leqslant n}x_ix_jx_k$ already gives every cubic term over $(x_2, \ldots, x_n, x_{n+1})$ exactly once, there is no cubic term in $L\cdot \sum_{2\leqslant i<j \leqslant n}x_ix_j$ (after taking module $2$). 
If $L\equiv 0$ or $1$, then we are done.
Otherwise, there is a variable that appears in $L$.
Without loss of generality, we may assume that $x_2\in L$ (i.e., $x_2$ appears in $L$). 
%

Let $Q(x_3, \ldots, x_{n})=\sum_{3\leqslant i<j \leqslant n}x_ix_j\in \mathbb{Z}_2[x_3, \ldots, x_{n}].$
Since $n\geqslant 4$, we have $Q\not\equiv 0$.
For every $x_ix_j\in Q$, since $x_2\in L$, the cubic term $x_2x_ix_j$ will appear in $L\cdot \sum_{2\leqslant i<j \leqslant n}x_ix_j$. 
To cancel it, exactly one between $x_i \cdot x_2x_j$ and $x_j \cdot x_2x_i$ must also appear in $L\cdot \sum_{2\leqslant i<j \leqslant n}x_ix_j$.
Thus, exactly one between $x_i$ and $x_j$  appears in $L$.

If $n\geqslant 5$, then $x_3x_4, x_4x_5, x_3x_5\in Q$.
Thus, exactly one between $x_3$ and $x_4$ is in $L$,
exactly one between $x_4$ and $x_5$ is in $L$, and 
exactly one between $x_3$ and $x_5$ is in $L$.
Clearly, this is a contradiction.

If $n=4$, then $Q=x_3x_4$. Either $x_3$ or $x_4$ appears in $L$. Thus, $L$ is a sum of two variables among $\{x_2, x_3, x_4\}$ plus a constant $0$ or $1$.
\end{proof}

\begin{lemma}\label{lem-in-affine}
Let $\mathcal{F}$ be non-$\mathcal{B}$ hard.
Let $f\in \mathcal{F}$ be an irreducible $2n$-ary $(2n\geqslant 8)$ signature with parity. 
Then,
\begin{itemize}
    \item  $\Holantb(\mathcal{F})$ is \#P-hard, or
    \item there is  a signature $g\notin \mathscr{A}$ of arity $2k<2n$  that is realizable from $f$ and $\mathcal{B}$, or
    \item $f\in \mathscr{A}$.
\end{itemize}
\end{lemma}
\begin{proof}

Again, we may assume that $f$ satisfies {\sc 2nd-Orth} and $f\in \int_{\mathcal{B}}\mathscr{A}$.
Also by Lemmas \ref{lem-affine-norm} and \ref{lem-affine-support}, we may assume that  $f(\alpha)=\pm 1$ for all $\alpha \in \mathscr{S}(f)$ and $\mathscr{S}(f)$ is an affine linear space. 
Let $\{x_1, \ldots, x_m\}$ be a set of free variables of $\mathscr{S}(f)$.
Then, on the support $\mathscr{S}(f)$,
every variable $x_i$ $(1\leqslant i \leqslant 2n)$ is expressible as
a unique affine linear combination over $\mathbb{Z}_2$ of these free variables, i.e., 
$x_i=L_i(x_1, \ldots, x_m)=\lambda^0_i+\lambda^1_i x_1+\ldots +\lambda^m_i x_m$, where $\lambda^0_i, \ldots, \lambda^{m}_i \in \mathbb{Z}_2$.
Clearly, for $1\leqslant i\leqslant m$, $L(x_i)=x_i$.
Then, 
\begin{equation*}
\begin{aligned}
\mathscr{S}(f)&=\{(x_1, \ldots, x_{2n})\in \mathbb{Z}_2^{2n}\mid x_1=L_1, \ldots, x_{2n}
=L_{2n}\}\\
&=\{(x_1, \ldots, x_{2n})\in \mathbb{Z}_2^{2n}\mid x_{m+1}=L_{m+1}, \ldots, x_{2n}=L_{2n}\}.
\end{aligned}  
\end{equation*}

Also, let $I(x_i)=\{1\leqslant k\leqslant m\mid \lambda^k_i=1\}.$
Clearly, for $1\leqslant i\leqslant m$, $I(x_i)=\{i\}$.
For $m+1\leqslant i\leqslant 2n$, we show that $|I_{x_i}|\geqslant 2$.
For a contradiction, suppose that there exists $m+1\leqslant i\leqslant 2n$ such that $|I_{x_i}|=0$ or $1$.
If $|I_{x_i}|=0$, then $x_i$ takes a constant value in $\mathscr{S}$.
Then, among $f_i^0$ and $f_i^1$, one is a zero signature. 
Thus, $f$ is reducible. Contradiction.
If $|I_{x_i}|=1$, then $x_i=x_k$ or $x_k+1$ for some free variable $x_k$.
Then, among $f_{ik}^{00}$, $f_{ik}^{01}$, $f_{ik}^{10}$ and $f_{ik}^{11}$, two are zero signatures.
Thus, $f$ does not satisfy {\sc 2nd-Orth}. Contradiction.

Since $f(\alpha)=\pm 1$ for all $\alpha \in \mathscr{S}(f)$ and each $\alpha \in \mathscr{S}(f)$ can be uniquely decided by its value on the first $m$ free variables, there exists a unique multilinear boolean polynomial $F(x_1, \ldots, x_m)\in \mathbb{Z}_2[x_1, \ldots, x_m]$ such that $$f(x_1, \ldots, x_m, \ldots, x_{2n})=\chi_{S}(-1)^{F(x_1, \ldots, x_m)}$$ where $S=\mathscr{S}(f)$. 
If $d(F)\leqslant 2$, then clearly $f\in \mathscr{A}$. We are done.
Thus, we may assume that $d(F)>2$ and hence $m>2$.
Remember that $F_{ij}^{ab}$ denotes the polynomial obtained by setting variables $(x_i, x_j)$ of $F$ to $(a, b)\in \mathbb{Z}_2^2$.
Then, $f_{ij}^{ab}=(-1)^{F^{ab}_{ij}}$ on  $\mathscr{S}(f)$.
We will show that for all $i, j\in [m]$, $d(F^{00}_{ij}+F^{11}_{ij})\leqslant 1$ and $d(F^{01}_{ij}+F^{10}_{ij})\leqslant 1$.
For brevity of notation, 
we prove this for $\{i, j\}=\{1, 2\}.$
The proof for arbitrary $\{i, j\}$ is the same by replacing $\{1, 2\}$ with $\{i, j\}$.
We first show that $d(F^{00}_{ij}+F^{11}_{ij})\leqslant 1$.
We use $S_0$ to denote $\mathscr{S}(f_{ij}^{00})$ and $S_1$ to denote $\mathscr{S}(f_{ij}^{11}).$
By Lemma~\ref{lem-disjoint-support}, there are two cases, $S_0=S_1$ or $S_0\cap S_1=\emptyset$.


\begin{itemize}
    \item
Suppose that $S_0=S_1$. For convenience, we use $L^0_{i}$ to denote $(L_{i})_{12}^{00}$ and $L^1_{i}$ to denote $(L_{i})_{12}^{11}$.
Then,
\begin{equation*}
\begin{aligned}
S_0=& \{(x_3, \ldots, x_{2n})\in \mathbb{Z}_2^{2n-2}\mid x_{m+1}=L^{0}_{m+1}, \ldots, x_{2n}=L^0_{2n}\}\\
S_1=& \{(x_3, \ldots, x_{2n})\in \mathbb{Z}_2^{2n-2}\mid x_{m+1}=L^{1}_{m+1}, \ldots, x_{2n}=L^1_{2n}\}.
    \end{aligned}
\end{equation*}
So $L^{0}_{i} \equiv L^{1}_{i}$ for all $i\geqslant m+1$.
Thus, either $\{1, 2\}\subseteq I(x_i)$ or $\{1, 2\}\cap I(x_i)=\emptyset$ for $i\geqslant m+1$.

Let ${S}_{+}=\{\alpha\in S_0\mid f^{00}_{ij}(\alpha)=f^{11}_{ij}(\alpha)\}$ and  ${S}_{-}=\{\alpha\in S_0\mid f^{00}_{ij}(\alpha)=-f^{11}_{ij}(\alpha)\}.$ 
Then, $\langle {\bf f}^{00}_{ij},{\bf f}^{11}_{ij}\rangle=1\cdot |{S}_{+}|-1 \cdot |{S}_{-}|=0$.
Since $S_0=S_+\cup S_-$, $|{S}_{+}|=|{S}_{-}|=\frac{1}{2}|S_0|$.
Note that $\mathscr{S}(\partial_{12}f)={S}_{+}$ and $\mathscr{S}(\partial^{-}_{12}f)={S}_{-}$.
By our assumption that $f\in \int_{\mathcal{B}}\mathscr{A}$, $\partial_{12}f, \partial^-_{12}f \in \mathscr{A}$.
Thus,
both ${S}_{+}$ and ${S}_{-}$ are affine linear subspaces of $S_0=S_1$.
Since $|{S}_{+}|=|{S}_{-}|=|S_0|/2$, there exists an (affine) linear polynomial $L(x_3, \ldots, x_{2n})$ such that 
$${S}_{+}=\{(x_3, \ldots, x_{2n})\in {S_0} \mid L(x_3, \ldots, x_{2n})=0\},$$
and $${S}_{-}=\{(x_3, \ldots, x_{2n})\in {S_0}  \mid L(x_3, \ldots, x_{2n})=1\}.$$
For $(x_3, \ldots, x_{2n})\in {S_0}$,  and $i\geqslant m+1$,
we can substitute the variable $x_i$ that appears in $L(x_3, \ldots, x_{2n})$ with $L_i^0 \equiv L_i^1$.
Then, we get an (affine) linear polynomial $L'(x_3, \ldots, x_m)\in \mathbb{Z}_2[x_1, \ldots, x_m]$ such that $L'(x_3, \ldots, x_m)=L(x_3, \ldots, x_m, x_{m+1}, \ldots, x_{2n})$ for $(x_3, \ldots, x_{2n})\in {S_0}$.
Thus, $${S}_{+}=\{(x_3, \ldots, x_{2n})\in {S_0} \mid L'(x_3, \ldots, x_{m})=0\},$$
and $${S}_{-}=\{(x_3, \ldots, x_{2n})\in {S_0}  \mid L'(x_3, \ldots, x_{m})=1\}.$$
Note that as $|{S}_{+}|=|{S}_{-}| >0$, the affine linear polynomial
$L'$ is non-constant, i.e., $d(L') =1$.
Then, for every $(x_3, \ldots, x_m)\in \mathbb{Z}^{m-2}_2$,
$$(-1)^{F_{12}^{00}(x_3, \ldots, x_m)}=(-1)^{F_{12}^{11}(x_3, \ldots, x_m)}\text{ if }  L'(x_3, \ldots, x_m)=0$$ and $$(-1)^{F_{12}^{00}(x_3, \ldots, x_m)}=-(-1)^{F_{12}^{11}(x_3, \ldots, x_m)} \text{ if } L'(x_3, \ldots, x_{m})=1.$$
Thus, $$(-1)^{F_{12}^{00}(x_3, \ldots, x_m)+F_{12}^{11}(x_3, \ldots, x_m)}=(-1)^{L'(x_3, \ldots, x_m)}.$$
Therefore, $F_{12}^{00}(x_3, \ldots, x_m)+F_{12}^{11}(x_3, \ldots, x_m)\equiv L'(x_3, \ldots, x_m).$
Then, $d(F_{12}^{00}+F_{12}^{11})=1$.

\item Suppose that ${S}_0\cap{S}_1=\emptyset$.
Then, there exists a variable $x_i$ where $i\geqslant m+1$ such that between $\{1, 2\}$, exactly one index is in $I(x_i)$.
Without loss of generality, we may assume that $i=m+1$, $1\in I(x_{m+1})$ and $2\notin I(x_{m+1})$.
Then, $x_{m+1}=x_1+K(x_3, \ldots, x_m)$ where $K\in \mathbb{Z}_2[x_3, \ldots, x_m]$ is an
(affine) linear polynomial.
Consider $S_0$. 
$$S_0= \{(x_3, \ldots, x_{2n})\in \mathbb{Z}_2^{2n-2}\mid x_1=x_2=0, x_{m+1}=x_1+K, x_{m+2}=L_{m+2} \ldots, x_{2n}=L_{2n}\}.$$
Since $x_1=x_2$ on $S_0$, for every $i\geqslant m+2$, 
if $x_1$ or $x_2$ appear in  $L_{i}$, we substitute each one of them with $x_{m+1}+K$. 
We get a linear polynomial $K_{i}\in \mathbb{Z}_2[x_3, \ldots, x_m, x_{m+1}]$.
Then, for every $(x_3, \ldots, x_{2n})\in {S}_0$, $L_{i}=K_{i}$.
Thus, $$S_0= \{(x_3, \ldots, x_{2n})\in \mathbb{Z}_2^{2n-2}\mid x_{m+1}+K=0, x_{m+2}=K_{m+2} \ldots, x_{2n}=K_{2n}\}.$$
Similarly, we have 
$$S_1= \{(x_3, \ldots, x_{2n})\in \mathbb{Z}_2^{2n-2}\mid x_{m+1}+K=1, x_{m+2}=K_{m+2} \ldots, x_{2n}=K_{2n}\}.$$
Let ${S}_\cup={S}_0\cup{S}_1$.
Then,
$$S_\cup= \{(x_3, \ldots, x_{2n})\in \mathbb{Z}_2^{2n-2}\mid x_{m+2}=K_{m+2} \ldots, x_{2n}=K_{2n}\}.$$
Thus, we can pick $x_3, \ldots, x_{m}, x_{m+1}$ as a set of free variables of $S_\cup$.

Consider $g=\partial_{12}f$. Clearly, $\mathscr{S}(g)
={S}_{\cup}$ since $S_0\cap S_1=\emptyset$.
Then, there exists a unique multilinear boolean polynomial $G(x_3, \ldots, x_{m+1})\in \mathbb{Z}_2[x_3, \ldots, x_{m+1}]$ such that $$g(x_3, \ldots, x_{2n})=\chi_{S_{\cup}}\cdot (-1)^{G(x_3, \ldots, x_{m+1})}.$$
For every $(x_3, \ldots, x_{2n})\in S_0$ that is uniquely decided by $(0, 0, x_3, \ldots, x_{m})\in \{(0, 0)\}\times\mathbb{Z}_2^{m-2}$, $x_{m+1}=K(x_3, \ldots, x_{m})$ and
$f_{12}^{00}(x_3, \ldots, x_{2n})=g(x_3, \ldots, x_{2n})$.
Thus, for every $(x_3, \ldots, x_{m})\in \mathbb{Z}_2^{m-2}$,
$$(-1)^{F^{00}_{12}(x_3, \ldots, x_{m})}=(-1)^{G(x_3, \ldots, x_{m}, K)}.$$
Also, for every $(x_3, \ldots, x_{2n})\in S_1$ that is uniquely decided by $(1, 1, x_3, \ldots, x_{m})\in \{(1, 1)\}\times\mathbb{Z}_2^{m-2}$,
 $x_{m+1}=K(x_3, \ldots, x_{m})+1$, and
$f_{12}^{11}( x_3, \ldots, x_{2n})=g(x_3, \ldots, x_{2n})$.
Thus, for every $(x_3, \ldots, x_{m})\in \mathbb{Z}_2^{m-2}$,
$$(-1)^{F^{11}_{12}(x_3, \ldots, x_{m})}=(-1)^{G(x_3, \ldots, x_{m}, K+1)}.$$
Thus, $F^{00}_{12}(x_3, \ldots, x_{m}) \equiv G(x_3, \ldots, x_{m}, K)$ and $F^{11}_{12}(x_3, \ldots, x_{m})\equiv G(x_3, \ldots, x_{m}, K+1)$. 

Since $f\in \int_{\mathcal{B}}\mathscr{A}$, $g=\partial_{12}f\in \mathscr{A}$.
Thus, $$g'(x_3, \ldots, x_m, x_{m+1})=(-1)^{G(x_3, \ldots, x_m, x_{m+1})}$$ is also in $\in \mathscr{A}$.
Let $y=x_{m+1}+K(x_3, \ldots, x_m)\in \mathbb{Z}[x_3, \ldots, x_{m+1}]$ be an affine linear combination of variables $x_3, \ldots, x_{m+1}$.
Since $g\in \mathscr{A}$, 
by Lemma~\ref{lem-congruity}, 
$$d[G(x_3, \ldots, x_m, K)+G(x_3, \ldots, x_m, K+1)]\leqslant 1.$$
Thus, $d(F^{00}_{12}+F^{11}_{12})\leqslant 1$.
Also if $d(G)=1$, then by Lemma~\ref{lem-congruity}
\begin{equation}\label{equ-01}
    d(F^{00}_{12}+F^{11}_{12})=0, \text{ i.e., } F^{00}_{12}+F^{11}_{12}\equiv 0 ~\mbox{or}~ 1.
\end{equation}
\end{itemize}

Similarly, we can show that $d(F^{01}_{12}+F^{10}_{12})\leqslant 1$.
Thus,  for all $i, j\in [m]$, $d(F^{00}_{ij}+F^{11}_{ij})\leqslant 1$ and $d(F^{01}_{ij}+F^{10}_{ij})\leqslant 1$.
By Lemma~\ref{lem-mutilinear-poly}, $d(F)\leqslant 3$.

If $d(F)\leqslant 2$, then clearly $f\in\mathscr{A}$. We are done.
Otherwise, $d(F)= 3$ and by  Lemma~\ref{lem-mutilinear-poly}, $F$ is a complete cubic multilinear polynomial over $m$ variables.
If we pick another set $X$ of $m$ free variables and substitute variables of $F$ by variables in $X$, 
then we will get 
a cubic multilinear polynomial $F'$ over variables in $X$.
 Same as the analysis of $F$,
$F'$ is also a complete cubic  polynomial. 
In particular, consider the variable $x_{m+1}$.
Recall that $|I(x_{m+1})|\geqslant 2$.
Without loss of generality, we assume that $1\in I(x_{m+1})$.
Then, $x_{m+1}=x_1+L(x_2, \ldots, x_{m})$ where $L(x_2, \ldots, x_{m})$ is an affine linear combination of variables $x_2, \ldots, x_{m}$.
We substitute $x_1$ in $F$ by $x_{m+1}+L$, and we get a complete cubic multilinear polynomial $F'(x_2, \ldots, x_{m+1})\in\mathbb{Z}_2[x_2, \ldots, x_{m+1}].$
By Lemma~\ref{lem-cubic-poly}, if $m\geqslant 5$, then $x_{m+1}=x_1$ or $x_{m+1}=\overline{x_1}$. Thus, $I(x_{m+1})=\{1\}$. This contradicts with $|I(x_{m=1})|\geqslant 2$.
Thus, $m\leqslant 4$.

If $m=4$,
then by Lemma~\ref{lem-cubic-poly}, $x_{5}={x_1}+ \epsilon$,
or $x_{5}={x_1+x_i+x_j}+\epsilon$, where $\epsilon = 0 ~\mbox{or}~1$,  for some $2\leqslant i < j\leqslant 4$.
Since $|I(x_5)|\geqslant 2$, the case that  $x_{5}={x_1}+\epsilon$ is impossible.
Similarly, for $i\geqslant m+2$, the variable $x_i$  is a sum of three variables in $\{x_1, x_2, x_3, x_4\}$ plus a constant $0$ or $1$. 
If there exist $x_i$ and $x_j$ for $5\leqslant i < j \leqslant 2n$ such that $I(x_i)=I(x_j)$. 
Then, $x_i=x_j$ or $\overline{x_j}$.
Thus, among $f^{00}_{ij}$, $f^{01}_{ij}$, $f^{10}_{ij}$ and $f^{11}_{ij}$, two are zero signatures. 
Thus, $f$ does not satisfy {\sc 2nd-Orth}. Contradiction.
Thus, $I(x_i)\neq I(x_j)$ for any $5\leqslant i < j \leqslant 2n$.
There are only ${4\choose 3}=4$ ways to pick three variables from $\{x_1, x_2, x_3, x_4\}$.
Thus, $2n\leqslant 4+4=8$.
By the hypothesis $2n \geqslant 8$ of the lemma, we have $2n=8$.
Clearly, $|\mathscr{S}(f)|=2^4=16$.
Due to {\sc 2nd-Orth}, for all $\{i, j\}\in [8]$, 
$|\mathscr{S}(f_{ij}^{00})|=|\mathscr{S}(f_{ij}^{01})|=|\mathscr{S}(f_{ij}^{10})|=|\mathscr{S}(f_{ij}^{11})|=4$.
\begin{itemize}
    \item 
If there exists $\{i, j\}$ such that $\mathscr{S}(f_{ij}^{00})=\mathscr{S}(f_{ij}^{11})$, 
then for any point $\alpha$ in $\mathscr{S}(f_{ij}^{00})=\mathscr{S}(f_{ij}^{11})$, regardless whether $f_{ij}^{00}(\alpha)=f_{ij}^{11}(\alpha)$ or $f_{ij}^{00}(\alpha)=-f_{ij}^{11}(\alpha)$,
either  
$\alpha\in \mathscr{S}(\partial^+_{ij}f)$ or  $\alpha\in \mathscr{S}(\partial^-_{ij}f)$. Thus,
$$\mathscr{S}(\partial^+_{ij}f)\cup \mathscr{S}(\partial^-_{ij}f)=\mathscr{S}(f_{ij}^{00})=\mathscr{S}(f_{ij}^{11}).$$
Also, by {\sc 2nd-Orth}, $$\langle{\bf f}_{ij}^{00},{\bf f}_{ij}^{11}\rangle=|\mathscr{S}(\partial^-_{ij}f)|-|\mathscr{S}(\partial^+_{ij}f)|=0.$$
Thus, $|\mathscr{S}(\partial^+_{ij}f)|=|\mathscr{S}(\partial^-_{ij}f)|=2$.
Note that every 6-ary signature in $\mathcal{B}^{\otimes}$ has support of size $8$, and every signature in $\mathcal{F}_6$ and  $\mathcal{F}^H_6$ has support of size $32$.
Thus, $\partial^+_{ij}f\notin \mathcal{B}\cup\mathcal{F}_6\cup\mathcal{F}^H_6$.
Then, by Corollary~\ref{lem-non-f6}, we get  \#P-hardness.
Similarly, if there exists $\{i, j\}$ such that $\mathscr{S}(f_{ij}^{01})=\mathscr{S}(f_{ij}^{10})$, then we have 
 $|\mathscr{S}(\partial^{\widehat+}_{ij}f)|=|\mathscr{S}(\partial^{\widehat-}_{ij}f)|=2$.
 Thus,  $\partial^{\widehat +}_{ij}f\notin \mathcal{B}^{\otimes}\cup\mathcal{F}_6\cup\mathcal{F}^H_6$.
 Again,  we get  \#P-hardness.
\item Otherwise, for all $\{i, j\}\in [8]$, $\mathscr{S}(f_{ij}^{00})\cap\mathscr{S}(f_{ij}^{11})=\emptyset$ and 
$\mathscr{S}(f_{ij}^{01})\cap\mathscr{S}(f_{ij}^{10})=\emptyset$.
Then, $\mathscr{S}(\partial^{+}_{ij}f)=\mathscr{S}(f_{ij}^{00})\cup\mathscr{S}(f_{ij}^{11})$.
Thus, $|\mathscr{S}(\partial^{+}_{ij}f)|=8$.
Clearly, $\partial^{+}_{ij}f\notin \mathcal{F}_6\cup\mathcal{F}^H_6$.
If $\partial^{+}_{ij}f\notin \mathcal{B}^{\otimes 3}$, then we get  \#P-hardness.
For a contradiction, suppose that $\partial^{+}_{ij}f\in \mathcal{B}^{\otimes 3}$.
Then, $$\partial^{+}_{ij}f=\chi_{\mathscr{S}(\partial^{+}_{ij}f)}(-1)^{G^+_{ij}} \text{ ~where~ } d(G^+_{ij})=1.$$
As we proved above in equation (\ref{equ-01}), $F^{00}_{ij}+F^{11}_{ij}\equiv 0$ or $1.$
Similarly, suppose $\partial^{\widehat +}_{ij}f\in \mathcal{B}^{\otimes 3}$, and 
we can show that $F^{01}_{ij}+F^{10}_{ij}\equiv 0$ or $1.$
Thus, for all $\{i, j\}\subseteq[8]$, $F^{00}_{ij}+F^{11}_{ij}\equiv 0$ or $1$ and $F^{01}_{ij}+F^{10}_{ij}\equiv 0$ or $1.$
Then, by Lemma~\ref{lem-mutilinear-poly}, $d(F)\leqslant 2$. Contradiction. 
\end{itemize}

Suppose that $m=3$. 
Remember that for $4\leqslant i \leqslant 2n$, $|I(x_i)|\geqslant 2$. Thus, $x_i$ is a sum of at least two variables in $\{x_1, x_2, x_3\}$ plus a constant $0$ or $1$.
Again, if there exist $x_i$ and $x_j$ for $4\leqslant i<j\leqslant 2n$ such that $I(x_i)=I(x_j)$, then among $f^{00}_{ij}$, $f^{01}_{ij}$, $f^{10}_{ij}$ and $f^{11}_{ij}$, two are zero signatures. 
Contradiction.
Thus, $I(x_i)\neq I(x_j)$ for any $4\leqslant i<j\leqslant 2n$.
There are ${3\choose 2}+{3\choose 3}=4$ different ways to pick at least 
two variables from $\{x_1, x_2, x_3\}$. Thus, $2n\leqslant 3+4=7$. Contradiction.
\end{proof}

\begin{theorem}\label{thm-holantb}
Suppose that $\mathcal{F}$ is non-$\mathcal{B}$ hard.
Then, $\Holantb( \mathcal{F})$ is \#P-hard.
\end{theorem}
\begin{proof}
Since $\mathcal{F}$ does not satisfy condition (\ref{main-thr}), $\mathcal{F}$ contains a signature $f\notin \mathscr{A}$. Suppose that $f$ has arity $2n$. We prove this theorem by induction on $2n$.

If $2n=2, 4$ or $6$, then by Corollary~\ref{lem-non-f6} and its remark, $\Holantb(\mathcal{F})$ is \#P-hard.

Inductively assume for some $2k\geqslant 6$,  $\Holantb(\mathcal{F})$ is \#P-hard when $2n\leqslant 2k$. We consider the case that $2n=2k+2\geqslant 8$.
First, suppose that $f$ is reducible.
If it is a tensor product of two signatures of odd arity,
then we can realize a signature of odd arity by factorization.
We get  \#P-hardness by Theorem~\ref{odd-dic}.
Otherwise, it is a tensor product of two signatures of even arity that are not both in $\mathscr{A}$ since $f\notin \mathscr{A}$.
Then, we can realize a non-affine signature of arity $2m\leqslant 2k$ by factorization.
By our induction hypothesis, we get  \#P-hardness.
Thus, we may assume that $f$ is irreducible.
If $f$ has no parity, then we get  \#P-hardness by Lemma~\ref{lem-parity}. Thus, we may also assume that $f$ has parity. Then by Lemma~\ref{lem-in-affine}, $\Holantb( \mathcal{F})$ is \#P-hard, or we can realize a non-affine signature of arity $2m\leqslant 2k$. By our induction hypothesis, we get  \#P-hardness.
\end{proof}

Since ${\mathcal{B}}$ is realizable from $f_6$ and $\{f_6\}\cup \mathcal{F}$ is non-$\mathcal{B}$ hard for any real-valued $\mathcal{F}$ that does not satisfy condition (\ref{main-thr}), we have the following result.
\begin{lemma}\label{lem-f6-b-hardness}
$\Holantb({f_6}, {\mathcal F})$ is \#P-hard.
\end{lemma}

Combining Theorem~\ref{lem-6-ary-f6} and Lemma~\ref{lem-f6-b-hardness}, we have the following result.
This concludes Sections \ref{sec-f6} and \ref{sec-holantb}, and we are done with the arity $6$ case.

\begin{lemma}\label{lem-6-ary}
If  $\widehat{\mathcal{F}}$ contains a signature $\widehat{f}$ of arity $6$ and
$\widehat{f}\notin \widehat{\mathcal{O}}{^\otimes}$, then
 $\holant{\neq_2}{\widehat{\mathcal{F}}}$ is \#P-hard.
 \end{lemma}
 
\section{Final Obstacle: an 8-ary Signature with Strong Bell Property}\label{sec-f8}
We have seen some extraordinary properties of the signature $f_8$.
Now, we formally analyze it. 
Remember that $f_8=\chi_T$ where
\begin{equation}\label{egn:T-defined-again}
 \begin{aligned}
  T=&\mathscr{S}(f_8)=\{(x_1, x_2, \ldots, x_8)\in \mathbb{Z}_2^{8} \mid~
x_1+x_2+x_3+x_4= 0, ~ x_5+x_6+x_7+x_8= 0,\\
  &\hspace{34ex}x_1+x_2+x_5+x_6= 0, ~ x_1+x_3+x_5+x_7= 0\}.\\
=&\{00000000, 00001111, 00110011, 00111100, 01010101, 01011010, 01100110, 01101001,\\
& ~10010110,  10011001, 10100101, 10101010, 11000011, 11001100, 11110000, 11111111\}.
\end{aligned}
\end{equation}
One can see that $\mathscr{S}(f_8)$ has the following structure: the sums of the first four variables, and the last four variables are both even; the assignment of the first four variables are either identical to, or complement of the assignment of the last four variables.  
    Another interesting description of   $\mathscr{S}(f_8)$ is as follows:
    One can take  4  variables, called them $y_1, y_2, y_3, y_4$.
    Then  on the support the remaining 4 variables 
    are mod 2 sums of ${4 \choose 3}$ subsets of $\{y_1, y_2, y_3, y_4\}$,
    and $y_1, y_2, y_3, y_4$ are free variables.
    (However, the 4 variables $(y_1, y_2, y_3, y_4)$ cannot be taken
    as $(x_1, x_2, x_3, x_4)$ in the above description (\ref{egn:T-defined-again}).
    But one \emph{can} take  $(y_1, y_2, y_3, y_4) = (x_1, x_2, x_3, x_5)$.
    More specifically, one can take any 3 variables $x_i, x_j, x_k$
    from $\{x_1, \ldots, x_8\}$
    first as free variables, which excludes one unique other $x_{\ell}$ from the remainder
    set $X' = \{x_1, \ldots, x_8\} \setminus \{x_i, x_j, x_k\}$, and \emph{then}
    one can take any one variable $x_{r} \in X'$ as the 4th free variable. \emph{Then}
    the remaining 4  variables    
    are the mod 2 sums of ${4 \choose 3}$ subsets of the 4 free variables $\{x_i, x_j, x_k, x_r\}$, and in particular $x_\ell = x_i + x_j + x_k$, on $\mathscr{S}(f_8)$.)
    We give the following Figure~\ref{fig-arity-8} to visualize the signature matrix $M_{1234}(f_8)$.
    A block with orange color denotes an entry $+1$. Other blank blocks are zeros.
    
          \begin{figure}[!h]
    \centering
    \includegraphics[height=2.4in]{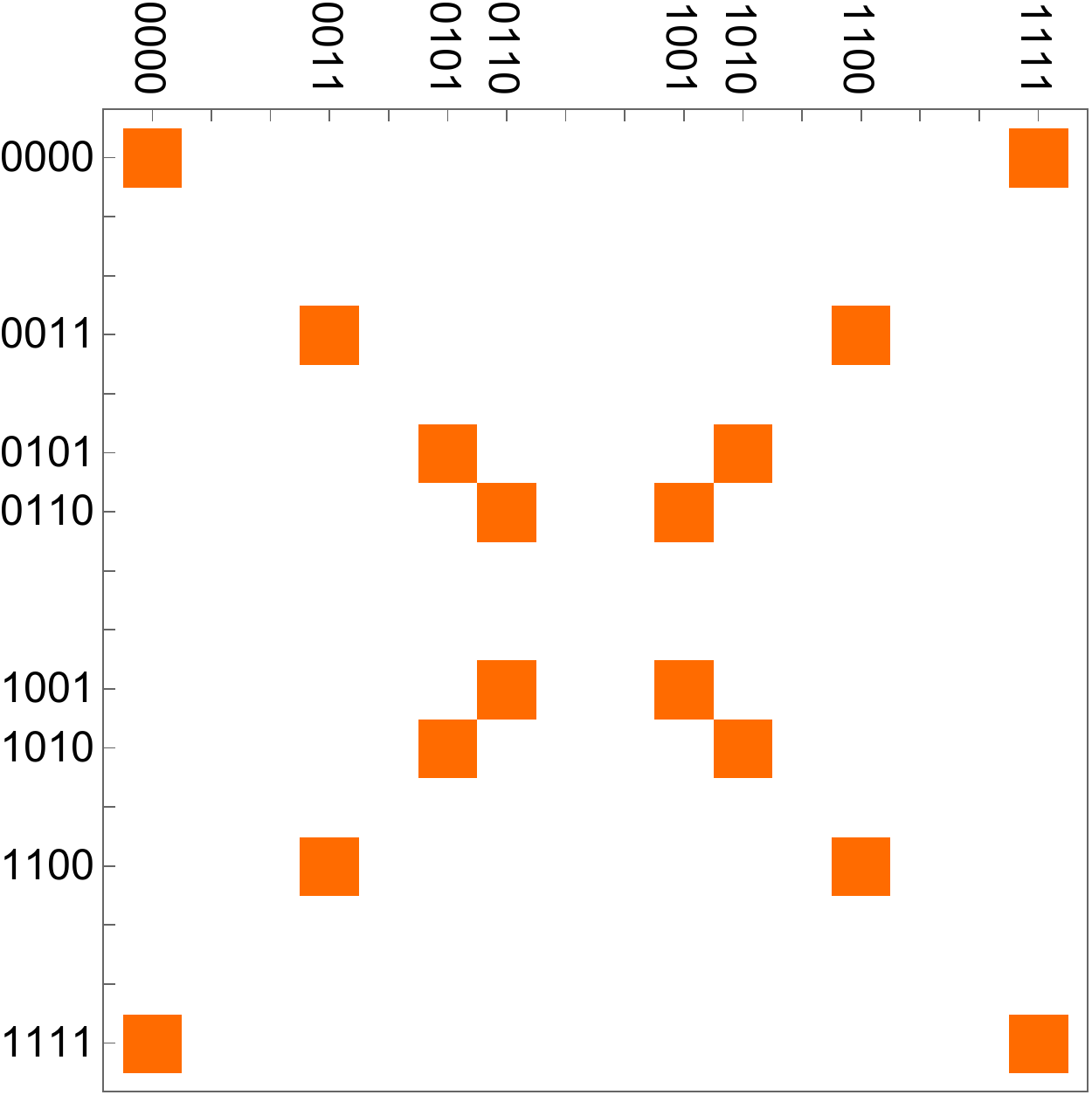}
    \caption{A visualization of $f_8$, which happens to be the same as $\widehat{f_8}=Z^{-1}f_8$}
    \label{fig-arity-8}
\end{figure}
    
    One can check that $f_8$ 
    satisfies both {\sc 2nd-Orth}  and $f_8\in \int \mathcal{O}^{\otimes}.$
 Also, $f_8$ is unchanged under the holographic transformation by $Z^{-1}$, i.e., $\widehat{f_8}=Z^{-1}f_8=f_8$.
\subsection{The discovery of $\widehat{f_8}$}
 In this subsection, we show how this extraordinary signature $\widehat{{f_8}}$ was discovered. 
 We use the notation $\widehat{{f_8}}$ since we  consider the problem $\holant{\neq_2}{\widehat{\mathcal{F}}}$ for complex-valued $\widehat{\mathcal{F}}$ satisfying {\sc ars}.
 We prove that if $\widehat{\mathcal{F}}$ contains an 8-ary signature $\widehat{f}$ where $\widehat{f}\notin\widehat{\mathcal{O}}{^\otimes}$, then $\holant{\neq_2}{\widehat{\mathcal{F}}}$ is \#P-hard or $\widehat{f_8}$ is realizable from $\widehat{f}$ (Theorem~\ref{thm-f8}).

Remember that ${\mathcal{D}}=\{\neq_2\}.$ 
Then ${\mathcal{D}}^{\otimes}=\{\lambda \cdot (\neq_2)^{\otimes n}\mid \lambda\in \mathbb{R}\backslash\{0\}, n\geqslant 1\}$ is the set of tensor products of binary disequalities $\neq_2$ up to a nonzero real scalar.
If for all pairs of indices $\{i, j\}$, $\widehat{\partial}_{ij}\widehat{f}\in {\mathcal{D}}^{\otimes}$, then we say $\widehat{f}\in \widehat{\int}{\mathcal{D}}^{\otimes}$.
Clearly, if $\widehat{f}\in {\mathcal{D}}^{\otimes}$ and $\widehat{f}$ has arity greater than $2$, then $\widehat{f}\in \widehat{\int}{\mathcal{D}}^{\otimes}.$
We first show the following result for signatures of arity at least $8$.

\begin{lemma}\label{lem-even>8}
Let $\widehat f \notin \widehat{\mathcal{O}}^{\otimes}$ be a signature of arity $2n\geqslant 8$ in $\widehat{\mathcal{F}}$.
Then, 
\begin{itemize}
    \item $\holant{\neq_2}{\widehat{\mathcal{F}}}$ is \#P-hard, or
    \item there is a signature $\widehat g\notin \widehat{\mathcal{O}}^{\otimes}$ of arity $2k\leqslant 2n-2$ that is realizable from $\widehat f$, or 
    \item there is an irreducible signature $\widehat{f^\ast}\in \widehat{\int}{\mathcal{D}}^{\otimes}$ of arity $2n$ that is realizable from $\widehat f$. 
\end{itemize}
\end{lemma}

\begin{proof}
Since $\widehat{f}\notin \widehat{\mathcal{O}}^{\otimes}$, $\widehat f\not\equiv 0$.
Again, we may assume that $\widehat{f}$ is irreducible.
Otherwise, by factorization, we can realize a nonzero signature of odd arity and we get \#P-hardness by Theorem~\ref{odd-dic}, or we can realize a signature of lower even arity that is not in $\widehat{\mathcal{O}}^{\otimes}$ and we are done.
 Under the assumption that $\widehat{f}$ is irreducible, we may further assume that $\widehat{f}$ satisfies {\sc 2nd-Orth} by Lemma~\ref{second-ortho}.
Consider signatures $\widehat{\partial}_{ij}\widehat{f}$ for all pairs of indices $\{i, j\}$. 
If there exists a pair $\{i, j\}$ such that $\widehat{\partial}_{ij}\widehat{f}\notin \widehat{\mathcal{O}}^{\otimes}$, then let $\widehat{g}=\widehat{\partial}_{ij}\widehat{f}$, and we are done.
Thus, we may also assume that $\widehat{f}\in \widehat{\int}\widehat{\mathcal{O}}^{\otimes}$.  

If for all pairs of indices $\{i, j\}$, we have $\widehat{\partial}_{ij}\widehat{f}\equiv 0$. 
Then, by Lemma \ref{lem-zero_2}, ${\widehat f}(\alpha)=0$ for all $\alpha$ with ${\rm wt}(\alpha)\neq 0$ or $2n$. 
Since $f\not\equiv 0$ and by {\sc ars}, $|{\widehat f}(\vec{0}^{2n})| =|{\widehat f}(\vec{1}^{2n})| \not =0$.
Clearly, such a signature  does not satisfy {\sc 2nd-Orth}. 
Contradiction.
Thus,
%
without loss of generality, we assume that $\widehat{\partial}_{12}\widehat{f}\not\equiv 0$.
Since $\widehat{\partial}_{12}\widehat{f} \in \widehat{\mathcal{O}}^{\otimes}$, without loss of generality, we may assume that in the UPF of $\widehat{\partial}_{12}\widehat{f}$, variables $x_3$ and $x_4$ appear in one  binary signature $b_1(x_3, x_4)$, $x_5$ and $x_6$ appear in one binary signature $b_2(x_5, x_6)$ and so on. 
Thus, we have 
\begin{equation*}
\widehat{\partial}_{12}\widehat{f}=  \widehat{b_1}(x_3, x_4)\otimes \widehat{b_2}(x_5, x_6)\otimes\widehat{b_3}(x_7, x_8)\otimes \ldots \otimes \widehat{b_{n-1}}(x_{2n-1}, x_{2n}).
\end{equation*}

By Lemma \ref{lin-wang},  all these binary signatures $\widehat{b_1}$, $\widehat{b_2}$, \ldots, $\widehat{b_{n-1}}$ are realizable from $f$ by factorization. 
Note that for nonzero binary signatures $\widehat{b_{i}}(x_{2i+1}, x_{2i+2})$ $(1\leqslant i\leqslant n-1)$, if we connect the variable $x_{2i+1}$ of two copies of $\widehat{b_{i}}(x_{2i+1}, x_{2i+2})$ using $\neq_2$ (mating two binary signatures), then we get $\neq_2$ up to a scalar.
We consider the following gadget construction on $\widehat{f}$.
Recall that in the setting of $\holant{\neq}{\widehat{\mathcal{F}}}$, variables are connected using $\neq_2$.
For $1\leqslant i\leqslant n-1$, by a slight abuse of names of variables, we connect the variable $x_{2i+1}$ of $\widehat{f}$ 
with the variable $x_{2i+1}$ of $\widehat{b_{i}}(x_{2i+1}, x_{2i+2})$ using $\neq_2$. 
We get a signature $\widehat{f'}$ of arity $2n$. 
(Note that, as a complexity reduction using
 factorization (Lemma~\ref{lin-wang}), we can only apply it a constant
 number of times. 
However, the arity $2n$ of $\widehat{f}$
is considered a constant, as $\widehat{f} \in \widehat{\mathcal{F}}$,  which is independent of
the input size of a signature grid to the
problem  $\holant{\neq_2}{\widehat{\mathcal{F}}}$.)
We denote this gadget construction by $ G_1$ and we write $\widehat{f'}$ as $G_1 \circ \widehat f$.
$G_1$ is constructed by extending variables of $\widehat{f}$ using binary signatures realized from $\widehat{\partial}_{12}\widehat{f}$. 
It does not change the irreducibility of $\widehat{f}$.  
Thus, $\widehat{f'}$ is irreducible since $\widehat{f}$ is irreducible.
Similarly, we may assume that $\widehat{f'}\in \widehat{\int}\widehat{\mathcal{O}}^{\otimes}$. Otherwise, we are done.

Consider the signature $\widehat{\partial}_{12}\widehat{f'}$. 
Since the above gadget construction $G_1$  does not touch variables $x_1$ and $x_2$ of $f$, $G_1$
commutes with the merging gadget $\widehat{\partial}_{12}$. 
(Succinctly, the commutativity can be expressed as  $\widehat{\partial}_{12}\widehat{f'}=\widehat{\partial}_{12}(G_1\circ \widehat{f})=G_1\circ \widehat{\partial}_{12}\widehat{f}.)$
Thus, $\widehat{\partial}_{12}\widehat{f'}$ can be realized by performing the  gadget construction $G_1$ on $\widehat{\partial}_{12}\widehat{f}$, which connects each binary signature $\widehat{b_{i}}(x_{2i+1}, x_{2i+2})$ in the UPF of $\widehat{\partial}_{12}\widehat{f}$ with another copy of $\widehat{b_{i}}(x_{2i+1}, x_{2i+2})$  (in the mating fashion). 
Thus, each binary signature $\widehat{b_{i}}$ in $\widehat{\partial}_{12}\widehat{f}$ is changed to $\neq_2$ up to a nonzero scalar after this gadget construction $G_1$. 
After normalization and renaming variables, we have 
\begin{equation}\label{eqn:lm5-partial12-f}
\widehat{\partial}_{12}\widehat{f'}= (\neq_2)(x_3, x_4)\otimes (\neq_2)(x_5, x_6)\otimes(\neq_2)(x_7, x_8)\otimes \ldots \otimes (\neq_2)(x_{2n-1}, x_{2n}).
\end{equation}
Thus, $\widehat{\partial}_{12}\widehat{f'}\in {\mathcal{D}}^{\otimes}$.
Moreover, for all pairs of indices $\{i, j\}$ disjoint with $\{1, 2\}$,
we have 
\begin{equation}\label{equ-ij-12-in-D}
    \widehat{\partial}_{(ij)(12)}\widehat{f'}\in {\mathcal{D}}^{\otimes}, \text{ ~and hence~ } \widehat{\partial}_{(ij)(12)}\widehat{f'}\not\equiv 0.
    \end{equation}
A fortiori, for all pairs of indices $\{i, j\}$ disjoint with $\{1, 2\}$, $\widehat{\partial}_{ij}\widehat{f'}\not\equiv 0$.

Now, we show that we can realize an irreducible signature $\widehat{f^\ast}$ of arity $2n$ from $\widehat{f'}$ such that $\widehat{f^\ast}\in \widehat{\int}{\mathcal{D}}^{\otimes}$.
We first prove the following claim.
  \begin{quote}
     {\bf Claim.} \emph{Let $\widehat{h}\in \widehat{\int}{\widehat{\mathcal{O}}}^{\otimes}$ be a signature of arity $2n\geqslant 8$. If 
     $\widehat\partial_{ij}\widehat{h}\in \mathcal{D}^{\otimes}$ for all $\{i, j\}$ disjoint with $\{1, 2\}$, then $\widehat{h}\in \widehat{\int}{{\mathcal{D}}}^{\otimes}.$}
  \end{quote}
    
 Clearly, we only need to show that $\widehat{\partial}_{1k}\widehat{h}\in {\mathcal{D}}^{\otimes}$ for all $2 \leqslant k \leqslant 2n$. Then, by symmetry we also have $\widehat{\partial}_{2k}\widehat{h}\in {\mathcal{D}}^{\otimes}$ for $k=1$ and all $3 \leqslant k \leqslant 2n$. This will prove $\widehat{h}\in \widehat{\int}{\mathcal{D}}^{\otimes}.$
    Consider $\widehat{\partial}_{1k}\widehat{h}$ for an arbitrary $2\leqslant k \leqslant 2n$.
   Since for all $\{i, j\}$ disjoint with $\{1, 2\}$, we have
    $\widehat{\partial}_{ij}\widehat{h}\in {\mathcal{D}}^{\otimes}$, a fortiori
    for all  $\{i, j\}$ disjoint with $\{1, 2\} \cup \{k\}$,
    \begin{equation}\label{eqn:partial-1kj-in-sec8}
    \widehat{\partial}_{(1k)(ij)}\widehat{h}\in {\mathcal{D}}^{\otimes}.
   \end{equation}
    Since $\widehat{h}$ has arity $2n\geqslant 8$, we can pick a pair of indices $\{i, j\}$ disjoint with $\{1, 2\} \cup \{k\}$.
    Since $\widehat{\partial}_{(1k)(ij)}\widehat{h}\in {\mathcal{D}}^{\otimes}$, which is nonzero,  a  fortiori  we have
    $\widehat{\partial}_{1k}\widehat{h} \not \equiv 0$.
    So we may consider  the UPF of $\widehat{\partial}_{1k}\widehat{h}$,
    which is known to be in $\widehat{\mathcal{O}}^{\otimes}$.
    For a contradiction,
    suppose that  there is a binary signature  $\widehat{b_1}$ (as a factor of $\widehat{\partial}_{1k}\widehat{h}$) such that $\widehat{b_1}$ is not an associate of $\neq_2$. 
    Among the two variables in the scope of $\widehat{b_1}$,  at least one is not $x_2$. 
    We pick such a variable $x_s$ where $x_s \neq x_2$. 
    Then, we consider another binary signature $\widehat{b_2}$ in the UPF of $\widehat{\partial}_{1k}\widehat{h}$. 
    \begin{itemize}
        \item 
    If $\widehat{b_2} =\lambda \cdot \neq_2$, for some  nonzero scalar
    $\lambda$, then we pick a variable $x_t$ in the scope of $\widehat{b_2}$ that is not $x_2$. 
    Consider $\widehat{\partial}_{(st)(1k)}\widehat{h}$. 
    When merging variables $x_s$ and $x_t$ of $\widehat{\partial}_{1k}\widehat{h}$, 
    we connect the variable $x_s$ of  $\widehat{b_1}$ with the variable $x_t$ of $\lambda  \cdot \neq_2$, and the resulting binary signature is just $\lambda \cdot\widehat{b_1}$, which is not an associate of $\neq_2$. 
    Thus, we have $\widehat{\partial}_{(st)(1k)}\widehat{h}\notin {\mathcal{D}}^{\otimes}$.
    \item Otherwise,  $\widehat{b_2}$ is not an associate of $\neq_2$. Since $\widehat{\partial}_{1k}\widehat{h}$ has arity $2n-2\geqslant 6$, we can find another binary signature $\widehat{b_3}$ in the UPF of $\widehat{\partial}_{1k}\widehat{h}$. 
    We pick a variable $x_t$ in the scope of $\widehat{b_3}$ that is not $x_2$. 
    Consider $\widehat{\partial}_{(st)(1k)}\widehat{h}$. 
    Now, when merging variables $x_s$ and $x_t$ of $\widehat{\partial}_{1k}\widehat{h}$, 
    the binary signature $\widehat{b_2}$ is untouched. 
    Thus, we have $\widehat{b_2}\mid \widehat{\partial}_{(st)(1k)}\widehat{h}$, which implies that
     $\widehat{\partial}_{(st)(1k)}\widehat{h}\notin  {\mathcal{D}}^{\otimes}$.
   
   \end{itemize}
   
   Note that in both cases, $\{s,t\} \cap (\{1,2\} \cup \{k\}) = \emptyset$. Therefore
   the two cases above both contradict
   (\ref{eqn:partial-1kj-in-sec8}) by picking $\{i, j\}=\{s, t\}$.
      Thus, $\widehat{\partial}_{1k}\widehat{h}\in {\mathcal{D}}^{\otimes}$ for all $2 \leqslant k \leqslant 2n$. Then similarly, we can show that $\widehat{\partial}_{2k}\widehat{h}\in {\mathcal{D}}^{\otimes}$ for all $3 \leqslant k \leqslant 2n$. 
      This finishes the proof of our Claim.
      
\vspace{.1in}

Remember that $\widehat{\partial}_{ij}\widehat{f'}\not\equiv 0$ for all $\{i, j\}$ disjoint with $\{1, 2\}$.
We consider the UPF of $\widehat{\partial}_{ij}\widehat{f'}$. Since $\widehat{f'}\in \widehat{\int}\widehat{\mathcal{O}}^{\otimes}$, there are two cases depending on whether variables $x_1$ and $x_2$ appear in one  binary signature or two distinct binary signatures.

\vspace{1ex}
\noindent{\bf Case 1.}
    For every $\{i, j\}$ disjoint with $\{1, 2\}$, in the UPF of $\widehat{\partial}_{ij}\widehat{f'}$, $x_1$ and $x_2$ appear in one nonzero binary signature
    $\widehat{b_{ij}}(x_1, x_2) \in \widehat{\mathcal{O}}$.  
    In other words, for every $\{i, j\}$ disjoint with $\{1, 2\}$, $$\widehat{\partial}_{ij}\widehat{f'}=\widehat{b_{ij}}(x_1, x_2)\otimes \widehat{g_{ij}}, ~~\text{ for some } \widehat{g_{ij}}\not\equiv 0.$$
    (These factors $\widehat{b_{ij}}$ and $ \widehat{g_{ij}}$ are nonzero since $\widehat{\partial}_{ij}\widehat{f'}\not\equiv 0$.)
  Then, $\widehat{g_{ij}} \sim \widehat{\partial}_{(12)(ij)}\widehat{f'}$, and by
    (\ref{equ-ij-12-in-D}), we have $\widehat{g_{ij}}\in{\mathcal{D}}^{\otimes}$. 
    Also for $\{k, \ell\}$ disjoint with both $\{i, j\}$
    and $\{1, 2\}$, $\widehat{\partial}_{(k\ell)(ij)}\widehat{f'}
    \not \equiv 0$ since
    $\widehat{\partial}_{(12)(k \ell)(ij)} \widehat{f'} = \widehat{\partial}_{(ij)(k \ell)(12)} \widehat{f'} 
    \not \equiv 0$.
    
    We first show that for any two pairs $\{i, j\}\neq \{k, \ell\}$ that are both disjoint with $\{1, 2\}$,
    $\widehat{b_{ij}}(x_1, x_2)\sim \widehat{b_{k\ell}}(x_1, x_2)$. 
    If $\{i, j\}$ is disjoint with $\{k, \ell\}$, then  $\widehat{b_{ij}}(x_1, x_2)\mid \widehat{\partial}_{(k\ell)(ij)}\widehat{f'}$ and $\widehat{b_{k\ell}}(x_1, x_2)\mid \widehat{\partial}_{(ij)(k\ell)}\widehat{f'}$. 
    Since $\widehat{\partial}_{(k\ell)(ij)}\widehat{f'}=\widehat{\partial}_{(ij)(k\ell)}\widehat{f'}\not\equiv 0$, by Lemma \ref{unique}, 
    we have $\widehat{b_{ij}}(x_1, x_2)\sim \widehat{b_{k\ell}}(x_1, x_2)$. 
    Otherwise, $\{i, j\}$ and $\{k, \ell\}$ are not disjoint. Since $\widehat{f'}$ has arity $\geqslant 8$, we can find another pair of indices $\{s, t\}$ such that it is disjoint with $\{1, 2\}\cup \{i, j\} \cup \{k, \ell\}$. 
    Then, by the above argument, we have $\widehat{b_{ij}}(x_1, x_2)\sim \widehat{b_{st}}(x_1, x_2),$ and  $\widehat{b_{st}}(x_1, x_2)\sim \widehat{b_{k\ell}}(x_1, x_2).$
    Thus, $\widehat{b_{ij}}(x_1, x_2)\sim \widehat{b_{k\ell}}(x_1, x_2).$
    We can use a binary signature $\widehat{b}(x_1, x_2)$  to denote these binary signature $\widehat{b_{ij}}(x_1, x_2)$ for all $\{i, j\}$ disjoint with $\{1, 2\}$. 
    Then, $\widehat{b}(x_1, x_2) \mid \widehat{\partial}_{ij}\widehat{f'}$ for all $\{i, j\}$ disjoint with $\{1, 2\}$. 
   Also,  $\widehat{b}(x_1, x_2)$ is realizable  from $\widehat f'$ by merging and factorization.
    
    Then, we consider the following gadget construction $G_2$ on $\widehat{f'}$. 
    By a slight abuse of variable names, 
    we connect the variable $x_1$ of $\widehat{f'}$ with the variable $x_1$ of $\widehat b(x_1, x_2)$ and we get a signature $\widehat{f^\ast}$. 
    Clearly, $G_2$ is constructed by extending variables of $\widehat{f'}$. It does not change the irreducibility of $\widehat{f'}$.
    Thus, $\widehat{f^\ast}$ is irreducible.
    Again, we may assume that $\widehat{f^\ast}\in \widehat{\int}\widehat{\mathcal{O}}^{\otimes}$.
    Consider $\widehat{\partial}_{ij}\widehat{f^\ast}$ for all $\{i, j\}$ disjoint with $\{1, 2\}$. 
    Since the above gadget construction $G_2$ only touches the variable $x_1$ of $f'$, it commutes with the merging operation $\widehat{\partial}_{ij}$.
    Thus, $\widehat{\partial}_{ij}\widehat{f^\ast}$ can be realized by performing the gadget construction $G_2$ on $\widehat{\partial}_{ij}\widehat{f'}$, i.e.,
    connecting the binary signature $\widehat b(x_1, x_2)$ in the UPF of $\widehat{\partial}_{ij}\widehat{f'}$ with itself (in the mating fashion), which changes $\widehat b(x_1, x_2)$ to $\neq_2$ up to some nonzero scalar $\lambda_{ij}$.
    Thus,  for all $\{i, j\}$ disjoint with $\{1, 2\}$, after renaming variables, we have
    $$\widehat{\partial}_{ij}\widehat{f^\ast}=\lambda_{ij}\cdot(\neq_2)(x_1, x_2) \otimes \widehat{g_{ij}}\in {\mathcal{D}}^{\otimes}.$$
        Thus, $\widehat{\partial}_{ij}\widehat{f^\ast}\in {\mathcal{D}}^{\otimes}$ for all $\{i, j\}$ disjoint with $\{1, 2\}$. By our {Claim}, $\widehat{f^\ast}\in \widehat{\int}{\mathcal{D}}^{\otimes}$.
        We are done with {Case 1.}

    
    \vspace{1ex}
    \noindent{\bf Case 2.}  There is a pair of indices $\{i, j\}$ disjoint with $\{1, 2\}$ such that $x_1$ and $x_2$ appear in two distinct nonzero binary signatures $\widehat{b_1'}(x_1, x_u)$ and $\widehat{b_2'}(x_2, x_v)$ in the UPF of $\widehat{\partial}_{ij}\widehat{f'}$. 
 In other words, there exits $\{i, j\}$ such that  \begin{equation}
    \widehat{\partial}_{ij}\widehat{f'}=\widehat{b'_1}(x_1, x_u)\otimes \widehat{b'_2}(x_2, x_v) \otimes \widehat{h_{ij}}, \text{ for some } \widehat{h_{ij}}\not\equiv 0.
    \end{equation}
    Since $\widehat{h_{ij}} \mid \widehat{\partial}_{(12)(ij)}\widehat{f'}$ and $\widehat{\partial}_{(12)(ij)}\widehat{f'}\in 
    {\mathcal{D}}^{\otimes}$, we have $\widehat{h_{ij}}\in {\mathcal{D}}^{\otimes}$.
    Also, after merging variables $x_1$ and $x_2$ (using $\neq_2$) in $\widehat{\partial}_{ij}\widehat{f'}$, variables $x_u$ and $x_v$ form a binary disequality up to a nonzero scalar (this binary signature on $x_u$ and $x_v$ must be a binary disequality
    because we already know $\widehat{\partial}_{(12)(ij)}\widehat{f'}\in 
    {\mathcal{D}}^{\otimes}$). 
    In other words, by connecting the variable $x_1$ of $\widehat{b'_1}(x_1, x_u)$ and the variable $x_2$ of $\widehat{b'_2}(x_2, x_v)$ (using $\neq_2$), we get $\lambda \cdot \neq_{2}(x_u, x_v)$ for some $\lambda\neq 0$.
    By Lemma \ref{lem-binary-sim}, we have
    $\widehat{b'_1} \sim \widehat{b'_2}$.
    Also, connecting the variable $x_u$ of $\widehat{b'_1}$  and the variable  $x_v$ of $\widehat{b'_2}$ (using $\neq_2$) will give the binary signature $\lambda \cdot \neq_{2}(x_1, x_2)$ as well.
    
    We consider the following gadget construction $G_3$ on $\widehat{f'}$.
    By a slight abuse of variable names, 
    we connect variables $x_1$ and $x_2$ of $\widehat{f'}$ with the variable $x_1$ of $\widehat{b'_1}$ and $x_2$ of $\widehat{b'_2}$ using $\neq_2$ respectively. 
    We get  a signature $\widehat{f^{\ast}}$.
    Again, $\widehat{f^{\ast}}$ is irreducible since the gadget construction $G_3$ does not change the irreducibility of $\widehat{f'}.$
    Also, we may assume that $\widehat{f^{\ast}}\in\widehat{\int}\widehat{\mathcal{O}}^{\otimes}$.
    Otherwise, we are done.
    Consider $\widehat{\partial}_{ij}\widehat{f^{\ast}}$.
    Similarly, by the commutitivity of the  gadget construction $G_3$ and the merging gadget $\widehat{\partial}_{ij}$, 
    $\widehat{\partial}_{ij}\widehat{f^{\ast}}$ can be realized by connecting variables $x_1$ and $x_2$ of $\widehat{\partial}_{ij}\widehat{f'}$ with the variable $x_1$ of $\widehat{b'_1}$ and the variable $x_2$ of $\widehat{b'_2}$ respectively.
    After renaming variables, we have
    \begin{equation}\label{form_ij}
        \widehat{\partial}_{ij}\widehat{f^{\ast}}=\lambda_{ij} \cdot (\neq_2)(x_1, x_u)\otimes (\neq_2) \left(x_2, x_v\right) \otimes \widehat h_{ij}\in \mathcal{D}^{\otimes}.
    \end{equation}

We now show that  $\widehat{\partial}_{12}\widehat{f^{\ast}}\in {\mathcal{D}}^{\otimes}$. Note that it is realized in the following way; we first connect variables $x_1$ and $x_2$ of $\widehat{f'}$ with the variable $x_1$ of $\widehat{b'_1}(x_1, x_u)$ and the variable $x_2$ of $\widehat{b'_2}(x_2, x_v)$ respectively (using $\neq_2$) to get $\widehat{f^{\ast}}$, and then after renaming variables $x_u$ and $x_v$  to $x_1$ and $x_2$ respectively, we merge them using $\neq_2$ (see Figure~\ref{fig:f^ast-f'}(a)). 
By 
associativity of gadget constructions, we can change the order; 
we first connect the variable $x_u$ of $\widehat{b'_1}(x_1, x_u)$ with the variable $x_v$ of $\widehat{b'_2}(x_2, x_v)$ (using $\neq_2$), and then we use the resulting binary signature to connect variables $x_1$ and $x_2$ of $\widehat{f'}$ (edges are connected using $\neq_2$). 
Note that connecting $x_u$ of $\widehat{b'_1}(x_1, x_u)$ with $x_v$ of $\widehat{b'_2}(x_2, x_v)$ gives $\lambda\cdot \neq_2$ up to a nonzero scalar $\lambda$, and $\lambda\cdot \neq_2$ is unchanged by extending both of its two variables with $\neq_2$ (see Figure~\ref{fig:f^ast-f'}(b)).
Thus, $\widehat{\partial}_{12}\widehat{f^{\ast}}$ is actually realized by merging $x_1$ and $x_2$ of $\widehat{f'}$ (using $\neq_2$) up to a nonzero scalar. 
Thus, we have $\widehat{\partial}_{12}\widehat{f^{\ast}}\sim \widehat{\partial}_{12}\widehat{f'},$ and hence  $\widehat{\partial}_{12}\widehat{f^{\ast}}\in {\mathcal{D}}^{\otimes}$,
by the form (\ref{eqn:lm5-partial12-f}) of $\widehat{\partial}_{12}\widehat{f'}$.

\begin{figure}[!h]
    \centering
    \includegraphics[height=4.8cm]{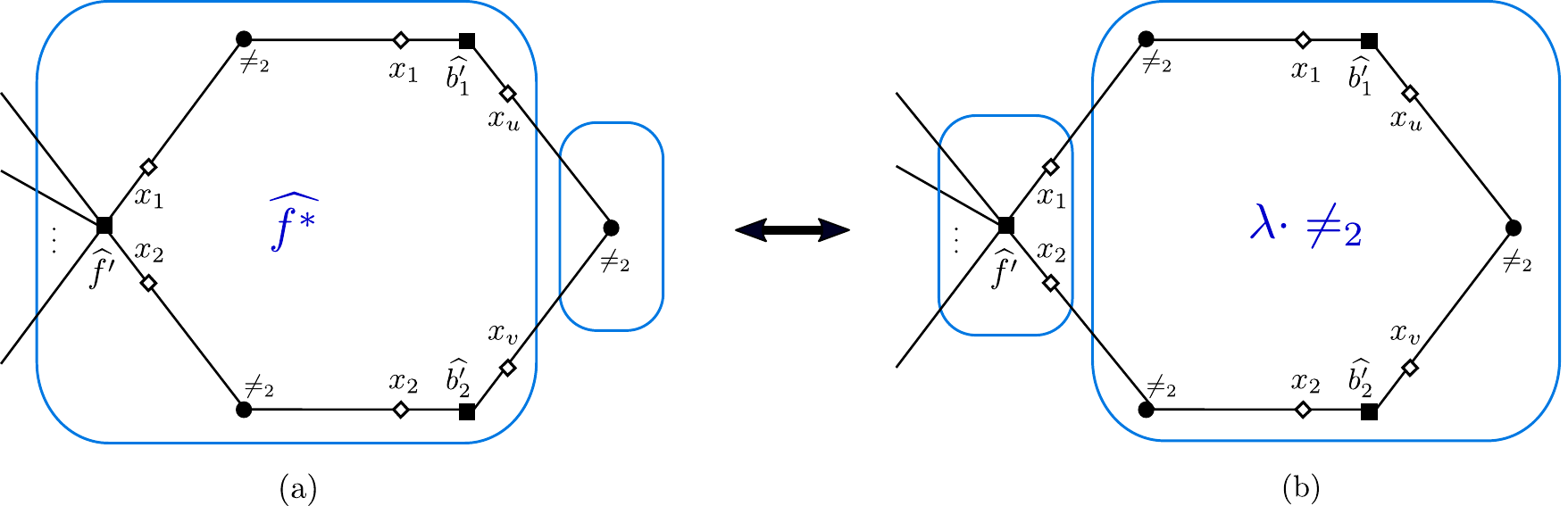}
    \caption{Gadget constructions of $\widehat{\partial}_{12}{\widehat{f^\ast}}$  and $\widehat{\partial}_{12}{\widehat{f'}}$}
    \label{fig:f^ast-f'}
\end{figure}

Then, we show that  $\widehat{\partial}_{st}\widehat{f^{\ast}}\in {\mathcal{D}}^{\otimes}$ for all pairs of indices $\{s, t\}$ disjoint with $\{1, 2, i, j\}$ and $\{s, t\} \neq \{u, v\}$ where $u$ and $v$ are named in (\ref{form_ij}). 
Clearly, $\widehat{\partial}_{st}\widehat{f^{\ast}}\not\equiv 0$ since $\widehat{\partial}_{(st)(12)}\widehat{f^{\ast}}\in{\mathcal{D}}^{\otimes}$.
 We first show that in the UPF of $\widehat{\partial}_{st}\widehat{f^{\ast}}$, $x_1$ and $x_2$ appear in two distinct nonzero binary signatures. 
 Otherwise, for a contradiction, suppose that there is a nonzero binary signature $\widehat{b^\ast}(x_1, x_2)$ such that $\widehat{b^\ast}(x_1, x_2)\mid \widehat{\partial}_{st}\widehat{f^{\ast}}$. 
 Then, $\widehat{b^\ast}(x_1, x_2)\mid\widehat{\partial}_{(ij)(st)}\widehat{f^{\ast}}=\widehat{\partial}_{(st)(ij)}\widehat{f^{\ast}}\not\equiv 0$. By the form (\ref{form_ij}) of $\widehat{\partial}_{ij}\widehat{f^{\ast}}$,  
 the only way that $x_1$ and $x_2$ can form a nonzero binary signature in $\widehat{\partial}_{(st)(ij)}\widehat{f^{\ast}}$ is that the merging gadget is actually merging $x_u$ and $x_v$. Thus, $\{s, t\}=\{u, v\}$. Contradiction. 
Therefore, for some $i'$ and $j'$, we have 
\begin{equation}\label{form_st}
\widehat{\partial}_{st}\widehat{f^{\ast}}=\widehat{b^\ast_{st1}}(x_1, x_{i'})\otimes \widehat{b^\ast_{st2}}(x_2, x_{j'})\otimes\widehat{h_{st}},
\end{equation}
for some $\widehat{b^\ast_{st1}}(x_1, x_{i'}),  \widehat{b^\ast_{st2}}(x_2, x_{j'}), \widehat{h_{st}} \not \equiv 0$ since  $\widehat{\partial}_{st}\widehat{f^{\ast}}\not\equiv 0$.
  Since $\widehat{h_{st}} \mid \widehat{\partial}_{(12)(st)}\widehat{f^\ast}$ and $\widehat{\partial}_{(12)(st)}\widehat{f^\ast}\in 
    {\mathcal{D}}^{\otimes}$, 
we have $\widehat{h_{st}} \in {\mathcal{D}}^{\otimes}$.
Also, by Lemma \ref{lem-binary-sim},
$\widehat{b^\ast_{st1}}  \sim \widehat{b^\ast_{st2}}$.
For a contradiction,
suppose that $\widehat{\partial}_{st}\widehat{f^{\ast}}\notin {\mathcal{D}}^{\otimes}$, then  $\widehat{b^\ast_{st1}}(x_1, x_{i'}) 
\not\sim (\neq_2)$,
and $\widehat{b^\ast_{st2}}(x_2, x_{j'})  
\not\sim (\neq_2)$.
Consider the signature $\widehat{\partial}_{(st)(ij)}\widehat{f^{\ast}}$. 
Since $\{s, t\}\neq \{u, v\}$,  by the form (\ref{form_ij}) of $\widehat{\partial}_{ij}\widehat{f^{\ast}}$, 
$x_1$ and $x_2$ appear in two binary signatures in the UPF of $\widehat{\partial}_{(st)(ij)}\widehat{f^{\ast}}.$
Remember that $\widehat{\partial}_{(st)(ij)}\widehat{f^{\ast}}=\widehat{\partial}_{(ij)(st)}\widehat{f^{\ast}}$.
By the form (\ref{form_st}) of $\widehat{\partial}_{st}\widehat{f^{\ast}}$, 
if
$\{i', j'\}= \{i, j\}$, 
then, after merging $x_i$ and $x_j$ of $\widehat{\partial}_{st}\widehat{f^{\ast}}$, $x_1$ and $x_2$ will form a new binary  signature in $\widehat{\partial}_{(ij)(st)}\widehat{f^{\ast}}$. Contradiction.
Thus, $\{i', j'\}\neq \{i, j\}$.
Then, when merging $x_i$ and $x_j$ of $\widehat{\partial}_{st}\widehat{f^{\ast}}$,
among $\widehat{b^\ast_{st1}}(x_1, x_{i'})$ and $\widehat{b^\ast_{st2}}(x_2, x_{j'})$, at least one binary signature is untouched. 
Thus, $\widehat{\partial}_{(ij)(st)}\widehat{f^{\ast}}$ has a factor that is not an associate of $\neq_2$. 
A contradiction with  $\widehat{\partial}_{(ij)(st)}\widehat{f^{\ast}}\in {\mathcal{D}}^{\otimes}$,
which is a consequence of (\ref{form_ij}).
Thus, $\widehat{\partial}_{st}\widehat{f^{\ast}}\in {\mathcal{D}}^{\otimes}$.

Then, we show that $\widehat{\partial}_{uv}\widehat{f^{\ast}} \in {\mathcal{D}}^{\otimes}$. 
Recall the form (\ref{form_ij}) of $\widehat{\partial}_{ij}\widehat{f^{\ast}}$. Clearly, $\{u, v\}$ is disjoint with $\{1, 2, i, j\}$.
Also, $\widehat{\partial}_{uv}\widehat{f^{\ast}}\not\equiv 0$ since $\widehat{\partial}_{(ij)(uv)}\widehat{f^{\ast}} \in {\mathcal{D}}^{\otimes}.$
Consider the UPF of $\widehat{\partial}_{uv}\widehat{f^{\ast}}$.
\begin{itemize}
    \item 
If $x_1$ and $x_2$ appear in one nonzero binary signature $\widehat{b^\ast_{uv}}(x_1, x_2)$, then $$\widehat{\partial}_{uv}\widehat{f^{\ast}}=\widehat{b^\ast_{uv}}(x_1, x_2)\otimes \widehat{g_{uv}} ~~~~\text{ for some } \widehat{g_{uv}} \not\equiv 0.$$
Then, we have
 $\widehat{g_{uv}} \sim
 \widehat{\partial}_{(12)(uv)}\widehat{f^{\ast}} \in 
{\mathcal{D}}^{\otimes}$ since $\widehat{\partial}_{12}\widehat{f^{\ast}}\in {\mathcal{D}}^{\otimes}$.
Also, since $\widehat{b^\ast_{uv}}(x_1, x_2)\mid \widehat{\partial}_{(ij)(uv)}\widehat{f^{\ast}}\in {\mathcal{D}}^{\otimes}$, we have $\widehat{b^\ast_{uv}}(x_1, x_2)\in {\mathcal{D}}^{\otimes}$. Hence, $\widehat{\partial}_{uv}\widehat{f^{\ast}}\in {\mathcal{D}}^{\otimes}$.
\item If $x_1$ and $x_2$ appear in two distinct nonzero binary signatures $\widehat{b^\ast_{uv1}}(x_1, x_{i'})$ and $\widehat{b^\ast_{uv2}}(x_2, x_{j'})$, then
 $$\widehat{\partial}_{uv}\widehat{f^{\ast}}=\widehat{b^\ast_{uv1}}(x_1, x_{i'})\otimes \widehat{b^\ast_{uv2}}(x_2, x_{j'})\otimes\widehat{h_{uv}} ~~~~\text{ for some } \widehat{h_{uv}}\not\equiv 0.$$
 
Then, we have $ \widehat{h_{uv}} \in {\mathcal{D}}^{\otimes}$ since $\widehat{\partial}_{(12)(uv)}\widehat{f^{\ast}} \in 
{\mathcal{D}}^{\otimes}$. 
By the form (\ref{form_ij}) of $\widehat{\partial}_{ij}\widehat{f^{\ast}}$, 
after merging variables $x_u$ and $x_v$ of  $\widehat{\partial}_{ij}\widehat{f^{\ast}}$,
variables $x_1$ and $x_2$ form a binary $\neq_2$ in $\widehat{\partial}_{(uv)(ij)}\widehat{f^{\ast}}=\widehat{\partial}_{(ij)(uv)}\widehat{f^{\ast}}$. 
On the other hand,
by the form of $\widehat{\partial}_{uv}\widehat{f^{\ast}}$, the only way that $x_1$ and $x_2$  form a binary after merging two variables in $\widehat{\partial}_{uv}\widehat{f^{\ast}}$ is to merge  $x_{i'}$ and $x_{j'}$. 
Thus, we have $\{i', j'\}=\{i, j\}$. Since $\widehat{f^{\ast}}$ has arity $2n \geqslant 8$, we can find another pair of indices $\{s, t\}$ disjoint with $\{1, 2, i, j, u, v\}$. 
When merging variables $x_s$ and $x_t$ in $\widehat{\partial}_{uv}\widehat{f^{\ast}}$, 
binary signatures $\widehat{b^\ast_{uv1}}(x_1, x_{i'})$ and $\widehat{b^\ast_{uv2}}(x_2, x_{j'})$ are untouched. 
Thus, we have $\widehat{b^\ast_{uv1}}(x_1, x_{i'})\otimes \widehat{b^\ast_{uv2}}(x_2, x_{j'})\mid \widehat{\partial}_{(st)(uv)}\widehat{f^{\ast}}$.
As showed above, we have $\widehat{\partial}_{st}\widehat{f^{\ast}}\in {\mathcal{D}}^{\otimes}$ and then  $\widehat{\partial}_{(st)(uv)}\widehat{f^{\ast}}\in {\mathcal{D}}^{\otimes}$.
Thus, $\widehat{b^\ast_{uv1}}(x_1, x_{i'})\otimes \widehat{b^\ast_{uv2}}(x_2, x_{j'})\in {\mathcal{D}}^{\otimes}$ and then $\widehat{\partial}_{uv}\widehat{f^{\ast}}\in {\mathcal{D}}^{\otimes}$.
\end{itemize}

So far, we have shown that $\widehat{\partial}_{12}\widehat{f^{\ast}}\in {\mathcal{D}}^{\otimes}$, $\widehat{\partial}_{ij}\widehat{f^{\ast}}\in {\mathcal{D}}^{\otimes}$ and $\widehat{\partial}_{st}\widehat{f^{\ast}}\in {\mathcal{D}}^{\otimes}$ for all $\{s, t\}$ disjoint with $\{1, 2, i, j\}$.
If we can further show that $\widehat{\partial}_{ik}\widehat{f^{\ast}}\in {\mathcal{D}}^{\otimes}$ for all $k\neq 1, 2, i, j$, and then symmetrically $\widehat{\partial}_{jk}\widehat{f^{\ast}}\in {\mathcal{D}}^{\otimes}$ for all $k\neq 1, 2, i, j$, then 
$\widehat{\partial}_{st}\widehat{f^{\ast}}\in {\mathcal{D}}^{\otimes}$ for all $\{s, t\}$  disjoint with $\{1, 2\}$.
Thus, by our {Claim},  $\widehat{f^\ast}\in \widehat{\int}\mathcal{D}^{\otimes}$. This will finish the proof of Case 2.


Now we prove $\widehat{\partial}_{ik}\widehat{f^{\ast}}\in {\mathcal{D}}^{\otimes}$ for all $k\neq 1, 2, i, j$. 
Since $\widehat{\partial}_{(ik)(12)}\widehat{f^{\ast}}\in \mathcal{D}^{\otimes}$,
we have  $\widehat{\partial}_{ik}\widehat{f^{\ast}}\not\equiv 0$. 
So we can consider the UPF
of $\widehat{\partial}_{ik}\widehat{f^{\ast}}$. 
\begin{itemize}
    \item 
If $x_1$ and $x_2$ appear in one nonzero binary signature, then  $$\widehat{\partial}_{ik}\widehat{f^{\ast}}=\widehat{b^\ast_{ik}}(x_1, x_2)\otimes \widehat{g_{ik}} ~~~~\text{ for some } \widehat{g_{ik}}\in \mathcal{D}^{\otimes}.$$ 
Here, $\widehat{g_{ik}}\in \mathcal{D}^{\otimes}$
since $\widehat{\partial}_{(ik)(12)}\widehat{f^{\ast}}
\in \mathcal{D}^{\otimes}$.
Since $\widehat{f^\ast}$ has arity $2n \geqslant 8$,
we can pick a pair of indices $\{s, t\}$ disjoint with $\{1, 2, i, j, k\}$, and merge variables $x_s$ and $x_t$ of $\widehat{\partial}_{ik}\widehat{f^{\ast}}$.
Then, $\widehat{b^\ast_{ik}}(x_1, x_2)\mid \widehat{\partial}_{(st)(ik)}\widehat{f^{\ast}}.$
Since $\widehat{\partial}_{st}\widehat{f^{\ast}}\in \mathcal{D}^{\otimes}$,
$\widehat{\partial}_{(st)(ik)}\widehat{f^{\ast}}=\widehat{\partial}_{(ik)(st)}\widehat{f^{\ast}}\in {\mathcal{D}}^{\otimes}$. Thus, $\widehat{b^\ast_{ik}}(x_1, x_2) \in {\mathcal{D}}^{\otimes}$ and then $\widehat{\partial}_{ik}\widehat{f^{\ast}}\in {\mathcal{D}}^{\otimes}$.
\item    If $x_1$ and $x_2$ appear in two nonzero distinct binary signatures,
then $$\widehat{\partial}_{ik}\widehat{f^{\ast}}=\widehat{b^\ast_{ik1}}(x_1, x_p)\otimes \widehat{b^\ast_{ik2}}(x_2, x_q)\otimes\widehat{h_{ik}} ~~~~\text{ for some } \widehat{h_{ik}} \in \mathcal{D}^{\otimes}.$$
Again, here $\widehat{h_{ik}} \in \mathcal{D}^{\otimes}$
since $\widehat{\partial}_{(ik)(12)}\widehat{f^{\ast}}
\in \mathcal{D}^{\otimes}$.
By connecting variables $x_1$ and $x_2$ of $\widehat{\partial}_{ik}\widehat{f^{\ast}}$, $x_p$ and $x_q$ will form a binary disequality up to a nonzero scalar
(this binary signature is disequality
because we know that
 $\widehat{\partial}_{(ik)(12)}\widehat{f^{\ast}}\in \mathcal{D}^{\otimes}$).
By Lemma \ref{lem-binary-sim}, 
as the type of binary signatures,
$\widehat{{b^\ast_{ik1}}} \sim \widehat{b^\ast_{ik2}}.$
Between $x_p$ and $x_q$, at least one of them is not $x_j$; suppose that it is $x_p$. 
We pick a variable $x_r$ in the scope of $\widehat{h_{ik}}$ that is also not $x_j$ (such a  variable $x_r$  exists as $2n \geqslant 8$). Then, by merging $x_p$ and $x_r$ of $\widehat{\partial}_{ik}\widehat{f^{\ast}}$, 
the binary signature $\widehat{b^\ast_{ik2}}(x_2, x_q)$ is untouched. 
Since $\{p, r\}$ is disjoint with $\{1, 2, i, j\}$, we have $\widehat{b^\ast_{ik2}}(x_2, x_q)\mid \widehat{\partial}_{(ik)(pr)}\widehat{f^{\ast}}\in {\mathcal{D}}^{\otimes}$. Thus, we have $\widehat{b^\ast_{ik2}}(x_2, x_q)
\in {\mathcal{D}}^{\otimes}$ and so does
$\widehat{{b^\ast_{ik1}}}(x_1, x_p)$, since we have shown that they are  associates
as the type of binary signatures. 
Thus, $\widehat{\partial}_{ik}\widehat{f^{\ast}}\in {\mathcal{D}}^{\otimes}$.
    \end{itemize}

As remarked earlier,  by symmetry, we also have $\widehat{\partial}_{jk}\widehat{f^{\ast}}\in {\mathcal{D}}^{\otimes}$ for all $k\neq 1, 2, i, j$. Thus, we are done with Case 2.

Thus, an irreducible signature $\widehat{f^\ast}\in \widehat\int\mathcal{D}^{\otimes}$ of arity $2n$ is realized from $\widehat f$.
\end{proof}
\begin{remark}
Since $\widehat{f^\ast}$  is realized from $\widehat f$ by gadget construction,  $\widehat{f^\ast}$ satisfies {\sc ars}
as $\widehat{f}$ does.\end{remark}
     
     
   
   We first give a condition (Lemma~\ref{lem-twice-divide-8}) in which we can quite straightforwardly get the \#P-hardness of $\holant{\neq}{\widehat{f}, \widehat{\mathcal F}}$ by {\sc 2nd-Orth} given $\widehat{f}\in \widehat{\int}\mathcal{D}^{\otimes}$ is  an  irreducible $8$-ary signature.
   
   \begin{lemma}\label{lem-eo-not-2ndorth}
   Let 
   $\widehat{f}=a(1, 0)^{\otimes 2n}
   +\bar{a}(0, 1)^{\otimes 2n}
   + (\neq_2)(x_i, x_j) \otimes \widehat{g_{\rm h}}$ 
   be an irreducible $2n$-ary signature, where $2n\geqslant 4$ and 
$\widehat{g_{\rm h}}$ is a nonzero {\rm EO} signature (i.e., with half-weighted support) of arity $2n-2$. Then, $\widehat{f}$ does not satisfy {\sc 2nd-Orth}.
   \end{lemma}
   \begin{proof}
   By renaming variables, without loss of generality, we may assume that $\{i, j\}=\{1, 2\}$. 
   
   For any input $00\beta\neq \vec{0}^{2n}$ of $\widehat{f}$, 
we have $\widehat{f}(00\beta)=
(\neq_2)(0, 0)\cdot \widehat{g_{\rm h}}(\beta)=0$.
Thus, $$|\widehat{f}_{12}^{00}|^2= \sum_{\beta \in \mathbb{Z}_2^{2n-2}}|\widehat{f}(00\beta)|^2=|\widehat{f}(\vec{0}^{2n})|^2.$$
On the other hand, since both $(\neq_2)(x_1, x_2)$ and $\widehat{g_{\rm h}}$ are nonzero EO signatures,
$(\neq_2)(x_1, x_2)\otimes \widehat{g_{\rm h}}$ is a nonzero EO signature.
Then, we can pick an input $01\gamma\in \mathbb{Z}_2^{2n}$ with  ${\rm wt}(01\gamma)=n$ 
such that $\widehat{f}(01\gamma)=(\neq_2)(0,1)\cdot \widehat{g_{\rm h}}(\gamma)\neq 0.$
Since $\gamma\in \mathbb{Z}_2^{2n-2}$, and ${\rm wt}(\gamma)=n-1\geqslant 1$, 
there exists 
a bit $\gamma_i$ in $\gamma$ 
such that  $\gamma_i =0$. 
Without loss of generality, we may assume that $01\gamma=010\gamma'$. 
Then, $$|\widehat{f}_{13}^{00}|^2\geqslant |\widehat{f}(\vec{0}^{2n})|^2+ |\widehat{f}(010\gamma')|^2> |\widehat{f}(\vec{0}^{2n})|^2=|\widehat{f}_{12}^{00}|^2.$$
Note that the constant
$\lambda$ for the norm squares must be the same  for all index pairs $\{i,j\} \subseteq [2n]$ in order to satisfy {\sc 2nd-Orth}
in Definition~\ref{def:second-order-othor}.
Thus, $\widehat{f}$ does not satisfy {\sc 2nd-Orth}.
   \end{proof}
   
    \begin{lemma}\label{lem-twice-divide-8}
    Let $\widehat{f}\in \widehat{\int} \mathcal{D}^{\otimes}$ be an irreducible 8-ary signature in $\widehat{\mathcal{F}}$.
    If there exists a binary disequality $(\neq_2)(x_i, x_j)$ and two pairs of indices $\{u, v\}$ and $\{s, t\}$ where $\{u, v\}\cap\{s, t\}\neq \emptyset$ such that $(\neq_2)(x_i, x_j)\mid \widehat\partial_{uv}\widehat f$ and $(\neq_2)(x_i, x_j)\mid \widehat\partial_{st}\widehat f$,
    then $\holant{\neq_2}{\widehat{\mathcal{F}}}$ is \#P-hard.
    \end{lemma}
\begin{proof}
For all pairs of indices $\{i, j\}$, since $\widehat{\partial}_{ij}\widehat{f} \in 
{\mathcal{D}}^{\otimes}$,  $\mathscr{S}(\widehat{\partial}_{ij}\widehat{f})$ is on half-weight. 
By Lemma \ref{lem-zero_2}, we have $\widehat{f}(\alpha) =0$ for all ${\rm wt}(\alpha)\neq 0, 4, 8$.
Suppose that $\widehat{f}(\vec{0}^{8})=a$ and by {\sc ars} $\widehat{f}(\vec{1}^{8})=\bar a$. 
We can write $\widehat{f}$ in the following form
$$\widehat{f}=a(1, 0)^{\otimes 8}+\bar{a}(0, 1)^{\otimes 8}+\widehat{f_{\rm h}},$$
where $\widehat{f_{\rm h}}$  is an EO signature of arity $8$.

Clearly, $\widehat{\partial}_{ij}\widehat{f}=\widehat{\partial}_{ij}\widehat{f_{\rm h}}$ for all $\{i,j\}.$
Then, $\widehat{f_{\rm h}}\in \widehat{\int} \mathcal{D}^{\otimes}$ since  $\widehat{f}\in \widehat{\int} \mathcal{D}^{\otimes}.$
In addition, since there exists a binary disequality $(\neq_2)(x_i, x_j)$ and two pairs of indices $\{u, v\}$ and $\{s, t\}$ where $\{u, v\}\cap\{s, t\}\neq \emptyset$ such that $(\neq_2)(x_i, x_j)\mid \widehat\partial_{uv}\widehat{f_{\rm h}},  \widehat\partial_{st}\widehat{f_{\rm h}}$,
by Lemma \ref{lem-eo}, 
$\widehat{f_{\rm h}}\in \mathcal{D}^{\otimes}$ and $(\neq_2)(x_i, x_j)\mid \widehat{f_{\rm h}}$.
Thus, 
$$\widehat{f}=a(1, 0)^{\otimes 8}+\bar{a}(0, 1)^{\otimes 8}+(\neq_{2})(x_i, x_j)\otimes \widehat{g_{\rm h}},$$
where $\widehat{g_{\rm h}}\in \mathcal{D}^{\otimes}$ is a nonzero EO signature or arity $6$ since $\widehat{f_{\rm h}}\in \mathcal{D}^{\otimes}$.
By Lemma~\ref{lem-eo-not-2ndorth},
$\widehat{f}$ does not satisfy {\sc 2nd-Orth}.
  Thus, $\holant{\neq_2}{\widehat{\mathcal{F}}}$ is \#P-hard by Lemma \ref{second-ortho}. 
\end{proof}

 For signatures in $\mathcal{D}^{\otimes}$, we give the following property. 
   Now we adopt the following notation for brevity. We use $(i, j)$ to denote the binary disequality $(\neq_2)(x_i, x_j)$ on variables $x_i$ and $x_j$.
    \begin{lemma}\label{lem-divide-twice}
    Let $\widehat{f}\in \mathcal{D}^{\otimes}$ be a signature of arity at least $6$.
    If there exist $\{u, v\}\neq \{s, t\}$ such that $(i, j)\mid \widehat \partial_{uv}\widehat{f}$ and $(i, j)\mid \widehat \partial_{st}\widehat{f}$, then 
    $(i, j)\mid \widehat{f}$.
    \end{lemma}
    \begin{proof}
    For a contradiction, suppose that $(i, j)\nmid \widehat{f}$. Thus $x_i$ and $x_j$ appear in two separate disequalities in the UPF of 
    $\widehat f$.
    Since $\widehat f\in \mathcal{D}^{\otimes}$, there exists $\{\ell, k\}$ such that $(i, \ell)\otimes(j, k)\mid \widehat f$.
    By merging two variables of $\widehat f$, 
    the only way to make $x_i$ and $x_j$ to form a binary disequality is by merging $x_\ell$ and $x_k$.
    By the  hypothesis of the lemma, $\{\ell, k\}=\{u, v\}=\{s, t\}$.
    Contradiction.
    \end{proof}

      \begin{theorem}\label{thm-f8}
Let $\widehat{f}\notin \widehat{\mathcal{O}}{^\otimes}$ be a signature of arity $8$ in $\widehat{\mathcal{F}}$.
 Then
\begin{itemize}
   \item  $\holant{\neq_2}{\widehat{\mathcal{F}}}$ is \#P-hard, or
  \item there exists some $\widehat{Q}\in \widehat{{\bf O}_2}$ such that $\holant{\neq_2}{\widehat{f}_8, \widehat Q\widehat{\mathcal{F}}}\leqslant_T\holant{\neq_2}{\widehat{\mathcal{F}}}.$
\end{itemize}
      \end{theorem}
\begin{proof}
By Lemma~\ref{lem-even>8}, we may assume that an irreducible signature $\widehat{f^\ast}$ of arity $8$ where $\widehat{f^\ast}\in \widehat{\int}\mathcal{D}^{\otimes}$ is realizable from $\widehat{f}$, and $\widehat{f^\ast}$ also satisfies {\sc ars}. 
Otherwise,  $\holant{\neq_2}{\widehat{\mathcal{F}}}$ is \#P-hard or we can realize a signature $\widehat{g}\notin \widehat{\mathcal{O}}^{\otimes}$ of arity $2, 4$ or $6$.
Then, by Lemmas \ref{lem-2-ary}, \ref{lem-4-ary} and \ref{lem-6-ary}, we get  \#P-hardness.
We will show that $\widehat{f_8}$ is realizable from $\widehat{f^\ast}$, or otherwise we get  \#P-hardness. For brevity of notation, we rename $\widehat{f^\ast}$ by $\widehat{f}$.
We first show that after renaming variables by applying a suitable permutation to $\{1,2, \ldots, 8\}$, 
for all $\{i, j\}\subseteq\{1, 2, 3, 4\}$, $(\ell, k)\mid \widehat{\partial}_{ij}\widehat{f}$ where $\{\ell, k\}=\{1, 2, 3, 4\}\backslash \{i, j\}$. 
   Furthermore, we show that
   either $ \holant{\neq_2}{\widehat{\mathcal{F}}}$ is \#P-hard, or 
   \begin{equation}\label{eqn:divisible56-57-67}
    (5, 6)\mid \widehat \partial_{12}\widehat f, ~~ (5, 7)\mid \widehat \partial_{13} \widehat f, ~~ (6, 7)\mid \widehat \partial_{23} \widehat f, \text { and }~~ (1, 2)\mid \widehat \partial_{56}\widehat{f} \text { ~or~ } (1, 3)\mid \widehat \partial_{56}\widehat{f}.
    \end{equation}

Consider $\widehat\partial_{12}\widehat f$.
Since $\widehat{ f} \in \widehat{\int}\mathcal{D}^{\otimes}$, 
$\widehat\partial_{12}\widehat f\in \mathcal{D}^{\otimes}$.
By renaming variables, without loss of generality, we may assume that 
\begin{equation}\label{f12}
\widehat\partial_{12}\widehat f=\lambda_{12}\cdot(3, 4)\otimes(5, 6)\otimes(7, 8),
\end{equation}
for some $\lambda_{12} \in \mathbb{R} \setminus \{0\}$.
Then, consider $\widehat\partial_{34}\widehat f$.
$\widehat\partial_{56}\widehat f$, and $\widehat\partial_{78}\widehat f$.
There are two cases. 
\begin{itemize}
    \item Case 1.
      $(1, 2)\mid \widehat\partial_{34}\widehat f, \widehat\partial_{56}\widehat f$ and $\widehat\partial_{78}\widehat f.$ 
     Then we can write $\widehat\partial_{56}\widehat{f}=(1, 2)\otimes \widehat h$ for some $\widehat h\in \mathcal{D}^{\otimes}$.
Clearly, $\widehat h\sim \widehat\partial_{(12)(56)}\widehat{f}$.
By the form (\ref{f12}) and commutativity, $\widehat\partial_{(12)(56)}\widehat{f}\sim (3,4)\otimes(7, 8).$
Thus,  $\widehat h\sim (3,4)\otimes(7, 8)$.
Then, for  some $\lambda_{56} \in \mathbb{R} \setminus \{0\}$,
\begin{equation}\label{f561}
\widehat\partial_{56}\widehat f=\lambda_{56}\cdot(1, 2)\otimes(3, 4)\otimes(7, 8).
\end{equation}
Similarly, we have 
\begin{equation*}\label{f781}
\widehat\partial_{78}\widehat f=\lambda_{78}\cdot(1, 2)\otimes(3, 4)\otimes(5, 6),
\end{equation*} 
and  
\begin{equation*}\label{f341}
\widehat\partial_{34}\widehat f=\lambda_{34}\cdot(1, 2)\otimes(5, 6)\otimes(7, 8),
\end{equation*}
for  some $\lambda_{78}, \lambda_{34} \in \mathbb{R} \setminus \{0\}$.

Let $\widehat{g}=(1, 2)\otimes(3, 4)$.
Let $\{i, j\}\subseteq\{1, 2, 3, 4\}$ and $\{\ell, k\}=\{1, 2, 3, 4\}\backslash\{i, j\}$.
If we merge variables $x_i$ and $x_j$ of $\widehat{g}$, i.e., if we form
$\widehat{\partial}_{ij} \widehat{g}$, then clearly variables $x_\ell$ and $x_k$ will form a disequality.
Thus, for all $\{i, j\}\subseteq\{1, 2, 3, 4\}$, $(\ell, k)\mid \widehat{\partial}_{ij}\widehat{g}$.
Then, $(\ell, k)\mid \widehat{\partial}_{ij}\widehat{g}\otimes (7, 8)\sim \widehat\partial_{(ij)(56)}\widehat f$ by (\ref{f561}),  and 
similarly $(\ell, k)\mid \widehat{\partial}_{ij}\widehat{g}\otimes (5, 6)\sim \widehat\partial_{(ij)(78)}\widehat f.$
By Lemma~\ref{lem-divide-twice},   $(\ell, k)\mid \widehat\partial_{ij}\widehat f$.



\item 
Case 2. Among $\widehat\partial_{34}\widehat f$, 
$\widehat\partial_{56}\widehat f$, and $\widehat\partial_{78}\widehat f$, there is at least one signature that is not divisible by $(1, 2)$.
Without loss of generality, suppose that $(1, 2)\nmid \widehat\partial_{56}\widehat f$.
Since $\widehat\partial_{56}\widehat f\in \mathcal{D}^{\otimes}$, there exists $\{u, v\}$
disjoint from $\{1, 2, 5, 6\}$ such that $(1, u)\otimes(2, v)\mid \widehat\partial_{56}\widehat f$.
Then, by merging variables $x_1$ and $x_2$ of $\widehat\partial_{56}\widehat f$,
we have $(u,v) \mid \widehat\partial_{(12)(56)}\widehat f$;
comparing it to  $\widehat\partial_{(56)(12)}\widehat f$ using the form of 
(\ref{f12}) and by unique factorization we get
$\{u, v\}=\{3,4\}$ or $\{7,8\}$.
Without loss of generality (i.e., this is still within the freedom of our naming variables subject to the choices made so far), we may assume that $\{u, v\}=\{3, 4\}$ and furthermore, $u=3$ and $v=4$.
Then,  for  some $\lambda'_{56} \in \mathbb{R} \setminus \{0\}$,
\begin{equation}\label{f562}
\widehat\partial_{56}\widehat f=\lambda'_{56}\cdot(1, 3)\otimes(2, 4)\otimes(7, 8).
\end{equation}
Then, consider $\widehat\partial_{78}\widehat f$. 
We show that $(5, 6)\mid \widehat\partial_{78}\widehat f$.
Otherwise, there exists $\{s, t\}$
disjoint from $\{5, 6, 7, 8\}$ such that $(5, s)\otimes(6, t)\mid \widehat\partial_{78}\widehat f$.
By merging two variables of $\widehat \partial_{78} \widehat f$, 
    the only way to make $x_5$ and $x_6$ form a binary disequality is to merge $x_s$ and $x_t$.
By the form (\ref{f12}), $(5, 6)\mid \widehat\partial_{(12)(78)}\widehat f$.
    Thus, $\{s, t\}=\{1, 2\}$.
    From
    $(5, s)\otimes(6, t)\mid \widehat\partial_{78}\widehat f$, and $\{s, t\}=\{1, 2\}$ we know that $x_1$ and $x_2$ will form a binary disequality in 
    $\widehat \partial_{(56)(78)} \widehat f$.
    Thus, $(1, 2)\mid \widehat \partial_{(56)(78)} \widehat f$. However, by (\ref{f562})
    $\widehat \partial_{(56)(78)} \widehat f \sim (1, 3)\otimes (2, 4)$.
    This is a contradiction to UPF.
    Thus, $\widehat \partial_{78} \widehat f=(5, 6)\otimes \widehat{g'}$ and $\widehat{g'}\sim \widehat \partial_{(56)(78)} \widehat f\sim (1, 3)\otimes (2, 4).$ 
    Then,  for  some $\lambda'_{78} \in \mathbb{R} \setminus \{0\}$,
    \begin{equation}\label{78-2}
\widehat\partial_{78}\widehat f=\lambda'_{78}\cdot(1, 3)\otimes(2, 4)\otimes(5, 6).
\end{equation}
Let $\{i, j\}\subseteq\{1, 2, 3, 4\}$ and $\{\ell, k\}=\{1, 2, 3, 4\}\backslash\{i, j\}$.
If we merge variables $x_i$ and $x_j$ of $\widehat{g'}$,
which is an associate of $(1, 3)\otimes (2, 4)$, then clearly variables $x_\ell$ and $x_k$ will form a disequality.
Thus, for all $\{i, j\}\subseteq\{1, 2, 3, 4\}$, $(\ell, k)\sim \widehat{\partial}_{ij}\widehat{g'}$.
Then, $(\ell, k)\mid \widehat{\partial}_{ij}\widehat{g'}\otimes (7, 8)\sim \widehat\partial_{(ij)(56)}\widehat f$ (by (\ref{f562})) and $(\ell, k)\mid \widehat{\partial}_{ij}\widehat{g'}\otimes (5, 6)\sim \widehat\partial_{(ij)(78)}\widehat f$  (by (\ref{78-2})).
By Lemma~\ref{lem-divide-twice},   $(\ell, k)\mid \widehat\partial_{ij}\widehat f$.
\end{itemize} 

Thus, in both cases, we have  $(\ell, k)\mid \widehat\partial_{ij}\widehat f$ where $\{i, j\}\sqcup\{\ell, k\}=\{1, 2, 3, 4\}$ is an arbitrary disjoint union of two pairs.
Now, we show that in both cases, (with possibly switching the names $x_7$ and $x_8$, which we are 
still free to do), we can have
\begin{equation}\label{eqn:3-divisibility}
(5, 6)\mid \widehat \partial_{12}\widehat f, ~~(5, 7)\mid\partial_{13} \widehat f,  
~~(6, 7)\mid \widehat \partial_{23} \widehat f.
\end{equation}

Clearly, by the form (\ref{f12}), we have $(5, 6)\mid \widehat \partial_{12}\widehat f$.
Consider $\partial_{13}\widehat f$.
We already know that $(2, 4)\mid \partial_{13}\widehat f$ (in both cases).
If $(5, 6)\mid \widehat \partial_{13}\widehat f$, then since $(5, 6)\mid \widehat \partial_{12}\widehat f$ and $\{1, 2\}\cap\{1, 3\} \ne \emptyset$, by Lemma~\ref{lem-twice-divide-8}, $\holant{\neq_2}{\widehat{\mathcal{F}}}$ is \#P-hard.
Thus,  $(5, 7)\mid \widehat \partial_{13}\widehat f$ or  $(5, 8)\mid \widehat \partial_{13}\widehat f$.
By renaming variables $x_7$ and $x_8$, we may assume that in both cases

\begin{equation}\label{new-13} \widehat \partial_{13}\widehat f=(2, 4)\otimes (5, 7) \otimes (6, 8).
\end{equation}
This renaming will not change  any of the above forms of $\widehat{\partial_{ij}}\widehat{f}$.
Consider $\widehat \partial_{23}\widehat f$.
We already have $(1, 4)\mid\widehat \partial_{23}\widehat f$.
We know $\widehat \partial_{23}\widehat f \in \mathcal{D}^{\otimes}$, and so in its UPF,
$(6, r) \mid \widehat \partial_{23}\widehat f$,
for some $r \in [8] \setminus \{1,2,3,4,6\}$.
If $(5, 6)\mid\widehat \partial_{23}\widehat f$, then since $(5, 6)\mid\widehat \partial_{12}\widehat f$ and $\{1, 2\}\cap\{2, 3\}\neq\emptyset$, by Lemma~\ref{lem-twice-divide-8}, we get \#P-hardness.
If $(6, 8)\mid\widehat \partial_{2 3}\widehat f$, then since $(6, 8)\mid\widehat \partial_{13}\widehat f$ by (\ref{new-13}) and $\{1, 3\}\cap\{2, 3\}\neq\emptyset$, again by Lemma~\ref{lem-twice-divide-8}, we get \#P-hardness.
Thus, we may assume that 
$r = 7$ and $(6, 7)\mid\widehat \partial_{2 3}\widehat f$.
Therefore,  we have established
(\ref{eqn:3-divisibility}) in  both cases.
Furthermore, in  Case 1, we have $(1, 2)\mid \widehat \partial_{56}\widehat{f}$ by form (\ref{f561}), and in Case 2,  we have $(1, 3)\mid \widehat \partial_{56}\widehat{f}$ by form (\ref{f562}).

Now, we show that for any $\alpha\in \mathbb{Z}_2^4$ with ${\rm wt}(\alpha)= 1$, $\widehat f_{1234}^{\alpha}\equiv 0$.
Since $(3,4)\mid \widehat{\partial}_{12}\widehat{f}$, $(\widehat{\partial}_{12}\widehat{f})_{34}^{00}\equiv 0$.
Since $\{1, 2\}$ is disjoint with $\{3, 4\}$, 
\begin{equation}\label{12zero}
    (\widehat{\partial}_{12}\widehat{f})_{34}^{00}=\widehat{\partial}_{12}(\widehat{f}_{34}^{00})=\widehat{f}_{1234}^{0100}+\widehat{f}_{1234}^{1000}\equiv 0.
\end{equation}
Since  $(1,4)\mid \widehat{\partial}_{23}\widehat{f}$, 
\begin{equation}\label{23zero}
   (\widehat{\partial}_{23}\widehat{f})_{14}^{00}=\widehat{\partial}_{23}(\widehat{f}_{14}^{00})=\widehat{f}_{1234}^{0010}+\widehat{f}_{1234}^{0100}\equiv 0. 
\end{equation}
Since  $(1,3)\mid \widehat{\partial}_{24}\widehat{f}$, 
\begin{equation}\label{13zero}
    (\widehat{\partial}_{13}\widehat{f})_{24}^{00}=\widehat{\partial}_{13}(\widehat{f}_{24}^{00})=\widehat{f}_{1234}^{0010}+\widehat{f}_{1234}^{1000}\equiv 0.
\end{equation}
Comparing (\ref{12zero}), (\ref{23zero}) and (\ref{13zero}), we have 
$$\widehat{f}_{1234}^{1000}=\widehat{f}_{1234}^{0100}=\widehat{f}_{1234}^{0010}\equiv 0.$$
Since  $(2,3)\mid \widehat{\partial}_{14}\widehat{f}$, 
\begin{equation*}
   (\widehat{\partial}_{14}\widehat{f})_{23}^{00}=\widehat{\partial}_{14}(\widehat{f}_{23}^{00})=\widehat{f}_{1234}^{0001}+\widehat{f}_{1234}^{1000}\equiv 0. 
   \end{equation*}
   Plug in $\widehat{f}_{1234}^{1000}\equiv 0$,  we have $\widehat{f}_{1234}^{0001}\equiv0$.
   Thus for any  $\alpha\in \mathbb{Z}_2^4$ with ${\rm wt}(\alpha)=1$, we have
   $\widehat f_{1234}^{\alpha} \equiv 0$.
   
   Also, for  $\alpha\in \mathbb{Z}_2^4$ with ${\rm wt}(\alpha)=3$ and any $\beta \in \mathbb{Z}_2^4$, 
   by {\sc ars} we have,  $$\widehat f_{1234}^{\alpha}(\beta)=\overline{\widehat f_{1234}^{\overline{\alpha}}(\overline{\beta})} = 0$$ since ${\rm wt}(\overline{\alpha})=1.$
   Thus, for any  $\alpha\in \mathbb{Z}_2^4$ with ${\rm wt}(\alpha)=3$, we also have
   $\widehat f_{1234}^{\alpha} \equiv 0$.
   
   Let $\alpha\in \mathbb{Z}_2^4$ be an assignment of the first four variables of $f$, and $\beta\in \mathbb{Z}_2^4$ be an assignment of the last four variables of $f$.
Thus, for any $\alpha, \beta \in \mathbb{Z}_2^4$,
$\widehat f(\alpha\beta)=0$ if ${\rm wt}(\alpha)= 1$ or $3$.
 Also, 
 since $\widehat{ f} \in \widehat{\int}\mathcal{D}^{\otimes}$, 
 by Lemma~\ref{lem-zero_2}, $\widehat f(\alpha\beta)=0$ if ${\rm wt}(\alpha)+{\rm wt}(\beta)\neq 0, 4$ and $8$.
 Then, we show that for any $\alpha\beta\in \mathscr{S}(\widehat f)$, 
 $$|\widehat f(\alpha\beta)|=|\widehat f(\overline \alpha\beta)|=|\widehat f(\alpha\overline\beta)|=|\widehat f(\overline\alpha\overline\beta)|.$$
 By {\sc ars}, $|\widehat f(\alpha\beta)|=|\widehat f(\overline\alpha\overline\beta)|$ and $|\widehat f(\overline \alpha\beta)|=|\widehat f(\alpha\overline\beta)|$.
 So, we only need to show that 
 \begin{equation}\label{eqn:alpha-beta-cross}
 |\widehat f(\alpha\beta)|=|\widehat f(\alpha\overline\beta)|.
 \end{equation}

 Pick an arbitrary $\{i, j\}\subseteq\{1, 2, 3, 4\}$ and an arbitrary $\{u, v\}\subseteq\{5, 6, 7, 8\}$.
 Let $\{\ell, k\}=\{1, 2, 3, 4\}\backslash \{i, j\}$ and $\{s, t\}=\{5, 6, 7, 8\}\backslash\{u, v\}$.
 Since $\widehat f$ satisfies {\sc 2nd-Orth}, by equation~(\ref{e5}), we have 
 $|\widehat{\bf f}_{ijuv}^{0000}|^2=|\widehat{\bf f}_{ijuv}^{0011}|^2.$
 Since $\widehat f(\alpha\beta)=0$ if  ${\rm wt}(\alpha)= 1$ or $3$, or ${\rm wt}(\alpha)+{\rm wt}(\beta)\neq 0, 4$ and $8$,
we get the equation, 
 \begin{equation}\label{equ-8ary-1}
|\widehat f_{ij\ell kuvst}^{00000000}|^2+|\widehat f_{ij\ell kuvst}^{00110011}|^2=|\widehat f_{ij\ell kuvst}^{00111100}|^2+|\widehat f_{ij\ell kuvst}^{00001111}|^2.
 \end{equation}
 Note that for $|\widehat{\bf f}_{ijuv}^{0000}|^2$, since
 we set $x_ix_j=00$, the only possible nonzero terms
 are for $x_\ell x_k =00$ or $11$; furthermore,  as we also
 set $x_ux_v=00$, then $x_sx_t = 00$
 if $x_\ell x_k=00$, and $x_sx_t = 11$
 if $x_\ell x_k=11$. The situation is
 similar for $|\widehat{\bf f}_{ijuv}^{0011}|^2$.

 Also, by considering $|\widehat{\bf f}_{ijst}^{0000}|^2=|\widehat{\bf f}_{ijst}^{0011}|^2,$ we have 
  \begin{equation}\label{equ-8ary-2}
|\widehat f_{ij\ell kuvst}^{00000000}|^2+|\widehat f_{ij\ell kuvst}^{00111100}|^2=|\widehat f_{ij\ell kuvst}^{00110011}|^2+|\widehat f_{ij\ell kuvst}^{00001111}|^2.
 \end{equation}
 Comparing equations (\ref{equ-8ary-1}) and (\ref{equ-8ary-2}), 
 we have 
 $$|\widehat f_{ij\ell kuvst}^{00000000}|^2=|\widehat f_{ij\ell kuvst}^{00001111}|^2, \text{ ~ and ~ } |\widehat f_{ij\ell kuvst}^{00110011}|^2=|\widehat f_{ij\ell kuvst}^{00111100}|^2.$$
 Also, by {\sc ars}, $$|\widehat f_{ij\ell kuvst}^{11111111}|^2=|\widehat f_{ij\ell kuvst}^{11110000}|^2.$$
 As $(i, j, k, \ell)$ is an arbitrary permutation of $(1, 2, 3, 4)$ and $(u, v, s, t)$ is an arbitrary permutation of $(5, 6, 7, 8)$, and
 $\widehat f(\alpha\beta)$ vanishes if ${\rm wt}(\alpha)+{\rm wt}(\beta)\neq 0, 4$ and $8$, the above have
 established (\ref{eqn:alpha-beta-cross})  for any $\alpha, \beta \in \mathbb{Z}_2^4$. 
 Hence, for all 
 $\alpha, \beta \in \mathbb{Z}_2^4$,
  $$|\widehat f(\alpha\beta)|=|\widehat f(\overline \alpha\beta)|=|\widehat f(\alpha\overline\beta)|=|\widehat f(\overline\alpha\overline\beta)|.$$

Note that  $\widehat{f}$ has at most $4+{4\choose 2}\cdot{4\choose 2}=40$ many possibly non-zero entries. 
  In terms of norms, these $40$ entries can be represented by  $\widehat f^{\vec{0}^8}$ and the following 9 entries in Table~\ref{tab:9-entry}.
  In other words, for every $\alpha\beta\in\mathbb{Z}_2^8$ where ${\rm wt}(\alpha)\equiv{\rm wt}(\beta)\equiv 0 \pmod{2}$ and ${\rm wt}(\alpha)+{\rm wt}(\beta)\equiv 0 \pmod{4}$, 
  exactly one entry among $\widehat f(\alpha\beta)$, $\widehat f(\overline\alpha\beta)$, $\widehat f(\alpha\overline\beta)$ and $\widehat f(\overline\alpha\overline\beta)$ appears in Table~\ref{tab:9-entry}.
 We also view these 9 entries in Table~\ref{tab:9-entry} as a 3-by-3 matrix denoted by $M=(m_{ij})_{i, j=1}^3$. 
 
 \begin{table}[!h]
\renewcommand{\arraystretch}{2}
    \centering
    \begin{tabular}{|l|c|c|c|c|}
    \hline
    \diagbox[width=1.8in, height=7ex]{$x_1x_2x_3x_4$}{$x_5x_6x_7x_8$} &  $\alpha_1=0110$ (Col 1) &    $\alpha_2=1010$ (Col 2)&  $\alpha_3=1100$ (Col 3)\\
    \hline
  $\alpha_1=0110$ (Row 1)& $m_{11}=\widehat{f}^{01100110}$      & $m_{12}=\widehat{f}^{01101010}$ & $m_{13}=\widehat{f}^{01101100}$  \\
            \hline
     $\alpha_2=1010$ (Row 2)& $m_{21}=\widehat{f}^{10100110}$      & $m_{22}=\widehat{f}^{10101010}$ & $m_{23}=\widehat{f}^{10101100}$  \\
     \hline
   $\alpha_3=1100$ (Row 3)& $m_{31}=\widehat{f}^{11000110}$      & $m_{32}=\widehat{f}^{11001010}$ & $m_{33}=\widehat{f}^{11001100}$  \\
   \hline
    \end{tabular}
    \caption{Representative entries of $\widehat{f}$ in terms of norms}
    \label{tab:9-entry}
\end{table}
 
 Let $\widehat f^{\vec{0}^8}=a$.
 First we show that 
 \begin{equation}\label{equ-row}
     |m_{i,1}|^2+|m_{i,2}|^2+|m_{i,3}|^2=|a|^2, \text{ ~ for } i=1, 2, 3.
 \end{equation}
 and 
 \begin{equation}\label{equ-column}
     |m_{1,j}|^2+|m_{2,j}|^2+|m_{3,j}|^2=|a|^2, \text{ ~ for } j=1, 2, 3.
 \end{equation}
  Let $(i, j, k)$ be an arbitrary permutation of $(1, 2, 3)$.
  Again, by equation (\ref{e5}), 
  $|{\bf \widehat{f}}^{0110}_{ijk8}|^2=|{\bf \widehat{f}}^{0000}_{ijk8}|^2$. Then, we have 
  $$|\widehat{f}^{01100110}_{ijk45678}|^2+|\widehat{f}^{01101010}_{ijk45678}|^2+|\widehat{f}^{01101100}_{ijk45678}|^2=|\widehat{f}^{00000000}_{ijk45678}|^2=|a|^2.$$
  By taking $(i, j, k)=(1, 2, 3), (2, 1, 3)$ and $(3, 1, 2)$, 
  we get equations (\ref{equ-row}) for $i=1, 2, 3$ respectively. 
  Similarly, by considering  $|{\bf \widehat{f}}^{0110}_{4ijk}|^2=|{\bf \widehat{f}}^{0000}_{4ijk}|^2$ where $(i, j, k)$ is an arbitrary permutation of $(5, 6, 7)$, we get equations (\ref{equ-column}). 
  
  Also, since $(5, 6)\mid \widehat{\partial}_{12}\widehat{f}$, we have $\widehat{\partial}_{12}\widehat{f}(x_3, \ldots, x_8)=0$ if $x_5=x_6$.
  Notice that $$m_{13}+m_{23}=\widehat{f}^{01101100}+\widehat{f}^{10101100}$$ is an entry of $\widehat{\partial}_{12}\widehat{f}$ on the input $101100$. Thus, $m_{13}+m_{23}=0$.
  Also, since $(5, 7)\mid \widehat{\partial}_{13}\widehat{f}$, we have
  $$m_{12}+m_{32}=0.$$ Since $(6, 7)\mid \widehat{\partial}_{23}\widehat{f}$,  we have $$m_{21}+m_{31}=0.$$
  Let $x=|m_{13}|=|m_{23}|$, $y=|m_{12}|=|m_{32}|$, and $z=|m_{21}|=|m_{31}|.$
  Plug $x$, $y$, $z$ into equations (\ref{equ-row}) and (\ref{equ-column}). 
  We have 
  \begin{equation*}
      \begin{aligned}
       &|m_{11}|^2+y^2+x^2=&|m_{11}|^2+z^2+z^2\\
       =&z^2+|m_{22}|^2+x^2=&y^2+|m_{22}|^2+y^2\\
        =&z^2+y^2+|m_{33}|^2=&x^2+x^2+|m_{33}|^2.\\
      \end{aligned}
  \end{equation*}
  Thus, $x=y=z$ and $|m_{11}|=|m_{22}|=|m_{33}|$.
  Consider $$m_{11}+m_{21}=\widehat{f}^{01100110}+\widehat{f}^{10100110} ~~~\text{ and }~~~ m_{12}+m_{22}=\widehat{f}^{01101010}+\widehat{f}^{10101010}.$$
  They are entries of $\widehat \partial_{12}\widehat{f}$ on inputs $100110$ and $101010$.
  By form (\ref{f12}) of $\widehat \partial_{12}\widehat{f}$, 
  we have $$m_{11}+m_{21}=m_{12}+m_{22}\in \mathbb{R}\backslash\{0\}.$$
  Remember that we also have  $(1, 2)\mid \widehat \partial_{56}\widehat{f} \text { or } (1, 3)\mid \widehat \partial_{56}\widehat{f}.$
  
  We first consider the case that $(1, 3)\mid \widehat \partial_{56}\widehat{f}.$
  Then  $$m_{21}+m_{22}=\widehat{f}^{10100110}+\widehat{f}^{10101010}=0.$$
  Thus, 
  $$m_{11}+m_{21}=m_{12}-m_{21}\in \mathbb{R}\backslash\{0\}.$$
  Since $|m_{12}|=|m_{21}|$, $|m_{22}|=|m_{11}|$ and $m_{21}+m_{22}=0,$ 
  $$|m_{12}|=|m_{21}|=|m_{22}|=|m_{11}|.$$ 
  Thus, $m_{11}=\overline{m_{21}}$ and $m_{12}=-\overline{m_{21}}.$
   Let $\mathfrak{Re}(x)$ the real part of a number $x$.
   Then, $$\mathfrak{Re}(m_{11})+\mathfrak{Re}(m_{21})=2\mathfrak{Re}(m_{21})=\mathfrak{Re}(m_{12})-\mathfrak{Re}(m_{21})=-2\mathfrak{Re}(m_{21}).$$
   Thus, $\mathfrak{Re}(m_{21})=0$.
   Then, $\mathfrak{Re}(m_{11})=\mathfrak{Re}(m_{21})=0$.
   Thus, $m_{11}+m_{21}\notin \mathbb{R}\backslash\{0\}$ since $\mathfrak{Re}(m_{11} + m_{21})=0.$
   Contradiction.
   
   

  Now, we consider the case that $(1, 2)\mid \widehat \partial_{56}\widehat{f}$. Then 
  $$m_{31}+m_{32}=\widehat{f}^{11000110}+\widehat{f}^{11001010}=0.$$
  Since $m_{12}+m_{32}=0$ and $m_{21}+m_{31}=0$, we have $m_{12}=-m_{21}.$
 Thus,  we have $$m_{11}+m_{21}=m_{12}+m_{22}=m_{22}-m_{21}\in \mathbb{R}\backslash\{0\}.$$
 Taking the imaginary part,
 $\mathfrak{Im}(m_{11}) + \mathfrak{Im}(m_{21}) =
 \mathfrak{Im}(m_{22}) - \mathfrak{Im}(m_{21})  =0$.
 Adding the two, we get $\mathfrak{Im}(m_{11}) +  \mathfrak{Im}(m_{22}) =0$,
and  thus, $m_{11}+m_{22}\in \mathbb{R}$.
  Since $|m_{11}|=|m_{22}|$, $m_{11}=\overline{{m_{22}}}.$
  Then,
     $\mathfrak{Re}(m_{11})=\mathfrak{Re}(m_{22}).$
   Also, since $m_{11}+m_{21}=m_{22}-m_{21}\in \mathbb{R}\backslash\{0\}$, $$\mathfrak{Re}(m_{11})+\mathfrak{Re}(m_{21})=\mathfrak{Re}(m_{22})-\mathfrak{Re}(m_{21})=\mathfrak{Re}(m_{11})-\mathfrak{Re}(m_{21}) \ne 0.$$
   Thus, $\mathfrak{Re}(m_{21})=0$,
   and $\mathfrak{Re}(m_{11})\not =0$.
   Suppose that $m_{21}=d\ii$ for some $d\in \mathbb{R}$.
   Then there exists $c\in \mathbb{R}\backslash\{0\}$ such that $m_{11}=c-d\ii$ and then $m_{22}=c+d \ii.$
   Remember that $m_{21}+m_{31}=0$. Thus, $m_{31}=-d\ii$.
   Consider $$m_{11}+m_{31}=\widehat{f}^{01100110}+\widehat{f}^{11000110}=c-2d\ii.$$
   It is an entry of the signature $\widehat{\partial}_{13}\widehat{f}$.
   Since $\widehat{\partial}_{13}\widehat{f}\in \mathcal{D}^{\otimes}$,
   $c-2d\ii\in \mathbb{R}$.
   Thus, $d=0$.
   Then, $m_{21}=0$ and $m_{11}\in \mathbb{R}$.
   Thus, $$x=|m_{13}|=|m_{23}|=y=|m_{12}|=|m_{32}|=z=|m_{21}|=|m_{31}|=0,$$
   and $$|m_{11}|=|m_{22}|=|m_{33}|=|a|=|\widehat{f}(\vec{0})|.$$
   Since $\widehat{f}\not\equiv 0$, $a\neq 0$.
   Thus, $$\mathscr{S}(\widehat f)=\{\delta\delta, \delta\overline\delta, \overline{\delta}\delta, \overline\delta\overline\delta\in \mathbb{Z}_2^8\mid \delta=0000, \alpha_1, \alpha_2, \alpha_3\},$$
   where $\alpha_1, \alpha_2, \alpha_3$ are named in Table~\ref{tab:9-entry}.
   It is easy to see that $\mathscr{S}(\widehat f)=\mathscr{S}(\widehat f_8)$.
   Since $m_{11}\in \mathbb{R}$, and $|m_{11}| =|a| \not =0$, we can normalize it to $1$.
  Since,  $\widehat{\partial}_{12}{\widehat f}\in \mathcal{D}^{\otimes}$,
  we have $$1=\widehat f(\alpha_1\alpha_1)+\widehat f(\alpha_2\alpha_1)=\widehat f(\alpha_1\alpha_2)+\widehat f(\alpha_2\alpha_2)=\widehat f(\alpha_1\overline{\alpha_1})+\widehat f(\alpha_2\overline\alpha_1)=\widehat f(\alpha_1\overline{\alpha_2})+\widehat f(\alpha_2\overline\alpha_2).$$
  Since, $\widehat f(\alpha_2\alpha_1)=\widehat f(\alpha_1\alpha_2)=\widehat f(\alpha_2\overline\alpha_1)=\widehat f(\alpha_1\overline{\alpha_2})=0$,
  $$\widehat f(\alpha_1\alpha_1)=\widehat f(\alpha_2\alpha_2)=\widehat f(\alpha_1\overline\alpha_1)=\widehat f(\alpha_2\overline{\alpha_2})=1.$$
  Similarly, since $\widehat{\partial}_{13}\widehat{ f}\in \mathcal{D}^{\otimes}$, 
  $$\widehat f(\alpha_1\alpha_1)=\widehat f(\alpha_3\alpha_3)=\widehat f(\alpha_1\overline\alpha_1)=\widehat f(\alpha_3\overline{\alpha_3})=1.$$
  By {\sc ars}, we have $$1=\overline{\widehat f(\alpha_1\alpha_1)}=\widehat f(\overline{\alpha_1}\overline{\alpha_1})=\widehat f(\overline{\alpha_1}{\alpha_1})=\widehat f(\overline{\alpha_2}\overline{\alpha_2})=\widehat f(\overline{\alpha_2}{\alpha_2})=\widehat f(\overline{\alpha_3}\overline{\alpha_3})=\widehat f(\overline{\alpha_3}{\alpha_3}).$$
  Also, since $\widehat{\partial}_{15}\widehat{ f}\in \mathcal{D}^{\otimes}$, 
  $$1=\widehat f(\alpha_1\overline{\alpha_1})=\widehat f^{01101001}+\widehat f^{11100001}=\widehat f^{00001111}+\widehat f^{10000111}=\widehat f^{00001111}.$$
  Then, by {\sc ars}, $\widehat f^{11110000}=\overline{\widehat f^{00001111}}=1.$
  Thus, $\widehat{f}(\gamma)=1$ for any $\gamma\in \mathscr{S}(\widehat f)$ with ${\rm wt}(\gamma)=4$.
  Remember that $\widehat{f}(\vec{0}^{8})=a$ where $|a|=1$. Then, $\widehat{f}(\vec{1}^{8})=\overline{a}$ by {\sc ars}.
  Suppose that $a=e^{\ii \theta}$.
  Let $\widehat{Q}=\left[\begin{smallmatrix}
  \rho & 0\\
  0 & \overline{\rho}\\ 
  \end{smallmatrix}\right]\in \widehat{\mathbf{O}_2}$ where $\rho=e^{-\ii \theta/8}.$
  Consider the holographic transformation by $\widehat{Q}$.
  $\widehat{Q}$ does not change the entries of $\widehat{f}$ on half-weighed inputs, but change the values of $\widehat{f}(\vec{0}^{8})$ and $\widehat{f}(\vec{1}^{8})$ to 1.
  Thus, $\widehat{Q}\widehat{f}=\widehat{f_8}$.
  Then, $\holant{\neq_2}{\widehat{f_8}, \widehat{Q}\widehat{\mathcal{F}}}\leqslant_T\holant{\neq_2}{\widehat{\mathcal{F}}}.$
\end{proof}
     Now, we want to show that $\holant{\neq_2}{\widehat{f}_8, \widehat Q\widehat{\mathcal{F}}}$ is \#P-hard for all $\widehat{Q}\in \widehat{{\bf O}_2}$ and all $\widehat{\mathcal{F}}$ where $\mathcal{F}=Z\widehat{\mathcal{F}}$ is a real-valued signature set that does not satisfy condition (\ref{main-thr}). 
If so, then we are done.
Recall that for all $\widehat{Q}\in \widehat{{\bf O}_2}$, $\widehat{Q}\widehat{\mathcal{F}}=\widehat{Q\mathcal{F}}$ for some $Q\in {\bf O}_2$.
Moreover, for all $Q\in {\bf O}_2$, and all real-valued $\mathcal{F}$ that 
does not satisfy condition (\ref{main-thr}), 
$Q\mathcal{F}$ is also a real-valued signature set 
that does not satisfy condition (\ref{main-thr}).
Thus, it suffices for us to show that $\holant{\neq_2}{\widehat{f}_8, \widehat{\mathcal{F}}}$ is \#P-hard for all real-valued $\mathcal{F}$ that 
does not satisfy condition (\ref{main-thr}).

 
 The following Lemma shows that $\widehat{f_8}$ gives non-$\widehat{\mathcal{B}}$ hardness (Definition~\ref{def:non-B-hard}).
 
 \begin{lemma}\label{lem-2-notb-8ary}
$\holant{\neq_2}{\widehat{f_8}, \widehat{\mathcal{F}}}$ is \#P-hard if $\widehat{\mathcal{F}}$ contains a nonzero binary signature $\widehat{b}\notin \widehat{\mathcal{B}}^{\otimes}$.
Equivalently, $\Holant({f_8}, {\mathcal{F}})$ is \#P-hard if ${\mathcal{F}}$ contains a nonzero binary signature ${b}\notin {\mathcal{B}}^{\otimes}$.
\end{lemma}
\begin{proof}
We prove this lemma in the setting of $\holant{\neq_2}{\widehat{f_8}, \widehat{\mathcal{F}}}$.
If $\widehat{b}\notin \widehat{\mathcal{O}}^{\otimes}$, then by Lemma~\ref{lem-2-ary}, we get \#P-hardness.
Thus, we may assume that $\widehat{b}\in \widehat{\mathcal{O}}^{\otimes}.$
Then, $\widehat{b}$ has parity. 
We first consider the case that $\widehat{b}$ has even parity, i.e., $\widehat{b}=(a, 0, 0, \bar{a})$.
Since $\widehat{b}\not\equiv 0$, $a\neq 0$. We can normalize $a$ to $e^{\ii\theta}$ where $0 \leqslant \theta<\pi$.
Then $\bar a=e^{-\ii \theta}$.
Since $\widehat{b}\notin \widehat{\mathcal{B}}$, $a\neq \pm 1$ and $a\neq \pm \ii$. Thus,  $\theta \neq 0$ and $\theta \neq \frac{\pi}{2}$. 

We connect variables $x_1$ and $x_5$ of $\widehat{f_8}$ with the two variables of $\widehat{b}$ (using $\neq_2$), and we get a $6$-ary signature denoted by $\widehat{g}$. We rename variables $x_2, x_3, x_4$ of $\widehat{g}$ to $x_1, x_2, x_3$ and variables $x_6, x_7, x_8$ to $x_4, x_5, x_6$.
Then, $\widehat{g}$ has the following signature matrix 
$$M_{123, 456}(\widehat{g})=\left[
\begin{matrix}
e^{-\ii \theta} & 0 & 0 &0 &0 &0 &0 &0\\
 0 & e^{\ii \theta}& 0 &0 &0 &0 &0 &0\\
 0&  0 & e^{\ii \theta}&0 &0 &0 &0 &0\\
 0& 0&  0 & e^{-\ii \theta}&0 &0 &0 &0 \\
 0& 0& 0&  0 & e^{\ii \theta}&0 &0 &0  \\
 0&  0& 0& 0&  0 & e^{-\ii \theta}&0 &0   \\
 0&  0&  0& 0& 0&  0 & e^{-\ii \theta}&0    \\
 0&  0&  0& 0& 0&  0 & 0 & e^{\ii \theta}   \\
\end{matrix}\right].$$

Now, we show that $\widehat{g}\notin \widehat{\mathcal{O}}^{\otimes}.$
For a contradiction, suppose that $\widehat{g}\in \widehat{\mathcal{O}}^{\otimes}.$
Notice that  $\mathscr{S}(\widehat{g})=\{(x_1, \ldots, x_6)\in \mathbb{Z}^6_2\mid x_1=x_4$, $x_2=x_5$ and $x_3=x_6\}$. 
Then, we can write $\widehat{g}$ as $$\widehat{g}=\widehat{b_1}(x_1, x_4)\otimes\widehat{b_2}(x_2, x_5)\otimes\widehat{b_3}(x_3, x_6),$$
where $\widehat{b_1}=(e^{\ii \theta_1}, 0, 0, e^{-\ii \theta_1})$, $\widehat{b_2}=(e^{\ii \theta_2}, 0, 0, e^{-\ii \theta_2})$ and $\widehat{b_3}=(e^{\ii \theta_3}, 0, 0, e^{-\ii \theta_3}).$
Then notice that $$\widehat{g}^{000000}=e^{-\ii\theta}=\widehat{b_1}(0, 0)\cdot\widehat{b_2}(0, 0)\cdot\widehat{b_3}(0, 0)=e^{\ii(\theta_1+\theta_2+\theta_3)},$$ and 
 $$\widehat{g}^{011011}=e^{-\ii\theta}=\widehat{b_1}(0, 0)\cdot\widehat{b_2}(1, 1)\cdot\widehat{b_3}(1, 1)=e^{\ii(\theta_1-\theta_2-\theta_3)}.$$
By multiplying the above two equations, we have 
 $$e^{-\ii2\theta}=e^{\ii(\theta_1+\theta_2+\theta_3)}\cdot e^{\ii(\theta_1-\theta_2-\theta_3)}=e^{\ii2 \theta_1}.$$
 Also, notice that
 $$\widehat{g}^{001001}=e^{\ii\theta}=\widehat{b_1}(0, 0)\cdot\widehat{b_2}(0, 0)\cdot\widehat{b_3}(1, 1)=e^{\ii(\theta_1+\theta_2-\theta_3)},$$ and 
 $$\widehat{g}^{010010}=e^{\ii\theta}=\widehat{b_1}(0, 0)\cdot\widehat{b_2}(1, 1)\cdot\widehat{b_3}(0, 0)=e^{\ii(\theta_1-\theta_2+\theta_3)}.$$
 By multiplying them, we have 
 $$e^{\ii2\theta}=e^{\ii(\theta_1+\theta_2-\theta_3)}\cdot e^{\ii(\theta_1-\theta_2+\theta_3)}=e^{\ii2 \theta_1}.$$
 Thus, $e^{\ii 2\theta}=e^{-\ii 2\theta}.$
 Then, $e^{\ii 4\theta}=1.$
 Since, $\theta\in[0, \pi)$,  $\theta=0$ or $\frac{\pi}{2}$. Contradiction.
 Thus, $\widehat{g}\notin\widehat{\mathcal{O}}^\otimes$.
 By Lemma~\ref{lem-6-ary}, we get \#P-hardness.
 
Now, suppose that $\widehat{b}$ has odd parity, i.e., $\widehat{b}(y_1, y_2)=(0, e^{\ii\theta}, e^{-\ii \theta}, 0)$ where $\theta\in [0, \pi)$ after normalization.
We still consider the $6$-ary signature $\widehat{g'}$ that is realized by connecting variables $x_1$ and $x_5$ of $\widehat{f_8}$ with the two variables $y_1$ and $y_2$ of $\widehat{b}$ (using $\neq_2$).
Then, after renaming variables, $\widehat{g'}$ has the following signature matrix
$$M_{123, 456}(\widehat{g'})=\left[
\begin{matrix}
 0 & 0 &0 &0 &0 &0 &0  &e^{-\ii \theta} \\
 0 &  0 &0 &0 &0 &0  & e^{\ii \theta}& 0\\
 0&  0 & 0 &0 &0 &e^{\ii \theta} &0 &0\\
 0& 0&  0 &0 & e^{-\ii \theta} &0 &0 &0 \\
 0& 0& 0 & e^{\ii \theta} &  0 &0 &0 &0  \\
 0&  0 &e^{-\ii \theta}& 0& 0&  0 & 0 &0   \\
 0 & e^{-\ii \theta}&  0&  0& 0& 0&  0 &0    \\
  e^{\ii \theta} & 0&  0&  0& 0& 0&  0 & 0    \\
\end{matrix}\right].$$
 Similarly, we can show that $\widehat{g'}\notin \widehat{\mathcal{O}}^{\otimes}$.
 Thus, by Lemma~\ref{lem-6-ary}, we get \#P-hardness.
\end{proof}

  We go back to real-valued Holant problems under the $Z$-transformation. 
 Consider the problem $\Holant(f_8, \mathcal{F})$.
 Remember that $f_8=\widehat{f_8}$.
We observe that, by Lemma~\ref{lem-2-notb-8ary}
the set $\{f_8\} \cup \mathcal{F}$ is non-$\mathcal{B}$ hard,
according to Definition~\ref{def:non-B-hard}.
Then if we apply Theorem~\ref{thm-holantb} to the
set $\{f_8\} \cup \mathcal{F}$ we see that
$\Holantb(f_8, \mathcal{F})$ is \#P-hard.
Now \emph{if 
 we were able to show} that $\mathcal{B}$ is realizable from $f_8$ then we would be done, since by Theorem~\ref{thm-f8}, we either already have the \#P-hardness for $\Holant(\mathcal{F})$,
 or we can realize $f_8$ from $\mathcal{F}$, and thus the following 
 reduction chain holds
 \[\Holant^b(f_8, \mathcal{F}) \leqslant_T 
 \Holant(f_8, \mathcal{F}) \leqslant_T 
 \Holant(\mathcal{F}).\]
 Thus we get the \#P-hardness of $\Holant(\mathcal{F})$
 in either way.

However, since $f_8$ has even parity and all its entries are non-negative, all gadgets realizable from $f_8$ have even parity and have non-negative entries. 
 Thus, $=_2^-$, $\neq_2$ and $\neq_2^-$ \emph{cannot} be realized from $f_8$ by gadget construction. 
 In fact,  it  is observed in \cite{realodd} that $f_8$ satisfies the following strong Bell property.
 \begin{definition}\label{def-strong-bell}
A signature $f$ satisfies the strong Bell property if for all pairs of indices $\{i, j\}$, and every $b\in \mathcal{B}$, the signature 
$\partial_{ij}^b f$  realized by merging $x_i$ and $x_j$ of $f$ using $b$ is in $\{b\}^{\otimes}.$
\end{definition}

\subsection{Holant problems with limited appearance and a novel reduction}
In this subsection, \emph{not using gadget construction} but critically based on the strong Bell property of $f_8$,
we prove that
$\Holantb(f_8, \mathcal{F})\leqslant_T \Holant(f_8, \mathcal{F})$ in a novel way. 
 We define the following Holant problems with limited appearance. 
 
 \begin{definition}
Let $\mathcal{F}$ be a signature set containing a signature $f$. The problem $\Holant(f^{\leqslant k}, \mathcal{F})$ contains all instances of $\Holant(\mathcal{F})$ where the signature $f$ appears at most $k$  times. 
 \end{definition}
 
 \begin{lemma}\label{lem-2toinfty}
 For any $b\in \mathcal{B}$, $\Holant(b, f_8,  \mathcal{F})\leqslant_T\Holant(b^{\leqslant 2}, f_8, \mathcal{F}).$
 \end{lemma}
 \begin{proof}
 Consider an instance $\Omega$ of $\Holant(b, f_8,  \mathcal{F})$. Suppose that $b$ appears $n$  times in $\Omega$.
 If $n\leqslant 2$, then $\Omega$ is already an instance of  $\Holant(b^{\leqslant 2}, f_8, \mathcal{F})$.
 Otherwise, $n\geqslant 3$. 
  Consider the gadget $\partial^b_{ij}f_8$  realized by connecting two variables $x_i$ and $x_j$ of $f_8$ using $b$.
  (This gadget uses $b$ only once.)
 Since $f_8$ satisfies the strong Bell property,
 $\partial^b_{ij}f_8=b^{\otimes 3}$.
 Thus, by replacing three occurrences of $b$ in $\Omega$ by $\partial^b_{ij}f_8$, we can reduce the number of occurrences of $b$ by 2. 
 We carry out this replacement a linear number of times to obtain
  an equivalent instance of $\Holant(b^{\leqslant 2}, f_8, \mathcal{F})$, of size linear in $\Omega$.
 \end{proof}
 
 Now, we are ready to prove the reduction $\Holantb(f_8, \mathcal{F})\leqslant_T \Holant(f_8, \mathcal{F})$. 
 Note that if  $\Holant(f_8, \mathcal{F})$ is \#P-hard, 
then the reduction  holds trivially.
For any $b\in \mathcal{B}$, if we connect a variable of $b$ with a variable of  another copy of $b$ using $=_2$, we get $\pm (=_2)$.
Also, for any $b_1, b_2\in \mathcal{B}$ where $b_1\neq b_2$ 
if we connect the two variables of $b_1$ with the two variables of $b_2$, we get a value 0.

\begin{lemma}\label{lem-0to1}
$\Holantb(f_8, \mathcal{F})\leqslant_T\Holant(f_8, \mathcal{F}).$
\end{lemma}
\begin{proof}
We prove this reduction in two steps.

\vspace{1ex}
\noindent{\bf Step 1.}
There exists a signature $b_1 \in \mathcal{B}\backslash\{=_2\}$ such that $\Holant(b_1, f_8, \mathcal{F})\leqslant_T \Holant(f_8, \mathcal{F}).$
\vspace{1ex}

We consider \emph{all} binary 
and 4-ary signatures realizable by gadget constructions from $\{f_8\}\cup\mathcal{F}$.
If a binary signature $g\notin\mathcal{B}$ is realizable from  $\{f_8\}\cup\mathcal{F}$, 
then by Lemma \ref{lem-2-notb-8ary},  
 $\Holant(f_8, \mathcal{F})$ is \#P-hard, and we are done.
If a binary signature $g\in \mathcal{B}\backslash\{=_2\}$ is realizable from  $\{f_8\}\cup\mathcal{F}$, then we are done by choosing $b_1=g$.
So we may assume that all binary signatures $g$ realizable from $\{f_8\}\cup\mathcal{F}$ are $=_2$ (up to a scalar) or the zero binary signature,
i.e., 
$g=\mu \cdot(=_2)$ for some $\mu \in \mathbb{R}$.
Similarly, if a nonzero 4-ary signature $h\notin \mathcal{B}^{\otimes2}$ is realizable, then we have $\Holant(f_8, \mathcal{F})$ is \#P-hard, by Lemma~\ref{lem-4-notb},
as  Lemma~\ref{lem-2-notb-8ary} says the  set $\{f_8\} \cup \mathcal{F}$ is non-$\mathcal{B}$ hard.
If a nonzero 4-ary signature $h\in \mathcal{B}^{\otimes 2}\backslash\{=_2\}^{\otimes 2}$ is realizable, then we can realize a binary signature $b_1\in \mathcal{B}\backslash\{=_2\}$ by factorization, and we are done.
Thus, we may assume that all 4-ary signatures $h$ realizable from $\{f_8\}\cup\mathcal{F}$ are $(=_2)^{\otimes 2}$ or the 4-ary zero signature, i.e., $h=\lambda \cdot (=_2)^{\otimes 2}$ for some $\lambda \in \mathbb{R}$.

Now, let $b_1$ be a signature in $\mathcal{B}\backslash\{=_2\}$. 
We show that $\Holant(b_1^{\leqslant 2}, f_8, \mathcal{F})\leqslant_T \Holant(f_8, \mathcal{F}).$
Consider an instance $\Omega$ of $\Holant(b_1^{\leqslant 2}, f_8, \mathcal{F})$.
\begin{itemize}
    \item If $b_1$ does not appear in $\Omega$, 
then $\Omega$ is already an instance of $\Holant(f_8, \mathcal{F})$.
\item If $b_1$ appears exactly once in  $\Omega$ (we may assume it does connect to itself), then we may consider
\emph{the rest of  $\Omega$ that connects to $b_1$} as a gadget 
realized from $\{f_8\}\cup\mathcal{F}$, which must have signature
$\lambda \cdot(=_2)$, for some $\lambda \in \mathbb{R}$. Connecting the two variables of  $b_1$ by $(=_2)$ for every $b_1\in \mathcal{B}\backslash\{=_2\}$ will always gives  $0$.
Thus, $\Holant(\Omega)= 0$.
\item Suppose $b_1$ appears exactly twice in  $\Omega$.
It is easy to handle when the two copies of $b_1$ form a gadget of arity 0 or 2
to the rest of $\Omega$.
We may assume they are connected to the rest of $\Omega$ in such a way that the rest of $\Omega$ 
forms a 4-ary gadget $h$ realized from $\{f_8\}\cup\mathcal{F}$.
We can name the  four dangling edges of $h$ in any specific ordering as  $(x_1, x_2, x_3, x_4)$. 
Then $$h(x_1, x_2, x_3, x_4)=\lambda\cdot (=_2)(x_1, x_j)\otimes (=_2)(x_k, x_\ell)$$ for some
partition $\{1, 2, 3, 4\} = \{1, j\}\sqcup\{k, \ell\}$,
and   some $\lambda \in \mathbb{R}$.
(Note that while we have named four specific dangling edges as $(x_1, x_2, x_3, x_4)$,
the specific partition  $\{1, 2, 3, 4\} = \{1, j\}\sqcup\{k, \ell\}$ and the value $\lambda$ are unknown  at this point.)
We consider the following three instances 
$\Omega_{12}$, $\Omega_{13}$, and $\Omega_{14}$, where $\Omega_{1s}$ $(s\in \{2, 3, 4\})$ is the instance formed by merging variables $x_1$ and $x_s$ of $h$ using $=_2$, and merging the other two variables of $h$ using $=_2$
(see Figure~\ref{fig:three-instances} where $h_1=h_2=(=_2)$ and $h=\lambda \cdot h_1\otimes h_2$).
Since $h$ is a gadget realized from $\{f_8\}\cup\mathcal{F}$, 
$\Omega_{12}$, $\Omega_{13}$, and $\Omega_{14}$  are instances of 
$\Holant(f_8, \mathcal{F})$.
Note that $\Holant(\Omega_{1s})=4\lambda$ when $s=j$ and $\Holant(\Omega_{1s})=2\lambda$ otherwise. 
Thus, by computing $\Holant(\Omega_{1s})$ for $s\in \{2, 3, 4\}$, we can get $\lambda$,
and if $\lambda\neq 0$ the partition $\{1, j\}\sqcup\{k, \ell\}$ of the four variables.
Thus we can get the exact structure of the 4-ary gadget $h$.
In either case
(whether $\lambda =0$ or not), we can compute the value of $\Holant(\Omega)$.
\end{itemize}
Thus,  $\Holant(b_1^{\leqslant 2}, f_8, \mathcal{F})\leqslant_T \Holant(f_8, \mathcal{F}).$
By Lemma \ref{lem-2toinfty},  $\Holant(b_1, f_8, \mathcal{F})\leqslant_T \Holant(f_8, \mathcal{F}).$

\vspace{2ex}
\begin{figure}[!h]
    \centering
    \includegraphics[height=4cm]{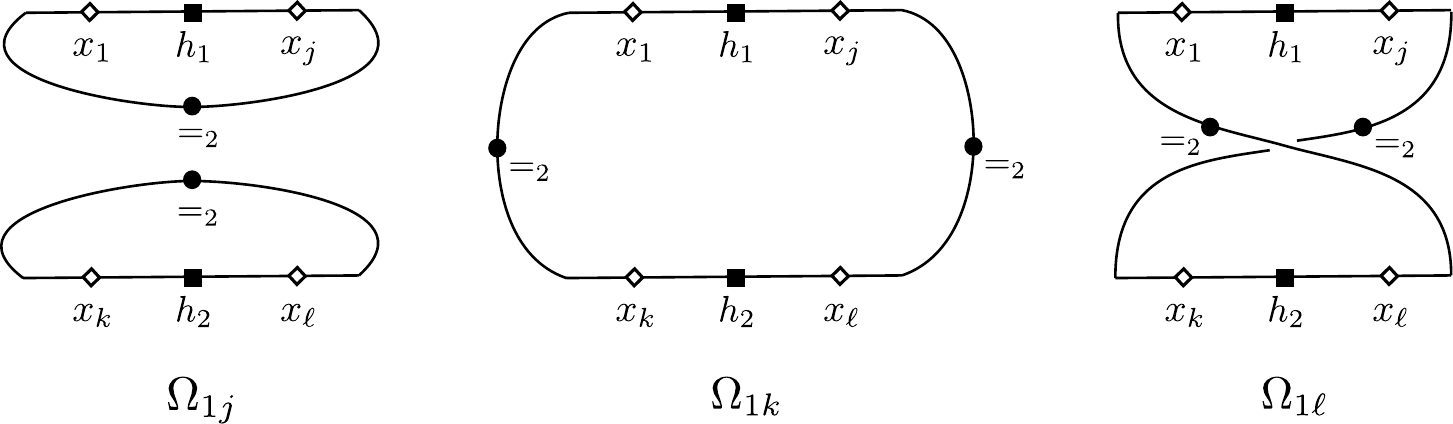}
    \caption{Instances $\Omega_{1j}$, $\Omega_{1k}$ and $\Omega_{1\ell}$}
    \label{fig:three-instances}
\end{figure}

\vspace{1ex}
\noindent{\bf Step 2.}
For any  $b_1\in \mathcal{B}\backslash\{=_2\}$, we have $\Holantb(f_8,\mathcal{F})\leqslant_T\Holant(b_1, f_8, \mathcal{F}).$

We show that we can get another  $b_2\in \mathcal{B}\backslash\{=_2, b_1\}$, i.e.,
for some binary signature $b_2\in \mathcal{B}\backslash\{=_2, b_1\}$ we have the reduction
$\Holant(b_2, b_1, f_8, \mathcal{F})\leqslant_T\Holant(b_1, f_8, \mathcal{F}).$
Then, by connecting one variable of $b_1$ and one variable of $b_2$ using $=_2$, we get the third signature in $\mathcal{B}\backslash\{b_1, b_2\}.$
Then, the lemma is proved.
The proof is similar to the proof in {Step 1}.
 We consider \emph{all} binary  and 4-ary gadgets realizable from $\{b_1, f_8\} \cup \mathcal{F}$. 
Still, we may assume that all realizable binary signatures are of the form $\mu \cdot (=_2)$ or $\mu \cdot b_1$ for some $\mu \in \mathbb{R}$, and all realizable 4-ary signatures are of form $\lambda \cdot (=_2)^{\otimes 2}$, $\lambda \cdot b_1^{\otimes 2}$ or $\lambda \cdot (=_2){\otimes b_1}$ for some $\lambda\in \mathbb{R}$.
Otherwise, we can show that $\Holant(b_1, f_8, \mathcal{F})$ is \#P-hard or we realize a signature $b_2\in \mathcal{B}\backslash\{=_2, b_1\}$ directly by gadget construction. 

Then, let $b_2$ be an arbitrary signature in $\mathcal{B}\backslash\{=_2, b_1\}$. 
We show that $$\Holant(b_2^{\leqslant 2}, b_1, f_8, \mathcal{F})\leqslant_T\Holant(b_1, f_8, \mathcal{F}).$$
Consider an instance $\Omega$ of $\Holant(b_2^{\leqslant 2}, b_1, f_8, \mathcal{F})$.
If $b_2$ does not appear in $\Omega$, then $\Omega$ is already an instance of $\Holant(b_1, f_8, \mathcal{F})$.
If $b_2$ appears exactly once in $\Omega$, then it is connected with a binary gadget $g$ where $g=\mu \cdot (=_2)$ or $g=\mu \cdot b_1$.
In both cases,  the evaluation is $0$. 
Thus, $\Holant(\Omega)=0$.
Suppose  $b_2$ appears exactly twice in $\Omega$.
Again it is easy to handle the case if
 the rest of 
$\Omega$ forms
a gadget of arity 0 or 2 to
 the two occurrences of $b_2$.
So we may assume the two occurrences of $b_2$ are connected to a 4-ary gadget $h= \lambda \cdot (=_2)^{\otimes 2}$, $\lambda \cdot b_1^{\otimes 2}$ or $\lambda \cdot (=_2){\otimes b_1}.$
We denote the four variables of $h$ by $(x_1, x_2, x_3, x_4)$,
by an arbitrary ordering of the four dangling edges. 
Then $h(x_1, x_2, x_3, x_4)=\lambda \cdot h_1(x_1, x_j)\otimes h_2(x_k, x_\ell)$ where $h_1, h_2\in \{=_2, b_1\}$,
for some $\lambda$ and $\{j, k, \ell \} = \{2,3,4\}$.
(Note that at the moment the values $\lambda$ and $j, k, \ell $ are unknown.)
We  consider the following three instances $\Omega_{12}$, $\Omega_{13}$ and $\Omega_{14}$, where $\Omega_{1s}$ $(s\in \{2, 3, 4\})$ is the instance  formed by connecting variables $x_1$ and $x_s$ of $h$ using $=_2$, and connecting the other two variables of $h$ using $=_2$ (again see Figure~\ref{fig:three-instances}).
Clearly, $\Omega_{12}$, $\Omega_{13}$ and $\Omega_{14}$ are instances of 
$\Holant(b_1, f_8, \mathcal{F})$.
Consider the evaluations of these instances.
We have three cases.
\begin{itemize}
    \item If $h_1=h_2=(=_2)$, then  $\Holant(\Omega_{1s})=4\lambda$ when $s=j$ and $\Holant(\Omega_{1s})=2\lambda$ when $s\neq j$.
    \item If $h_1=h_2=b_1$, then 
    $\Holant(\Omega_{1s})=0$ when $s=j$. If $M(b_1)$ is the 2 by 2 matrix 
    form for the binary signature $b_1$ where we list its first variable
    as row index and second variable as column index, then
     we have $\Holant(\Omega_{1k})= \lambda \cdot 
    {\rm tr} (M(b_1) M(b_1)^{\tt T})$,
    and  $\Holant(\Omega_{1 \ell})= \lambda \cdot {\rm tr} (M(b_1)^2)$,
    where ${\rm tr}$ denotes trace. For 
    $b_1 = (=_2^{-})$ or $(\ne_2^+)$,
    the matrix $M(b_1)$ is symmetric, and the value
    $\Holant(\Omega_{1s}) = 2 \lambda$ in both cases $s =k$ or $s = \ell$.
    For $b_1 = (\ne_2^-)$, $M(b_1)^{\tt T} = - M(b_1)$, and we have
    $\Holant(\Omega_{1k}) = 2 \lambda$,
    and $\Holant(\Omega_{1 \ell}) = - 2 \lambda$.
    \item If one of $h_1$ and $h_2$ is $=_2$ and the other is $b_1$, then  $\Holant(\Omega_{1s})=0$ for all $s\in \{j, k, \ell\}.$
\end{itemize}
Thus, if the values of $\Holant(\Omega_{1s})$ for $s\in \{2, 3, 4\}$ are not all zero, 
then $\lambda \ne 0$ and the third case is impossible, and
we can tell whether $h$ is in the form  $\lambda \cdot (=_2)^{\otimes 2} $ or $\lambda \cdot (b_1)^{\otimes 2}$.  Moreover we can get the exact structure of $h$, i.e., the value $\lambda$ and the decomposition
form of $h_1$ and $h_2$. 
Otherwise, the values of $\Holant(\Omega_{1s})$ for $s\in \{2, 3, 4\}$ are all zero. 
Then we can write $h=\lambda\cdot (=_2)(x_1, x_j)\otimes b_1(x_k, x_\ell)$ or $h=\lambda \cdot b_1(x_1, x_j) \otimes (=_2)(x_k, x_\ell)$, including possibly $\lambda = 0$, which means $h \equiv 0$. 
(Note that if 
$\lambda \not = 0$, this uniquely identifies that we are in the third case; if $\lambda = 0$ then this form is still formally valid, even though we cannot  say  this uniquely identifies the third case. But 
when $\lambda = 0$ all three cases are the same, i.e., $h \equiv 0$.)  At this point we still do not know the exact value of 
$\lambda$ and the 
decomposition
form of $h$.

We further consider the following three instances $\Omega'_{12}$, $\Omega'_{13}$ and $\Omega'_{14}$, where $\Omega'_{1s}$ $(s\in \{2, 3, 4\})$ is the instance  formed by connecting variables $x_1$ and $x_s$ of $h$ using $b_1$, and connecting the other two variables of $h$ using $=_2$.
(In other words, we replace the labeling $=_2$ of the edge that is connected to the variable $x_1$ in each instance illustrated in Figure~\ref{fig:three-instances} by $b_1$.)
It is easy to see that $\Omega'_{12}$, $\Omega'_{13}$ and $\Omega'_{14}$ are instances of 
$\Holant(b_1, f_8, \mathcal{F})$.
Consider the evaluations of these instances.
\begin{itemize}
    \item If $h_1=(=_2)(x_1, x_j)$, then $\Holant(\Omega'_{1s})=0$ when $s=j$.
    Also we have $\Holant(\Omega_{1k})= \lambda \cdot 
    {\rm tr} (M(b_1)^2)$, and $\Holant(\Omega_{1 \ell})= 
    \lambda \cdot 
    {\rm tr} (M(b_1) M(b_1)^{\tt T})$.
     For 
    $b_1 = (=_2^{-})$ or $\ne_2^+$,
    the matrix $M(b_1)$ is symmetric, and the value
    $\Holant(\Omega_{1s}) = 2 \lambda$ in both cases $s =k$ or $s = \ell$.
     For $b_1 = (\ne_2^-)$, $M(b_1)^{\tt T} = - M(b_1)$, and we have
    $\Holant(\Omega_{1k}) = -2 \lambda$,
    and $\Holant(\Omega_{1 \ell}) =  2 \lambda$.
    \item If $h_1=b_1(x_1, x_j)$, then $\Holant(\Omega'_{1s})=4\lambda$ when $s=j$ and $\Holant(\Omega'_{1s})=2\lambda$ when $s\neq j$.
\end{itemize}
Thus, by further computing $\Holant(\Omega'_{1s})$ for $s\in \{2, 3, 4\}$, we can get the exact structure of $h$. 

Therefore, by querying $\Holant(b_1, f_8, \mathcal{F})$ at most 6 times, 
we can compute $h$ exactly.
Then, we can compute $\Holant(\Omega)$ easily.
Thus,  $\Holant(b_2^{\leqslant 2}, b_1, f_8, \mathcal{F})\leqslant_T \Holant(b_1, f_8, \mathcal{F}).$
By Lemma~\ref{lem-2toinfty},  $\Holant(b_2, b_1, f_8, \mathcal{F})\leqslant_T \Holant(b_1, f_8, \mathcal{F}).$
The other signature in $\mathcal{B}\backslash\{=_2, b_1, b_2\}$ can be realized by connecting $b_1$ and $b_2$.
Thus,  $\Holantb(f_8, \mathcal{F})\leqslant_T \Holant(b_1, f_8, \mathcal{F}).$

Therefore, $\Holantb(f_8, \mathcal{F})\leqslant_T\Holant(f_8, \mathcal{F}).$
\end{proof}

Since $\Holantb(f_8, \mathcal{F})\leqslant_T\Holant(f_8, \mathcal{F})$ and $\{f_8\}\cup \mathcal{F}$ is non-$\mathcal{B}$ hard for any real-valued $\mathcal{F}$ that does not satisfy condition (\ref{main-thr}),
by Theorem~\ref{thm-holantb},
we have the following result.
\begin{lemma}\label{lem-holantb-8ary}
$\Holant(f_8, \mathcal{F})$ is \#P-hard.
\end{lemma}
Combining Theorem~\ref{thm-f8} and Lemma \ref{lem-holantb-8ary}, we have the following result.
\begin{lemma}\label{lem-8-ary}
If $\widehat{\mathcal{F}}$ contains a signature $\widehat{f}$ of arity $8$ and $\widehat{f}\notin\widehat{\mathcal{O}}{^\otimes}$,
 then $\holant{\neq_2}{\widehat{\mathcal{F}}}$ is \#P-hard.
\end{lemma}

\section{The Induction Proof: $2n\geqslant 10$}\label{sec-10-ary}
Now, we show that our induction framework works for signatures of arity $2n \geqslant 10$. 
\begin{lemma}\label{lem-10-ary}
If $\widehat{\mathcal{F}}$ contains a signature $\widehat{f}$ of arity $2n\geqslant 10$ and $\widehat{f}\notin\widehat{\mathcal{O}}{^\otimes}$,
  then, 
  \begin{itemize}
      \item $\Holant(\neq_2 \mid \widehat{\mathcal F})$ is {\rm \#}P-hard, or
      \item a signature $\widehat g\notin \widehat{\mathcal{O}}^{\otimes}$ of arity $2k\leqslant 2n-2$ is realizable from $\widehat f$.
  \end{itemize}
\end{lemma}


\begin{proof}
By Lemma~\ref{lem-even>8}, we may assume that an irreducible signature $\widehat{f^\ast}$ of arity $2n\geqslant 10$ where
 $\widehat{f^\ast}\in \widehat{\int}\mathcal{D}^{\otimes}$ is realizable, and $\widehat{f^\ast}$ satisfies {\sc ars}.
We show that $\widehat{f^\ast}$ does not satisfy {\sc 2nd-Orth}, and hence we get  \#P-hardness. 

For all pairs of indices $\{i, j\}$, since $\widehat{\partial}_{ij}\widehat{f^\ast} \in 
{\mathcal{D}}^{\otimes}$,  $\mathscr{S}(\widehat{\partial}_{ij}\widehat{f^\ast})$ is on half-weight.
By Lemma \ref{lem-zero_2}, we have $\widehat{f^\ast}(\alpha) =0$ for all ${\rm wt}(\alpha)\neq 0, n, 2n$.
Suppose that $\widehat{f^\ast}(\vec{0}^{2n})=a$ and $\widehat{f^\ast}(\vec{1}^{2n})=\bar a$ by {\sc ars}. 
We can write $\widehat{f^\ast}$ in the following form
$$\widehat{f^\ast}=a(1, 0)^{\otimes 2n}+\bar{a}(0, 1)^{\otimes 2n}+\widehat{f_{\rm h}^\ast}.$$
where $\widehat{f_{\rm h}^\ast}$ is an EO signature of arity $2n\geqslant 10$.
 
%
%

Clearly, $\partial_{ij}\widehat{f^\ast} = \partial_{ij}\widehat{f_{\rm h}^\ast}$ for all $\{i,j\}$.
Then, $\widehat{f_{\rm h}^\ast}
\in \widehat{\int}\mathcal{D}^{\otimes}$
since $\widehat{f^{\ast}}\in \widehat{\int}\mathcal{D}^{\otimes}$.
Since $\widehat{f_{\rm h}^\ast}$  is an EO signature of arity at least $10$ and  $\widehat{f_{\rm h}^\ast}
\in \widehat{\int}\mathcal{D}^{\otimes}$,
by Lemma~\ref{lem-eo}, we have
$\widehat{f_{\rm h}^\ast}
\in \mathcal{D}^{\otimes}$. Recall that all signatures in $\mathcal{D}^{\otimes}$ are nonzero by definition.
Pick some $\{i, j\}$ such that $(\neq_2)(x_i, x_j)\mid \widehat{f_{\rm h}^\ast}.$
Then, $$\widehat{f^\ast}=a(1, 0)^{\otimes 2n}+\bar{a}(0, 1)^{\otimes 2n}+\widehat{b^\ast}(x_i, x_j)\otimes \widehat{g_{\rm h}^\ast},$$ 
where $\widehat{g^\ast_{\rm h}}\in \mathcal{D}^{\otimes}$ is a nonzero EO signature since $\widehat{f_{\rm h}^\ast}
\in \mathcal{D}^{\otimes}$. 
By Lemma~\ref{lem-eo-not-2ndorth},
$\widehat{f^\ast}$ does not satisfy {\sc 2nd-Orth}.
Thus,  $\holant{\neq_2}{\widehat{\mathcal{F}}}$ is \#P-hard by Lemma \ref{second-ortho}.
\end{proof}

\begin{remark}
 Indeed, following from our proof, we can also show that there is no irreducible signature $\widehat{f}$ of arity $2n \geqslant 10$ that satisfies both  {\sc 2nd-Orth} and $\widehat f\in \widehat{\int}\widehat{\mathcal{O}}^{\otimes}$.
\end{remark}

Finally, 
we give the proof of Theorem \ref{main-theorem}.
We restate it here.

\begin{theorem}
Let $\mathcal{F}$ be a set of real-valued signatures.
If $\mathcal{F}$ satisfies the tractability condition {\rm(\ref{main-thr})} in Theorem \ref{thm-main-thr},
then $\Holant(\mathcal{F})$ is polynomial-time computable;
otherwise, $\Holant(\mathcal{F})$ is \#P-hard.
\end{theorem}
\begin{proof}
By Theorem \ref{thm-main-thr}, if  $\mathcal F$ satisfies  condition \rm{(\ref{main-thr})}, then $\Holant(\mathcal F)$ is P-time
computable. 
Suppose that $\mathcal F$ does not satisfy  condition \rm{(\ref{main-thr})}. 
If $\mathcal{F}$ contains a nonzero signature of odd arity,
then by Theorem~\ref{odd-dic}, $\Holant(\mathcal{F})$ is \#P-hard.
We show $\holant{\neq_2}{\widehat{\mathcal{F}}}\equiv_T\Holant(\mathcal{F})$ is \#P-hard when $\mathcal{F}$ is a set of signatures of even arity.
Since ${\mathcal{F}}$ does not satisfy  condition \rm{(\ref{main-thr})}, $\widehat{\mathcal{F}}\not\subseteq \mathscr{T}$.
Since  $\widehat{\mathcal{O}}^{\otimes}\subseteq \mathscr{T}$,
   there is a signature $\widehat{f}\in \widehat{\mathcal{F}}$ of arity $2n$ such that $\widehat{f}\notin \widehat{\mathcal{O}}^{\otimes}$.
   We prove this theorem by induction on $2n$.
   
 When $2n\leqslant 8$, by  Lemmas \ref{lem-2-ary}, \ref{lem-4-ary}, \ref{lem-6-ary}, \ref{lem-8-ary},  $\holant{\neq_2}{\widehat{\mathcal{F}}}$ is \#P-hard. 
 
Inductively, suppose for some $2k\geqslant 8$,
if $2n\leqslant 2k$, then $\holant{\neq_2}{\widehat{\mathcal{F}}}$ is \#P-hard.
 We consider $2n=2k+2\geqslant 10$.
 By Lemma~\ref{lem-10-ary}, $\holant{\neq_2}{\widehat{\mathcal{F}}}$ is \#P-hard, or
 $\holant{\neq}{\widehat g, \widehat{\mathcal{F}}}\leqslant_T \holant{\neq}{ \widehat{\mathcal{F}}}$ for some $\widehat g\notin \widehat{\mathcal{O}}^{\otimes}$ of arity $\leqslant 2k$.
 By the induction hypothesis, $\holant{\neq}{\widehat g, \widehat{\mathcal{F}}}$ is \#P-hard.
 Thus, $\holant{\neq_2}{\widehat{\mathcal{F}}}$ is \#P-hard.
\end{proof}

\section*{Acknowledgement}
\addcontentsline{toc}{section}{\protect\numberline{}{Acknowledgement}}
We would like to thank Professor Mingji Xia for pointing out that 
the strong Bell property can be used to prove Lemma~\ref{lem-2toinfty},
which led to the inspiration to prove Lemma~\ref{lem-0to1}.
Without this inspiration this paper may languish for much more time.
We would also like to thank Professor Zhiguo Fu for many
valuable discussions and verification of parts of the proof.
The proof of the present paper is built on his substantial previous work.
Despite their invaluable 
contributions, they generously declined our invitation for co-authorship. 

\addcontentsline{toc}{section}{\protect\numberline{}{References}}
\bibliography{real-holant}{}
\newpage


\end{document}